\documentclass[12pt,a4paper]{report}

\usepackage[utf8]{inputenc}  
\usepackage{amsmath,amssymb} 
\usepackage{graphicx}        
\usepackage[left=1in,right=1in,top=1in,bottom=1in]{geometry} 
\usepackage{caption}         
\usepackage{subcaption}      
\usepackage{float}           
\usepackage{titlesec}        
\usepackage{abstract}        
\usepackage[hidelinks]{hyperref} 
\usepackage{xcolor}          
\usepackage{dsfont}
\usepackage{bm}
\usepackage{algorithm}
\usepackage{algpseudocode}
\usepackage{amsmath}
\usepackage{mdframed}
\usepackage{amsthm}
\usepackage{mathrsfs}
\usepackage{multicol}
\usepackage{amsfonts}
\usepackage{amssymb}
\usepackage{bbold}
\usepackage[toc,page]{appendix} 
\usepackage{pdfpages}

\usepackage[english]{babel}
\usepackage[backend=bibtex,style= numeric-comp, sorting=none,sortcites=true]{biblatex}
\newtheorem{fact}{Fact}
\newtheorem{claim}{Claim}
\newtheorem{theorem}{Theorem}

\newtheorem{lemma}{Lemma}
\newtheorem{proposition}{Proposition}
\newtheorem{corollary}{Corollary}
\newtheorem{definition}{Definition}

\newtheorem{corollaryprime}{Corollary}[corollary]

\newtheorem{theoremAlph}{Theorem}

\usepackage{enumerate}
\usepackage{enumitem}

\newcommand{\matr}[1]{\bm{#1}}
\newcommand{\prlt}{\xi}
\newcommand{\tprlt}{\tilde \xi}
\newcommand{\FT}{\mathcal{F}}
\newcommand{\BL}{\mathcal{B}}
\newcommand{\BLs}{\mathscr{B}}
\newcommand{\TL}{\mathcal{D}}
\newcommand{\TLs}{\mathscr{D}}
\newcommand{\Ls}{\mathscr{L}}
\newcommand{\Pop}{\mathcal{P}}
\newcommand{\conj}[1]{\overline{#1}}
\newcommand{\Pl}{\stackrel{\text{Pl.}}{=}}

\newcommand{\bw}{W}

\DeclareMathOperator*{\argmax}{arg\,max}

\usepackage{xcolor}

\newcommand{\Hil}{\mathscr{H}}

\usepackage{csquotes}  
\usepackage[T1]{fontenc}
\usepackage{newtxtext} 
\usepackage{geometry}  
\algdef{SE}[SUBALG]{Indent}{EndIndent}{}{\algorithmicend\ }%
\algtext*{Indent}
\algtext*{EndIndent}

\definecolor{darkblue}{rgb}{0.0,0.0,0.3}

\title{\textbf{Prolate Spheroidal Wave Functions and the Accuracy and Dimensionality of Spectral Analysis}}
\author{Timothy Stroschein\thanks{
                  Eidgenössische Technische Hochschule}}
\date{}

\begin{document}

\begin{titlepage}
    \centering
    \vspace*{1cm}
    \includegraphics[width=0.3\textwidth]{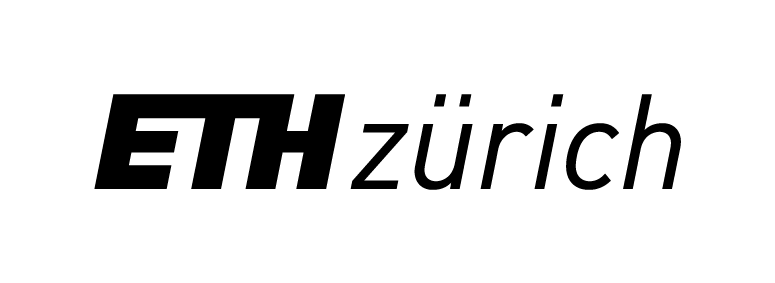} 
    \vspace{1.5cm}
    
    {\Large\textbf{Master Thesis}\par}
    \vspace{1.5cm}
    
    {\huge\bfseries Prolate Spheroidal Wave Functions and the Accuracy and Dimensionality of Spectral Analysis \par}
    \vspace{2cm}
    
    \vspace{1cm}
    

    {\Large\textbf{Timothy Stroschein}\par}
    \vspace{0.5cm} 
    {\small \texttt{timothy.stroschein@phys.chem.ethz.ch}\par}  
    \vspace{0.5cm} 
    {\large
    Eidgenössische Technische Hochschule \par
    }
    \vspace{2cm}
    
    {\large{Supervised by\par}}
    \vspace{0.3cm}
    \large \textbf{Prof. Markus Reiher}\par
    \vspace{2cm}
    
    {A Thesis Submitted to Attain the Degree of \par}

    \vspace{0.5cm}
     {\Large Master of Science in Physics\par}
    \vspace{1cm}
    
    {\large August 14, 2024\par}
\end{titlepage}

    
    
    
    
    
    
    

    

\shipout\null

\thispagestyle{empty} 
\section*{Acknowledgements}

This work is inspired by a project that Prof. Markus Reiher drew my attention to. 
I would like to express my gratitude for his guidance and support.
During the course of this project, some serendipity has led to unexpected research paths and
his encouragement and advice have been invaluable. I consider myself truly fortunate to call him my future PhD supervisor
and look with a great deal of joy at the many projects we will be able to work on over the next few years.\\ 

I would also like to thank Dr. Mihael Erakovic for the many insightful discussions and pieces of advice. My appreciation also extends 
to Prof. Gian Michele Graf, with whom conversations after lectures or upon request have been deeply appreciated. I would also like 
to thank Prof. Renato Renner for graciously agreeing to supervise this thesis formally. \\

My gratitude also goes to my family and friends, who have been an immense source of inspiration and support.
Special thanks to Vinicius Mohr.\\

Finally, I am grateful to my mother Ekaterina Eremenko for her exceptional ability to communicate science and complex mathematics,
and particularly for her enduring support throughout my life.

\newpage 

\shipout\null

\thispagestyle{empty} 


\noindent\hspace{-1cm}
\begin{minipage}{\linewidth}
\begin{displayquote}
\emph{"The unreasonable effectiveness of mathematics in the natural sciences is a gift which we neither understand nor deserve."} \\
--- Eugene Wigner, 1960
\end{displayquote}
\end{minipage}

\vspace{\fill} 

\noindent\hspace{1cm}
\begin{minipage}{\linewidth}
\begin{displayquote}
\emph{"... The other side of the coin is the maladroit ineffectiveness of mathematics at matching precise definitions to our intuitive preconceptions. (The truth might be that our intuition is vague and inconsistent, and nothing precise could possibly match it.)"} \\
--- Paul R. Halmos, 1985
\end{displayquote}
\end{minipage}



\newpage
\shipout\null
\thispagestyle{empty}

\chapter*{Preface}

\section*{On the Need for Approximation Theory}

While the underlying physical laws necessary to understand chemistry and material sciences are widely regarded to be long known,
the complexity of the resulting equations quickly becomes intractable for all but the simplest systems.
Thus, the curse of dimensionality is the greatest limitation of our current understanding of the physical world.\\

On the other hand, the need for an ever-increasing computational dimensionality to understand larger systems does not match our intuition. 
Across various fields, models operating at fewer degrees of freedom are often found to be of great predictive power. 
An emergence phenomena is observed: The laws of the macroscopic do not depend on the details of the microscopic.\\

While in the context of systems that exhibit a certain symmetry, this emergence phenomenon is often well leveraged through approximate models, it is not well understood. The approximate treatment is a well-informed guess but is rarely systematically derived 
from the fundamental equations of the microscopic layer underneath.
In more complex scenarios that are typical for chemistry and biology, it is often not feasible 
to guess, based on intuition, how the degrees of freedom can be efficiently compressed while still obtaining a good approximation.\\

Thus, it seems that a successful description of complex systems that are natural to the applied sciences,
require a mathematical 
language, different from the one that led to the effectiveness of mathematics in physics that Wigner referred to in his famous essay. 
We will broadly refer to this language as approximation theory.\\

The theoretical foundation of physics is based on mathematically sterile scenarios. Theorists are allowed to strip away the complex,
messy interactions of the real world in order to focus on understanding fundamental interactions. If an experimenter can sufficiently replicate these 
idealized scenarios, and the deviation between the experimental measurements and the theoretical predictions is small, the theoretical description is considered 
successful. 
Here, we highlight the asymmetric relation between theory and experiment that has been, to a large extent, dominant in physics.
The theory is permitted to idealize and derive exact solutions, whereas it is accepted that experimental quantification
will always be subject to finite precision.\\

For the success of theoretical and mathematical principles in the more complex and messy systems of chemistry and biology,
we can no longer afford such asymmetry. The spirit of the kind of approximation theory we dream of leverages 
the fact that our experience of the real world will always be of finite precision.
By embracing this truth, we recognize that exact solutions to complex equations are not only unnecessary but also not worth the computational effort.
Instead, we want to leverage the permission of theoretical predictions of finite precision to derive 
new mathematical equations in a smaller parameter space. \\

The broader mathematical understanding that remains elusive but necessary for the natural sciences is a
fundamental understanding of the optimal relationship between dimensionality and precision of approximation.
Surely, much of the approximation theory that we have been sketching here has various instances at different levels and has already been applied for a long time.
And they all form a part of the mathematical lexicon we are referring to. \\

Nevertheless, there exists a specific branch of approximation theory that has reached a rare level of mathematical clarity and depth
in the understanding of dimensionality and precision of approximation. Despite its potential, 
much of the natural sciences seem to have remained unaware of its existence.\\

H. Landau, H. Pollak and D. Slepian have already asked in the early 1960s about the optimal 
concentration a band-limited function might have in a given time interval. This simple question has led to much 
serendipity and eventually allowed them to derive an approximation theorem of particular significance. 
The space of band-limited functions that are also concentrated in a given time interval is to high 
precision approximated in a $2\bw T / \pi$ dimensional space \cite{ProI, ProIII, OnBandwidth}. \\
  
The significance of this result can be appreciated when compared to another theorem,
which is widely regarded as one the most impactful results to our information age. 
After Shannon's success in the discrete communication model, he sought to extend his 
theory to the physical scenario of communication across band-limited channels. Here, he offered a theorem on the optimal transmission rate of
a band-limited channel.
However, his proof of the theorem relied on the assumption that the space of signals used for communication is of 
dimension $2\bw T / \pi$ \cite{NoiseCommunication}. \\

It was only after Landau, Pollak, and Slepian introduced their approximation theorem that it became clear: 
Band-limited signals concentrated within a time interval are effectively approximated in a $2\bw T / \pi$ dimensional space.
We will refer to the theory that their seminal series of papers has initiated as Prolate Fourier Theory \cite{ProI, ProII,ProIII, ProIV, ProV}.\\

In Prolate Fourier Theory, we recognize an essential cornerstone of approximation theory, which seems invaluable for further theoretical progress in the natural sciences. 

\newpage 

\section*{Thesis Overview}

The main result of this thesis is an efficient protocol to determine the frequencies of a signal $C(t)= \sum_k |a_k|^2 e^{i \omega_k t}$, 
which is given for a finite time to a high degree of precision.
More specifically, it is a theorem that gives a bound on the error of the frequency estimates that are obtained through the protocol.
Equally noteworthy, however, is the theoretical framework that underlies this theorem, as it is general and allows application to a wide range of problems.\\

In Chapter \ref{chap:prolate_theory_intro}, we provide an introductory review of Prolate Fourier Theory and establish subsequently used notation.\\

Chapter \ref{chap:prlt_bound} presents our main result within Prolate Fourier Theory. This chapter establishes a bound on the maximum value a prolate spheroidal 
wave function may attain outside of its concentration region. It also proves that the concentration 
property of prolates translates to a high degree to their derivatives.
The underlying proof relies on a commutation relation, which is often regarded as a coincidence in the field.
This chapter also provides a novel geometrical insight into the commutation relation.  \\

Chapter \ref{chap:Sampling} builds on the new concentration identities established in Chapter \ref{chap:prlt_bound} 
to derive more precise error bounds for truncated sampling formulas. Chapter \ref{chap:Sampling} is non-expository at this 
stage. \\

Chapter \ref{chap:DimRedSepAna} develops an approximation theory for
spectral analysis through low-dimensional subspaces. This approach already has widespread applications
in quantum chemistry. In this chapter, we introduce a formalism that enables an uncertainty quantification for the method and derive spectral inequalities.\\

Chapter \ref{chap:FilterDiag} formulates Filter Diagonalization as a numerical scheme that can access the approximation theory of Chapter \ref{chap:DimRedSepAna}.
In this context, prolates are identified as the optimal choice for filter functions. The bounds of Chapter \ref{chap:prlt_bound} are applied 
to provide the essential precision estimate.

\tableofcontents

\chapter{An Introduction to Prolate Fourier Theory}
\label{chap:prolate_theory_intro}

In this chapter we offer a pedagogical introduction to Prolate Fourier Theory, that is based on the simple 
problem of finding the band limited function that is maximally concentrated in a time interval.
This approach is 
inspired by the second paper of the prolate series by Slepian, Landau and Pollak \cite{ProII}, but we will
also draw from the first paper \cite{ProI}. \\

An integral operator
forms the heart of Prolate Fourier Theory. The spectrum of this prolate spheroidal integral operator 
is of information theoretical significance, as it reveals the essential dimensionality of band limited signals that are also time concentrated.
Different from the original papers, we will strongly emphasize a commutation relation of the integral operator 
with a second order differential operator. \\

\section{Preliminaries}

\subsection{Notation} 

We consider the Hilbertspace $\Ls^2_\infty$ of square integrable functions on the real line
with the the usual inner product
\begin{align*}
    \langle f, g \rangle_\infty = \int_{-\infty}^{\infty} \conj{f(t)} g(t) dt.
\end{align*}
Analogously we denote the Hilbertspace of square integrable functions in $[-A,A]$ as $\Ls^2_A$ and write 
\begin{align*}
    \|f\|_A^2 = \int_{-A}^{A} \|f(x)\|^2 dx.
\end{align*}
Thus in this notation $\|f\|_\infty^2$ denotes the $\ell_2$-norm over 
the real line. Sometimes we will refer to the $\ell_2$-norm of a functions as its energy.  Our convention of the Fourier transform is
$$\FT[f](\omega) = \int_{-\infty}^{\infty} f(t) e^{-i \omega t} dt, \quad \quad \FT^{-1}[f](t) = \frac{1}{2 \pi} \int_{-\infty}^{\infty} f(\omega) e^{i \omega t} d \omega.$$ 
By Plancherel's theorem the Fourier Transform is an isometry on $\Ls^2_\infty$ and we have 
\begin{align*}
    \int_{-\infty}^{\infty} \conj{f(t)} g(t) dt \stackrel{\text{Pl.}}{=} \frac{1}{2\pi} \int_{-\infty}^{\infty} \conj{\FT[f](\omega)} \FT[g](\omega) d\omega.
\end{align*}
Occasionally we will denote an application of Plancherel's theorem with a $\stackrel{\text{Pl.}}{=}$ sign.
The convolution theorem writes as 
\begin{align*}
    \FT[ g f](\omega) = \frac{1}{2\pi} \FT[g] \ast \FT[f](\omega) \\
    \FT^{-1}[ g f](\omega) =  \FT^{-1}[g] \ast \FT^{-1}[f](\omega). 
\end{align*}
If $f \in \Ls^2_\infty$ is continuously differentiable we have $\FT[f'](t) = i \omega \FT[f](\omega)$. 
It is also worth to recall the standard result of Fourier analysis that 
$\left\{\frac{1}{\sqrt{2 A}} e^{\pi i n t / A}: n \in \mathbb{Z}\right\}$ 
is an orthonormal basis of $\Ls^2_A$. By Parseval's identity 
we have for every orthonormal basis $\left(\varphi_j\right)_{j=1}^{\infty}$ on an Hilbertspace $\mathscr{H}$
$$ \|f\|^2=\sum^{\infty}\left|\left(\varphi_j,  f \right)\right|^2, \quad \forall \varphi \in \mathscr{H}.$$

We will also use the notion of finite Fourier transforms defined as 
$$\FT_T[f](\omega) = \int_{-T}^{T} f(t) e^{-i \omega t} dt, \quad \quad \FT^{-1}_\bw[f](t) = \frac{1}{2 \pi} \int_{-\bw}^{\bw} f(\omega) e^{i \omega t} d \omega.$$ 

\subsection{Bandlimited Functions}

\subsubsection{Simple known Identities}
\label{sec:bandlimited_known}

We denote with $\chi_\bw(x)$ the characteristic function of the interval $[-\bw, \bw]$ and 
write for the sampling function $\rho_\bw(x) = \frac{\sin{\bw x}}{\pi x}$ . Note that the sampling function is even 
$\rho_\bw(x) = \rho_\bw(-x)$ and $|\rho_\bw(x)| \leq \frac{1}{\pi|x|}$ for $|x|$ large enough.
Thus $\rho_\bw(x)$ is square integrable and $\chi_\bw(y)$ is the Fourier transform of $\rho_\bw(x)$,
\begin{align}
    \rho_\bw(x) = \frac{1}{2\pi} \int_{-\infty}^{\infty} \chi_\bw(y) e^{i y x} d y     \label{eq:rho_int_rep}\\  
    \chi_{\bw}(x)= \int_{-\infty}^{\infty} e^{-i y x} \rho_{\bw}(y) d y. \notag
\end{align}
The first identity follows from computation, the second 
from the fact that $\FT$ and $\FT^{-1}$ are inverse operators.
From the integral expression we see $\sup_{x\in \mathbb{R}} \rho_\bw(x) = \rho_\bw(0)= \frac{\bw}{\pi}$.
Through Plancherel's theorem we infer 
$$\|\rho_\bw\|_\infty^2 = \frac{1}{2\pi } \|\chi_\bw\|_\infty^2 = \frac{\bw}{\pi}.$$

We denote with $\BLs_\bw$ the class of square integrable functions whose 
Fourier transform is supported on $[-\bw, \bw]$ and call it the space 
of $\bw$-band limited functions\footnote{In the literature 
this space is sometimes also referred to as 
the Paley Wiener Space. If the context is clear we will omit the $\bw$ and simply say 
a function is band limited.}. A prototypical member of $\BLs_\bw$ is given by $\rho_\bw(x)$ as can be seen 
form (\ref{eq:rho_int_rep}).
We denote with $\BL_\bw$ the operation of band limiting a function $f$,
such that its Fourier Transform is supported on the interval $[-\bw, \bw]$ (an ideal band pass filter)
\begin{align}
    \BL_\bw[f](t) & : = \frac{1}{2\pi} \int_{-\bw}^{\bw} e^{i \omega t}  \FT[f](\omega)   d\omega = \frac{1}{2\pi} \int_{-\infty}^{\infty}  \chi_\bw (\omega) e^{i \omega t}  \FT[f](\omega) \notag \\
    & \Pl \int_{-\infty}^{\infty} \conj{\FT^{-1}[ \chi_\bw(\omega) e^{- i\omega t}](s)} (\FT^{-1}\FT)[f](s) ds \notag\\
                        & = \int_{-\infty}^{\infty} \rho_\bw(s-t) f(s) ds. \label{eq:bandlimiting_int}
\end{align}
The first equality is the definition and in the last 
we used $\FT^{-1}[ \chi_\bw(\omega) e^{- i\omega t}](s)= \rho_\bw(s-t)$ 
as Fourier transform with respect to $\omega$.
Similarly we have 
\begin{align}
    \int_{-\infty}^{\infty} \rho_\bw(s-x) \rho_\bw(s-y) & \Pl  \frac{1}{2\pi} \int_{-\infty}^{\infty} \chi_\bw(t)^2 e^{it(x-y)} dt  \notag \\
    & = \rho_\bw(x-y).   \label{eq:double_rho_int}
\end{align}
From this identity and (\ref{eq:bandlimiting_int}) wa can deduce that $\BL_\bw^2 = \BL_\bw$ as square integrability implies 
that the order of the integrals can be exchanged.
From Plancherel's theorem we also have for all $f \in \Ls^2_\infty$ and $g \in \BLs_\bw$
\begin{align*}
    \langle f - \BL_\bw[f], g \rangle_\infty = 0 .
\end{align*}
Thus we have established the well known fact that 
$\BL_\bw$ is the orthogonal projection operator onto the space of band limited functions $\BLs_\bw$. 
It is worth to refresh that orthogonal projection operators are self adjoint and idempotent.\\

From the integral representation of a band limited function, 
\begin{align}
    f(t) = \int_{-\bw}^{\bw} e^{i \omega t} \FT[f](\omega) d\omega  \label{eq:bandlimited}
\end{align}
we have that every band limited function
can be extended to an entire function on the complex plane. formally this can be proven through a simple application of Morera's theorem.\footnote{
    Morera's theorem states that if a function $f$ is continuous on a domain $D$ and 
    $\oint_{\gamma} f(z) dz = 0$ for every closed piecewise $C^1$ curve $\gamma$ in $D$, then $f$ is holomorphic on $D$.
    If $D$ is simply connected Cauchy's integral theorem guarantees that the converse is also true.
    The order of the integrals can be exchanged and $\BL_\bw[f](t)$ inherits from $e^{i \omega t}$ the property of being an entire function.
}\\

\subsubsection*{Time Limited Functions}

Similarly we can consider the class of $T$-time limited functions $\TLs_T \subset \Ls_\infty^2$ that 
are supported on $[-T,T]$. 
For the corresponding orthogonal
projection operator $\TL_T$ we have 
\begin{align*}
    \TL_T f(t) = f(t) \chi_T(t).
\end{align*}
Naturally $\TLs_T$ is closely related to $\Ls^2_T$ with the sole difference 
that members of $\TLs_T$ are defined on the whole real line.
In fact, there is indeed a trivial isometry between $\Ls^2_T$ and $\TLs_T$.
For a function $f$ that is formally only defined in $\Ls^2_T$ we denote with $\TL_T f (t) \in \Ls_\infty^2$ the function that coincides with $f$ on $[-T,T]$ and is zero elsewhere.
We will prefer to work in $\TLs_T$ as here functions are defined on the whole real line.\footnote{Through the 
trivial isometry between $\Ls^2_T$ and $\TLs_T$ known results in $\Ls^2_T$ as for example the completeness relation 
of the exponential basis, similarly apply in $\TLs_T$. Statements we will derive in $\TLs_T$ also apply in $\Ls^2_T$.}\\

\subsubsection{Additional simple Remarks}

Further, we have that every band limited function is bounded on the real line, as can be seen 
from (\ref{eq:bandlimited}). By applying the Cauchy-Schwarz inequality and noting that 
$|e^{i \omega t}| =1$ for $t \in \mathbb{R}$,
\begin{align*}
    |f(t)|^2 & = \left| \frac{1}{2 \pi} \int_{-\bw}^{\bw} e^{i \omega t} \FT[f](\omega) d\omega \right|^2  \leq \frac{1}{2\pi} \int_{-\bw}^{\bw} |e^{i \omega t}|^2 d\omega \int_{-\bw}^{\bw}  |\FT[f](\omega)|^2  d\omega \\
    & =   \frac{2 \bw }{2\pi} \| \FT[f]\|_{\infty}^2  \Pl 2 \bw \|f\|_{\infty}^2.
\end{align*}
Here we also used  $\| \FT[f]\|_{\bw}^2 = \| \FT[f]\|_{\infty}^2$ for $f \in \BLs_\bw$.
From this we can further deduce
that every band limited function is
$L_p$ integrable over the real line for all $p\geq 2$,
\begin{equation*}
    \int_{-\infty}^{\infty} |f(t)|^p dt \leq C^{p-2} \int_{-\infty}^{\infty} |f(t)|^2 dt < \infty.
\end{equation*}
As band limited functions are entire we have that all their derivatives are also band limited
\begin{align*}
    \FT[f^{(m)}](t) = (i \omega)^m \FT[f](\omega) \quad \Rightarrow \quad f^{(m)} \in \BL_\bw,
\end{align*}
where we denoted $f^{(m)}(t) = \frac{d^m }{dt^m} f(t)$. Thus $\BL_\bw$ is closed 
with respect to differentiation. \\

$\BL_\bw$ is also closed with respect to convolution, 
i.e. if $f,g \in \BL_\bw$ then $f \ast g \in \BL_\bw$ as can be seen 
from the convolution theorem. Similarly the product of two band limited function results 
in a bandlimitted function with twice the bandwidth,
\begin{align*}
    \FT[fg](\omega) & = \frac{1}{2\pi} \FT[f] \ast \FT[g](\omega) =  \frac{1}{2\pi} \int_{-\infty}^{\infty} \FT[f](\omega - s) \FT[g](s) ds
\end{align*}
The last equality implies that $\FT[fg](\omega) = 0$ for $|\omega| > 2\bw$ and thus $fg \in \BL_{2\bw}$.\\

\newpage 

\section{From Optimal Fourier Space Concentration to  Prolate Fourier Theory}
\label{sec:prolate_theory_intro}

\label{sec:Fourier_concentration_time_limited}

We are interested in the problem of finding 
the function $f$ supported in the interval $[-T, T]$ whose Fourier transform 
is optimally concentrated in the interval $[-\bw, \bw]$.\\

More formally we are interested in the members of $\TLs_T$ that maximize the ratio of
energy in the $\bw$-band to total energy $\| f \|_T^2$. 
We denote $F(\omega) = \FT[f](\omega)$ and are lead to following optimization problem,
\begin{align}
    \gamma_0' := \max_{f \in \TLs_T} \frac{1}{2\pi}\frac{\int_{-\bw}^{\bw} |F(\omega)|^2 d\omega}{\int_{-T}^{T} |f(t)|^2dt}.  \label{eq:optimization_Problem_DT}
\end{align}
We can rewrite this expression such that the nature of our question as a concentration problem becomes more apparent.
Considering the numerator we have with Plancherel,
\begin{align*}
    \int_{-\infty}^{\infty} |f(t)|^2dt = \frac{1}{2\pi} \int_{-\infty}^{\infty} |F(\omega)|^2 d\omega = \frac{1}{2\pi}  \| F\|_\infty^2. 
\end{align*}
Plugging this into (\ref{eq:optimization_Problem_DT}) we obtain 
\begin{align}
    \gamma_0' = \max_{f \in \TLs_T} \frac{\|F\|_\bw^2}{\|F\|_\infty^2}.  \label{eq:optimization_Problem_DT_Pl}
\end{align}
We already see that $\gamma_0' \leq 1$. Equality can not be attained, as the Fourier transform of a time-limited function 
is supported across the whole real line and we have $\|F\|_\bw^2 < \|F\|_\infty^2$ \footnote{
    More formally: For $f \in \TLs_T$ we have that its Fourier transform $F$ is entire. By the identity theorem for analytic functions, $F$ cannot be equal to zero on an open subset such as $(T,\infty)$ unless 
    $F$ is identically zero. Thus $\|F\|_\bw^2 < \|F\|_\infty^2$.}. Therefore $\gamma_0' < 1$.\\
 
Focusing on the numerator we can write 
\begin{align}
    \frac{1}{2\pi} \int_{-\bw}^{\bw} |F(\omega)|^2 d\omega & = \frac{1}{2\pi} \int_{-\bw}^{\bw} \int_{-T}^{T} \int_{-T}^{T} \conj{f(s)} f(t) e^{-i \omega (s-t)} ds dt d\omega \notag \\
    & = \int_{-T}^{T} \int_{-T}^{T} \conj{f(s)} f(t) \rho_\bw(t-s) ds dt = \langle f, \BL_\bw \TL_T f \rangle_\infty.  \label{eq:optimization_Problem_DT_2}
\end{align}
In the first line we used that $f$ is supported on $[-T,T]$ and in the second line that $\BL_\bw \TL_T$ has the integral representation
\begin{align}
    \BL_\bw \TL_T f(t) = \int_{-T}^{T} \rho_\bw(t-s) f(s) ds.  \label{eq:BLTL_operator}
\end{align}
Thus our search for time-limited functions whose Fourier transform is optimally concentrated in the $\bw$-band,
leads to the integral operator $\BL_\bw \TL_T$ acting on $\TLs_T$.\\

And $\BL_\bw \TL_T$ is indeed a nice operator, in the sense that it is compact and allows access to the spectral theorem. 
Moreover, $\BL_\bw \TL_T$ is a Hilbert-Schmidt operator in the more general space of square integrable functions $\Ls^2_\infty$.\footnote{Every Hilbert-Schmidt operator is compact
and has a spectral decomposition, if it is self adjoint. For a mathematically detailed exposition on integral operators see \cite{courant_methoden_1993} chapter 3 section 5 or \cite{Reed1980MethodsOM}. 
Within $\Ls_\infty^2$  $\BL_\bw \TL_T$ is an integral operator with integral kernel $K(t,s)=\rho_\bw(t-s) \chi_T(s)$}
To compute the Hilbert-Schmidt norm we recall
$ \rho_{\bw}(t-s) = \mathcal{F}^{-1}\left[\chi_{\bw}(\omega) e^{-i \omega t}\right](s)$  from section \ref{sec:bandlimited_known} and have
\begin{align}
    \|\BL_\bw \TL_T\|_{\text{HS},\Ls^2_\infty}^2 & = \int_{-\infty}^{\infty} \int_{-T}^{T} |\rho_\bw(t-s)|^2 dt ds = \int_{-T}^{T}  \int_{-\infty}^{\infty}  |\rho_\bw(t-s)|^2 dt ds \notag \\
    & \Pl  \frac{1}{2\pi }\int_{-T}^{T}  \int_{-\infty}^{\infty}  \left|\chi_{\bw}(\omega) e^{-i \omega t}\right|^2 d\omega ds = \frac{1}{2\pi }\int_{-T}^{T} 2 \bw ds  \notag\\
   & = \frac{2 \bw T}{\pi} < \infty.  \label{eq:BLTL_HSopNorm}
\end{align}
Thus $\BL_\bw \TL_T$ is indeed Hilbert-Schmidt in $\Ls^2_\infty$. 
\footnote{ $\BL_\bw \TL_T$ was in the literature mostly only considered as a Hilbert-Schmidt operator acting in $\Ls^2_T$.}\\

Returning to the subspace of time-limited functions $\TLs_T$ we note that here $\BL_\bw \TL_T$ acts equivalently to $ \TL_T \BL_\bw \TL_T$. In this form we immediately see that $\BL_\bw \TL_T$ is self adjoint in $\TLs_T$
and that the spectral theorem applies. 
Going a bit further we can deduce from (\ref{eq:optimization_Problem_DT_2}) that $\BL_\bw \TL_T$ is positive definite in $\TLs_T$. 
To formally prove this note that the Fourier transform of a time-limited function is entire.
From the identity theorem for analytic functions, we have that an entire function that equals zero on an open subset must be identically zero.
In particular $\|F\|_\bw^2 = 0$ implies $F=0$. Since the Fourier transform is an isometry this ultimately implies  
$f=0$. Therefore (\ref{eq:optimization_Problem_DT_2}) has indeed
\begin{align*}
    \langle \TL_T f, \BL_\bw \TL_T f \rangle_\infty  & =  \langle  f, \TL_T \BL_\bw \TL_T f \rangle_\infty = \|F\|_\bw^2 \geq 0
\end{align*} 
and equality is attained if and only if $\TL_T f=0$. Thus $\BL_\bw \TL_T$ is indeed positive definite in $\TLs_T$.\\

From this, we can also infer, that all eigenvalues of $\BL_\bw \TL_T$ acting in $\Ls^2_\infty$ are real and non-negative. Let $\varphi \in \Ls^2_\infty$ be an eigenfunction of $\BL_\bw \TL_T$ with eigenvalue $\lambda \in \mathbb{C}$
\begin{align*}
    \BL_\bw \TL_T \varphi = \lambda \varphi.
\end{align*}
Acting from the left with the time-limiting operator we obtain 
\begin{align*}
    \TL_T \BL_\bw \TL_T \varphi = \lambda \TL_T \varphi.
\end{align*}
If $\TL_T \varphi \neq 0$ we have that $\TL_T \varphi$ is an eigenfunction of $\TL_T \BL_\bw \TL_T$ and thus $\lambda > 0$.
On the other hand if $\TL_T \varphi = 0$ it is clear that $\lambda = 0$. Thus we have that all eigenvalues of $\BL_\bw \TL_T$ are real and non-negative.
We call the eigenfunctions of $\BL_\bw \TL_T$ that do not map to zero prolates (also known as Prolate Spheroidal Wave Functions PSWF) and denote them as $\prlt_n$ and 
their corresponding eigenvalues as $\gamma_n$.\\

We have already more than enough to answer our original question 
posed in (\ref{eq:optimization_Problem_DT}). By the spectral theorem, 
the projected prolates $\TL_T \prlt_n$ form a complete orthogonal system of $\TLs_T$ and we have established that all non-trivial eigenvalues of $\BL_\bw \TL_T$ are real and positive.
Let $\gamma_0$ be the leading eigenvalue with eigenfunction $\prlt_0$. We can rewrite equation (\ref{eq:optimization_Problem_DT}) a final time 
using (\ref*{eq:optimization_Problem_DT_2}) and $f = \TL_T f$ for $f \in \TLs_T$
\begin{align*}
    \gamma_0' & = \max_{f \in \TLs_T} \frac{1}{2\pi}\frac{\int_{-\bw}^{\bw} |F(\omega)|^2 d\omega}{\int_{-T}^{T} |f(t)|^2dt} = \max_{f \in \TLs_T}  \frac{\langle \TL_T f, \BL_\Omega \TL_T f \rangle }{ \| f\|_T^2} \\
    & =  \max_{f \in \TLs_T}  \frac{\langle f, \TL_T \BL_\Omega \TL_T f \rangle }{ \| \TL_T f \|_\infty^2} = \gamma_0.
\end{align*}
In the last step, we used the variational principle for self-adjoint operators. Thus the optimal energy concentration of a time-limited function in the $\bw$-band is given by the leading eigenvalue $\gamma_0$ and is attained by $\TL_T \prlt_0$.
From our considerations given in (\ref{eq:optimization_Problem_DT_Pl}) we also have the bound $\gamma_0 < 1$ and with descending ordering of the eigenvalues
we have
\begin{align*}
    1 > \gamma_0 \geq \gamma_1 \geq \gamma_2 \geq \ldots > 0.
\end{align*}

The particular question for optimal Fourier concentration we posed at the beginning of this section is just one of many
Fourier concentration problems one might be interested in for applications in larger information processing routines. 
The natural arising of $\BL_\bw \TL_T$ in this context gives 
us an indication of the importance of prolates for achieving optimal bounds. We are motivated to further investigate their properties.

\subsection{Double Completeness and Further Properties}

We have already established that Prolates form an orthogonal basis of $\TLs_T$.
But as eigenfunctions of a product of two projection operators, they also form 
a complete orthogonal system in $\BLs_\bw$ in correspondence to $\BL_\bw$. For $g,f \in \BLs_\bw$ we have 
$ f = \BL_\bw f$. Recalling that orthogonal projection operators are self-adjoint, we have 
\begin{align*}
    \langle g , \BL_\bw \TL_T f \rangle & =  \langle g , \BL_\bw \TL_T \BL_\bw  f \rangle
    = \langle \BL_\bw \TL_T \BL_\bw g  , f  \rangle = \langle \BL_\bw \TL_T g   , f  \rangle,
\end{align*}
and $\BL_\bw \TL_T$ is indeed self adjoint in $\BLs_\bw$. But even more so, $\BL_\bw \TL_T$ is in fact positive definite in $\BLs_\bw$,
\begin{align*}
    \langle f, \BL_\bw \TL_T f \rangle & = \langle \BL_\bw f, \TL_T f \rangle = \langle f,  \TL_T f \rangle \\
    & = \int_{-T}^{T} \conj{f(t)} f(t) dt =\|f\|_T^2.
\end{align*}
Similar to previous arguments we have $\|f\|_T^2 = 0$ implies $f=0$ as $f\in \BLs_\bw$ is entire.\footnote{In the original paper the positive 
definite spectrum was formally proven through a reference to a theorem of Bochner, see page 58 in \cite{ProI}. 
By Bochners Theorem, every function $\rho$ whose Fourier transform is a probability measure, is totally positive. In particular $\rho_\bw$ is totally 
positive and induces an positive definite convolution operator \cite{Bochner1932}.}\\
Therefore the eigenfunctions of $\BL_\bw \TL_T$ in $\BLs_\bw$ 
are exactly the non zero eigenfunctions $\prlt_n$ of $\BL_\bw \TL_T$ in $\Ls^2_\infty$ with eigenvalue $\gamma_n$.
By the spectral theorem, prolates form a complete orthogonal system in $\BLs_\bw$.\\

It is worth highlighting the peculiar double orthogonal basis relation we have just derived. 
$\TL_T \prlt_n$ is a orthonormal basis of $\TLs_T$ and $\prlt_n$ is a orthonormal basis of $\BLs_\bw$.
In other words prolates form an orthogonal system on the interval $[-T, T]$ and the whole real line. The 
double completeness relation turns out to be very practical for many derivations.\\

Of special interest is also how the eigenvalues characterize the energy concentration of prolates
\begin{align}
    \|\prlt_n\|_T^2 = \gamma_n \|\prlt_n\|_\infty^2.  \label{eq:prolate_interval_energy}
\end{align} 
To establish this relation we recall $\int_{-\infty}^{\infty} \rho_\bw(t-x) \rho_\bw(t-y) = \rho_\bw(x-y)$ (see eq.\ref{eq:double_rho_int}) and
compute
\begin{align*}
    \| \prlt_n\|_\infty^2 & = \int_{-\infty}^\infty \prlt_n(t)^2 dt = \frac{1}{\gamma_n^2} \int_{-\infty}^\infty  \int_{-T}^{T}  \int_{-T}^{T} \rho_\bw(t-x) \rho_\bw(t-y) \prlt_n(x) \prlt_n(y) dx dy dt \\
   & = \frac{1}{\gamma_n^2} \int_{-T}^{T} \int_{-T}^{T}   \rho_\bw(x-y) \prlt_n(x) \prlt_n(y) dx dy = \frac{1}{\gamma_n}  \int_{-T}^{T} \prlt_n(x)^2 dx = \frac{1}{\gamma_n}  \|\prlt_n\|_T^2.
\end{align*}
We see that eigenvalues of $\BL_\bw \TL_T$ describe the energy concentration of prolates in the interval $[-T, T]$. \\

Naturally, prolates are band-limited. In Section \ref{sec:bandlimited_known} we have seen how band-limited functions 
can be extended to entire functions from their integral representation. As the extension is unique, prolates can also be extended from their 
integral eigenfunction representation and
\begin{align}
    \gamma_n \prlt_n(t) = \int_{-T}^{T} \rho_\bw(t-s) \prlt_n(s) ds  \label{eq:prlt_eigfunc_extension}
\end{align}
holds in fact for all $t \in \mathbb{C}$. Since prolates are eigenfunctions of an operator with real eigenvalues,
they can be chosen to be real-valued on the real line. As the integral kernel $\rho_\bw$ is even, prolates can be 
constructed to be either even $\prlt_e(t) = \prlt(t) + \prlt(-t)$ or odd $\prlt_o(t) = \prlt(t) - \prlt(-t)$. Indeed we have 
\begin{align*}
     \BL_\bw \TL_T \prlt_e(t) & = \int_{-T}^{T} \rho_\bw(t-s) (\prlt(s) + \prlt(-s)) ds = \gamma \prlt(t) + \int_{-T}^{T} \rho_\bw(t-s) \prlt(-s) ds \\
     &  = \gamma \prlt(t) + \int_{-T}^{T} \rho_\bw(-t-s) \prlt(s) ds = \gamma \prlt(t) + \gamma \prlt(-t) = \gamma \prlt_e(t).
\end{align*}
and the calculation for $\prlt_o$ is analogous.\\

We have already derived a fair share of interesting properties of prolates.
We will make a pause here and take a look at results that can be summarized from the 
original paper on their discovery \cite{ProI}.

\begin{theoremAlph}
    The eigenvalues of $\BL_\bw \TL_T$ are non degenerate and permit ordering 
    \begin{align*}
        1> \gamma_0 > \gamma_1 > \gamma_2 > \ldots > 0.
    \end{align*}
    The corresponding eigenfunctions $\prlt_n$ are real-valued and entire, such that 
    \begin{align*}
        \int_{-T}^{T} \rho_\bw(t-x) \prlt(s) ds = \gamma_n \prlt_n(t) 
    \end{align*}
    holds for all $t \in \mathbb{C}$. They are also eigenfunctions to a second integral operator
    \begin{align}
        \int_{-T}^{T} e^{i \frac{ \bw \tau t}{T} } \prlt_n(t) dt  = \mu_n \sqrt{\frac{T}{\bw}} \prlt_n(\tau),  \label{eq:ori_exp_eig_int}
    \end{align}
    where $\frac{|\mu_n|^2}{2\pi} = \gamma_n$. The functions $\prlt_n$ are even or odd with $n$. They satisfy
    the energy concentration relation
    \begin{align*}
        \|\prlt_n\|_T^2 = \gamma_n \|\prlt_n\|_\infty^2.  
    \end{align*}
    In particular, they form a complete orthogonal system in $\Ls^2_T$ and $\BLs_\bw$, and can be 
    chosen to be normalized such that
    \begin{align*}
        \int_{-\infty}^{\infty} \prlt_n(t) \prlt_m(t) dt = \delta_{nm} \qquad \text{and}  \qquad \int_{-T}^{T} \prlt_n(t) \prlt_m(t) dt = \gamma_n \delta_{nm}.
    \end{align*}
    Further, we have unique solutions to the following optimization problems \footnote{Understood up to multiplicative constants. To 
    be precise, the optimization problem (\ref{eq:original_optimization1}) was actually solved in \cite{ProII}.}
    \begin{alignat}{2}
        \gamma_0  & = \max_{ f\in \BLs_\bw} \frac{\| f \|_T^2}{ \|f \|_{\infty}^2}  \quad  &&\text{ attained by } \prlt_0.  \label{eq:original_optimization1}\\
        \gamma_0  &  =  \max_{f \in \Ls^2_\infty} \frac{ \| \BL_\bw \TL_T f \|_\infty^2}{\| f\|_\infty^2} \quad  &&\text{ attained by } \TL_T \prlt_0.\\
        \gamma_0^2 & = \max_{f \in \BLs_\bw} \frac{ \| \BL_\bw \TL_T f \|_\infty^2}{\| f\|_\infty^2} \quad  &&\text{ attained by } \prlt_0. \label{eq:original_optimization3}
    \end{alignat}
    \label{thm:prolate_original}
\end{theoremAlph}

We see that most of the statements in Theorem \ref{thm:prolate_original} are already familiar to us by our previous analysis, 
based on the integral operator $\BL_\bw \TL_T$. The optimization problems (\ref{eq:original_optimization1})-(\ref{eq:original_optimization3})
can be solved similarly to how we solved (\ref{eq:optimization_Problem_DT}). However, the non-degeneracy of the eigenvalues
and the additional integral relation (\ref{eq:ori_exp_eig_int}) have not been established yet.\\

And they were indeed not derived solely from the integral operator $\BL_\bw \TL_T$, but rather through what Slepian much later referred to as a 'lucky accident' \cite{SlepianComment}. 
Originally the completeness relation of prolates in
$\BLs_\bw$ and $\TLs_T$ were also proven through this 'lucky accident.' Here, Hilbert-Schmidt operator theory has allowed us to derive the completeness in a slightly more direct manner.\\

In the next section, we shall explore the lucky accident and collect the remaining puzzle pieces.

\section{The "Lucky Accident": The Prolate Spheroidal Wave Equation}
\label{sec:PSWEq}

From the late 19th century until the mid 50s Eigenfunctions of the Laplace operator in most general coordinates were of much 
interest to the mathematical physicist.\footnote{ Efficient numerical methods 
to generate eigenfunctions and to understand the spectrum of the Laplacian are also today a very active research area.
} Especially elliptical coordinates are due to their generality of
great relevance and find naturally many applications in molecular physics or scattering problems.
A surprise was their significance to the communication and computation engineer. \\

The eigenfunctions of the Laplace operator $\nabla \Phi = -k^2 \Phi$ when solved by separation of variables 
in spheroidal coordinates  
yields a differential equation of Sturm-Liouville type for the angular part
\begin{align}
    \left[\left(1-z^2\right) S_{mn}^{\prime}(z)\right]^{\prime}+\left[\lambda_{mn}+c^2\left(1-z^2\right)-\frac{m^2}{1-z^2}\right] S_{mn}(z)=0. \label{eq:Meixner_dgl}
\end{align}
Above we are following the convention of Meixner \cite{meixner_mathieusche_1954}. In modern literature 
the differential equation is more commonly written in the form
\begin{align*}
    \left[\left(1-z^2\right) S_{mn}^{\prime}(z)\right]^{\prime}+\left[\chi_{nm}+c^2 z^2-\frac{m^2}{1-z^2}\right] S_{mn}(z)=0
\end{align*}
and we have the relation $\chi_{mn}(c) = \lambda_{mn}(c) + c^2$. The function $S_{mn}$ is called angular Prolate Spheroidal Wave Functions (PSWF),
the oblate case is of the same form with a minus sign in front of $c^2$ in the differential equation. For $c \to 0$ we obtain the well-known differential 
equation for associated Legendre polynomials and thus $\lambda_{mn}(0) = n(n+1)$. \\

Meixner and Schäfke have extensively studied the prolate spheroidal wave functions in their book \cite{meixner_mathieusche_1954}. They derived
many very general integral relations for $S_{mn}$, some of which take especially for $m=0$ a simple form.\footnote{Section 3.8 page 312 in \cite{meixner_mathieusche_1954}.}
For our analysis it suffices to only consider $m=0$ and we will simply write $\lambda_{mn}(c)= \lambda_{n}$ and $S_{mn}(c,z) = S_n(z)$.\footnote{The case $m \neq 0$ is only relevant for higher dimensional
generalization of Prolates (also known as Slepians) and is treated in \cite{ProIV}.} For convenience
we define the prolate spheroidal differential operator $\Pop_c$,
\begin{align*}
   \Pop_c:  y \mapsto\left(1-z^2\right) y^{\prime \prime}-2 z y^{\prime}+ c^2 (1-z^2)  y.
\end{align*}
Occasionally we will simply refer to the operator $\Pop_c$ as prolate spherical wave equation (PSWEq). 
The following is a selection of known results, on the prolate spheroidal wave equation \cite{meixner_mathieusche_1954} .

\begin{theoremAlph}
    Let $c\geq 0$. The eigenvalue problem for the prolate spheroidal wave equation
    \begin{align} 
        \Pop_c y(z) = -\lambda_n y(z)   \label{eq:Sn_dgl}
    \end{align}
    with $z\in[-1,1]$ has eigenvalues 
    \begin{align*}
        -c^2 < \lambda_0(c) < \lambda_1(c) < \lambda_2(c) < \ldots  \longrightarrow \infty.
    \end{align*}
    The eigenvalues are analytic functions of $c$.
    The corresponding eigenfunctions $S_n(z)$ form a complete orthogonal system in $\Ls^2_1$. The functions
    $S_n(c,z)$ can be extended to entire functions of complex variables $z$ (and $c$) and are real valued for real $z$.
    They are odd if $n$ is odd and even if $n$ is even. $S_n(z)$ has exactly $n$ zeros in $(-1,1)$.
    In addition to eq.(\ref{eq:Sn_dgl}) PSWF $S_n$ satisfy integral equations
    \begin{align} 
        \int_{-1}^1 \rho_c(s-z) S_{ n}(c, s) d s & =\gamma_n S_{n}( z)  \label{eq:S_band_kernel } \\
        \int_{-1}^1 e^{i c s z} S_{n}(c, s) d s & =\frac{\mu_n}{\sqrt{c}} S_{n}(c, z), \label{eq:S_exp_kernel}
    \end{align}
    where $\gamma_n = \frac{|\mu_n|^2}{2\pi}$ and $\conj{\mu}_n = (-1)^n \mu_n$. \footnote{
        The $\mu_n$ are related to the radial part of the prolate spheroidal wave equation
        $\mu_n \propto R_{0n}(c,1)$ \cite{ProI}. For a very long time, the common numerical expansion formulae for $R_{0n}(c,1)$ exhibited catastrophic cancellation effects. 
        The first paper that recognized and solved the problem seems to be only from 2002 \cite{Buren2002AccurateCO} and only 
        came last year along with a published code to calculate $R_{0n}(c,1)$.
        This may also explain why applications of Prolate Fourier Theory have not been as prevalent and much of the theory was somewhat left in the dust. }

    \label{thm:Prolate_Spheroidal_Wave_Equation}
\end{theoremAlph}

Much of Theorem \ref{thm:Prolate_Spheroidal_Wave_Equation}
follows from Sturm-Liouville theory. 
Note that equations (\ref{eq:Sn_dgl})-(\ref{eq:S_exp_kernel}) hold in fact for all complex $z$ as all equations are analytic 
and can be uniquely extended to the entire complex plane. \\

With the connection to the prolate spheroidal wave equation all puzzle pieces 
to understand Theorem \ref{thm:prolate_original} are finally in place. 
From (\ref{eq:S_band_kernel }) we can deduce that 
$S_{n}(c ,\frac{t}{T})$ with $c= \bw T$ satisfies the integral relation $\BL_\bw \TL_T S_{n}(c ,\frac{t}{T}) = \gamma_n S_{n}(c ,\frac{t}{T})$
\begin{align*}
    \int_{-T}^{T} \rho_\bw(s-t)  S_{0,n}(c ,\frac{t}{T}) dt & = \int_{-T}^{T} \frac{\sin{\bw(s-t)}}{{\pi(s-t)}} S_{0,n}(c ,\frac{t}{T})  dt \\
    & = \int_{-1}^{1} T \frac{\sin{\bw T (\frac{s}{T}-t)}}{{\pi T(\frac{s}{T} -t )}} S_{0,n}(c ,t) dt\\
    & = \int_{-1}^{1} \rho_c(\frac{s}{T} -t) S_{0,n}(c ,t) dt  = \gamma_n S_{0,n}(c ,\frac{s}{T}). 
\end{align*}
Due to the completeness relation of the angular PSWF $S_n(\frac{t}{T})$ in $\Ls_T^2 \equiv \TLs_T$, they coincide indeed with the prolates $\prlt_n(t)$
as the Fourier concentration optimizers, that we are interested in. 
The remaining non-degeneracy in $\gamma_n$ was proven in \cite{ProI} as a consequence of the non-degeneracy in $\lambda_n$ (or equivalently $\chi_n$). 
It was also shown that the ordering where $\lambda_n$ is increasing and $\gamma_n$ is decreasing in $n$ is indeed correct. The
additional integral relation  (\ref{eq:ori_exp_eig_int}) for $\prlt_n$ follows from (\ref{eq:S_exp_kernel})
\begin{align*}
    \int_{-T}^{T} e^{i \frac{ \bw \tau t}{T} } \prlt_n(t) dt & = \int_{-T}^{T} e^{i \frac{ \bw \tau t}{T} } S_n(\frac{t}{T}) dt = T \int_{-T}^{T} e^{i \bw \tau t} S_n(t) dt\\
    & = T \int_{-T}^{T} e^{i c \frac{\tau}{T} t} S_n(t) dt = \mu_n \frac{T}{\sqrt{\bw T}}  S_n(\frac{\tau}{T}) dt = \mu_n \sqrt{\frac{T}{\bw}} \prlt_n(\tau).
\end{align*}

\subsection{Generalized PSWEq and Commutation Relations}
\label{sec:generalized_PSWEq}

From their relation to $S_{0n}$ and (\ref{eq:Meixner_dgl}) we also obtain a differential equation to which 
prolates are eigenfunctions. The analog of (\ref{eq:Meixner_dgl}) reads \footnote{
    In the more common convention we would have 
    $$ \left(T^2-t^2\right)  \frac{d^2 \prlt_n(t)}{ d t^2}-2 t \frac{d \prlt_n(t)}{d t}+ \left( \chi_n(c)- \bw^2 t^2\right) \prlt_n(t)=0$$
    With $\prlt_n(t)= S_n(\frac{t}{T})$ and $z = \frac{t}{T}$ we can indeed rewrite this as 
    \begin{align*}
        \left(T^2 -t^2\right)&  \frac{d^2 S_n(\frac{t}{T})}{ d t^2}  -2 t \frac{d S_n(\frac{t}{T})}{d t}+  \bw^2 t^2 S_n(\frac{t}{T})  = \left(T^2-t^2 \right)  \frac{1}{T^2} \frac{d^2  S_n(z) }{dz^2}- 2 \frac{t}{T} \frac{d S_n(z)}{ d z} + \bw^2 t^2  S_n(z)\\
        \quad & = \left(1-z^2 \right)  \frac{d^2 S_n(z)}{dz^2}  - 2 z \frac{d S_n(z)}{ d z} + \bw^2 T^2  z^2  S_n(z)  = - \chi_n(c) S_n(z) = - \chi_n(c) \prlt_n(t).
    \end{align*}}
\begin{align}
   \left[ \left(T^2-t^2\right) \prlt_n'(t)\right]' +\left[\lambda_{n}(c)+ \bw^2(T^2 -t^2)\right] \prlt_n(t)=0, \label{eq:Prlt_dgl}
\end{align}
valid for all $t \in \mathbb{C}$. Correspondingly we can define a generalized prolate spheroidal differential operator to whom $\prlt_n$ are eigenfunctions 
\begin{align}
    & \mathcal{P}_{\bw T}:  y  \mapsto \left(T^2-t^2\right)y'' - 2 t y'  + \bw^2(T^2- t^2)  y  \notag \\
    & \mathcal{P}_{\bw T}[\prlt_n](t) = - \lambda_n(c) \prlt_n(t)  \label{eq:Prlt_dgl_op}.
\end{align}
With $\mathcal{P}_{\bw T}$ in hand we could have alternatively also 
shown that $\BL_\bw\TL_T$  does indeed commute with $\BL_\bw \TL_T$ and that the 
eigenfunctions of the two operators necessarily have to coincide. \\

In fact the stronger statement that $\Pop_{\bw T}$ commutes with $\TL_T$ and $\BL_\bw$ separately is true. 
Let $\varphi \in C^2(\mathbb{R})$ be a test function. We have in the distributional sense 
\begin{alignat*}{2}
    \Pop_{\bw T} \TL_T \varphi & = (T^2-t^2) (\chi_T \varphi'' +2 \chi_T' \varphi' +  \chi_T'' \varphi) - 2 t \left(\chi_T \varphi' + \chi_T'\varphi\right) + \bw^2(T^2-t^2) \chi_T \varphi\\
    & = \chi_T  \left[(T^2-t^2)\varphi'' - 2 t \varphi + \bw^2 (T^2-t^2) \varphi\right] \\
    & \hspace{3cm}+ (\delta_{-T}' - \delta_T')(T^2-t^2) \varphi + (\delta_{-T} -\delta_T)(2 (T^2 -t^2) \varphi' -2t\varphi)\\
    & = \TL_T \Pop_{\bw T}  \varphi  +  2t \varphi (\delta_{-T} -\delta_T ) - 2t\varphi  (\delta_{-T} -\delta_T )\\
    & = \TL_T \Pop_{\bw T}  \varphi +0.
\end{alignat*}
Next we want to proof that $\Pop_{\bw T}$ commutes with the Fourier transform up to a swap in $T$ and $\bw$
\begin{align}
    \Pop_{\bw T}  \FT = \FT  \Pop_{T \bw}. \label{eq:Pop_FT_comm}
\end{align}
To establish (\ref{eq:Pop_FT_comm}) we recall that for smooth square integrable functions $\varphi$
the Fourier transform has $\FT[ \partial^\alpha \varphi] = (i \omega)^\alpha \FT[\varphi]$ and $\FT^{-1}[ t^\alpha \varphi] = (-i)^\alpha \partial^\alpha\FT[\varphi]$.
We have 
\begin{align*}
    \FT^{-1} \Pop_{\bw T}  \FT \varphi & = \FT^{-1}[ (T^2 -\omega^2) \FT''\varphi - 2 \omega \FT' \varphi + \bw^2 (T^2 -\omega^2) \FT \varphi]\\
    & = \FT^{-1} \left[ [(T^2 -\omega^2) \FT'\varphi]' + \bw^2 (T^2 -\omega^2) \FT \varphi \right]\\
    & =  -it \FT^{-1}\left[ (T^2 -\omega^2) \FT'\varphi \right] + \FT^{-1} \left[\bw^2 (T^2 -\omega^2) \FT \varphi \right]\\
    & =  -it \left( T^2  \FT^{-1}\FT'\varphi  - \FT^{-1} [\omega^2 \FT'\varphi] \right) + T^2 \bw^2 \varphi -  \bw^2 \FT^{-1}[ \omega^2  \FT \varphi]\\
    & = -it \left( - it T^2   \varphi  + [\FT^{-1} \FT'\varphi]'' \right) + T^2 \bw^2 \varphi +  \bw^2 [\FT^{-1} \FT \varphi]''\\
    & = -t^2 T^2 \varphi - it[ -i t \varphi]'' + T^2 \bw^2 \varphi +  \bw^2  \varphi''\\
    & = (\bw^2 -t^2) \varphi'' - 2 t \varphi' + T^2(  \bw^2 -t^2) \varphi  \\
    & = \Pop_{T \bw} \varphi.
\end{align*}
Since $\FT$ is an isometry on the Hilbertspace of square integrable functions and $\FT^{-1}$ is its inverse equation (\ref{eq:Pop_FT_comm}) follows. 
The same relation also holds for the inverse Fourier transform
$$\Pop_{\bw T}  \FT^{-1} = \FT^{-1}  \Pop_{T \bw},$$
since $\conj{\Pop_{\bw T}} = \Pop_{\bw T}$. 
Indeed simply taking the complex conjugate of the previous calculation we obtain
$$\FT^{} \Pop_{\bw T}  \FT^{-1}  \conj{\varphi} = \Pop_{T \bw} \conj{\varphi},$$
and $\conj{\varphi}$ is still just a test function. \\

Naturally $\Pop_{T \bw}$ satisfies analogous to $\Pop_{\bw T}$ the commutation relation $$\Pop_{T \bw} \TL_\bw = \TL_\bw \Pop_{T \bw}.$$
Note that the operation of band limiting a function $\BL_\bw \varphi $ writes in terms of Fourier transforms as $\BL_\bw \varphi= \FT^{-1} \TL_{\bw} \FT \varphi$.
We can finally proof that the prolate differential operator $\Pop_{\bw T}$ and the integral operator $\BL_\bw$ commute,
\begin{align*}
    \Pop_{\bw T} \BL_\bw  \varphi & =  \Pop_{\bw T}  \FT^{-1} \TL_\bw \FT  \varphi = \FT^{-1}  \Pop_{T \bw }  \TL_\bw \FT \varphi =  \FT^{-1}  \TL_\bw  \Pop_{T \bw }\FT \varphi \\
    & = \FT^{-1}  \TL_\bw  \FT \Pop_{\bw T }  \varphi  = \BL_\bw  \Pop_{\bw T }\varphi.
\end{align*}

The fact that $\Pop_{\bw T}$ commutes with $\BL_\bw$ and $\TL_T$ separately was first noted and proven 
by Walter in \cite{walter_differential_1992}, much after the majority of known results of Prolates Fourier Theory have been established. He also showed more generally that a second 
order differential operator $P$ with quadratic coefficients commutes with $\BL_\bw \TL_T$ if and only if it is of the form,\footnote{\cite{walter_differential_1992} has a different convention for the Prolate differential operator, $\tilde P_{\bw T} = P_{\bw T} - \bw^2 T^2$. However, the statement does note care about constant offsets.}
$$P=a P_{\bw T}+b.$$

From standard Sturm Liouville theory it can also be established 
that $c^2 - \Pop_{\bw T}$ is positive definite on the interval $(-T,T)$. Recall that $c = \bw T$.
For convince we define $\tilde \Pop_{\bw T} = c^2 -  \Pop_{\bw T}$. Note that $\tilde \Pop_{\bw T} \prlt_n(t)= (c^2 + \lambda_n) \prlt_n(t)= \chi_n \prlt_n(t)$.
Let $\varphi \in \Ls^2_T \setminus \{0\}$ be sufficiently regular. We have with partial integration 
\begin{align*}
    \langle \varphi, \tilde \Pop_{\bw T} \varphi \rangle_T & = \int_{-T}^{T} \conj{\varphi(t)} (\bw^2 T^2 - \Pop_{\bw T})  \varphi(t) dt \\
    & = \int_{-T}^{T} \conj{\varphi(t)} \left( -[(T^2 -t^2)\varphi'(t)]' + t^2 \varphi(t)  \right) dt\\
    & = \int_{-T}^{T} (T^2 -t^2) |\varphi'(t)|^2 + t^2 |\varphi(t)|^2 dt > 0.
\end{align*}
Here we used, that boundary term from the partial integration vanishes.
We denote with $\TLs_T^*$ the subclass of time limited functions that are sufficiently regular on the open interval $(-T,T)$ and we have just seen that 
$\tilde \Pop_{\bw T}$ is positive definite in $\TLs_T^*$.\\

That an operator such as $\tilde \Pop_{\bw T}$ is positive definite on $\TLs_T^*$ might be a trivial statement to someone familiar with Sturm-Liouville theory.
However, we can go a little bit further and show that $\tilde \Pop_{\bw T}$ is also positive definite on 
the space of band limited functions $\BL_\bw$. This can indeed conveniently be shown using the commutation relations we previously derived.
Formally we will define the subclass of band limited functions $\BL_\bw^*$ whose Fourier transform is guaranteed to be in $\TLs_\bw^*$.
Then we have for $\varphi \in \BL_\bw^*$
\begin{align*}
    \langle \varphi, \tilde \Pop_{\bw T} \varphi \rangle &  \Pl \frac{1}{2\pi}\langle \FT \varphi, \FT \tilde \Pop_{\bw T} \varphi \rangle = \frac{1}{2\pi}\langle \FT \varphi, \FT \tilde \Pop_{\bw T} \varphi \rangle_\bw = \frac{1}{2\pi}\langle \FT \varphi,  \tilde \Pop_{T \bw} \FT \varphi \rangle_\bw \geq 0.
\end{align*}
Where the last inequality follows from the fact that $\tilde \Pop_{T \bw}$ is positive definite on $\TLs_\bw^*$.
From the previously seen identity theorem argument equality is attained if and only if $\varphi=0$. Thus we 
have established that $\tilde \Pop_{\bw T}$ is positive definite on $\BL_\bw^*$.\\

The following theorem summarizes the operator properties and commutation relations we have seen in this section.

\begin{theoremAlph}
    Withhin the subclass of sufficiently regular square integrable functions $\Ls^{2}_\infty{}^*$ we have in a distributional sense
    \begin{alignat*}{2}
        \Pop_{\bw T} \TL_T &=  \TL_T \Pop_{\bw T} , \hspace{2cm}  \Pop_{\bw T} \FT  &&= \FT \Pop_{ T \bw} ,\\
        \Pop_{\bw T} \BL_\bw &= \BL_\bw \Pop_{ \bw T},  \hspace{2cm} \Pop_{\bw T} \FT^{-1}   &&= \FT^{-1} \Pop_{ T \bw}.
    \end{alignat*}
    In particular, $$  \BL_\bw \TL_T \Pop_{\bw T} =  \Pop_{\bw T} \BL_\bw \TL_T.$$
    The integral operator $\BL_\bw \TL_T:\Ls^2_\infty \to \Ls^2_\infty $ is Hilbert-Schmidt and positive definite in $\TLs_T$ and $\BLs_\bw$. Its Hilbert-Schmidt norm is
    \begin{align*}
        \|\BL_\bw \TL_T\|_{HS,\Ls^2_\infty}^2 = \sum_{n=0}^\infty \gamma_n = \frac{2\bw T}{ \pi}.
    \end{align*}
    The differential operator $\tilde \Pop_{\bw T} = c^2 - \Pop_{\bw T}$ is positive definite in $\TLs_T^*$ and $\BLs_\bw^*$.
\end{theoremAlph}

\subsubsection{Remark on the lucky accident}

We have seen how the question of optimal Fourier concentration within intervals, naturally leads to the integral operator 
$\BL_\bw \TL_T$, that is at the heart of Prolate Fourier Theory. It is important to point out that the connection 
to the prolate spheroidal wave equation was much more incidental to the development of the theoretical apparatus to which 
we merely gave an introductory taste.\\

Much after the initial paper from 1961, Slepian referred to the fact, that the eigenfunctions of the integral operator $\BL_\bw \TL_T$ coincide with the eigenfunctions 
of the Prolate Spheroidal Wave Equation as a "lucky accident" \cite{SlepianComment} (1983). He and his colleagues had a feeling 
that there seemed to be a bigger structure underneath, that would allow us to carry out derivations in a more general manner such that the lucky 
accident may seem less incidental.\\




Apart from higher-dimensional generalizations and the discretization of the theory, the Prolate Spheroidal Wave Equation (PSWEq)\
did not play a significant role in the majority of the results derived by Landau, Pollak, and Slepian. Subsequent research primarily focused 
on understanding the \emph{essential dimensionality} of signal processing. 
Their work enabled a rigorous mathematical foundation for a long-standing engineering intuition that was coined the "Engineering Folk Theorem".
The mathematical essence of this theorem is expressed in the spectrum of the prolate spheroidal integral operator.

\section{The Engineering Folk Theorem}

A longstanding engineering intuition, which took root shortly after the development of telecommunications, 
suggests that any signal of bandwidth $2\bw$ and duration $ 2T$ can be reconstructed from $2\bw T / \pi$ samples points \cite{Nyquist1928,NoiseCommunication}. 
This intuition was eventually coined the "Engineering Folk Theorem" and can be reformulated in terms of the dimensionality of signals. \\

\noindent  \textbf{Engineering Folk Theorem:} \emph{The space of signals of bandwidth $2\bw$  and duration $2 T$ is esssentially $2 \bw T / \pi $-dimensional.}\\

Much of the formulation of the engineering folk theorem seems hand-wavy as there are no signals that are band limited and simultaneously also time limited.
Nevertheless, this folk theorem has had a profound impact on our modern understanding of communication.\\

In his 1948 seminal paper "A Mathematical Theory of Communication" Claude Shannon presented 
a model that introduced a discrete channel to describe communication across electrical devices \cite{MatCommunication}. An abstract apparatus that transmits messages from a discrete alphabet and 
operates at discrete time instances. The information theoretical implications he could derive from his model are today regraded 
as one of the most important scientific contributions to the information age we live in.\\

However, today and also back then, most actual electrical devices, operate with continuous electrical signals at
continuous time instances. Shannon was very much aware that a continuous theory of communication could more accurately describe
communication devices of reality. In his 1949 paper "Communication in the Presence of Noise" he 
made an attempt to reconcile his discrete communication model with continuous signals. Here he provided perhaps his most 
celebrated result, the optimal transmission rate at which discrete information can be transmitted over a band limited channel \cite{NoiseCommunication}.
However, he proved his theorem only up the very engineering folk theorem that we stated above. \\

One of the deepest insights gained by Prolate Fourier Theory, is that the mathematical truth behind the engineering folk theorem is expressed 
in the spectrum of prolate spheroidal integral operator. Landau and Pollak gave the fist mathematical rigorous instance of the engineer folk theorem, through 
an approximation theorem that leverages the double completeness property of prolates \cite{ProIII}. Any band limited signal that is also concentrated in $[-T,T]$ is well approximated 
within $[-T,T]$ by the projection of the signal onto the first $\sim 2\bw T / \pi$ prolates. 
The details of the approximation statement rely on the spectrum of the prolate spheroidal integral operator.
It is found that the first $\tilde 2\bw T /\pi - A \log(\tilde 2\bw T /\pi )$
eigenvalues $\gamma_n$ are very close to one, while eigenvalues $\gamma_n$ for $n > \tilde 2\bw T /\pi + A \log(\tilde 2\bw T /\pi )$ are very small.
Here $A$ is some constant \cite{ProIII}.\\

Thus, the spectrum of the prolate spheroidal integral operator exhibits strong clustering of the eigenvalues at $1$ and $0$ with a logarithmically sharp transition region 
in between \cite{ProIII,FUCHS1964317,SlepianAsymp,landau_eigenvalue_1980,landau_density_1993,OnBandwidth}. After the spectrum 
of the prolate spheroidal integral operator was better understood, Slepian gave a second version of the engineering folk theorem that leaned 
deeper into physical interpretation \cite{OnBandwidth}.\\

\chapter{Supremum Bounds for Prolate Spheroidal Wave Functions: New Insights and Implications}
\label{chap:prlt_bound}

\section{Content and Motivation}

Prolate Spheroidal Wave Functions (PSWFs) form a sequence of band-limited functions that are optimally concentrated within a time interval. 
They can be characterized as solutions to the optimization problem,
\begin{align*}
    \prlt_n = \argmax_{f \in \{ \prlt_0, \cdots, \prlt_{n-1}\}^\perp } \frac{\|\TL_T \BL_\bw f \|^2}{\|\BL_\bw f \|^2}.
\end{align*}
Prolate Fourier Theory has provided significant insights, such as the $2\bw T$-Theorem,
which has been a long standing intuition on the essential dimensionality of signals, often referred to as the "Engineering Folk Theorem" \cite{ProIII, OnBandwidth,NoiseCommunication}.
Consequently, PSWFs are a cornerstone of approximation theory. The energy of prolates outside their concentration region is given by $1-\gamma_n$ and frequently serves as a 
fundamental approximation bound in various derivations \cite{ProI, ProII, ProIII,OnBandwidth}.\\

This work investigates how well the ideal concentration property with respect to the square norm extends to a supremum bound outside the concentration region. 
Specifically, the conjecture was
\begin{align}
    \sup_{|t| \geq T}  \prlt_n(t)^2 \sim (1- \gamma_n),  \label{eq:goal_bund}
\end{align}
up to a less significant prefactor.\\

Establishing such a bound solely from the integral relations central to Prolate Fourier Theory and techniques from harmonic analysis
proved to be very challenging. However, by using the prolate spheroidal wave equation (PSWEq), a straightforward proof became feasible.
David Slepian used the fact that prolates are also eigenfunctions to 
a second-order differential equation in his generalizations of Prolate Fourier Theory to higher dimensions and discrete sequences \cite{ProIV, ProV}. 
He famously described the commutation relation between the integral operator and the differential operator as a "lucky accident" \cite{SlepianComment}.
The miraculous commutation relation was later coined the "Prolate Spheroidal Phenomena" and inspired Grünbaum and colleagues
to develop bispectral theory \cite{Grnbaum2004ThePS}. The Prolate Spheroidal Phenomena has lead to further mathematical miracles in the recent years
and unexpected connections \cite{UVSpectrum}.\\

However, the geometry behind the commutation relation still seems to be mysterious and unexplored. 
Landau has demonstrated that the spectrum of the prolate integral operator
$\BL_\bw \TL_T$ exhibits a certain \emph{spectral signature} \cite{landau_density_1993},
\begin{align}
    \gamma_{[\tilde c ]-1} \geq \frac{1}{2} \geq \gamma_{[\tilde c ]+1},
\end{align}
where  $\tilde c = 2\bw T / \pi $. We conjecture, that this spectral signature has a dual in the spectrum 
of the prolate spheroidal wave equation, 
\begin{align}
    \lambda_{[\tilde c ]-1} \leq 0 \leq \lambda_{[\tilde c ]+1}.  \label{eq:conj_dual}
\end{align}
We give a Bohr-Sommerfeld argument in support of this conjecture, which resulted from a conversation with
Gian Michele Graf. Additionally, the proof of $(\ref{eq:goal_bund})$ also implies, that the optimal concentration property of prolates 
extends up to a less important prefactor to their derivatives. This enhanced regularity property allows
to derive sharper error bounds for sampling formulas as we elaborate in Chapter \ref{chap:Sampling}.

\subsection{Result Overview}
\label{sec:bound_res_overview}

The following theorem summarizes the main results of this chapter. 
\begin{theorem}
    \label{thm:prlt_bound}
    The energy concentration of Prolates transfers to its derivatives as 
    \begin{alignat}{4}
        &\| \prlt_n' \|_{T}^2  && =  &&\gamma_n     && C_{\text{intra},n} \label{eq:thm_derivConcen_intra}\\
        &\| \prlt_n' \|_{>T}^2 && = (1- &&\gamma_n) && C_{\text{extra},n}  \label{eq:derivConcen_extra}.
    \end{alignat}
    Prolates have bounds within and outside of the concentration region
    \begin{alignat}{3}
        \sup_{|t| \geq T}  \prlt_n(t)^2 &\leq (1- &&\gamma_n) &&C_{\text{extra},n}  \label{eq:thm_prop_bound_extra} \\
        \sup_{|t| \leq T}  \prlt_n(t)^2 &\leq &&\gamma_n && \tilde C_{\text{intra},n}  \label{eq:thm_prop_bound_intra},
    \end{alignat}
    where $ \tilde C_{\text{intra},n} = C_{\text{intra},n}  + \delta_{n0} \frac{\prlt_0(T)^2}{\gamma_0}$.
\end{theorem}
\noindent Throughout this chapter we assume prolates $\prlt_n$ that are normalized across the real line. The $C$-factors of Theorem \ref{thm:prlt_bound} are defined as,
\begin{align*}
    C_{\text{extra},n} &= \sqrt{\| \prlt_n' \|_{\infty}^2 - \frac{\lambda_n}{T}  \frac{\prlt_n(T)^2}{1-\gamma_n} }, \quad \text{ and } \quad C_{\text{intra},n} = \sqrt{\| \prlt_n' \|_{\infty}^2 +  \frac{\lambda_n }{T}  \frac{\prlt_n(T)^2 }{\gamma_n}}.
\end{align*}
They have bounds,
\begin{itemize}
    \item For $\lambda_n < 0$  \begin{align} C_{\text{extra},n} \leq \left(C_n - \frac{\lambda_n}{2T}\right) \quad \text{ and } \quad C_{\text{intra},n} < \bw,  \label{eq:lambda_neg_C_estimate}
    \end{align}
    \item For $\lambda_n \geq 0$ \begin{align} C_{\text{extra},n} < \bw \quad \text{ and } \quad C_{\text{intra},n} \leq \left(C_n + \frac{\lambda_n}{2T}\right),
     \end{align}
\end{itemize}
where $C_n = \sqrt{\|\prlt_n'\|_{\infty}^2 + \frac{\lambda_n^2}{4T^2}}$. Further we have the bound $\|\prlt_n'\|_\infty < \bw$. We also recall from the prolate spheroidal wave equation $-c^2 < \lambda_0 < \lambda_1 < \ldots \to \infty$ where
$ c= \bw T$. In the interesting case of $c$ large and $\lambda_n < 0$  we have the asymptotic bound,
\begin{align*}
    C_{\text{extra},n} & < \bw ( c + c^{-1} +  O(c^{-2})).
\end{align*}
Meanwhile, the asymptotic behavior for the eigenvalues $1- \gamma_n$ at large $c$ is given by,
\begin{align}
    1 -\gamma_n = \frac{4 \sqrt{\pi} 2^{3 n} c^{n+\frac{1}{2}} e^{-2 c}}{n!}\left[1-\frac{6 n^2-2 n+3}{32 c}+O\left(c^{-2}\right)\right].  \label{eq:gamma_n_asympt}
\end{align}
The first term in (\ref{eq:gamma_n_asympt}) is due to Fuchs \cite{FUCHS1964317}, and the second term was added by Slepian \cite{SlepianAsymp}.
In particular, for $n$ fixed, $C_{\text{extra},n}$ varies in comparison to $1 -\gamma_n$ only insignificantly with respect to $c$. \\

Therefore, Theorem \ref{thm:prlt_bound} shows, that the
concentration property of prolates with respect to the energy norm approximately extends to supremum bounds outside and within their concentration region.
This further confirms prolates as optimal information processing basis. In Chapter \ref{chap:FilterDiag}, the bound (\ref{eq:thm_prop_bound_extra})
is crucial for the proof of a high precision signal processing routine.\\

The above results will be established in Section \ref{sec:exposition_prolate_bound}. The Bohr-Sommerfeld argument in support of conjecture (\ref{eq:conj_dual})
is given in Section \ref{sec:transition_signature}.

\section{Main Exposition} 
\label{sec:exposition_prolate_bound}

In Chapter \ref{chap:prolate_theory_intro} we have seen that Prolates are eigenfunctions to two integral operators and a differential operator, 
\begin{align}
    & \int_{-T}^{T} \rho_\bw(s-t)  \prlt_n(t) dt  = \gamma_n \prlt_n(s),   \label{eq:init_sinc_eig} \\
    & \int_{-T}^{T} e^{i \frac{ \bw \tau t}{T} } \prlt_n(t) dt  = \mu_n \sqrt{\frac{T}{\bw}} \prlt_n(\tau),  \label{eq:FT_eig}\\
    (T^2 -t^2) \prlt_n''(t)&  - 2 t \prlt_n'(t) + \bw^2 (T^2 -t^2) \prlt_n(t)  = - \lambda_n \prlt_n(t). \label{eq:genPSWEq}
\end{align}

\noindent Let $\tprlt_n$ denote the dual prolate of $\prlt_n$, defined as
$$\tprlt_n(\omega) = \sqrt{\frac{T}{\bw}} \prlt_n(\frac{T}{\bw} \omega).$$
By construction, we have
$$\|\tprlt_n\|_{\bw}^2 = \frac{T}{\bw} \int_{-\bw}^{\bw} \prlt_n(\frac{T}{\bw} \omega)^2 d \omega  = \| \prlt_n\|_{T}^2 = \gamma_n.$$
Due to symmetry, $\tprlt_n$ is a prolate function with the roles of $T$ and $\bw$ reversed and satisfies
the eigenfunction property,
\begin{align*}
    \BL_T \TL_\bw \tprlt_n(\omega) = \gamma_n \tprlt_n(\omega).
\end{align*}
Thus, all identities for $\prlt_n$ can be directly transferred to $\tprlt_n$ by exchanging $T$ and $\bw$ in the equations.\\

The following fact illustrates how prolates $\prlt_n$ and their dual $\tprlt_n$ are related 
through the Fourier transform.


\begin{fact} 
    We have Fourier transform relations
    \begin{align}
        \FT [\prlt_n](\omega) & =  \frac{2 \pi}{\mu_n}  \TL_\bw \tprlt_n \left(\omega \right)   \hspace{1.5cm}    \text{for all } \omega \in \mathbb{R} \label{eq:prolate_FT}\\
        \FT  [\TL_T \prlt_n](\omega) & =  \conj{\mu}_n  \tprlt_n \left(\omega \right)    \hspace{2.1cm}   \text{for all } \omega \in \mathbb{C}  \label{eq:prolate_FT_T}.
    \end{align}
    Vice versa,
    \begin{align}
        \FT [\tprlt_n](t) & =  \frac{2 \pi}{\mu_n}  \TL_T \prlt_n \left(t \right)   \hspace{1.5cm}    \text{for all } t \in \mathbb{R} \label{eq:Tprolate_FT}\\
        \FT [\TL_\bw \tprlt_n](t) & =  \conj{\mu}_n  \prlt_n \left(t \right)    \hspace{2.1cm}   \text{for all } t \in \mathbb{C}  \label{eq:Tprolate_FT_T}.
    \end{align}
    \label{fact:prolate_FT}
\end{fact}
Where we recall that $\TL_A$ is the truncation operator defined by $\TL_A f(t) = \chi_A(t) f(t) $. The crucial
eigenfunction relation to establish Fact \ref{fact:prolate_FT} is given by (\ref{eq:FT_eig}).

\begin{proof}[\textit{Proof of Fact \ref{fact:prolate_FT}}]
    We start with eq.(\ref{eq:prolate_FT_T}) and have
    \begin{align*}
        \FT [\TL_T \prlt_n](\omega) & = \int_{-T}^{T} e^{ - i \omega t } \prlt_n(t) dt = \int_{-T}^{T} e^{ - i \frac{t \bw}{T} \frac{T\omega  }{\bw} } \prlt_n(t) dt
    \end{align*}
    In this form we can apply (\ref{eq:FT_eig}) with $\tau = \frac{T\omega}{\bw}$ and obtain
    \begin{align*}
        \FT [\TL_T \prlt_n](\omega) = \conj{\mu}_n  \sqrt{\frac{T}{\bw}} \prlt_n(\frac{T}{\bw} \omega)= \conj{\mu}_n \tilde \prlt_n \left(\omega  \right).
    \end{align*}
    As we took the complex conjugate of (\ref{eq:FT_eig}) we also applied that prolates are real valued on the real line.
    By symmetry we have that $\tilde \prlt_n$ satisfies the analogous relation
    \begin{align}
        \prlt_n \left(t \right)  & = \frac{2 \pi}{\mu_n} \FT^{-1} [ \TL_\bw \tprlt_n](t). \label{eq:prlt_FT_p1}
    \end{align}
    Taking the Fourier transform of (\ref{eq:prlt_FT_p1}) gives
    \begin{align*}
        \FT [\prlt_n](\omega) & = \frac{2 \pi}{\mu_n}  \TL_\bw \tprlt_n(\omega).
    \end{align*}
    Equations (\ref{eq:Tprolate_FT}) and (\ref{eq:Tprolate_FT_T}) follow immediately.
\end{proof}

\begin{lemma}
    We have 
    \begin{align}
        \| \prlt_n' \|_{T}^2 & = \gamma_n \| \prlt_n' \|_{\infty}^2 + 2 \prlt_n(T) \prlt_n'(T)  \label{eq:deriv_concentration_T}
    \end{align}
    and 
    \begin{align*}
        \| \prlt_n' \|_{\infty}^2 & = \frac{1}{\gamma_n} \int_{-\bw}^{\bw} \omega^2 \tprlt_n(\omega)^2 d\omega.
    \end{align*}
    Further,
    \begin{align*}
        \| \prlt_n' \|_{\infty}^2  & < \bw^2 \| \prlt_n \|_{\infty}^2.
    \end{align*}
    \label{lem:deriv_concentration}
\end{lemma}
The proof of Lemma \ref{lem:deriv_concentration} applies the integral eigenfunction relation $(\ref{eq:init_sinc_eig})$ for the dual prolates $\BL_T \TL_\bw \tprlt_n = \gamma_n \tprlt_n$.

\begin{proof}[\textit{Proof of Lemma \ref{lem:deriv_concentration}}]
    We start with $\|\prlt_n'\|_\infty^2$. Applying Plancherel and (\ref{eq:prolate_FT}) gives
    \begin{align*}
        \int_{-\infty}^{\infty} \prlt_n'(t)^2 dt & \Pl  \frac{1}{2\pi}  \int_{-\infty}^{\infty}  \omega^2 |\FT[\prlt_n](\omega)|^2 d\omega = \frac{1}{2\pi} \frac{4 \pi^2}{|\mu_n|^2} \int_{-\bw}^{\bw} \omega^2 \tprlt_n^2(\omega) d\omega \\
        & = \frac{1}{\gamma_n} \int_{-\bw}^{\bw} \omega^2 \tprlt_n^2(\omega) d\omega.
    \end{align*}
    Partial integration and the fact that $\prlt_n(t) \prlt_n'(t)$ is an odd function gives for $\| \prlt_n' \|_T^2$, 
    \begin{align}
        \int_{-T}^{T} \prlt_n'(t) \prlt_n'(t) dt = 2 \prlt_n(T) \prlt_n'(T) - \int_{-T}^{T} \prlt_n(t) \prlt_n''(t) dt.  \label{eq:lem_123}
    \end{align}
    Relation (\ref{eq:Tprolate_FT_T}) of Fact \ref{fact:prolate_FT} can be differentiated and applied to the control the integral term in (\ref{eq:lem_123}). 
    We obtain,
    \begin{align*}
        \int_{-T}^{T} \prlt_n''(t) \prlt_n(t) dt & \stackrel{\text{eq.\ref{eq:Tprolate_FT_T}}}{=} \frac{1}{|\mu_n|^2} \int_{-T}^{T} \conj{\FT''} [\TL_\bw \tprlt_n](\omega) \FT [ \TL_\bw \tprlt_n](\omega) d \omega \\
        & = \frac{i^2}{|\mu_n|^2} \int_{-T}^{T} \int_{-\bw}^{\bw} \int_{-\bw}^{\bw} \omega^2 e^{i t (\omega -s)} \tprlt_n(\omega) \tprlt_n(s) d \omega ds dt \\
        & = i^2 \frac{2\pi }{|\mu_n|^2} \frac{1}{2\pi} \int_{-\bw}^{\bw} \int_{-\bw}^{\bw} \int_{-T}^{T} \omega^2 e^{i t (\omega -s)} \tprlt_n(\omega) \tprlt_n(s) dt d\omega ds \\
        & \stackrel{(*)}{=} i^2 \frac{1 }{\gamma_n} \int_{-\bw}^{\bw} \int_{-\bw}^{\bw} \omega^2 \rho_T(\omega -s) \tprlt_n(\omega) \tprlt_n(s) d\omega ds \\
        & = i^2 \int_{-\bw}^{\bw} \omega^2 \tprlt_n(\omega)^2 d\omega = -\gamma_n \| \prlt_n' \|_{\infty}^2.
    \end{align*}
    In  $(*)$ the integral representation of $\rho_T$ was used.
    For the inequality we simply have 
    \begin{align*}
        \| \prlt_n' \|_{\infty}^2 & = \frac{1}{\gamma_n} \int_{-\bw}^{\bw} \omega^2 \tprlt_n(\omega)^2 d\omega
        <  \frac{\bw^2}{\gamma_n} \int_{-\bw}^{\bw}  \tprlt_n(\omega)^2 d\omega = \frac{\bw^2}{\gamma_n} \gamma_n \|\prlt_n\|^2_\infty \\
        & = \bw^2 \| \prlt_n \|_{\infty}^2,
    \end{align*}
    and can conclude. 
\end{proof}

\begin{lemma}
    The energy concentration of the derivatives of Prolates has
    \begin{alignat}{4}
        &\| \prlt_n' \|_{T}^2  && =  &&\gamma_n     && \Bigl(\| \prlt_n' \|_{\infty}^2 +  \frac{\lambda_n }{T}  \frac{\prlt_n(T)^2 }{\gamma_n}  \Bigr)  \label{eq:derivConcen_intra}\\
        &\| \prlt_n' \|_{>T}^2 && = (1- &&\gamma_n) && \Bigl( \| \prlt_n' \|_{\infty}^2 - \frac{\lambda_n}{T}  \frac{\prlt_n(T)^2}{1-\gamma_n}  \Bigr) \label{eq:derivConcen_extra}.
    \end{alignat}
    \label{lem:deriv_concentration}
\end{lemma}

Finally, we also make use of the generalized prolate spheroidal wave equation (\ref{eq:genPSWEq}) to prove Lemma \ref{lem:deriv_concentration}.

\begin{proof}[\textit{Proof of Lemma \ref{lem:deriv_concentration}}]
    Evaluating (\ref{eq:genPSWEq}) at $t=T$ gives
    \begin{align*}
        \prlt_n'(T)= \frac{1}{2} \frac{\lambda_n}{T} \prlt_n(T).
    \end{align*}
    Plugging into (\ref{eq:deriv_concentration_T}) yields
    \begin{align*}
    \| \prlt_n' \|_{T}^2 & = \gamma_n \| \prlt_n' \|_{\infty}^2 +  \frac{\lambda_n}{T} \prlt_n(T)^2  = \gamma_n  \Bigl(\| \prlt_n' \|_{\infty}^2 +  \frac{\lambda_n }{T}  \frac{\prlt_n(T)^2 }{\gamma_n}  \Bigr) 
\end{align*}
    We may alternatively also write
    \begin{align*}
        \| \prlt_n'\|_\infty^2 - \| \prlt_n' \|_{T}^2 & = (1-\gamma_n) \| \prlt_n' \|_{\infty}^2 - \frac{\lambda_n}{T} \prlt_n(T)^2\\
        & = (1-\gamma_n)  \Bigl( \| \prlt_n' \|_{\infty}^2 - \frac{\lambda_n}{T}  \frac{\prlt_n(T)^2}{1-\gamma_n}  \Bigr).
    \end{align*}
\end{proof}
For convenience we denote 
\begin{align*}
    C_{\text{extra},n} &= \sqrt{\| \prlt_n' \|_{\infty}^2 - \frac{\lambda_n}{T}  \frac{\prlt_n(T)^2}{1-\gamma_n} }, \qquad \text{ and } \qquad C_{\text{intra},n} &= \sqrt{\| \prlt_n' \|_{\infty}^2 +  \frac{\lambda_n }{T}  \frac{\prlt_n(T)^2 }{\gamma_n}}.
\end{align*}
\begin{proposition}
    For $t \in \mathbb{R}$ we have bounds
    \begin{alignat}{3}
        \sup_{|t| \geq T}  \prlt_n(t)^2 &\leq (1- &&\gamma_n) &&C_{\text{extra},n}  \label{eq:prop_bound_extra} \\
        \sup_{|t| \leq T}  \prlt_n(t)^2 &\leq &&\gamma_n && \tilde C_{\text{intra},n}  \label{eq:prop_bound_intra},
    \end{alignat}
    where $ \tilde C_{\text{intra},n} = C_{\text{intra},n}  + \delta_{n0} \frac{\prlt_0(T)^2}{\gamma_0}$.
    \label{prop:bounds}
\end{proposition}
Here $\delta_{n0}$ is the Kronecker delta.

\begin{proof}[\textit{Proof of Proposition \ref{prop:bounds}}]
    We show the first inequality. Wlog we assume $t\geq T$. We have by the Fundamental Theorem of Calculus
    \begin{align}
        \prlt_n(t)^2 &= - \int_{t}^{\infty} [\prlt(x)^2]' dx = -2  \int_{t}^{\infty}  \prlt(x) \prlt'(x) dx \leq 2 \int_{\tau}^{\infty} |\prlt_n(x) \prlt_n'(x)| dx  \notag \\
        & \leq \int_{|x|\geq T} |\prlt_n(x) \prlt_n'(x)| dx  \leq \sqrt{\int_{|x|\geq T} \prlt_n(x)^2 dx \int_{|x|\geq T} \prlt_n'(x)^2 dx }. \label{eq:FundCalc}
    \end{align}
    Recall that $ \|\prlt_n\|_{>T}^2 = (1-\gamma_n) \|\prlt_n\|_{\infty}^2$ and with
    Lemma \ref{lem:deriv_concentration} 
    $$ \int_{|x|\geq T} \prlt_n'(x)^2 dx = (1-\gamma_n) \Bigl( \| \prlt_n' \|_{\infty}^2 - \frac{\lambda_n}{T}  \frac{\prlt_n(T)^2}{1-\gamma_n}  \Bigr) = (1-\gamma_n)C_{\text{extra},n}^2.$$
    Plugging into (\ref{eq:FundCalc}) yields
    \begin{align*}
        \prlt_n(t)^2 \leq \left(1-\gamma_n \right) C_{\text{extra},n}.\\
    \end{align*}
    For the second inequality the cases $n>0$ and $n=0$ are treated separately.
    \begin{itemize}
        \item For $n>0$: Recall that $\prlt_n$ has $n$ zeros in $[-T,T]$. Thus for $n>0$ there exists 
        at least one $ \tau \in [-T,T]$ with $\prlt_n(\tau)^2 = 0$. 
        Consider $t\in [-T,T]$ and assume wlog that $\operatorname{sign}{t} = \operatorname{sign}{\tau}$.
        Then we have similarly as above 
        \begin{align*}
            \prlt_n(t)^2 & = \prlt_n(t)^2 - \prlt_n(\tau)^2  = 2 \int_{\tau}^{t} \prlt_n(x) \prlt_n'(x) dx \\
            &\leq 2 \left|\int_{\tau}^{t} |\prlt_n(x) \prlt_n'(x)| dx \right| \leq \int_{-T}^{T} |\prlt_n(x) \prlt_n'(x)| dx  \leq \sqrt{  \|\prlt_n\|_T^2 \|\prlt_n'\|_T^2} \\
            & = \sqrt{\gamma_n^2 \left( \| \prlt_n' \|_{\infty}^2 + \frac{\lambda_n}{T} \frac{\prlt_n(T)^2}{\gamma_n} \right)}  = \gamma_n C_{\text{intra},n}
        \end{align*}
        In the last line (\ref{eq:derivConcen_intra}) was applied.
        \item For $n=0$: Assume wlog $t \in [-T,0]$. Then we have
            \begin{align*}
                \prlt_0(t)^2- \prlt_0(-T)^2 = 2 \int_{-T}^{t} \prlt_0(x) \prlt_0'(x) dx \leq \gamma_0 \sqrt{ \| \prlt_0' \|_{\infty}^2 + \frac{\beta_0 \prlt_n(T)^2}{\gamma_0}},
            \end{align*}
            Rearranging gives 
            \begin{align*}
                \prlt_0(t)^2 & \leq \gamma_0 \left(\sqrt{ \| \prlt_0' \|_{\infty}^2 + \frac{\beta_0 \prlt_n(T)^2}{\gamma_0}} + \frac{\prlt_0(T)^2}{\gamma_0}\right)\\
                & = \gamma_0 \left( C_{\text{intra},0} + \frac{\prlt_0(T)^2}{\gamma_0}\right).
            \end{align*}
    \end{itemize}
\end{proof}

The following corrolary stablishes that $\prlt_n(T)^2 \sim (1-\gamma_n)$ up to factors, that vary less importantly in $\bw T$. 

\begin{corollary}
    We have bounds for Prolates at the edge of the concentration region
    \begin{alignat}{3}
        \prlt_n(T)^2 & \leq (1-&&\gamma_n) &&\left(C_n - \frac{\lambda_n}{2T} \right)  \label{eq:bound_T_extra}\\
        \prlt_n(T)^2 & \leq &&\gamma_n  &&\left(C_n + \frac{\lambda_n}{2T} \right) \label{eq:bound_T_intra}, 
    \end{alignat}
    where $C_n = \sqrt{\|\prlt_n'\|_{\infty}^2 + \frac{\lambda_n^2}{4T^2}}$. The second identity assumes $n>0$.
\end{corollary}

\begin{proof}
    Evaluating the square of (\ref{eq:prop_bound_extra}) at $t=T$ gives 
    \begin{align*}
        \prlt_n(T)^4 \leq (1-\gamma_n)^2 \left( \|\prlt_n'\|_{\infty}^2  - \frac{\lambda_n}{T} \frac{\prlt_n(T)^2}{1-\gamma_n} \right) =  (1-\gamma_n)^2 \|\prlt_n'\|_{\infty}^2 - (1-\gamma_n) \frac{\lambda_n}{T} \prlt_n(T)^2.
    \end{align*}
    Left-hand side reappears on the right and completing the square yields,
    \begin{align*}
        \left( \prlt_n(T)^2 + (1-\gamma_n) \frac{\lambda_n}{2T} \right)^2 & \leq (1-\gamma_n)^2 \left( \|\prlt_n'\|_{\infty}^2 + \frac{\lambda_n^2}{4T^2}\right)  = (1-\gamma_n)^2 C_n^2.
    \end{align*}
    Taking the square root and rearranging gives,
    \begin{align*}
        \prlt_n(T)^2 & \leq (1-\gamma_n) \left( \sqrt{\|\prlt_n'\|_{\infty}^2 + \frac{\lambda_n^2}{4T^2}} - \frac{\lambda_n}{2T} \right) = (1-\gamma_n) \left( C_n - \frac{\lambda_n}{2T} \right).
    \end{align*}
    We can proceed analogously to obtain a second bound from (\ref{eq:prop_bound_intra}). We assume $n>0$.
    \begin{align*}
        \prlt_n(T)^4 & \leq \gamma_n^2 \left( \|\prlt_n'\|_{\infty}^2  + \frac{\lambda_n}{T} \frac{\prlt_n(T)^2}{\gamma_n} \right) =  \gamma_n^2 \|\prlt_n'\|_{\infty}^2 + \gamma_n \frac{\lambda_n}{T} \prlt_n(T)^2\\
        \left(\prlt_n(T)^2- \gamma_n \frac{\lambda_n}{2T} \right)^2  & \leq \gamma_n^2 \left( \|\prlt_n'\|_{\infty}^2  + \frac{\lambda_n^2}{4T^2} \right) =  \gamma_n^2 C_n^2\\
        \prlt_n(T)^2 & \leq \gamma_n \left( C_n + \frac{\lambda_n}{2T} \right).
    \end{align*}
\end{proof}

\noindent The case  $n=0$ in (\ref{eq:bound_T_intra}) can be derived in a manner similar to the derivation given in (\ref{eq:prop_bound_intra}), and the details of this calculation are omitted here.
With Theorem \ref{thm:prlt_bound} now established, the remaining task is to determine the bounds for the  prefactors $C_{\text{extra},n}$ and $C_{\text{intra},n}$.


\subsection{Bounds on the Prefactors}

We now provide a bound on $C_{\text{extra},n}$. 
For $\lambda_n \geq 0$ we can estimate,
\begin{align}
    C_{\text{extra},n}^2& = \|\prlt_n'\|_{\infty}^2 - \frac{\lambda_n}{T} \frac{\prlt_n(T)^2}{1-\gamma_n}  \label{eq:C_ext_est1}\\
                        & \leq \|\prlt_n'\|_{\infty}^2 \leq \bw^2 \notag.
\end{align}    
For $\lambda_n < 0$ we use (\ref{eq:bound_T_extra}) to estimate $\prlt_n(T)^2$ in  (\ref{eq:C_ext_est1}),
\begin{align*}
    C_{\text{extra},n}^2 &= \|\prlt_n'\|_{\infty}^2 - \frac{\lambda_n}{T} \frac{\prlt_n(T)^2}{1-\gamma_n} \leq \|\prlt_n'\|_{\infty}^2 - \frac{\lambda_n}{T}\left(C_n - \frac{\lambda_n}{2T}\right) \\
                         &  = \|\prlt_n'\|_{\infty}^2 + \frac{\lambda_n^2}{2T^2} - \frac{\lambda_n}{T} C_n = C_n^2 - \frac{\lambda_n}{T} C_n + \frac{\lambda_n^2}{4T^2} \\
                         & = \left(C_n - \frac{\lambda_n}{2T}\right)^2.
\end{align*}
For $C_{\text{intra},n}$ we estimate in the case $\lambda_n <0$,
\begin{align*}
    C_{\text{intra},n}^2 & = \|\prlt_n'\|_{\infty}^2 + \frac{\lambda_n}{T} \frac{\prlt_n(T)^2}{\gamma_n} \\
                         & < \|\prlt_n'\|_{\infty}^2.
\end{align*}
For $\lambda_n \geq 0$ we apply (\ref{eq:bound_T_intra}) to obtain,
\begin{align*}
    C_{\text{intra},n}^2 & = \|\prlt_n'\|_{\infty}^2 + \frac{\lambda_n}{T} \left(C_n + \frac{\lambda_n}{2T}\right) = \|\prlt_n'\|_{\infty}^2 + \frac{\lambda_n^2}{4 T^2} +  C_n \frac{\lambda_n}{T} + \frac{\lambda_n^2}{4 T^2}\\
                         & = C_n^2 + \frac{\lambda_n}{T} C_n + \frac{\lambda_n^2}{4 T^2} = \left(C_n + \frac{\lambda_n}{2T}\right)^2.
\end{align*}

\noindent In summary,
\begin{itemize}
    \item For $\lambda_n < 0$  $$C_{\text{extra},n} \leq \left(C_n - \frac{\lambda_n}{2T}\right) \quad \text{ and } \quad C_{\text{intra},n} < \|\prlt_n'\|_{\infty}.$$
    \item For $\lambda_n \geq 0$ $$C_{\text{extra},n} \leq \|\prlt_n'\|_\infty \quad \text{ and } \quad C_{\text{intra},n} \leq \left(C_n + \frac{\lambda_n}{2T}\right).$$
\end{itemize}

\subsection{Asymptotic Estimate for Supremum Outside of the Concentration Region}

\label{sec:asymptotic_bounds}

\noindent We can further estimate $C_{\text{extra},n}$ for $\lambda_n \geq 0$ to obtain an asymptotic bounds in terms of $c = \bw T$. 
We recall  $-c^2 < \lambda_0 < \lambda_1 < \ldots \to \infty$ and make a simple estimate, 
\begin{align*}
    C_{\text{extra},n}& \leq \sqrt{ \|\prlt_n'\|^2_T + \frac{\lambda_n^2}{4T^2}} - \frac{\lambda_n}{2T}\leq \sqrt{ \bw^2 + \frac{\lambda_0^2}{4T^2}} - \frac{\lambda_0}{2T}\\
    & \leq \sqrt{ \bw^2 + \frac{\bw^4 T^2}{4}} + \frac{\bw^2 T }{2} = \bw \left( \sqrt{ 1 + \frac{c^2}{4}} + \frac{c}{2} \right) .
\end{align*}

\noindent Applying asymptotic expansion formulas, we obtain the behavior of $C_{\text{extra},n}$ for large and small $c$,
\begin{align}
    \sqrt{ 1 + \frac{c^2}{4}} + \frac{c}{2} \approx \begin{cases}
        c + \frac{1}{c} + O(\frac{1}{c^2}) & \text{ for } c \gg 1\\
        1 + c + O(c^2) & \text{ for } c \ll 1
    \end{cases}. \label{eq:C_extra_asympt}
\end{align}
More accurate asymptotic expression can be derived from the asymptotic behavior of the eigenvalues $\lambda_n$ of the PSWEq \cite{meixner_mathieusche_1954,SlepianAsymp}.\footnote{Recall that we use a different convention 
for the eigenvalues of the PSWEq than in \cite{SlepianAsymp}: $\lambda_n = \chi_n - c^2$.}
Their asymptotic behavior for large $c$ is given by,
\begin{align}
    \lambda_n = c^2\left( -1 + \frac{2n+1}{c} + O(c^{-2})\right).
\end{align}
We find to leading order for large $c$,
\begin{align*}
    C_{\text{extra},n} & < \sqrt{ \bw^2 + \frac{\lambda_n^2}{4T^2}} - \frac{\lambda_n}{2T} \\
    &   = \bw c \left( 1 -  \frac{2n+1}{c} +O(c^{-2}) \right).
\end{align*}
Here we used,
\begin{align*}
    \frac{\lambda_n }{2 T} = \bw \frac{c}{2} \left( -1 +  \frac{2n+1}{c} +O(c^{-2})\right).
\end{align*}
Using the asymptotic formula (\ref{eq:gamma_n_asympt}) for $1-\gamma_n$ we obtain,
\begin{align}
    (1-\gamma_n) C_{\text{extra},n} = \bw \frac{4 \sqrt{\pi} 2^{3 n} c^{n+\frac{3}{2}} e^{-2 c}}{n!}\left[1-\frac{6 n^2+ 62n +35 }{32 c}+O\left(c^{-2}\right)\right], \label{eq:sup_extra_asympt}
\end{align}
for $n \leq \operatorname*{argmax}_n \text{ such that } (\lambda_n <  0)$. \\

The sequence of asymptotic expressions $f_n$ defined (\ref{eq:sup_extra_asympt}), 
is up to the leading-order term increasing for $n \leq 2c$.
A more detailed analysis of the asymptotic behavior of these bounds will be presented in future work.



\newpage 

\section{A Dual of the Transition Signature in the Spectrum of the PSWEq}
\label{sec:transition_signature}

The spectrum of the prolate integral operator $\BL_\bw \TL_T \prlt_n = \gamma_n \prlt_n$
exhibits a spectral transition.  
Let $\tilde c = 2\bw T / \pi$ be the characterizing phase space volume of prolates, which are concentrated in the interval $[-T,T]$  and consist of frequencies within 
$[-\bw,\bw]$.\footnote{ Here $\tilde c$ has the interpretation of the volume of a rectangle in phase space, $\tilde c = 2\bw \times 2T / 2 \pi$.  }
For $n \ll \tilde c $ the prolate eigenvalues $\gamma_n$ are nearly equal $1$ and for $n \gg \tilde c $ they are very close to $0$.
The transition region in between is logarithmically narrow,
\begin{align}
    \#\left\{k \mid \gamma \geq \gamma_k(\tilde c) \geq 1-\gamma\right\} \leq \frac{A}{\gamma(1-\gamma)} \log \tilde c,  \label{eq:transition_width}
\end{align}
where $A$ is a constant. Landau has proven a key spectral signature, that 
 precisely locates the transition region at $\tilde c$, 
\begin{align}
    \gamma_{[\tilde c]-1} \geq \frac{1}{2} \geq \gamma_{[\tilde c]+1}.  \label{eq:transition_signature_landau}
\end{align}
Equation (\ref{eq:transition_width}) and (\ref{eq:transition_signature_landau}) are drawn from \cite{landau_density_1993}.
We conjecture that a dual of the transition signature also exists in the spectrum of the PSWEq,
\begin{align}
    \lambda_{[\tilde c]-1} \leq 0 \leq \lambda_{[\tilde c]+1}.  \label{eq:transition_signature_conj}
\end{align}
In following we present an argument supporting the conjecture with a Bohr-Sommerfeld quantization condition.
The Bohr-Sommerfeld quantization condition corresponds to the leading order term of an exact quantization condition, as discussed in \cite{BohrSommerfeldQuant}
and also known from the context of WKB theory \cite{reed_iv_1978}.
We extend our gratitude to Gian Michele Graf for the insightful discussion that informed the development of this argument.

\subsection{Bohr-Sommerfeld Quantization Argument}

Consider the PSWEq,
\begin{align}
 \Pop_{c} =  \frac{d}{dx}\left[( 1-x^2) \frac{d}{dx} \right] + c^2\left(1 -x^2\right).  \label{eq:PSWEq_bohr}
\end{align}
where $\Pop_{c}$ acts on the space of sufficiently regular functions defined on the interval $[-1,1]$.
In Section \ref{sec:PSWEq} and  \ref{sec:generalized_PSWEq} we have seen that prolates are eigenfunctions to the PSWEq with $\Pop_{c} \prlt_n = -\lambda_n \prlt_n$.
We now turn to a Hamiltonian phase space interpretation.\\

Consider the momentum operator $p = -i \hbar \frac{d}{dx}$ and the domain $x \in [-1,1]$. We obtain from $\Pop_{c}$ a Hamiltonian function, 
\begin{align*} 
 H_c(x,p) =  \frac{1}{\hbar^2} p(1-x^2)p - c^2(1-x^2),
\end{align*}
which is defined on the phase space  $(x,p) \in [-1,1] \times \mathbb{R}$.
For an eigenfunction $\Pop_{c} \prlt_n = - \lambda_n \prlt_n$, the Hamiltonian 
takes the value $\lambda_n$ in phase space of $\prlt_n$,
\begin{align*}
 H_c(x,p)  = \lambda_n,
\end{align*}
for all $(x,p) \in \prlt_n$. \\

Within the Bohr-Sommerfeld quantization rule, the eigenvalues 
$\lambda_n$ are approximated through phase space volumes, 
\begin{align*}
 \left| \left\{ (x,p): H_c(x,p) \leq \lambda_n \right\} \right| = 2\pi \hbar n.
\end{align*}
To analyze the critical points of $H_c$, we compute its gradient and Hessian,
\begin{align*}
    \nabla H_c & =  \frac{2}{\hbar^2} \begin{pmatrix}  -xp^2 + c^2 x \\ + p(1-x^2) \end{pmatrix}, \\
    \delta^2 H_c & =  \frac{2}{\hbar^2} \begin{pmatrix}  c^2 - p^2 & -2xp \\ -2xp & 1 - x^2  \end{pmatrix}.
\end{align*}
The stationary points where $\nabla H_c=0$ are given at $(x,p)_g = (0,0)$ and $(x,p)_\text{trans} = (1, c)$. The point
$(x,p)_g$ corresponds to the ground state and has energy $H_c(x,p)_g =  -c^2$. The point $(x,p)_\text{trans}$ 
has energy,
\begin{align*}
 H_c(x,p)_\text{trans}  = \lambda_\text{trans} = 0. 
\end{align*}
For the Hessian we have,
\begin{align*}
    \delta^2 H_c (x,p)_\text{trans} = \frac{2}{\hbar^2} \begin{pmatrix}  0 & -2c^2 \\ -2c^2 & 0  \end{pmatrix} = - \frac{4c^2}{\hbar^2} \sigma_x,
\end{align*}
where $\sigma_x$ is the Pauli matrix. Thus the Hessian is indefinite and $(x,p)_\text{trans}$ is saddle point.\\

To calculate the encircled phase space area of the orbit, we note that 
\begin{align*}
 H_c(p,x) =  \left(\frac{p^2}{\hbar^2}  -c^2 \right) (1-x^2) \leq 0 
\end{align*}
is satisfied if and only if $|p| \leq \hbar c$. The phase space volume is given 
by a rectangle of area $4 \hbar c$. Thus, 
\begin{align*}
 \left| \left\{ (x,p): H_c(x,p) \leq 0 \right\} \right| = 4 \hbar c = 2\pi \hbar n_\text{trans}.
\end{align*}
This implies a saddle point at $n_\text{trans} = \frac{2 c}{\pi}$ and ultimately,
\begin{align*}
 \lambda_{ \lfloor\frac{2 c}{\pi}  \rfloor} \leq 0 \leq \lambda_{ \lceil \frac{2 c}{\pi} \rceil }.
\end{align*}
From Section \ref{sec:generalized_PSWEq}, we recall that the eigenvalues of $\bw T$-prolate coincide with the eigenvalues of the PSWEq (\ref{eq:PSWEq_bohr}) by 
setting $c = \bw T $. Thus, the transition signature in the spectrum of the generalized PSWEq (\ref{eq:transition_signature_conj}) is proven up to 
correction terms from the Bohr-Sommerfeld quantization condition.


\chapter{A Sampling Instance of the Engineering Folk Theorem }
\label{chap:Sampling}

This chapter is at this stage non-expositionary. We merely state the main results in a rough form. \\

From the refined concentration identities of Chapter \ref{chap:prlt_bound} we can 
derive improved error bounds for the truncation of the prolate sampling formula, that has been introduced by Walter and Shen \cite{walter_sampling_2003}.\\

Sampling formulas lead the first generation of pioneers of communication theory, such as Shannon and Nyquist, to the Engineering Folk Theorem \cite{NoiseCommunication, Nyquist1928}.
It states, that the space of signals through which we communicate across a band limiting channel is approximately $2\bw T/\pi $ dimensional.
Shannon's seminal and widely known theorem on the optimal channel capacity was actually only proven up to this intuition \cite{NoiseCommunication}.
Shannon has been led to this intuition through what is now known as the Whittaker-Shannon interpolation formula.
To reconstruct a signal of bandwidth $2 \bw$ and duration $2T$ to high precision, $2\bw T / \pi$ sample points should suffice.
However, they could not provide rigorous proof for this intuition through sharp error estimates. \\

Here we note that this long-standing intuition can be met with mathematical rigour, through the prolate sampling formula and the refined 
error estimates, that we derive in this chapter.\\

\section{Related Work}

\begin{theorem}[Whittaker-Shannon interpolation formula]
    Let $f \in \BL_\bw$. Then we have 

    \begin{alignat}{2}
        f(t)  &= \frac{\pi}{\bw} &&\sum_{n \in \mathbb{Z}} f\left(\frac{n \pi}{\bw}\right) \rho_\bw\left(t-\frac{n \pi}{\bw} \right).  \label{eq:shannonInterpolation}
    \end{alignat}
    The series converges uniformly over $\mathbb{R}$.
    \label{thm:shannonInterpolation}
\end{theorem}

We give a proof for this particular form of the Whittaker-Shannon interpolation in Section \ref{sec:samp_appendix} .
Walter and Shen, have derived through Prolate Fourier Theory an improved interpolation formula that converges faster and is more accurate than the Whittaker-Shannon interpolation formula \cite{walter_sampling_2003}.
After translating into our notation, the following theorem summarizes a selection of key results of their work.

\begin{theoremAlph}[Walter, Shen]
    Band limited functions $f \in \BL_\bw$ have discrete prolate expansions with respect to a $\bw$-Prolate basis $\{ \prlt_n \}$,
    \begin{align}
        f (t) = \frac{\pi}{\bw} \sum_{n=0}^{\infty}  \sum_{k= -\infty}^{\infty} f\left(\frac{k \pi}{\bw}\right) \prlt_n\left(\frac{k \pi}{\bw}\right) \prlt_n(t). \label{eq:ProlateSampling}
    \end{align}
    Prolates $\prlt_n$ satisfy discrete orthogonality relations
    \begin{align*}
        \sum_{n=0}^{\infty} \prlt_n\left(\frac{k \pi}{\bw} \right) \prlt_n \left(\frac{l \pi}{\bw}\right) & = \frac{\bw}{\pi} \delta_{kl} \\
        \sum_{k\in \mathbb{Z}} \prlt_n\left(\frac{k\pi}{\bw}\right) \prlt_m\left(\frac{k\pi}{\bw}\right) & = \frac{\bw}{\pi} \delta_{nm}
    \end{align*}
    for all $k,l \in \mathbb{Z}$ and $n,m \in \mathbb{N}$.
\end{theoremAlph}

\section{Results}

\begin{theorem}
    \label{thm:SamplingEngFolk}
    Consider a band limited function $f \in \BL_\bw$ in the time domain $[-T,T]$. Let $\{\prlt_n\}$ be the sequence of $\bw T $-prolates 
    and $c= T\bw/\pi$. We denote the truncated sampling series $f_{N,c}(t)$ as,
    \begin{align*}
        f_{N,c}(t) = \frac{\pi}{\bw} \sum_{n=0}^{N-1}  \sum_{k= -[T\bw/ \pi]}^{[T\bw/ \pi]} f\left(\frac{k \pi}{\bw}\right) \prlt_n\left(\frac{k \pi}{\bw}\right) \prlt_n(t).
    \end{align*}
    The sampling series $f_{N,c}(t)$ approximates $f$ within the interval $[-T,T]$ with a precision guarantee of,
    \begin{align*}
        \|f -f_{N,c}\|_{T}^2 \leq &  \frac{\pi}{\bw} \left( \|f \|_{>T}^2+ 2 \sqrt{\|f \|_{>T}^2\|f'\|_{>T}^2}\right) \sum_{n=0}^{N-1} \gamma_n(1-\gamma_n) C_{n} \\
        & +  \sum_{n=N}^{\infty}  \gamma_n |a_n|^2.
    \end{align*}
    Here $a_n= \int_{-\infty}^{\infty} \prlt_n(t) f(t) dt $ denotes the $n$-th coefficient of $f$ in a $\bw T$-prolate expansion and $C_n = (1+2 C_{\text{extra},n}^{ \frac{1}{2}})$ with $C_{\text{extra},n}$ as in Theorem \ref{thm:prlt_bound}.
\end{theorem}

The following convergence statement does not seem to be known in the literature.
\begin{proposition}
    The expansion of a band limited function $f \in \BL_\bw$ in a $\bw$-prolate basis converges uniformly over $\mathbb{R}$.
    \label{prop:ProlateExpansionConvergence}
\end{proposition}

\begin{corollary}
    The prolate sampling formula (\ref{eq:ProlateSampling}) converges uniformly over $\mathbb{R}$.
\end{corollary}

\begin{corollary}
    We denote with $\mathscr{P}_{WT}(N)$ the subspace spanned by the first $N+1$ prolates. Let $f \in \mathscr{P}_{WT}(N)$
    and $f_{N,c}$ as in Theorem \ref{thm:SamplingEngFolk}. Then we have
    \begin{align}
        \|f - f_{N+1,c}\|_{T}^2 \leq \frac{\pi \| f\|  }{\bw} (1-\gamma_{N}) C_{N} \sum_{n=0}^{N} \gamma_n(1-\gamma_n) C_{n} \label{eq:prlt_samp_trunc}
    \end{align}
    with $C_n$ as in Theorem \ref{thm:SamplingEngFolk}.\\
    
    In particular,  $2\bw T / \pi$ sample points suffice, to reconstruct a modulation of the first $N < 2\bw T / \pi - A \log (2\bw T / \pi )$ prolates within the
    interval $[-T,T]$ to high precision.
\end{corollary}

Here $A$ is a constant that describes the width of the transition region in spectrum of prolate spheroidal integral operator \cite{ProIII,landau_density_1993}.
For comparison, the truncation estimate of Walter and Shen had
\begin{align}
    \|f - f_{N,c}\|_{T}^2 \sim (1-\gamma_N)^{1/2},
\end{align}
while the estimate (\ref{eq:prlt_samp_trunc}) gives a sharper bound with
\begin{align}
    \|f - f_{N,c}\|_{T}^2 \sim (1-\gamma_N)^{2}.
\end{align}
So we get an improvement of four orders of magnitude in the precision parameter $(1-\gamma_N)$.

\section{Derivations}




\begin{fact}
    We have for all $ y \in \mathbb{R}$
    \begin{align}
        \sum_{n=0}^\infty \prlt_n(y)^2 = \frac{\bw}{\pi}  \label{eq:ProlateTsquareSum}
    \end{align}
    And 
    \begin{align}
        \rho_\bw(x-y) & = \sum_{n =0}^\infty \prlt_n(x) \prlt_n(y) \label{eq:SamplingFunctionExpansion}
    \end{align}
    with uniform and absolute convergence on $\mathbb{R}$, both as function 
    of $x$ or $y$.
    \label{fact:SamplingFunctionExpansion} 
\end{fact}
\begin{proof}[Proof of Fact \ref{fact:SamplingFunctionExpansion}]
    Consider $\rho(\cdot, y): \mathbb{R} \to \mathbb{R}, x \mapsto \rho_\bw(x-y)$. 
    We have $\|\rho(\cdot, y )\|_\infty^2 = \| \rho_\bw\|_\infty^2 = \frac{\bw}{\pi}$.
    Expanding $\rho(\cdot, y)$ in the prolate basis gives 
    \begin{align*}
        \rho(x-y) & = \lim_{N\to \infty} \underbrace{\sum_{n =0}^N \prlt_n(x) \prlt_n(y)}_{:= \rho_N(x-y)},
    \end{align*}
    where we used 
    \begin{align*}
        c_n = \int_{-\infty}^{\infty} \rho(x-y) \prlt_n(x) dx = \prlt_n(y).
    \end{align*}
    By Parseval's identity we have,
    \begin{align*}
        \sum_{n=0}^{\infty} \prlt_n(y)^2  = \frac{\bw}{\pi}.
    \end{align*}
    Similarly we define 
    $f (\cdot,y): [-\bw, \bw] \to \mathbb{C}, f(\omega)=  \frac{1}{2\pi} e^{-i \omega y}$ and expand in 
    the Fourier prolate basis $\tprlt_n$,

    \begin{align*}
        c_n & = \frac{1}{\gamma_n} \frac{1}{2\pi} \int_{-\bw}^{\bw} e^{i \omega y} \tprlt_n(\omega) d \omega = \frac{1}{\gamma_n} \frac{1}{2\pi} \FT_\bw [\tprlt_n](y) = \frac{2\pi }{|\mu_n|^2} \frac{\conj{\mu_n}}{2\pi}\prlt_n(y)\\
        & = \frac{1}{\mu_n} \prlt_n(y).
    \end{align*}
    We thus have with convergence in $\Ls^2_\bw$ (as function of $\omega$)
    \begin{align*}
        f(\omega,y) & = \lim_{N \to \infty}\underbrace{\sum_{n=0}^{N} \frac{1}{\mu_n} \tprlt_n(\omega) \prlt_n(y)}_{:= f_N(\omega,y)}.
    \end{align*}
    To proof uniform convergence we consider
    \begin{align*}
        \rho_\bw(x-y) -\rho_N(x,y) & = \int_{-\bw}^{\bw} \frac{1}{2\pi} e^{i \omega (x-y)} - \sum_{n =0}^{N} \frac{1}{\mu_n} e^{i \omega x} \tprlt_n(\omega) \prlt_n(y) d \omega.\\
            & = \int_{-\bw}^{\bw} \left(\frac{1}{2\pi} e^{-i \omega y} - \sum_{n =0}^{N} \frac{1}{\mu_n} \tprlt_n(\omega) \prlt_n(y) \right) e^{i\omega x} d \omega.\\
            & = \int_{-\bw}^{\bw} \left(f(\omega,y) - f_N(\omega,y) \right) e^{i\omega x} d \omega.
    \end{align*}
    Applying on the last line Cauchy Schwarz inequality we find
    \begin{align*}
        (\rho_\bw(x-y) -\rho_N(x,y))^2 \leq 2\bw \int_{-\bw}^{\bw} \left|f(\omega,y) - f_N(\omega,y) \right|^2 d \omega,
    \end{align*}
    where we used $|e^{i\omega x}| =1$. The right-hand side is already independent of $x$ 
    and approaches zero as $N \to \infty$ due to the  $\Ls^2_\bw$ convergence of $f_N(\cdot,y)$.
    By symmetry the same argument with respect to $y$ holds.
\end{proof}

\begin{fact}
    \label{fact:prlt_discrete_ortho}
    Prolates $\prlt_n$ satisfy discrete orthogonality relations
    \begin{align*}
        \sum_{n=0}^{\infty} \prlt_n\left(\frac{k \pi}{\bw} \right) \prlt_n \left(\frac{l \pi}{\bw}\right) & = \frac{\bw}{\pi} \delta_{kl} \\
        \sum_{k\in \mathbb{Z}} \prlt_n\left(\frac{k\pi}{\bw}\right) \prlt_m\left(\frac{k\pi}{\bw}\right) & = \frac{\bw}{\pi} \delta_{nm}
    \end{align*}
    for all $k,l \in \mathbb{Z}$ and $n,m \in \mathbb{N}$.
\end{fact}
We have $\rho_\bw(\frac{k \pi}{\bw}) = 0 $ for all $k \in \mathbb{Z} \setminus \{0\}$ and  $\rho_\bw(0) = \frac{\bw}{\pi}$.
From this and eq.(\ref{eq:SamplingFunctionExpansion}) we get for all $k,l \in \mathbb{Z}$
$$\sum_{n=0}^{\infty} \prlt_n\left(\frac{k \pi}{\bw} \right) \prlt_n \left(\frac{l \pi}{\bw}\right) = \frac{\bw}{\pi} \delta_{kl}. $$
Expanding $\prlt_n(t)$ in the sampling functions basis $ \sqrt{\frac{\pi}{\bw}} \rho_\bw(t - \frac{n\pi}{\bw})$ gives
with 
\begin{align*}
    \int_{-\infty}^{\infty} \prlt_n(t)\rho_\bw(t - \frac{n\pi}{\bw}) dt = \prlt_n\left(\frac{n\pi}{\bw}\right),
\end{align*}
\begin{align*}
    \prlt_n(t) = \frac{\pi}{\bw} \sum_{k\in \mathbb{Z}} \prlt_n\left(\frac{k\pi}{\bw}\right) \rho_\bw(t - \frac{k\pi}{\bw}).
\end{align*}
With Parseval's identity this implies 
\begin{align*}
    \int_{-\infty}^{\infty } \prlt_n(t) \prlt_m(t) dt  & = \frac{\pi}{\bw} \sum_{k\in \mathbb{Z}} \prlt_n\left(\frac{k\pi}{\bw}\right) \prlt_m\left(\frac{k\pi}{\bw}\right)
    = \delta_{nm}.          
\end{align*}

\begin{proposition}
    Consider $f \in \BL_\bw$ and the expansion in the Prolate basis
    \begin{align*}
        f(t) = \sum_{n=0}^{\infty} a_n \prlt_n(t),  \quad \text{ with } \quad a_n = \int_{-\infty}^{\infty} f(t) \prlt_n(t) dt.
    \end{align*}
    We have absolute and uniform convergence of the series in $\mathbb{R}$.
    \label{prop:ProlateExpansionConvergence}
\end{proposition}
\begin{proof}
    Absolute convergence follows from Parseval's identity and eq.(\ref{eq:ProlateTsquareSum}),
    \begin{align*}
       \sum_{n=0}^{\infty} |a_n \prlt_n(t)| \leq \sqrt{\sum_{n=0}^{\infty} |a_n|^2 \sum_{n=0}^{\infty} \prlt_n(t)^2} = \|f\| \sqrt{\frac{\bw}{\pi}} < \infty.
    \end{align*}
    Unifrom convergence:
    Recall that we have $\prlt_n(t) = \frac{2\pi}{\mu_n} \FT^{-1}_\bw[\tprlt_n](t)$ and 
    $\FT[\prlt_n](\omega) = \frac{2\pi}{\mu_n} \tprlt_n(\omega)$. We denote $\FT[f](\omega) = F(\omega)$ and write the expansion of $F$ in $\Ls^2_\bw$ with respect to the Fourier Prolate basis $\tprlt_n$,
    \begin{align*}
        F =  \lim_{N\to \infty} \underbrace{\sum_{n=0}^{N} \frac{b_n}{\gamma_n} \tprlt_n}_{:=F_N}, \quad \text{ with } \quad b_n = \int_{-\bw}^{\bw} F(\omega) \tprlt_n(\omega) d\omega,
    \end{align*}
    where the $\gamma_n$ factor comes from the normalization on $[-\bw, \bw]$.
    From Plancherel's theorem we get 
    \begin{align*}
        a_n & \Pl \frac{1}{2\pi} \int_{-\bw}^{\bw}  \conj{\FT[\prlt_n](\omega)}F(\omega)  d\omega = \frac{1}{2\pi} \int_{-\bw}^{\bw} F(\omega) \frac{2\pi}{\conj{\mu_n}} \tprlt_n(\omega) d\omega \\
        & = \frac{1}{\conj{\mu_n}} b_n.
    \end{align*}
    With this we have
    \begin{align*}
        f(t) - f_N(t) & = \frac{1}{2\pi} \int_{-\bw}^{\bw} F(\omega) e^{i \omega t} - \sum_{n=0}^{N} a_n \frac{2\pi}{\mu_n}\tprlt_n(\omega)e^{i\omega t}  d\omega \\
        & = \frac{1}{2\pi} \int_{-\bw}^{\bw} (F(\omega) - \sum_{n=0}^{N}  \frac{2\pi}{|\mu_n|^2 } b_n\tprlt_n(\omega)) e^{i \omega t} d \omega.
    \end{align*}
    We recognize $F_N(\omega) = \sum_{n=0}^{N} \frac{b_n}{\gamma_n} \tprlt_n(\omega)$ in the integrand and can apply Cauchy Schwarz inequality to find
    \begin{align*}
        |f(t) - f_N(t)|^2 \leq \frac{2\bw}{2\pi} \int_{-\bw}^{\bw} |F(\omega) - F_N(\omega)|^2 d\omega.
    \end{align*}
    Here we used $|e^{it\omega}|=1$ for $t\in \mathbb{R}$. The right hand side approaches zero as $N \to \infty$ due to the $\Ls^2_\bw$ convergence of $F_N$
    and is independent of $t$.
\end{proof}
Similarly the proof generalizes to show uniform convergence on any subspace $D \subset \mathbb{C}$ 
that ensures that $\operatorname*{imag}(z)$ is bounded for all $z \in D$.\\

Inserting the prolate expansion of $\rho_\bw$ in the Whittaker-Shannon interpolation formula
gives the prolate interpolation formula
\begin{align*}
    f(t) = \frac{\pi}{\bw} \sum_{k\in \mathbb{Z}} \sum_{n=0}^{\infty} f\left(\frac{k \pi}{\bw}\right) \prlt_n\left(\frac{k \pi}{\bw}\right) \prlt_n(t),
\end{align*}
that has similarly guaranteed absolute and uniform convergence over $\mathbb{R}$.
The advantage over the Whittaker-Shannon formula is that the optimal concentration 
properties of the prolates basis ensure fast convergence 
and allow for precise error estimates in truncations.  However, the essential 
time concentration of prolates is only given for prolates before the 
transition region $n < \frac{2 \bw T}{\pi}$. Fortunately, a result of 
Pollak and Landau states that bandlimitted functions that are also 
time concentrated in the interval $[-T,T]$ can be well approximated in 
the subspace of first $N$ prolates \cite{ProIII}.

\subsubsection{Error estimates with respect to L2 norm}
These are especially interesting for numerical integration. 
Denote $\|f \|_{>T}^2 = \int_{|t|>T} |f(t)|^2 dt=\|f \|_{\infty}^2 - \|f \|_{T}^2$

\begin{lemma}
    Let $f: \mathbb{R} \to \mathbb{C}$ be differentiable such that $f,f' \in \mathscr{L}^2$ and 
    $T,\bw >0$. We denote $c = \frac{T\bw}{\pi}$ and have 
    \begin{align*}
        \sum_{|k|\geq c} \left|f \left(\frac{\pi k}{\bw} \right) \right|^2 &\leq \|f \|_{>T}^2+ 2 \sqrt{\|f \|_{>T}^2\|f'\|_{>T}^2}, \\
        \sum_{|k| \leq c} \left|f \left(\frac{\pi k}{\bw} \right) \right|^2 & \leq \|f \|_{T}^2+ 2 \sqrt{\|f \|_{T}^2\|f'\|_{T}^2}.
    \end{align*}
    \label{lem:ErrorEstimatesTrucation}
\end{lemma}
\begin{proof}[Proof of Lemma \ref{lem:ErrorEstimatesTrucation}]
    We compare $ \sum_{|k|>[c]} |f(\frac{\pi k}{\bw})|^2$ to $|\int_{|x|>T} f(x)^2 dx|$. 
    From the integral mean value theorem we have $\int_{\frac{\pi k}{\bw}}^{\frac{\pi (k+1)}{\bw}} f(x)^2 dx =   f(x_k)^2$
    with $\frac{\pi k}{\bw} \leq x_k \leq \frac{\pi (k+1)}{\bw}$. Thus,
    \begin{align*}
        \left|  \left|\int_{\frac{\pi k}{\bw}}^{\frac{\pi (k+1)}{\bw}} f(x)^2 dx\right| -  |f(\frac{\pi k}{\bw})|^2 \right| & = \left| |f(x_k)|^2 -  |f(\frac{\pi k}{\bw})|^2  \right|  \leq \left| f(x_k)^2 -  f(\frac{\pi k}{\bw})^2  \right| \\
        & = \left| 2 \int_{\frac{\pi k}{\bw}}^{x_k} f(x) f'(x) dx \right| \leq 2 \int_{\frac{\pi k}{\bw}}^{x_k} |f(x) f'(x)| dx \\
        &  \leq 2 \int_{\frac{\pi k}{\bw}}^{\frac{\pi (k+1)}{\bw}} |f(x) f'(x)| dx,
    \end{align*}
    where I applied the reversed triangle inequality and the fundamental theorem of calculus.
    Lower bounding the left-hand side and rearranging gives
    \begin{align*}
        |f(\frac{\pi k}{\bw})|^2 & \leq \left|\int_{\frac{\pi k}{\bw}}^{\frac{\pi (k+1)}{\bw}} f(x)^2 dx\right| + 2 \int_{\frac{\pi k}{\bw}}^{\frac{\pi (k+1)}{\bw}} |f(x) f'(x)| dx.
    \end{align*}
    Here we assumed $k>0$, for $k<0$ we make an analogous bound with $\int_{\frac{\pi (k-1)}{\bw}}^{\frac{\pi k}{\bw}}$.
    Summing over $k$ gives
    \begin{align*}
        \sum_{|k| \geq c} |f(\frac{\pi k}{\bw})|^2  & \leq \int_{|x|>T} |f(x)|^2 dx + 2 \int_{|x|>T} |f(x) f'(x)| dx \\
        & \leq \|f \|_{>T}^2 + 2\sqrt{\|f \|_{>T}^2 \|f'\|_{>T}^2}.
    \end{align*}
    And similarly,
    \begin{align*}
        \sum_{|k|\leq c} |f(\frac{\pi k}{\bw})|^2   & \leq \int_{|x|<T} |f(x)|^2 dx + 2 \int_{|x|<T} |f(x) f'(x)| dx \\
        & \leq \|f \|_{T}^2 + 2\sqrt{\|f \|_{T}^2 \|f'\|_{T}^2}.
    \end{align*}
\end{proof}

\begin{proposition}
    Consider $f \in \BL_\bw$ and an interval in the time domain $[-T,T]$. Let $c= T\bw$ and 
    \begin{align*}
        f_{N,c}(t) = \frac{\pi}{\bw} \sum_{n=0}^{N-1}  \sum_{k= -[c]}^{[c]} f\left(\frac{k \pi}{\bw}\right) \prlt_n\left(\frac{k \pi}{\bw}\right) \prlt_n(t).
    \end{align*}
    Then we have 
    \begin{align}
        \|f -f_{N,c}\|_{T}^2 \leq \sum_{n=N}^{\infty} \gamma_n |a_n|^2 + \frac{\pi}{\bw} \sum_{n=0}^{N-1} \gamma_n(1-\gamma_n) \sum_{|k|>[c]} f\left(\frac{k \pi}{\bw}\right)^2. \label{eq:ErrorEstimatesTrucation}
    \end{align}
\end{proposition}

\begin{proof}
    Consider the auxiliary function
    \begin{align*}
        f_N(t) = \frac{\pi}{\bw} \sum_{n=0}^{N-1}  \sum_{k = -\infty}^\infty f\left(\frac{k \pi}{\bw}\right) \prlt_n\left(\frac{k \pi}{\bw}\right) \prlt_n(t).
    \end{align*}
    We have 
    \begin{align*}
        \|f -f_{N,c}\|_{T}^2 & = \|(f -f_N) - (f_N -f_{N,c})\|_{T}^2  = \|f -f_N\|_{T}^2 + \|f_N -f_{N,c}\|_{T}^2.
    \end{align*}
    The last equality follows from the fact that $f -f_N$ and $f_N -f_{N,c}$ are orthogonal in $\TLs_T$.
    The bound on the first term on the righ-hand side is obtained from the discrete orthogonality relations of Prolates.
    \begin{align*}
        \int_{-T}^{T}  \left| f(t) - f_N(t) \right|^2 dt & = \int_{-T}^{T} \left| \frac{\pi}{\bw} \sum_{n=N}^{\infty}\sum_{k\in \mathbb{Z}}  f\left(\frac{k \pi}{\bw}\right) \prlt_n\left(\frac{k \pi}{\bw}\right) \prlt_n(t) \right|^2 dt\\
        & =  \int_{-T}^{T} \left| \frac{\pi}{\bw} \sum_{n=N}^{\infty} \sum_{k\in \mathbb{Z}}   \sum_{m} a_m \prlt_m\left(\frac{k \pi}{\bw}\right) \prlt_n\left(\frac{k \pi}{\bw}\right) \prlt_n(t) \right|^2 dt\\
        & = \int_{-T}^{T} \left| \frac{\pi}{\bw}  \sum_{n=N}^{\infty}  \sum_{m} a_m   \delta_{nm} \frac{\bw}{\pi} \prlt_n(t) \right|^2 dt\\
        & = \int_{-T}^{T} \left|  \sum_{n=N}^{\infty}  a_n  \prlt_n(t) \right|^2 dt  =  \sum_{n=N}^{\infty}  |a_n|^2   \gamma_n
    \end{align*}
    In the third line we applied the discrete orthogonality relations of prolates given in Fact \ref{fact:prlt_discrete_ortho}.
    For the second term we have, 
    \begin{align*}
        \int_{-T}^{T}  \left| f_N(t) - f_{N,c}(t) \right|^2 dt & = \int_{-T}^{T} \left| \frac{\pi}{\bw} \sum_{|k|> [c]} \sum_{n=0}^{N-1} f\left(\frac{k \pi}{\bw}\right) \prlt_n\left(\frac{k \pi}{\bw}\right) \prlt_n(t) \right|^2 dt\\
        & =  \sum_{n=0}^{N-1} \gamma_n \left( \frac{\pi}{\bw} \sum_{|k|> [c]}  f\left(\frac{k \pi}{\bw}\right) \prlt_n \left(\frac{k \pi}{\bw}\right) \right)^2\\
        & = \frac{\pi^2}{\bw^2}   \sum_{n=0}^{N-1} \gamma_n \left( \sum_{|k|> [c]} \sum_{m=0}^\infty a_m  \prlt_m\left(\frac{k \pi}{\bw}\right) \prlt_n \left(\frac{k \pi}{\bw}\right) \right)^2\\
        & \leq \frac{\pi^2}{\bw^2} \sum_{n=0}^{N-1} \gamma_n  \sum_{|k|> [c]}  f\left(\frac{k \pi}{\bw}\right)^2 \sum_{|k|> [c]} \prlt_n \left(\frac{k \pi}{\bw}\right)^2.
    \end{align*}
    Lemma \ref{lem:ErrorEstimatesTrucation} and Theorem \ref{thm:prlt_bound} have together

    \begin{align*}
        \sum_{|k|> [c]} \prlt_n \left(\frac{k \pi}{\bw}\right)^2 & \leq  \|\prlt_n  \|_{>T}^2+ 2 \sqrt{\|\prlt_n  \|_{>T}^2\|\prlt_n '\|_{>T}^2} \\
        & \leq (1-\gamma_n)( 1 + 2 C_{\text{extra},n}^{\frac{1}{2}})
    \end{align*}
    Thus, (\ref{eq:ErrorEstimatesTrucation}) is established.
\end{proof}
\newpage 

\section{Appendix}

\label{sec:samp_appendix}

\begin{theorem}[Whittaker-Shannon interpolation formula]
    Let $f \in \BL_\bw$. Then we have 

    \begin{alignat}{2}
        f(t)  &= \frac{\pi}{\bw} &&\sum_{n \in \mathbb{Z}} f\left(\frac{n \pi}{\bw}\right) \rho_\bw\left(t-\frac{n \pi}{\bw} \right).  \label{eq:shannonInterpolation}
    \end{alignat}
    The series converges absolutely and uniformly over $\mathbb{R}$.
    \label{thm:shannonInterpolation}
\end{theorem}
With  
\begin{align*}
    \frac{\pi}{\bw} \rho_\bw\left(t-\frac{n \pi}{\bw}\right) & = \frac{\sin( \bw t - n \pi )}{ \bw t - n \pi}
\end{align*}
we obtain the expression that is more common in the literature.\footnote{ In the
literature band limited functions 
are also often treated with the convention $\bw = 2\pi W$. Then eq.\ref{eq:shannonInterpolation} reads 
\begin{align*}
    \sum_{n \in \mathbb{Z}} f\left(\frac{n \pi}{\bw}\right) \frac{\sin( \bw t - n \pi )}{  \bw t - n \pi} \longrightarrow  \sum_{n \in \mathbb{Z}} f\left(\frac{n }{2 W}\right) \frac{\sin( \pi( 2Wt - n))}{ \pi( 2Wt - n) },
\end{align*}
which coincides with the formula given in Shannon's original paper.} Note that 
Theorem \ref{thm:shannonInterpolation} ultimately 
also implies that $\left\{ \sqrt{\frac{\pi}{\bw}} \rho_\bw(t- \frac{n\pi}{\bw}) : n \in \mathbb{Z} \right\}$ 
is a complete orthonormal basis of $\BL_\bw$. This also implies by Parseval's indentity that 
\begin{align*}
    \| f \|_{\infty}^2 = \frac{\bw}{\pi}\sum_{n \in \mathbb{Z}} \left|f\left(\frac{n \pi}{\bw}\right) \right|^2.
\end{align*}

\begin{proof}[\textit{Proof of Theorem \ref{thm:shannonInterpolation}}]
    We consider the Fourier series of the complex exponential $\omega \mapsto e^{i \omega x}$
    on $[-\bw, \bw]$,
    \begin{align*}
        \frac{1}{2\pi}e^{i x \omega} & = \frac{1}{2\bw} \sum_{n \in \mathbb{Z}} \rho_\bw \left(x-\frac{n \pi}{\bw} \right) e^{i n \pi \omega / \bw}
    \end{align*}
    where we used 
    \begin{align*}
        \frac{1}{2\pi} \int_{-\bw }^{\bw} e^{i x \omega} e^{- i n \pi \omega / \bw} d \bw = \rho_\bw \left(x-\frac{n \pi}{\bw} \right).
    \end{align*}
    With this we obtain after exchange of sum and integral
    \begin{align*}
        f(t) & =  \int_{-\bw}^{\bw} \frac{1}{2\pi} e^{i t \omega} \FT[f](\omega ) d \omega \\
            & = \frac{1}{2\bw} \sum_{n \in \mathbb{Z}} \rho_\bw\left(t-\frac{n \pi}{\bw}\right)  \int_{-\bw}^{\bw} \FT[f](\omega) e^{i n \pi \omega / \bw} d \omega \\
            & = \frac{2 \pi}{2\bw} \sum_{n \in \mathbb{Z}} \rho_\bw\left(t-\frac{n \pi}{\bw}\right)  f\left(\frac{n \pi}{\bw}\right).
    \end{align*}
    Since $f$ and $\rho_\bw$ are both square integrable the last series expression converges absolutely and exchange of sum and integral is justified by Fubini Tonelli theorem.\\
    To proof unifrom convergence we denote $f_N(t) = \frac{\pi}{\bw} \sum\limits_{n =-N}^N \rho_\bw\left(t-\frac{n \pi}{\bw}\right)  f\left(\frac{n \pi}{\bw}\right)$
    \begin{align*}
        f(t) - f_N(t) & = \frac{1}{2\pi} \int_{-\bw}^{\bw}  e^{i t \omega} \FT[f](\omega) - \frac{\pi}{\bw} \sum\limits_{n =-N}^N  \left(\frac{n \pi}{\bw}\right) e^{i \omega(t - \frac{n \pi}{\bw})} d \omega \\
        & = \frac{1}{2\pi} \int_{-\bw}^{\bw} ( \FT[f](\bw) - \underbrace{\frac{\pi}{\bw}\sum\limits_{n =-N}^N f\left(\frac{n \pi}{\bw}\right) e^{- i \omega \frac{n \pi}{\bw}}}_{:= F_N^*(\omega)} ) e^{i t \omega}  d \omega .
    \end{align*}
    Note that $F_N^*(\omega)$ is the Fourier series of $\FT[f](\omega)$ on $[-\bw, \bw]$. Applying Cauchy Schwarz inequality we find
    \begin{align*}
        (f(t) - f_N(t))^2 \leq \frac{2\bw}{2\pi} \int_{-\bw}^{\bw} | \FT[f](\omega) - F_N^*(\omega)|^2 d \omega.
    \end{align*}
    The right-hand side is independent of $t$ and approaches zero as $N \to \infty$, since the Fourier series of a square integrable function 
    converges with respect to the $L^2$ norm.
\end{proof}

\chapter{Dimension Reduction for Spectral Analysis: The Precision of Subspace Eigenvalue Problems}
\label{chap:DimRedSepAna}



\section{Dimension Reduction through Subspace Eigenvalue Problems}

\subsection{Introduction}

Eigenvalue problems are prevalent across natural and information sciences, where many
questions are expressed in this form. This can indeed be seen as a reason for the "unreasonable effectiveness of mathematics in the natural sciences" that Eugene Wigner expressed
in his famous essay.\\

However, a fast increase in computational dimensionality with the system size often limits our ability to
apply our physical understanding to larger systems and to unveil more complex phenomena.
On the other hand, algebraic eigenvalue problems are particularly lucrative from a computational point of view due to well-established diagonalization algorithms. 
And for most applications, it is rarely the full spectrum that is of interest for the 
overarching goal, but a smaller number of eigenvalues, that is independent of the size of the system. A 
quantum chemist learns the most about a molecule from the few lowest eigenvalues of its Hamiltonian, rather than the full spectrum.\\

Therefore projection-based techniques, that map high dimensional eigenvalue problems to lower dimensional 
Generalized Eigenvalue Problems (GEP) of smaller size, have emerged to counteract the curse of dimensionality and
are already common practice in the numerical literature of quantum chemistry and other computational fields \cite{Neuhauser1990BoundSE, Mandelshtam, Baiardi_2021, barshan_supervised_2011}.\\

Despite their cross-disciplinary significance, GEPs have only received little mathematical attention. And if 
so, mostly in an attempt to obtain analogous classical results from perturbation theory in the setting of a GEP \cite{stewart_perturbation_1978, elsner_perturbation_1982}.
In particular, spectral inequalities for GEPs that are tailored to their applications in numerical methods are elusive in the literature.
Here we will present a mathematical framework for GEPs as dimension reduction schemes. We derive spectral inequalities that bound 
the deviation of the eigenvalues of a GEP to a target operator that is of higher dimension. 
We aim to facilitate rigorous derivations for dimension reduction techniques on spectral analysis.
The derived bounds will be applied in the Prolate Filter Diagonalization Method which is subject to the next Chapter \ref{chap:FilterDiag}.\\

We briefly sketch the intuition of the method and how it is applied in the numerical literature. 

\subsection{Intuition of GEPs for Numerical Spectral Analysis}
\label{sec:intuition_GEP}

Projection-based GEPs are already common practice in the numerical literature of quantum chemistry.
Here a spectral analysis of the electronic structure of a large quantum system with Hamiltonian $\matr{H}$ is desired.
Rather than the full spectrum or only the ground state energy, one asks for the eigenvalues of $\matr{H}$ in a certain interval $(E_a,E_b)$.\\

GEP-based methods, rely on some other numerical subroutine or assume some prior spectral information to generate a set of guess vectors $v_1, \cdots, v_k$.
The guess vectors aim to span an spectral subspace of $\matr{H}$, that corresponds to eigenvectors $\varphi_i$ with eigenvalues in the energy range of interest,
\begin{align} 
 \operatorname{span}\{v_1, \cdots, v_k \} \approx \operatorname{span} \{ \varphi_1, \cdots, \varphi_m \}. \label{eq:approx_span}
\end{align}
A guess vector matrix $\matr{V} = [v_1, \cdots, v_k]$ is then used to map the originally high dimensional eigenvalue problem into a GEP of smaller 
dimension, 
\begin{align}
 \matr{V}^\dagger \matr{H} \matr{V} b = \mu \matr{V}^\dagger \matr{V} b. \label{eq:gen_eig_intro} 
\end{align}
The eigenvalues $\mu$ of the generalized eigenvalue problem are then said to approximate the eigenvalues $\lambda_i$ of $\matr{H}$ in the interval $[E_a,E_b]$.\\

So far the exposition is more based on intuition rather than mathematical rigour, but 
forms a promising and flexible direction towards dimension reduction. If the guess vectors were to perfectly span a spectral subspace of the operator $\matr{H}$,
the eigenvalues of the GEP would coincide with the eigenvalues of the original operator. However, 
in practice, the guess subspace will not perfectly agree with the spectral subspace of the operator. This is also where interesting
numerical opportunities for relaxation arise. Although intuitive, the exact meaning of 
the approximation in (\ref{eq:approx_span}) is, so far, mathematically not clear.\\ 

To enable rigorous dimension reduction and approximation guarantees in the context of spectral analysis, a projection-based spectral theory 
is developed in this chapter. In the following, we illustrate how the machinery can be applied to give precision guarantees for GEP-based protocols.

\subsection{Result Highlights}
\label{sec:ResHigh}
We assume a Hermitian operator $\matr{H}$ on a Hilbert space $\Hil$ whose spectrum is subject to study. We call $\mathcal{E} \subset \Hil$
a \emph{spectral subspace} of $\matr{H}$ if it is spanned by eigenvectors of $\matr{H}$.
We make the following assumptions on the guess vector generating protocol $\mathfrak{P}$.

\begin{definition}[$\varepsilon$-Dimension Reduction Protocols]
    \label{asump:protocol}
    \noindent 
 Given a Hermitian operator $\matr{H}$ and a target spectral subspace $\mathcal{E}$, a protocol $\mathfrak{P}$ produces a set of guess vectors $v_i \in \Hil$. For $M$ generated guess vectors, the guess vector matrix is denoted as $\matr{V}_M = [v_1, v_2, \ldots, v_M]$. A signal-noise decomposition of $\matr{V}_M$ has,
    \begin{align*}
 \matr{V}_M = \matr{B}_M + \matr{N}_M \quad \text{such that} \quad 
 \operatorname{span}(\matr{B}_M) \subseteq \mathcal{E} \quad \text{and} \quad 
 \operatorname{span}(\matr{N}_M) \subseteq \mathcal{E}^\perp.
    \end{align*}     

 The protocol $\mathfrak{P}$ is said to have an \emph{$\varepsilon$-guarantee} if it provides a sequence of noise estimates $\{\varepsilon_{M}\}_{M \in \mathbb{N}}$ with,
    \begin{align}
         \lambda_1(\matr{N}_M^2) < \varepsilon_{M},  \label{eq:eps_guarantee}
    \end{align}
 where $\lambda_1(\matr{N}_M^2)$ denotes the largest eigenvalue of $\matr{N}_M^2$.
\end{definition}

Typically a one specifies the target space $\mathcal{E}$ through an interval $(E_a, E_b)$. Then, $\mathcal{E} = \mathcal{E}_{(E_a,E_b)}$ refers to the subspace
of all eigenvectors of $\matr{H}$ with eigenvalues in $(E_a,E_b)$. 
Given a protocol $\mathfrak{P}$, the following algorithm is suggested to determine eigenvalues of interest.

\begin{algorithm}[H]
\caption{$\varepsilon$-Dimension Reduction for Spectral Analysis}
\label{alg:1}
\text{For $\mathfrak{P}$, $\matr{H}$, $\mathcal{E}$ and $\{\varepsilon_{M}\}$ as in Definition \ref{asump:protocol}.}
\begin{algorithmic}[1] 
    \State Pick $M$ \label{line:pickM}
    \Indent
        \State Obtain form $\mathfrak{P}$ matrices $\matr{V}_M^\dagger \matr{H} {\matr{V}_M}$ and $\matr{V}_M^2$
        \State Compute $\matr{U}$ s.t. $ \matr{U}^\dagger \matr{V}^2_M  \matr{U} = \operatorname{diag}(\lambda_1(\matr{V}_M^2), \cdots , \lambda_M(\matr{V}_M^2))$
        \State Set $m_{\text{detect}} = \min_m \lambda_m(\matr{V}_M^2) \quad \text{s.t.} \quad \lambda_m(\matr{V}_M^2) \geq \varepsilon_{M}$  \Comment{Assumed dimension of $\mathcal{E}$}
        \State if $m=M$
        \Indent
            \State Go back to line \ref{line:pickM} with some $M_{\text{new}} >M$
            \EndIndent
        \State Compute $ \matr{A} = \matr{U}^\dagger \matr{V}_M^\dagger \matr{H} {\matr{V}_M} \matr{U}$
        \State Set $\matr{V}^2_{*} = \operatorname{diag}(\lambda_1(\matr{V}_M^2), \cdots, \lambda_{m_\text{detect}}(\matr{V}_M^2))$ and $ \matr{H}_{\matr{V}_{*}} = [\matr{A}_{ij}]_{i,j=1}^{m_\text{detect}}$
        \State Compute generalised eigenvalues $( \mu_1, \cdots, \mu_{m_{\text{detect}}})$ of $(\matr{H}_{\matr{V}_{*}}, \matr{V}_{*})$ 
    \EndIndent
    \State Return $( \mu_1, \cdots, \mu_{m_{\text{detect}}})$, $\lambda_{m_\text{detect}}(\matr{V}_M^2)$, $M$ 
    \end{algorithmic}
\end{algorithm}

\noindent The eigenvalues $( \mu_1, \cdots, \mu_{m_{\text{detect}}})$ are the obtained approximations 
of the eigenvalues of $\matr{H}$ in the target spectral subspace $\mathcal{E}$. In line 8 of Algorithm \ref{alg:1} the $m_{\text{detect}}$-dimensional sub-matrices of 
$\matr{V}_M^2$ and $\matr{A}$ are taken, that correspond to a well-conditioned guess space. 
In line 3, the eigenvalues of the Gram matrix are in descending order, $\lambda_1(\matr{V}_M^2) \geq \cdots \geq \lambda_M(\matr{V}_M^2)$.\\
The following guarantees on the approximation algorithm can be given.
 
\subsubsection*{Guarantees of Algorithm \ref{alg:1}}
\emph{R1: The determined dimension $m_{\text{detect}}$ is guaranteed to be lower or equal to the true dimensionality of the target spectral subspace $\mathcal{E}$,
\begin{align*}
 m_{\text{detect}} \leq \operatorname{dim}( \mathcal{E}).
\end{align*}}

If $M$ was chosen sufficiently larger than  $m_{\text{detect}}$ and one observes a significant drop in the spectrum of the Gramm matrix,
that is $\lambda_{m_\text{detect}}(\matr{V}_M^2) \gg \lambda_{m_\text{detect}+1}(\matr{V}_M^2)$, one assumes that the guess vectors can sufficiently span the target spectral subspace
and 
\begin{align}
 m_{\text{detect}} = \operatorname{dim}( \mathcal{E}).  \label{eq:dim_detect_condition}
\end{align}
If (\ref{eq:dim_detect_condition}) is true we will write $m = m_{\text{detect}} = \operatorname{dim}( \mathcal{E})$.
Let $\lambda_1 \geq \cdots  \geq \lambda_m$ denote the eigenvalues of $\matr{H}$ within the spectral subspace $\mathcal{E}$.
We can give precision guarantees on the determined eigenvalues $\mu_i$.\\

\noindent \emph{R2: If (\ref{eq:dim_detect_condition}) is true, then 
\begin{align}
 |\mu_i - \lambda_i| & \leq \frac{\lambda_1(\matr{N}_M^2)}{\lambda_{m}(\matr{V}_M^2)} \max_{1, m}\left|\mu_i(\matr{N}_*) - \lambda_i \right|  \label{eq:R2_corr_bound}
\end{align}
for all $i \in [m]$.}\\

The perturbation coefficient $\lambda_1(\matr{N}_M^2)$ can be estimated with $\varepsilon_{M}$.
Here  $\matr{N}_*$ is the noise matrix of an signal noise decomposition of $\matr{V}_{*}$ and $\mu_i(\matr{N}_*)$
denotes the $i$-th largest eigenvalue of the GEP problem corresponding to $(\matr{H}, \matr{N}_*)$ .
If no additional information on the noise $\matr{N}_M$ is available, but the operator $\matr{H}$ is bounded one 
can obtain from (\ref{eq:R2_corr_bound}) a more practical bound.\\

\noindent \emph{R2.1: If $\matr{H}$ is bounded and (\ref{eq:dim_detect_condition}) is true, then 
\begin{align*}
 |\mu_i - \lambda_i| & \leq \frac{\lambda_1(\matr{N}_M^2)}{\lambda_{m}(\matr{V}_M^2)} \max \left\{ \lambda_{\text{max}}(\matr{H})  - \lambda_i, | \lambda_{\text{min}}(\matr{H})  -\lambda_i | \right\} \\
    & \leq \frac{\lambda_1(\matr{N}_M^2)}{\lambda_{m}(\matr{V}_M^2)} \left( \lambda_{\text{max}}(\matr{H})  -\lambda_{\text{min}}(\matr{H}) \right).
\end{align*}
}

For approximations of high but finite-dimensional matrices, the bound of \emph{R2.1} can still result in good approximation guarantees. 
However, this bound does not leverage all information on the guess vectors, that a protocol may provide. In addition, Hamiltonian
operators of physical systems are unbounded. \\

To address these issues, we provide an integrated spectral inequality that makes use of a \emph{spectral measure}. 
Allowing for a quantum mechanical operator, with both discrete ($\sigma_d$) and continuous spectra ($\sigma_c$) we assume for 
the guess vectors an expansion
\begin{align*}
 |v_l\rangle & = \sum_{\lambda_k \in \sigma_d } \alpha_{lk} |\varphi_k\rangle + \int_{\sigma_c} \int_{\Gamma_\lambda }\alpha_l(k) |\varphi_k\rangle dk   d\lambda.
\end{align*}
The spectral measure specific to a GEP $(\matr{H}, \matr{V})$ is defined as 
\begin{align}
 |\alpha( E )|^2 = \sum_{\lambda_k \in \sigma_d}  \delta(E - \lambda_k ) \sum_{l =1}^M  |a_{lk}|^2 + \int_{ \sigma_c} \delta(E - \lambda ) \int_{\Gamma_\lambda} \sum_{l =1}^M  |a_l(k)|^2 dk d\lambda. \label{eq:spec_meas_def}
\end{align}

\noindent On the common example $\mathcal{E} = \mathcal{E}_{(E_a,E_b)}$ the integrated spectral inequalities give the following precision guarantees.\\

\noindent \emph{R3: If $\mathcal{E} = \mathcal{E}_{(E_a,E_b)}$, (\ref{eq:dim_detect_condition}) is true and $\lambda_m(\matr{V}_M^2) \geq (m+1) \lambda_1(\matr{N}^2_M)$, then 
\begin{align}
 \frac{\int_{-\infty}^{E_a}\left(E -  \lambda_i\right)|\alpha(E)|^2 dE}{\lambda_{m}(\matr{V}_M^2) - \lambda_1(\matr{N}^2_M)} \leq \mu_i - \lambda_i & \leq \frac{\int_{E_b}^{\infty} \left(E - \lambda_i \right) |\alpha(E)|^2 dE }{\lambda_{m}(\matr{V}_M^2)- \lambda_{1}(\matr{N}^2_M)}. \label{eq:R3_int_bound}
\end{align}
The tails of the spectral measure have, \begin{align*}
     \int_{-\infty}^{E_a} |\alpha(E)|^2 dE + \int_{E_b}^{\infty} |\alpha(E)|^2 dE = \operatorname{Tr}[\matr{N}^2_*]\leq m \lambda_1(\matr{N}^2_M).
\end{align*} }

Intuition has, that a protocol, which aims to generate guess vectors for a spectral subspace $\mathcal{E} = \mathcal{E}_{(E_a, E_b)}$,
will also generate a spectral measure concentrated around the target interval $(E_a,E_b)$. If a protocol can capture this 
intuition with a bounding envelope on the spectral measure, then the integrated spectral inequalities of \emph{R3} can be applied to give
a sharp approximation bound.

\subsection{Overview of the Detailed Exposition}

In Section \ref{sec:SpecTheoryGEP} we give a pedagogical overview 
on spectral theory for Generalized Eigenvalue Problems as dimension reduction schemes.\\

In Section \ref{sec:SpecIneq} we derive the first volume of spectral inequalities for GEPs 
in the context of dimension reduction. The guarantees $R2$ and $R2.1$ correspond 
to Corollary  \ref{cor:specral_stability_signoise} and Corollary \ref{cor:user_friend}. \\

In Section \ref{sec:IntegratedSpectralInequalities} we significantly build upon the theory of the previous section.
A spectral measure ultimately allows us to derive integrated spectral inequalities.
In contrast to classical matrix inequalities, they can be applied to unbounded operators 
with a mixed spectrum. Guarantee \emph{R3} is a Corollary of Theorem \ref{thm:integrated_spectral_stability}.\\

In Section \ref{sec:NumTheoSimbiosis} we elaborate on implications of the theory for numerical methods.
A numerical practice is suggested that effectively solves the \emph{singularity problem} of the Gram matrix
and enables efficient dimension detection. 
The singularity problem has been identified as a notable issue in the numerical literature, as discussed in sources such as, \cite{GenLDA, Mandelshtam}.
Guarantee \emph{R1} corresponds to Claim \ref{claim:detected_dimension}.

\newpage 

\section{Spectral Theory of GEPs}

\subsection{Foreword on the Formalism and Related Work}

Generalized Eigenvalue Problems have been of interest in the computational literature for a long time. 
However, extensive mathematical expositions for GEPs mostly obscure their relation to a underlying higher-dimensional operator.
The equation of interest is then
\begin{align}
 \matr{A} x = \mu \matr{B} x,  \label{eq:gen_eig_else}
\end{align}
where $\matr{B}$ is some positive definite matrix \cite{FORD1974337, ghojogh2023eigenvalue, stewart_perturbation_1978, GEPQuant}. While 
equation $(\ref{eq:gen_eig_else})$ is the more common approach to represent GEPs in the literature, it has the 
drawback of being agnostic to a hidden underlying operator.
However, this hidden operator $\matr{H}$, is often of central interest in applications of GEPs in quantum chemistry.\\

The pedagogical framework we provide in this Section differs in this aspect.
We will right away consider a GEP as a subspace eigenvalue problem for a higher-dimensional operator,
\begin{align}
 \matr{V}^\dagger \matr{H} \matr{V} x = \mu \matr{V}^\dagger \matr{V} x. \label{eq:gen_eig_here}
\end{align}
From the perspective of quantum chemistry, equation  $(\ref{eq:gen_eig_here})$ is more natural.
But equation $(\ref{eq:gen_eig_here})$ is also of pedagogical value, as the spectral theory for GEPs is essentially a corollary of the familiar spectral theory for Hilbert spaces.\\

An additional advantage of considering  $(\ref{eq:gen_eig_here})$ is, that it implies the spectral theory for a GEP of the form $(\ref{eq:gen_eig_else})$.
Since positive definite matrices have a well-defined square root and inverse, it is indeed trivial to construct a hidden operator $\matr{H}$ from $\matr{A}$ and $\matr{B}$:
\begin{align*}
 \matr{H} = \matr{B}^{-1/2 \dagger } \matr{A} \matr{B}^{-1/2}.
\end{align*}
Therefore, any GEP of the form \ref{eq:gen_eig_else} can be endowed with 
an interpretation as a subspace eigenvalue problem of a higher-dimensional operator.
In Chapter \ref{chap:FilterDiag} a similar interpretation through a hidden high-dimensional operator,
allows to derive the filter diagonalization method. In the next subsection we introduce the mathematical formalism that is subsequently used.

\subsection{GEPs as Dimension Reduction Schemes}

\label{sec:SpecTheoryGEP}
We consider a complex-valued Hilbert space $\Hil$ with inner product $\langle \cdot, \cdot \rangle$. Formally we
also allow for infinite dimensional Hilbert spaces. 
Through GEPs, the study of operators acting on the high dimensional space $\Hil$ is mapped to a computation-friendly problem in a low dimensional 
Euclidean space. The theory holds equivalently regardless of whether $\Hil$ or the
Euclidean space is real or complex-valued.
We start with a simple definition of the Generalized Eigenvalue Problem.

\begin{definition}
    \label{def:gen_eig}
 Let $\matr{H}: \Hil \to \Hil$ be a linear operator acting on a Hilbert space $\Hil$. 
 Let $\matr{V}= [v_1, \cdots v_{M}]: \mathbb{C}^{M} \to \Hil$ be a linear operator defined by
    \begin{align*}
 \matr{V} x = \sum_{i=1}^M x_i v_i,
    \end{align*}
 where $v_i \in \Hil$.
 The equation
    \begin{align}
 \matr{V}^\dagger \matr{H} \matr{V} x = \mu \matr{V}^\dagger \matr{V} x    \label{eq:gen_eig}
    \end{align}
 with $\mu \in \mathbb{C}$ and $x \in \mathbb{C}^M$, is called the Generalized Eigenvalue Problem (GEP) of $\matr{H}$ with respect to subspace basis $\matr{V}$,
 denoted as $(\matr{H}, \matr{V})$.
A GEP is called \emph{hermitian} if $\matr{H}$ is hermitian. A GEP is called \emph{compact} if $\matr{H}$ is a compact operator. 
 We define the $\textup{rank of a GEP}$ as the rank of its Gram matrix $\matr{V}^\dagger \matr{V}$,
    \begin{align*}
 m = \operatorname*{rank}(\matr{V}^\dagger \matr{V}).
    \end{align*}
\end{definition}
We will use a shorthand notation and write $\matr{V}^2 = \matr{V}^\dagger \matr{V}$ and 
$\matr{H}_V = \matr{V}^\dagger \matr{H} \matr{V}$. The entries of the matrices are given by
\begin{align*}
 (\matr{H}_V)_{ij} & = \langle v_i, \matr{H} v_j \rangle \quad \text{and} \quad \matr{V}^2_{ij} = \langle v_i, v_j \rangle.
\end{align*}
If $\matr{H}$ is self-adjoint, it follows from the conjugate symmetry of the inner product that $\matr{H}_V$ is also self-adjoint,
\begin{align*}
 \conj{(\matr{H}_V)_{ij}} = \conj{\langle v_i, \matr{H} v_j \rangle} = \langle \matr{H} v_j, v_i \rangle = \langle v_j, \matr{H} v_i \rangle = (\matr{H}_V)_{ji}.
\end{align*}
Sometimes we call $\matr{V}^2$ the \emph{basis matrix} of the GEP, as the columns of $\matr{V}$ span the subspace within which the operator $\matr{H}$ is studied.
We also slightly abuse notation and denote the space that is spanned by the columns of $\matr{V}$ simply as $\operatorname{span}(\matr{V})$ or just $\matr{V}$.
With the same motivation we sometimes call $\matr{H}$ the \emph{target operator} of the GEP.\\

\subsection{Generalized Spectral Theorem}

\label{sec:GeneralizedSpectralTheorem}
Generalized Eigenvalue Problems in the sense of Definition \ref{def:gen_eig} have access to a spectral theory analogous to the familiar spectral theory of hermitian operators.
The following theorem summarizes known identities that have 
been repeatedly rediscovered in different computational fields.

\begin{theorem}[Genralized Spectral Theorem]
    \label{thm:gen_specral_theorem}
 Let $(\matr{H}, \matr{V})$ be a Hermitian Generalized Eigenvalue Problem of rank $m$, as 
 in Definition \ref{def:gen_eig}. There exist a matrix $\matr{\Phi} \in \mathbb{C}^{M \times M}$
 and a diagonal matrix $\matr{\Lambda} \in \mathbb{R}^{M \times M}$ such that 
    \begin{align}
 \matr{V}^\dagger \matr{H} \matr{V} \matr{\Phi} = \matr{V}^\dagger \matr{V} \matr{\Phi}  \matr{\Lambda}, \label{eq:thm_gen_mat_form}
    \end{align}
 where $\matr{\Phi}$ satisfies
    \begin{align}
 \matr{\Phi}^\dagger\matr{V}^\dagger \matr{V}  \matr{\Phi}  = \matr{I}_m. \label{eq:ortho_phi}
    \end{align}
\end{theorem}
Here $\matr{I}_m$ is a pseudo identity matrix with ones on the first $m$ diagonal entries and zeros elsewhere.
The generalized unitary relation given in (\ref{eq:ortho_phi}) implies that 
the eigenvectors $x_i$ of a hermitian GEP that are outside of the kernel of $\matr{V}^2$ can be chosen such that
$$\langle x_i, \matr{V}^2 x_j \rangle = \delta_{ij}.$$
A constructive proof of the Generalized Spectral Theorem \ref{thm:gen_specral_theorem} for GEPs is given in the Appendix, Section \ref{sec:proof_specral_theorem}.\\

Theorem \ref{thm:gen_specral_theorem} can indeed understood as a generalization of the known Spectral Theorem in finite-dimensional Euclidean space. If $\matr{V}$ is a linear isometry and 
therefore has $\matr{V}^2 = \matr{I}$ and $M=m$, (\ref{eq:thm_gen_mat_form}) reduces 
to an eigendecomposition of an hermitian matrix $\matr{H}_V$ and (\ref*{eq:ortho_phi}) implies that $\matr{\Phi}$ is a unitary matrix. 
Concerning infinite dimensional Hilbert spaces, it is also worth noting that Theorem \ref{thm:gen_specral_theorem}
does not require $\matr{H}$ to be a compact operator, as the operator version of the classical spectral theorem does. 
It then only requires $\matr{V}$ to be of finite dimension, that is $M < \infty$ or to correspond to a spectral subspace of $\matr{H}$.\\

However, given that the Gram matrix $\matr{V}^2$ induces at least a semi-definite inner product, the spectral theorem for GEPs is not a surprise.\footnote{The term semi-definite inner products is used here for an inner product, that does not satisfy positive definiteness, but only positive semidefiniteness.
It is at least on the orthogonal complement of the null space of $\matr{V}^2$ positive definite.}
It is indeed closely related to the known spectral theorem on a Hilbert space $\Hil_V$
with inner product $\langle x, y\rangle_V =\langle \matr{V} x, \matr{V} y\rangle $, and in 
Section \ref{sec:proof_specral_theorem2} we give a second and shorter proof of Theorem \ref{thm:gen_specral_theorem} that is based on this viewpoint.
From this perspective 
Theorem \ref{thm:gen_specral_theorem} is a mere Euclidean implication of the Spectral Theorem on the Hilbert space with an inner product, that is induced by the basis matrix of the GEP.
Returning to our original numerical motivation, that very "inner product" will be in practice subject to noise. Therefore we will 
here not delve too deep into this perspective. \\

The Generalized Spectral Theorem also states, that the rank of a GEP is of greater interest than its dimensionality. 
This is because the null space of $\matr{V}^2$ is a subspace of the null space of $\matr{H}_V$. Therefore we can 
always bring a GEP into a block diagonal form, such that we can effectively ignore the uninteresting nullspace and consider a 
GEP of dimensionality $m$ instead (details in the proof given in Section \ref{sec:proof_specral_theorem}). In Section 
\ref{sec:NumTheoSimbiosis} we will transfer this insight to a numerical setting where there is noise. 
Then the concept of a rank of a matrix is indeed slightly more delicate, and singularities in the basis matrix $\matr{V}^2$
have caused some confusion in the literature on GEPs. \\

We call eigenvalues that correspond to eigenvectors in the nullspace of  $\matr{V}^2$ \emph{spurious}. They can take arbitrary values, and we set them formally equal to zero.
\begin{align*}
 x \in \operatorname*{Ker}(\matr{V}) \quad \Rightarrow \quad \matr{H}_V x = \mu \matr{V}^2 x = 0 \quad \text{ for all } \quad \mu.
\end{align*}
Eigenvalues that correspond to eigenvectors outside the nullspace of $\matr{V}^2$ are called \emph{proper}. \\

We denote classical eigenvalues of an hermitian operator $\matr{H}$ as $\lambda_i(\matr{H})$ and generalized 
eigenvalues similarly as $\mu_i(\matr{H}, \matr{B})$. We adopt the convention that proper generalized and classical eigenvalues are ordered in descending order,
\begin{alignat*}{4}
    &\mu_1(\matr{H},\matr{V}) && \geq \mu_2(\matr{H},\matr{V}) &&\geq \cdots &&\geq \mu_m(\matr{H},\matr{V}) \\
    &\lambda_1(\matr{H}) && \geq \lambda_2(\matr{H}) && \geq \cdots &&\geq \cdots.
\end{alignat*}
We denote the set of proper eigenvalues of a GEP as
$$\mu(\matr{H},\matr{V}) = \{ \mu_i(\matr{H}, \matr{V}): i = 1, \cdots, m\}.$$

The scientific significance of eigenvalue problems often arises from the
variational principle, as many optimization problems can be formulated as a Rayleigh quotient. 
The variation-based interpretation of eigenvalues is therefore perhaps one of the main reasons for their omnipresence in various scientific fields
and will be transferred to GEPs in the following section.

\subsection{The Variational Principle }
\label{sec:RayleighQuotient}
Rayleigh quotients are an essential tool in both analytical and numerical spectral analysis. In the context of GEPs, 
Generalized Rayleigh quotients are given as
\begin{align*}
 r(\matr{H}, \matr{V}, x) = \frac{\langle x, \matr{V}^\dagger \matr{H} \matr{V} x \rangle}{ \langle x, \matr{V}^\dagger \matr{V} x \rangle} = \frac{ \langle \matr{V} x,  \matr{H} \matr{V} x \rangle}{ \langle \matr{V}x, \matr{V} x \rangle}.
\end{align*}
Thus, a generalized Rayleigh quotient corresponds to a usual Rayleigh quotient, with respect to a vector $\matr{V} x$ in the Hilbert space $\Hil$. 
So while a generalized Rayleigh quotient might have a correspondence to high dimensional space, its degrees of freedom $x\in \mathbb{C}^M$ are low dimensional. This also motivates the computational interest of GEPs as a framework for dimension reduction.\\

The eigenvalues $\mu_i$ of a GEP can alternatively be given by the Rayleigh quotient of their generalized eigenvectors $x_i$,
\begin{align*}
    \mu_i(\matr{H}, \matr{V}) = r(\matr{H}, \matr{V}, x_i).
\end{align*}
To be consistent with our convention that spurious eigenvectors are equal to zero, we formally set 
the Rayleigh quotient equal to zero for $x \in \operatorname*{Ker}(\matr{V})$. 
From the Raleigh quotient, it is also easy to establish the following Fact. 

\begin{fact}
    \label{fact:spec_range}
 For any hermitian and compact GEP $(\matr{H}, \matr{V})$ the spectrum of the GEP is contained in the spectral range of the target operator $\matr{H}$,
    \begin{align*}
        \mu(\matr{H},\matr{V}) \subset [\lambda_{\text{min}}(\matr{H}), \lambda_{\text{max}}(\matr{H})].
    \end{align*}
\end{fact}
The statement follows from a simple expansion in the eigenbasis of $\matr{H}$ and is given in Section \ref{sec:proof_spec_range}.
The statement can also be adapted to the more general case of operators that are not compact, but only bounded, as the proof only relies on Rayleigh quotients.\\

A slightly deeper instance of the variational principle is given in the min-max theorem, also known as the Courant-Fischer-Weyl principle \cite{Caurant1993}.
It has been used in many derivations and also has a strong numerical flavour to it, due to its characterization as optimization problem.
Naturally, the concept also extends to GEPs.

\begin{theorem}[min-max characterization]
 Let $(\matr{H}, \matr{V})$ be a hermitian GEP as in Definition \ref{def:gen_eig}. Let $S_k \subset \mathbb{C}^M$ denote a $k$ dimensional subspace.
    \begin{align*}
        \mu_k(\matr{H}, \matr{V}) & = \operatorname*{max}_{\substack{ S_k }} \operatorname*{min}_{\substack{ x \in S_k }} \frac{\langle x, \matr{V}^\dagger \matr{H} \matr{V} x \rangle}{   \langle x , \matr{V}^\dagger\matr{V} x \rangle}\\
                                 & = \operatorname*{min}_{\substack{ S_{k-1} }} \operatorname*{max}_{\substack{ x \in S_{k-1}^\perp}} \frac{\langle x, \matr{V}^\dagger \matr{H} \matr{V} x \rangle}{   \langle x , \matr{V}^\dagger\matr{V} x \rangle}
    \end{align*}
    \label{thm:min_max_GEP}
\end{theorem}
The proof of Theorem \ref{thm:min_max_GEP} is a trivial adaptation of the classical proof by changing to the inner product induced by the basis matrix $\matr{V}^2$ of the GEP 
and is for completeness given in Section \ref{sec:proof_min_max_GEP}.
The generalized min-max correspondence is the starting point for the derivation of spectral inequalities in the proofs of Section \ref{sec:SpecIneq}. 

\subsection{Correspondence to Target Eigenvalues}

Here, we formally justify the numerical significance 
of GEPs, by relating the generalized eigenvalues to normal eigenvalues of the target operator $\matr{H}$. \\

Let $(\lambda_i, \varphi_i)$ denote the normal eigenpairs of $\matr{H}$.
For a set of indices $\mathcal{I} \subset \mathbb{N}$ we introduce a corresponding \emph{spectral subspace} $\mathcal{E}_{\mathcal{I}}(\matr{H})$ and a set of eigenvalues
\begin{align*}
 \mathcal{E}_{\mathcal{I}}(\matr{H}) & = \operatorname*{span}\{ \varphi_{i\in \mathcal{I}} \} \\
\lambda_\mathcal{I}(\matr{H}) & = \{ \lambda_i(\matr{H}): i \in \mathcal{I}\}
\end{align*}
We say that a GEP $(\matr{H}, \matr{B})$ \emph{coincides} with a spectral subspace $\mathcal{E}_{\mathcal{I}}(\matr{H})$, if the basis matrix
$\matr{B} = [b_1, \cdots ,b_M]$ has $\operatorname*{span}\{b_1, \cdots, b_M\} = \mathcal{E}_{\mathcal{I}}(\matr{H})$. 
The following Fact justifies the terminology.

\begin{fact}
    \label{fact:gen_eig_to_normal}
 A hermitian GEP $(\matr{H}, \matr{B})$ that \emph{coincides} with a spectral subspace $\mathcal{E}_{\mathcal{I}}(\matr{H})$, has proper eigenvalues that coincide with the eigenvalues of the spectral subspace $\mathcal{E}_{\mathcal{I}}(\matr{H})$,
    \begin{align*}
        \mu(\matr{H},\matr{B}) = \lambda_\mathcal{I}(\matr{H}).
    \end{align*}
\end{fact}
A derivation of the simple Fact is given in Section \ref{sec:ProofFact}.\\

For many applications of GEPs in numerical spectral analysis, the spectral subspace is 
defined according to an energy interval $[a,b]$ of interest. In such a setting we have 
$\mathcal{I} = \{J , J+1, \cdots , J+M \}$ such that $\lambda_{J-1}(\matr{H}) <  a \leq \lambda_{J}(\matr{H}) $ and $\lambda_{J+M} \leq b <\lambda_{J+M+1} $.
Then, we may also write $\mathcal{E}_{[a,b]}$ and $\lambda_{[a,b]}(\matr{H})$. 
In practice, the number of eigenvalues $M$ in a region of interest is often unknown. \\

However, in the scenario of relevance here, the guess vectors are not guaranteed to perfectly span an eigenspace $\mathcal{E}_{\mathcal{I}}$,
but rather only approximately so. Numerical methods will generate 
guess vectors $v_i$ that also have at least small contributions outside of the spectral subspace of interest. To treat this
from a mathematical perspective we will use an orthogonal decomposition of the basis matrix $\matr{V}$. 

\subsection{Orthogonal Decomposition of the Basis Matrix}
Here, we will decompose the basis matrix, into two parts, to one of which the spectral correspondence of Fact \ref{fact:gen_eig_to_normal} applies. 
We assume a compact and hermitian operator $\matr{H}$ and expand the guess vectors $v_i$ in terms the eigenvector basis $\varphi_i$ of $\matr{H}$,
\begin{align*}
 v_i & = \sum_{j} \alpha_{j} \varphi_j = \sum_{j \in \mathcal{I}} \alpha_{j} \varphi_j + \sum_{j \not \in \mathcal{I}} \alpha_{j} \varphi_j\\
    & = b_i + n_i.
\end{align*}
Thus, $(b_i,n_i)$ yields here the unique decomposition of $v_i$, such that $b_i \in \mathcal{E}_{\mathcal{I}}$ and $n_i \in \mathcal{E}_{\mathcal{I}}^\perp$.
Accordingly, we decompose the basis matrix $\matr{V}$ as
\begin{align*}
 \matr{V} = \matr{B} + \matr{N},
\end{align*}
with $\matr{B} = [b_1, \cdots, b_m]$ and $\matr{N} = [n_1, \cdots, n_m]$.
Since the columns of $\matr{B}$ correspond to a linear combination of vectors orthogonal to the basis that spans the columns of $\matr{N}$, 
the decomposition has 
\begin{align*}
 \matr{B}^\dagger \matr{N} = \matr{N}^\dagger \matr{B} = 0.
\end{align*}
Since the action of a linear and self-adjoint operator does not mix its eigensubspaces, we also have 
\begin{align*}
 \matr{B}^\dagger \matr{H} \matr{N} = \matr{N}^\dagger \matr{H} \matr{B} = 0.
\end{align*}

The motivated orthogonal decomposition is summarized in the following definition. 
\begin{definition}
    \label{def:ortho_decomp}
 We call $[\matr{B}, \matr{N}] \in \mathbb{C}^{M \times M} \times \mathbb{C}^{M \times M}$ an orthogonal decomposition to the GEP of $\matr{H}$ to basis $\matr{V}$,
 if 
    \begin{enumerate}
        \item $\matr{V} = \matr{B} + \matr{N}$,
        \item $\matr{B}^\dagger \matr{N} = \matr{N}^\dagger \matr{B} = 0$,
        \item $\matr{B}^\dagger \matr{H} \matr{N} = \matr{N}^\dagger \matr{H} \matr{B} = 0$.
    \end{enumerate}
 An orthogonal decomposition $[\matr{B}, \matr{N}]$ is called a \emph{signal-noise decomposition}, if the space $\matr{B}$ is a subspace of an spectral subspace of $\matr{H}$,
   $$\operatorname{span} (\matr{B}) \subset \mathcal{E}_\mathcal{I}(\matr{H}).$$ We call a signal-noise decomposition is \emph{complete} if the space $\matr{B}$ coincides with a spectral subspace,
    $$\operatorname{span} (\matr{B}) = \mathcal{E}_\mathcal{I}(\matr{H}).$$
\end{definition}

With the notion of a signal-noise decomposition of a GEP, the mathematical framework that allows us to 
rigorously express the intention of \emph{guess vectors approximately spanning a spectral subspace} is completed. 
For a GEP $(\matr{H},\matr{V})$ with signal noise decomposition $[\matr{B}, \matr{N}]$ we have
\begin{align}
 \left(\matr{B}^\dagger \matr{H} \matr{B} + \matr{N}^\dagger \matr{H} \matr{N} \right) b &= \mu \left(\matr{B}^\dagger \matr{B} + \matr{N}^\dagger \matr{N} \right) b. \label{eq:gen_eig_signal_noise}
\end{align}
Thus, a noisy guess vector matrix $\matr{V}$ results in perturbations on both sides of the GEP equation. Moreover, the disturbing matrices $\matr{H}_N$ and $\matr{N}^2$ are not independent
of each other but are correlated. This is where our analysis finally starts to differ from simple generalizations of classical spectral theory.
We are not aware of spectral inequalities that capture the behavior in a noisy scenario such as in (\ref{eq:gen_eig_signal_noise}).
Nor do we know of spectral inequalities that relate the eigenvalues of a low-dimension GEP to the eigenvalues of high dimensional operator $\matr{H}$ that is hidden within the GEP.\\

In the next section, we study the implication of the situation on the spectrum of the GEP and derive spectral inequalities. 
We will give bounds on the deviation of eigenvalues of a GEP to the
eigenvalues of the target operator $\matr{H}$. 

\newpage 

\section{Spectral Stability of Guess Vector Projections}
\label{sec:SpecIneq}

The rationale of projection-based GEPs forms a promising direction towards dimension reduction in spectral analysis. 
In particular, the generation of guess vectors enables opportunities to incorporate prior information and to save computational cost. 
However, to understand the accuracy that can be achieved with such approximate methods and to assess opportunities to overcome 
the curse of dimensionality, a mathematical analysis with rigorous precision statements is desired.\\

Classical perturbation theory typically examines the effects of a perturbing operator $\matr{E}$ on a operator $\matr{H}$ with a known spectrum.
In such theoretical treatments, the spectrum of the perturbed operator $\matr{H} + \matr{E}$ is then of primary interest.\\

The scenario for GEPs as a numerical dimension reduction scheme differs in two essential ways.
First, the perturbation does not arise in the operator $\matr{H}$ that is subject to study, but in the basis matrix $\matr{V}$ that is used to project the operator.
Second, the clean basis matrix $\matr{B}$ and the noisy perturbation $\matr{N}$ are both unknown to us, but not their sum $\matr{V} = \matr{B} + \matr{N}$.
More swiftly, through the GEP we can efficiently "measure" a spectrum subject to noise and want to relate it to the hidden high dimensional operator $\matr{H}$.\\

Surprisingly, spectral inequalities that capture the precision of GEPs as dimension reduction schemes appear elusive in the literature.
In the following, we present the first volume of results, that we obtained in this setting.

\subsection{Results}

The following inequalities capture how GEPs as dimension reduction schemes react to noise in the guess vector matrix $\matr{V} = \matr{B} + \matr{N}$.
Here, $\matr{N}$ is considered as the perturbation.

\begin{theorem}[Spectral Stability]
 Let $(\matr{H}, \matr{V})$ be a GEP with rank $m$ as in Definition \ref{def:gen_eig}, with an orthogonal decomposition $[\matr{B}, \matr{N}]$, as in Definition \ref{def:ortho_decomp}. For all $i \in \{1,\cdots, m\}$ the spectral stability inequality holds,
    \begin{align}
 \left| \mu_i(\matr{H}, \matr{V}) - \mu_i(\matr{H}, \matr{B})  \right| & \leq \frac{\lambda_{1}({\matr{N}^2})}{\lambda_m({\matr{V}^2})}  \max_{j \in {1,m}} \bigl| \mu_j(\matr{H}, \matr{N})- \mu_i(\matr{H}, \matr{B}) \bigr|.  \label{eq:specral_stability_main} 
    \end{align}
    \label{thm:specral_stability}
\end{theorem}
The spectral stability (\ref{eq:specral_stability_main}) gives an error bound in terms of a precision factor $\lambda_{1}({\matr{N}^2}) / \lambda_m({\matr{V}^2})$.
Thus, $\lambda_m({\matr{V}^2})^{-1}$ quantifies the stability of a GEP with respect to noise and can be considered as the condition number for projection-based GEPs.
Importantly, as $\matr{V}^2$ is accessible and of small dimension, a GEP-based algorithm can determine the quality of its conditioning, independently of a theoretical guarantee. \\

In the special case of a signal noise decomposition, the generalized eigenvalues $\mu_i(\matr{H}, \matr{B})$ coincide with the eigenvalues of the target operator $\matr{H}$ (Fact \ref{fact:gen_eig_to_normal}).
We immediately obtain a corollary.
\begin{corollary}
    \label{cor:specral_stability_signoise}
 A hermitian GEP $(\matr{H}, \matr{V})$ of rank $m$ with a complete signal noise decomposition $[\matr{B}, \matr{N}]$ 
 to an eigensubspace $\mathcal{E}_{\mathcal{I}}(\matr{H})$ has 
    \begin{align}
 \left| \mu_i(\matr{H}, \matr{V}) - \lambda_{\mathcal{I},i}(\matr{H}) \right| & \leq \frac{\lambda_{1}({\matr{N}^2})}{\lambda_m({\matr{V}^2})}  \max_{j \in {1,m}} \bigl| \mu_j(\matr{H}, \matr{N})- \lambda_{\mathcal{I},i}(\matr{H})\bigr|.  \label{eq:specral_stability_signoise} 
    \end{align}
\end{corollary}
Here, $\lambda_{\mathcal{I},i}(\matr{H})$ denotes eigenvalues $\lambda_j(\matr{H})$ with $j \in \mathcal{I}$ that are given in 
descending order; for example $\lambda_{\mathcal{I},1}(\matr{H}) = \max_{j \in \mathcal{I }} \lambda_j(\matr{H})$. \\

Corollary \ref{cor:specral_stability_signoise} is the strongest bound for dimension reduction protocols $\mathfrak{P}$ provided in this section. The bound is characterized by the leading eigenvalue of the noise Gramm matrix $\matr{N}^2$.
A protocol  $\mathfrak{P}$ that generates guess vectors and wants to access a precision guarantee on the obtained eigenvalues, must
provide a bound on the operator norm of the noise Gram matrix $\matr{N}^2$ it may generate. Since operator norm estimates are a well-established mathematical discipline,
we consider this to be a feasible requirement.\footnote{A simple operator norm estimate can often be derived through the Gerschgorin circle theorem.}\\

However, the bound also assumes control over $\mu_{1,m}(\matr{H},\matr{N})$, which is difficult to obtain solely from the accessible GEP $(\matr{H}, \matr{V})$. 
Corollary \ref{cor:user_friend} illustrates how we can still obtain user-friendly inequalities from the more general bound 
given above.
\begin{corollary}
    \label{cor:user_friend}
 A hermitian and compact GEP $(\matr{H}, \matr{V})$ of rank $m$ with a complete signal noise decomposition $[\matr{B}, \matr{N}]$ has 
    \begin{align*}
 \operatorname*{min}_{j}\left| \mu_i(\matr{H}, \matr{V}) - \lambda_{j}(\matr{H}) \right| & \leq \frac{\lambda_{1}({\matr{N}^2})}{\lambda_m({\matr{V}^2})}  \left( \lambda_{\text{max}}(\matr{H}) -  \lambda_{\text{min}}(\matr{H})\right)
    \end{align*}
\end{corollary}
Corrolary \ref{cor:user_friend} follows from Corollary \ref{cor:specral_stability_signoise} by applying that the eigenvalues of a compact GEP are contained in the range of the target operator.
If the spectral subspace $\mathcal{E} = \mathcal{E}_{(a,b)}$ corresponds to an energy interval $(a,b) \subset [ \lambda_{\text{min}}, \lambda_{\text{max}}]$ the bound can be improved by 
considering $\operatorname*{max}\{ | \lambda_{\text{max}}(\matr{H}) - a|, |b - \lambda_{\text{min}}(\matr{H})|\}$ instead of $(\lambda_{\text{max}}(\matr{H}) -  \lambda_{\text{min}}(\matr{H}))$.\\

Similarly, if it is known, that the guess vectors are supported in a larger spectral subspace, 
$\operatorname{span}( \matr{V}) \subset \mathcal{E}_{(E_A,E_B)}$ one can obtain from Corollary \ref{cor:specral_stability_signoise} a sharper bound 
with 
\begin{align*}
 E_A \leq \mu_m(\matr{H}, \matr{N}) \qquad \text{ and } \qquad \mu_1(\matr{H}, \matr{N}) \leq E_B.
\end{align*}
Especially for operators $\matr{H}$ with dense spectra, we expect Corollary \ref{cor:user_friend} 
to provide a satisfactory and user-friendly estimate. However, we also point out that 
the estimates necessary from Corollary \ref{cor:specral_stability_signoise} to Corollary \ref{cor:user_friend} are 
quite pessimistic for many applications. This shortcoming will be further addressed in Section \ref{sec:Limitations}. \\

The basis for the previous results is given by the spectral inequalities of the following Theorem. It will also form the basis 
for more sophisticated inequalities that will be developed in Section \ref{sec:IntegratedSpectralInequalities}.

\begin{theorem}[Spectral Inequalities]
    \label{thm:spectral_inequalities}
 For a GEP $(\matr{H}, \matr{V})$ of rank $m$ with an orthogonal decomposition $[\matr{B}, \matr{N}]$ we have
    \begin{align}
        \mu_j(\matr{H}, \matr{V}) \leq \mu_i(\matr{H}, \matr{B}) + \frac{\lambda_{k + \emph{s} p}({\matr{N}^2})}{\lambda_k({\matr{V}^2})} \bigl( \mu_l(\matr{H}, \matr{N}) - \mu_i(\matr{H}, \matr{B}) \bigr)  \label{eq:upper_spectral_Weyl} \\
        \mu_j(\matr{H}, \matr{V}) \geq  \mu_i(\matr{H}, \matr{B}) + \frac{\lambda_{k - \emph{s} q}({\matr{N}^2})}{\lambda_k({\matr{V}^2})} \bigl( \mu_l(\matr{H}, \matr{N}) - \mu_i(\matr{H}, \matr{B}) \bigr)  \label{eq:lower_spectral_Weyl}
    \end{align}
 where \begin{align*}
 p  & = j-i-l-m+2,\\
 q &= -2m +1+ i +l -j,\\
 s & = \operatorname*{sgn} \left( \mu_l(\matr{H}, \matr{N}) - \mu_i(\matr{H}, \matr{B}) \right) \in \{ \pm 1,0\} .
    \end{align*}
\end{theorem}
All eigenvalue indices in Theorem \ref{thm:spectral_inequalities} need to be chosen such that they are valid.
On the example of (\ref{eq:upper_spectral_Weyl}) this means $i,j,l,k,(k+sp) \in \{1, \cdots, m\}$.\\
The proof of Theorem \ref{thm:spectral_inequalities} also shows, that it is possible 
to first upper (\ref{eq:upper_spectral_Weyl}) or lower bound (\ref{eq:lower_spectral_Weyl}) the factor 
$\bigl( \mu_l(\matr{H}, \matr{N}) - \mu_i(\matr{H}, \matr{B}) \bigr) $ and then to choose $s$ accordingly.
On the example of the upper estimate (\ref{eq:upper_spectral_Weyl}) this means if we take
\begin{align*}
 \bigl( \mu_l(\matr{H}, \matr{N}) - \mu_i(\matr{H}, \matr{B}) \bigr) \leq \max \{ A, 0 \} := \tilde A,
\end{align*}
we will be sure to have
\begin{align*}
    \mu_j(\matr{H}, \matr{V}) \leq \mu_i(\matr{H}, \matr{B}) + \frac{\lambda_{k + p}({\matr{N}^2})}{\lambda_k({\matr{V}^2})} \tilde A.
\end{align*}

In the next subsection, we briefly discuss the possibilities arising from the 
presented inequalities towards dimension reduction, but also their limitations.

\subsection{Discussion}
More clarity on the relationship between dimensionality and precision of computation is 
of much interest in our contemporary scientific landscape. The provided spectral inequalities 
enable precision guarantees for projection-based GEPs. They form a small step towards 
an understanding of how we can effectively overcome the curse of dimensionality in spectral analysis
by allowing for a small but controlled error. \\

The presented inequalities showcase a precision factor $\frac{\lambda_{1}({\matr{N}^2})}{\lambda_m({\matr{V}^2})}$,
that is rather satisfactory. The error bound in terms of $\lambda_{1}({\matr{N}^2})$ only imposes a feasible requirement 
on a protocol, to access precision guarantees. A method that aims to circumvent high dimensional 
computations through a low dimensional guess space only has to derive an operator norm bound 
on the noise Gramm matrix associated with the guess vector space. A bound for an operator norm 
enjoys the advantage of many established mathematical approaches and is usually feasible to derive.\\
The second part of the precision factor $\lambda_m({\matr{V}^2})$ is especially from a computational perspective
very attractive. Control on the lowest non-zero eigenvalue of the Gramm matrix $\lambda_m({\matr{V}^2})$ 
does not require any theoretical guarantees. Indeed, while the theoretical toolbox for 
operator norm estimates is vast, lower bounds on the smallest eigenvalues are less common. Instead, $\lambda_m({\matr{V}^2})$ 
can be directly computed from the guess vector matrix. In Section \ref{sec:NumTheoSimbiosis} we 
further elaborate on the computational opportunities that arise from the spectrum of guess vector matrix $\matr{V}^2$. In particular
dimension detection and avoidance, of the so-called \emph{singularity problem} is addressed.\\

However, the second part of the error bound in Corollary \ref{cor:specral_stability_signoise}  $\max_{j \in {1,m}} \bigl| \mu_j(\matr{H}, \matr{N})- \mu_i(\matr{H}, \matr{B}) \bigr|$
 depends the noise GEP eigenvalues $\mu_1(\matr{H}, \matr{N})$ and $\mu_m(\matr{H}, \matr{N})$. 
Simple estimates, that do not require any prior information on $\matr{N}$ allow this factor to be replaced by the spectral range of the target operator $\lambda_{\text{max}}(\matr{H}) -  \lambda_{\text{min}}(\matr{H})$
as seen in Corollary \ref{cor:user_friend}. Although this bound is perhaps a bit pessimistic, it provides 
a simple tool to provide guarantees on the analysis of high dimensional operators. 

\subsubsection{Example on Second Quantization Electronic Hamiltonian}
\label{sec:Example_sec_quant}
The second quantized molecular Hamiltonian is particularly prone to the curse of dimensionality.
We assume the Born-Oppenheimer approximation. 

Here the electronic Hamiltonian to describe $N$ orbitals 
is given by Pauli strings of length $N$ and has dimension $2^{N}$. 
\begin{align*}
    H = \sum_l h_l P_l^{(N)}
\end{align*}
The Hamiltonian commutes with the particle number operator $\matr{n}$. 
The ground state energy for a number of electrons $n_e$ is then given by 
the lowest eigenvalue of the Hamiltonian in the subspace of $n_e$ electrons.\\

An estimate of the 
spectral range of $\matr{H}$ can be easily computed from the Pauli amplitudes $h_j$ that have been 
obtained from a Hartree-Fock ansatz \cite{MarkusReiher2017},
\begin{align*}
\lambda_{\text{max}}(H) - \lambda_{\text{min}}(H) = \sum_l |h_l|.
\end{align*}
If a protocol determines now for example the $m$ lowest eigenvalues of the Hamiltonian,
though a projected GEP, Corollary \ref{cor:user_friend} gives a guarantee on the precision of the obtained eigenvalues,
\begin{align*}
 \min_j \left| \mu_i(\matr{H}, \matr{V}) - \lambda_{j}(\matr{H}) \right| & \leq \frac{\lambda_{1}({\matr{N}^2})}{\lambda_m({\matr{V}^2})}  \sum_l |h_l|.
\end{align*}
Thus, the spectral inequalities indeed enable rigorous precision guarantees for algorithms within quantum chemistry.
After the protocol has generated the guess vector matrix $\matr{V}^2$ and $\matr{H}_{\matr{V}}$, the remaining computational 
cost to determine $m$ eigenvalues boils down to the cost of diagonalizing $m \times m$ matrices.\\

However, there are also some limitations to the presented results of this section that we would like to point out.

\subsubsection{Limitations and need for Full-Scale Analysis}
\label{sec:Limitations}
We have seen that the inequalities of this section can 
already provide first precision guarantees for GEPs. In particular Corollary \ref{cor:user_friend} can 
be used to provide good precision guarantees for the low-dimensional spectral analysis of 
operators with a narrow spectral range or high but still finite-dimensional matrices, such as a Pauli Hamiltonian.\\

However the estimate (\ref{eq:bad_estimate}) that was necessary in the derivation 
of Corollary \ref{cor:user_friend} from Corollary \ref{cor:specral_stability_signoise} is rather pessimistic for many applications,
\begin{align}
    \lambda_{\min}(\matr{H}) \leq \mu_m(\matr{H}, \matr{N}), \qquad  \qquad \mu_1(\matr{H}, \matr{N}) \leq \lambda_{\max}(\matr{H}). \label{eq:bad_estimate}
\end{align}
A protocol that generates guess vectors that aim to span a spectral subspace $\mathcal{E}_{(a,b)}(\matr{H})$, is very likely to exhibit 
\emph{target clustering}. This means that even the undesired noisy components of the guess vectors, mostly 
correspond to eigenvectors of $\matr{H}$ with eigenvalues close to the target spectral range $(a,b)$.
The simple estimate  (\ref{eq:bad_estimate}) does not honor the property, even though 
it is very much in favor of the approximation. \\

A second and more severe objection is that (\ref{eq:bad_estimate}) assumes that the operator subject to study 
is bounded. In the context of quantum chemistry, this is not exactly true. On the quantum scale, physical systems are actually described by
unbounded operators. In addition, these operators have a mixture of discrete and continuous spectra. The second quantized electronic Hamiltonian
of the previous example was itself only a finite-dimensional approximation to an infinite-dimensional operator.\\

From this perspective, we have derived a rigorous precision bound on a low dimensional solution to a high-dimensional operator, that 
was itself only an approximation of an infinite-dimensional operator. This is still a step forward. We also have the intuition that the error from the full truth to 
the second quantized Hamiltonian is very small indeed. \\

But it would be a disappointment if we had to keep deriving different inequalities for  
lower-dimensional approximations of operators at different scales. 
A satisfactory bound should not depend on the range of $\matr{H}$ in 
such a sensitive way. We would rather have it lean deeper on the nature of the guess vectors. 
They are the objects that we can control, not the operator $\matr{H}$ that is subject to study.\\

To overcome the objections and meet the requirements, a more sophisticated continuation of the theory is provided in 
 Section \ref{sec:IntegratedSpectralInequalities}.
A bound is presented, that has the means to encompass all the approximation errors across the different scales 
in a single integrated inequality. The inequity does not demand a need to 
consider the discrete and continuous spectrum of the target operators independently of each other.
If the target clustering of the guess vectors is 
appropriately quantified, the bound can also yield more precise estimates than Corollary \ref{cor:user_friend}.

\subsection{Derivations}

\subsubsection{From Corollary \ref{cor:specral_stability_signoise} to Corollary \ref{cor:user_friend}}
\label{sec:proof_user_friend}

\begin{proof}[\textit{From Corollary \ref{cor:specral_stability_signoise} to Corollary \ref{cor:user_friend}}]
 Let $(\matr{H}, \matr{V})$ and $[\matr{B}, \matr{N}]$ be as in Corollary \ref{cor:user_friend}. We assume Corrolary \ref{cor:specral_stability_signoise}. 
 By Fact \ref{fact:spec_range} we have that the spectrum of the GEP $(\matr{H}, \matr{N})$ is contained in the range of the target operator. In particular
    \begin{align*}
        \mu_1(\matr{H}, \matr{N}) \leq \lambda_{\text{max}}(\matr{H}) \quad \text{and} \quad \mu_m(\matr{H}, \matr{N}) \geq \lambda_{\text{min}}(\matr{H}).
    \end{align*}
 Trivially, $\lambda_{\mathcal{I},i}(\matr{H})$ can be lower or upper estimated by $\lambda_{\text{min}}$ or $\lambda_{\text{max}}$. Overall we obtain
    \begin{align*}
 \max_{j \in {1,m}} \bigl| \mu_j(\matr{H}, \matr{N})- \lambda_{\mathcal{I},i}(\matr{H})\bigr| \leq \lambda_{\text{max}}(\matr{H}) - \lambda_{\text{min}}(\matr{H}).
    \end{align*}
 For the left hand side of (\ref{eq:specral_stability_signoise}) we have
    \begin{align*}
 \min_{j} \left| \mu_i(\matr{H}, \matr{V}) - \lambda_{j}(\matr{H}) \right| \leq \left| \mu_i(\matr{H}, \matr{V}) - \lambda_{\mathcal{I},i}(\matr{H}) \right|,
    \end{align*}
 since $\lambda_{\mathcal{I},i}(\matr{H})$ is feasible in the minimization. 
\end{proof}

\subsubsection{From min-max Principle to Spectral Inequalities}

\begin{proof}[\textit{Proof of Theorem \ref{thm:spectral_inequalities}}]
 Let $(\matr{H}, \matr{V})$ and $[\matr{B}, \matr{N}]$ be as in Theorem \ref{thm:spectral_inequalities}. Since the 
 GEP $(\matr{H}, \matr{V})$ has rank $m$ we can WLOG regrad it in an $m$-dimensional Euclidean space $\mathbb{C}^m$.
 By the min-max Theorem \ref{thm:min_max_GEP}, there exist subspaces $\mathcal{S}_{H_B} \subset \mathbb{C}^m$ with dimension $i-1$, $\mathcal{S}_{H_N}$ with dimension $l-1$ and 
    $\mathcal{S}_{H_V}$ with dimension $j$ such that 
    \begin{align*}
        \mu_i( \matr{H}, \matr{B}^2) & = \max_{\substack{x \in \mathcal{S}_{H_B}^\perp }} \frac{\langle x, \matr{B}^\dagger \matr{H} \matr{B} x \rangle}{\langle x, \matr{B}^2 x \rangle }\\
        \mu_l( \matr{H}, \matr{N}^2) & = \max_{\substack{x \in \mathcal{S}_{H_N}^\perp }} \frac{\langle x, \matr{N}^\dagger \matr{H} \matr{N} x \rangle}{\langle x, \matr{N}^2 x \rangle }\\
        \mu_j( \matr{H}, \matr{V}^2) & = \min_{\substack{x \in \mathcal{S}_{H_V}}} \frac{\langle x, \matr{V}^\dagger \matr{H} \matr{V} x \rangle}{\langle x, \matr{V}^2 x \rangle }.
    \end{align*}
 For any $x\in \mathcal{S}_{H_B}^\perp \cap \mathcal{S}_{H_N}^\perp  \cap \mathcal{S}_{H_B}$ (to be proven non-empty later), we have 
    \begin{align}
        \mu_j( \matr{H}, \matr{V}^2) & \leq \frac{\langle x, \matr{H}_V x \rangle}{\langle x, \matr{V}^2 x \rangle} = \frac{\langle x, \matr{H}_B x \rangle + \langle x, \matr{H}_N x \rangle}{\langle x, \matr{V}^2 x \rangle} \leq \mu_i( \matr{H}, \matr{B}) \frac{\langle x, \matr{B}^2 x \rangle }{\langle x, \matr{V}^2 x \rangle} + \frac{\langle x, \matr{H}_N x \rangle}{\langle x, \matr{V}^2 x \rangle} \notag\\ 
        & = \mu_i( \matr{H}, \matr{B}) \frac{\langle x ,(\matr{B}^2 + \matr{N}^2 -\matr{N}^2) x\rangle }{\langle x, \matr{V}^2 x \rangle}  + \frac{\langle x, \matr{N^2} x \rangle}{\langle x, \matr{V}^2 x \rangle} \frac{\langle x, \matr{H}_N  x \rangle}{ \langle x, \matr{N^2} x \rangle} \notag \\
        & = \mu_i( \matr{H}, \matr{B}) + \frac{\langle x, \matr{N^2} x \rangle}{\langle x, \matr{V}^2 x \rangle} \left( \frac{\langle x, \matr{H}_N  x \rangle}{ \langle x, \matr{N^2} x \rangle} - \mu_i( \matr{H}, \matr{B}) \right) \notag\\
        & \leq \mu_i( \matr{H}, \matr{B}) + \frac{\langle x, \matr{N^2} x \rangle}{\langle x, \matr{V}^2 x \rangle} \left(  \mu_l( \matr{H}, \matr{N})- \mu_i( \matr{H}, \matr{B}) \right) \label{eq:ineq4}
    \end{align}
 In the first inequality we used $x\in \mathcal{S}_{H_V}$, in the second $x\in \mathcal{S}_{H_B}^\perp$, and in the third $x\in \mathcal{S}_{H_N}^\perp$.
 In the case $ x \in \operatorname*{Ker}(\matr{N})$ where the second and first line is not well defined, inequality (\ref{eq:ineq4}) is trivially satisfied.
 Depending on the sign of $\left( \mu_l(\matr{H}, \matr{N}) - \mu_i( \matr{H}, \matr{B}) \right)$ we either have to lower or upper estimate the noise 
 to signal ratio $\langle x, \matr{N^2} x \rangle / \langle x, \matr{V}^2 x \rangle$.\\

 We first handle the case where $ \left( \mu_l(\matr{H}, \matr{N}) - \mu_i( \matr{H}, \matr{B}) \right)$ is non-negative, that is $s=1$.
 We set $d = k +j-i-l-m+2$ and consider
 subspaces $\mathcal{S}_N$ with dimension $d-1$ and $\mathcal{S}_V$ with dimension $k$ such that
    \begin{align*}
        \lambda_d( \matr{N}^2 ) & = \max_{\substack{x \in \mathcal{S}_{N}^\perp }} \frac{\langle x, \matr{N}^2 x \rangle}{\langle x, x \rangle }\\
        \lambda_k( \matr{V}^2 ) & = \min_{\substack{x \in \mathcal{S}_{V}}} \frac{\langle x, \matr{V}^2 x \rangle}{\langle x, x \rangle }.
    \end{align*}
 Then $ x\in \mathcal{S}_N^\perp \cap \mathcal{S}_V$ has
    \begin{align*}
 \frac{\langle x, \matr{N^2} x \rangle}{\langle x, \matr{V}^2 x \rangle} \leq \frac{\lambda_d( \matr{N}^2)}{\lambda_k( \matr{V}^2)}.
    \end{align*}
 The overlap $\mathcal{S}_{H_B}^\perp \cap \mathcal{S}_{H_N}^\perp \cap \mathcal{S}_N^\perp$ has codimension at most $i +l +d -3 = j+k -m-1$, while
    $\mathcal{S}_{H_V} \cap \mathcal{S}_V$ has dimension at least $j+k -m$. Thus $\mathcal{S} = \mathcal{S}_{H_B}^\perp \cap \mathcal{S}_{H_N}^\perp \cap \mathcal{S}_N^\perp \cap \mathcal{S}_{H_V} \cap \mathcal{S}_V$
 has dimensionality of at least one. Choosing any non-zero $x \in \mathcal{S}$ proves the first bound (\ref{eq:upper_spectral_Weyl}), in the case $s=1$.\\

 For the second case, where $ \left( \mu_l(\matr{H}, \matr{N}) - \mu_i( \matr{H}, \matr{B}) \right)$ is negative, we
 consider subspaces  $\mathcal{S}_N$ with dimension $d$ and $\mathcal{S}_V$ with dimension $k-1$ such 
 that $ x\in \mathcal{S}_N \cap \mathcal{S}_V^\perp$ has
    \begin{align*}
 \frac{\langle x, \matr{N^2} x \rangle}{\langle x, \matr{V}^2 x \rangle} \geq \frac{\lambda_d( \matr{N}^2)}{\lambda_k( \matr{V}^2)}.
    \end{align*}
 Now, $\mathcal{S}_{H_B}^\perp \cap \mathcal{S}_{H_N}^\perp \cap \mathcal{S}_V^\perp$ has codimension at most $i +l +k -3$.
 We pick  $d =k -p = k- (j-i-l-n+2)$. Then 
    $\mathcal{S}_{H_V} \cap \mathcal{S}_N$ has dimension at least $j+d-n = k+i+l -2$ therefore a non-trivial overlap with $\mathcal{S}_{H_B}^\perp \cap \mathcal{S}_{H_N}^\perp \cap \mathcal{S}_V^\perp$.
 With this, inequality (\ref{eq:upper_spectral_Weyl}) is proven.\\

 The second inequality follows as a corollary of the first. We consider the first inequality 
 with respect to the negative $-\matr{H}$ and relabeled indices $i' = m+1-i$, $j' = m+1-j$ and $l' = m+1-l$.
 We use the simple relation 
    $$\mu_{i'}(-\matr{H}, \matr{B}) = - \mu_{m+1-i'}(\matr{H}, \matr{B}) = - \mu_{i}(\matr{H}, \matr{B})$$ 
 multiple times.  Inequality (\ref{eq:upper_spectral_Weyl}) then gives 
    \begin{align}
 -  \mu_j(\matr{H}, \matr{V})  \leq - \mu_i(\matr{H}, \matr{B}) - \frac{\lambda_{k + s' p'}({\matr{N}^2})}{\lambda_k({\matr{V}^2})} \bigl( \mu_l(\matr{H}, \matr{N}) - \mu_i(\matr{H}, \matr{B}) \bigr).  \label{eq:2828}
    \end{align}
 We have 
    \begin{align*}
 s' & = \operatorname*{sgn} \left( \mu_{l'}(\matr{H}, \matr{N}) - \mu_{i'}(\matr{H}, \matr{B}) \right) \\
            & = \operatorname*{sgn} \left( - \mu_l(\matr{H}, \matr{N}) - (-1) \mu_i(\matr{H}, \matr{B}) \right) = -s.
    \end{align*}
 For $p'$ we have
    \begin{align*}
 p' & = j'-i'-l'-m+2 = -j + i +l -2m +1 = q
    \end{align*}
 Plugging into (\ref{eq:2828}) we find the second inequality,
    \begin{align*}
        \mu_j(\matr{H}, \matr{V}) \geq \mu_i(\matr{H}, \matr{B}) + \frac{\lambda_{k-sq}({\matr{N}^2})}{\lambda_k({\matr{V}^2})} \bigl( \mu_{l}(\matr{H}, \matr{N}) - \mu_i(\matr{H}, \matr{B}) \bigr).
    \end{align*}
\end{proof}

\subsubsection{From Spectral inequalities to Spectral Stability}

\emph{\textit{Derivation of Theorem \ref{thm:specral_stability} from Theorem \ref{thm:spectral_inequalities}}.}
 Let $(\matr{H}, \matr{V})$ and $[\matr{B}, \matr{N}]$ be as in Theorem \ref{thm:specral_stability}.
 We assume Theorem (\ref{thm:spectral_inequalities})
 and consider eq.\ref{eq:upper_spectral_Weyl} at $i=j$ and $l=1$. Then, $ p= -m+1$, and we have 
    \begin{align*}
        \mu_i(\matr{H}, \matr{V}) - \mu_i(\matr{H}, \matr{B}) \leq \frac{\lambda_{k+s_1(1-m)}({\matr{N}^2})}{\lambda_m({\matr{V}^2})} \bigl( \mu_1(\matr{H}, \matr{N})- \mu_i(\matr{H}, \matr{B}) \bigr).
    \end{align*}
 Analogously for a lower bound eq.\ref{eq:lower_spectral_Weyl} at $i=j$ and $l =m$ has $q = -m +1$ and
    \begin{align*}
        \mu_i(\matr{H}, \matr{V}) - \mu_i(\matr{H}, \matr{B}) \geq \frac{\lambda_{k-s_m(1-m)}({\matr{N}^2})}{\lambda_m({\matr{V}^2})} \bigl( \mu_m(\matr{H}, \matr{N})- \mu_i(\matr{H}, \matr{B}) \bigr).
    \end{align*}
 Since $s_1 \geq s_m$ we distinguish between three cases, $(s_1,s_m)\in \{(-1,-1), (1,1), (-1,1)\}$.

    \begin{itemize}
        \item Case I: $s_1 =-1$, that is $\mu_1(\matr{H}, \matr{N}) < \mu_i(\matr{H}, \matr{B})$. 
 This scenario occurs when for example $\matr{N}$ only contains vectors spanned by eigenvectors of $\matr{H}$ with energy lower than the ones spanning $\matr{B}$.\\
 Then, we also have  $s_m = -1$ and $\mu_i(\matr{H}, \matr{V}) < \mu_i(\matr{H}, \matr{B})$. The remaining indices have
        \begin{alignat*}{5}
            &s_1 &&: \qquad && k-(1-m) && = m \quad &&\text{ with } \quad k=1\\
            &s_m &&: \qquad && k+(1-m) && = 1 \quad &&\text{ with } \quad k=m
         \end{alignat*}
 We find in this case inequalities for the absolute value
    \begin{align*}
 \frac{\lambda_{m}({\matr{N}^2})}{\lambda_1({\matr{V}^2})}   \bigl| \mu_1(\matr{H}, \matr{N}) & - \mu_i(\matr{H}, \matr{B}) \bigr|  \leq \\
        & \left| \mu_i(\matr{H}, \matr{V})  - \mu_i(\matr{H}, \matr{B})  \right| \leq \frac{\lambda_{1}({\matr{N}^2})}{\lambda_m({\matr{V}^2})}   \bigl| \mu_m(\matr{H}, \matr{N}) - \mu_i(\matr{H}, \matr{B}) \bigr|
    \end{align*}
 In this scenario we have 
    \begin{align*}
 \bigl| \mu_1(\matr{H}, \matr{N}) - \mu_i(\matr{H}, \matr{B}) \bigr|  \leq \bigl| \mu_m(\matr{H}, \matr{N}) - \mu_i(\matr{H}, \matr{B}) \bigr|
    \end{align*}
 (\ref{eq:specral_stability_main}) holds.
    \item Case II: $s_m =1$, that is $ \mu_m(\matr{H}, \matr{N}) \geq  \mu_i(\matr{H}, \matr{B})$.  This occurs for example if $\matr{N}$ only contains 
 energy vectors above the target range.\\
 Then we also have $s_1 =1$ and $\mu_i(\matr{H}, \matr{V}) \geq \mu_i(\matr{H}, \matr{B})$.
 Then, remaining indices have 
    \begin{alignat*}{5}
        &s_1 &&: \qquad && k-(1-m) && = 1 \quad &&\text{ with } \quad k=m\\
        &s_m &&: \qquad && k+(1-m) && = m \quad &&\text{ with } \quad k=1
     \end{alignat*}
 and we yield
    \begin{align*}
 \frac{\lambda_{m}({\matr{N}^2})}{\lambda_1({\matr{V}^2})}   \bigl| \mu_m(\matr{H}, \matr{N}) & - \mu_i(\matr{H}, \matr{B}) \bigr| \leq \\
        & \left| \mu_i(\matr{H}, \matr{V}) - \mu_i(\matr{H}, \matr{B})  \right| \leq \frac{\lambda_{1}({\matr{N}^2})}{\lambda_m({\matr{V}^2})}   \bigl| \mu_1(\matr{H}, \matr{N}) - \mu_i(\matr{H}, \matr{B}) \bigr|
    \end{align*}
 As above this inequality implies (\ref{eq:specral_stability_main}).
    \item Case III: $s_m=-1$ and $s_1=1$. Intuitively this corresponds to the most common case, when the noise contains eigenvector contributions above and below the target energy range. Then
 the generalized eigenvalues $\mu_l(\matr{H}, \matr{N})$ of the noise can take values above, below and also within the target range.\\
 We have 
        \begin{alignat*}{5}
           &s_1=1 &&: \qquad && k+(1-m) && = 1 \quad &&\text{ with } \quad k=m\\
           &s_m=-1&&: \qquad && k+(1-m) && = 1 \quad &&\text{ with } \quad k=m
        \end{alignat*}
 and obtain
        \begin{align*}
 \frac{\lambda_{1}({\matr{N}^2})}{\lambda_m({\matr{V}^2})} \bigl( \mu_m(\matr{H}, \matr{N}) & - \mu_i(\matr{H}, \matr{B}) \bigr) \leq \\
            & \mu_i(\matr{H}, \matr{V}) - \mu_i(\matr{H}, \matr{B})  \leq \frac{\lambda_{1}({\matr{N}^2})}{\lambda_m({\matr{V}^2})}   \left( \mu_1(\matr{H}, \matr{N}) - \mu_i(\matr{H}, \matr{B}) \right)
        \end{align*}
 Depending on the sign of  $\mu_i(\matr{H}, \matr{V}) - \mu_i(\matr{H}, \matr{B}) $ either the left or right hand side will give an upper bound 
 for the absolute value. Taking the maximum of the two candidates will surely give a valid upper bound.
    \end{itemize}
 With this, all possible cases are covered and the proof is complete.

\newpage 

\section{An Integrated Spectral Inequality}
\label{sec:IntegratedSpectralInequalities}

\subsection{Motivation: Full-Scale Analysis}
The purpose of this section is to establish a spectral inequality, that 
overcomes an overly sensitive dependence on the spectral range 
of the target operator $\matr{H}$ and that allows us to incorporate more information on the energy distribution of the guess vectors.
On the one hand, this enables a more precise estimate of a low-dimensional approximation in the spectral analysis of operators. 
On the other hand, additional information about the energy distribution of the guess vectors is essential for deriving bounds for unbounded operators, which are of particular physical interest.\\

Throughout this section, we assume that a protocol aims to generate guess vectors for spectral subspace $\mathcal{E}_{(E_a,E_b)}(\matr{H})$ to study 
all the eigenvalues in the range $(E_a,E_b)$. 
For the guess vectors $v_l$ we write 
\begin{align*}
 v_l = \sum_{k} a_{lk} \varphi_k.
\end{align*}
We say that the guess vector $v_l$ exhibits \emph{target clustering}
if the weights of the energy coefficients $|a_{lk}|^2$ and the distance of the eigenvalues $\lambda_k$ to the target range $(E_a,E_b)$
have a negative relationship. That is, even the noise weights $|a_{lk}|^2$ are concentrated around the target energy range $(E_a,E_b)$.\\

Most protocols that aim to generate guess vectors for spectral subspace $\mathcal{E}_{(E_a,E_b)}$ of an operator $\matr{H}$,
exhibit target clustering inherent to their machinery. Theorem \ref{thm:specral_stability}  does not allow us to incorporate this information
to a satisfactory extent, even though it can be crucial for an accurate estimate. The problem of the spectral stability inequality of Theorem \ref{thm:specral_stability}
\begin{align*}
 \left| \mu_i(\matr{H}, \matr{V}) - \mu_i(\matr{H}, \matr{B})  \right| & \leq \frac{\lambda_{1}({\matr{N}^2})}{\lambda_m({\matr{V}^2})}  \max_{j \in {1,m}} \bigl| \mu_j(\matr{H}, \matr{N})- \mu_i(\matr{H}, \matr{B}) \bigr|,
\end{align*}
is that $\mu_{1,m}(\matr{H}, \matr{N})$ has to be either estimated by the leading eigenvalues $\lambda_{1,m}(\matr{H})$, or it has to be known 
that the noise contributions are zero outside a larger spectral range, that is $\operatorname{span}(\matr{V}) \subset \mathcal{E}_{(E_A,E_B)}$ with $E_A < E_a$ and $E_b < E_B$.\\

The first kind of estimate $\mu_{1,(m)}(\matr{H}, \matr{N}) \leq_{(\geq)}  \lambda_{1,(m)}(\matr{H})$ can still be interesting for approximating 
a high-dimensional matrix, into a lower dimension as seen in on the example of the second quantized electronic Hamiltonian in Section \ref{sec:Example_sec_quant}.
This is perhaps already interesting towards more clarity on the precision of dimension reduction algorithms in the context of quantum chemistry.\\

However, since we have already been  working on precision guarantees for approximate methods, we could also argue for a slightly more demanding position. Even the high-dimensional second quantized electronic Hamiltonian 
is only an approximation to an infinite-dimensional and unbounded operator that, moreover, consists of both, point and continuous spectrum.\footnote{Point spectrum corresponds 
to bound states, while a continuous spectrum corresponds to scattering states.} A mathematical estimate of the solution of an approximate model 
is already a step forward, and this is also where we expect the most important deviations to occur. 
Iterating error estimates across different scales is commonly referred to as \emph{multi-scale analysis}.
But could we instead right away encompass all the approximation errors across the different scales in a single inequality?
We refer to such an attempt as \emph{full-scale analysis}. \\

For a full-scale analysis, in the context of physical systems, a simple 
estimate $\mu_{1,(m)}$ by leading eigenvalues of the operator $\matr{H}$ is not sufficient, either because the operator is unbounded or because
the estimate would become too lose. Therefore, it is necessary to honor the clustering property of the guess vectors for an accurate estimate. 
The second way of achieving an estimate of $\mu_{1,(m)}$, which was mentioned above, assumes 
\begin{align}
 \operatorname{span}(\matr{V}) \subset \mathcal{E}_{(E_A,E_B)} \label{eq:span_assumption}
\end{align}
and can be regarded as one incarnation of the clustering property. However, the assumption (\ref{eq:span_assumption}) is rather 
strong. In algorithmic protocols or theoretical derivations, often a guess vector $v_l$ will have some very small contributions outside of $\mathcal{E}_{(E_A,E_B)}$. 
Nevertheless, intuition has it that the physical theory should not depend on whether (\ref{eq:span_assumption}) is exactly true or only approximately so, in an overly sensitive manner. \\

In Subsection $\ref{sec:IntSpecResults}$ we will introduce an integrated spectral inequality that can capture this intuition
and allow for a full-scale analysis. It is instructive to first introduce some preliminary definitions, which enable 
such an "integrated" spectral inequality.

\subsection{Preliminary Definitions: The Spectral Measure and Projections of the Guess Space}
A pivotal object for this section is a \emph{spectral measure} that is introduced 
in Section \ref{sec:spec_meas}.
The spectral measure as we will use for our approximation purposes, is similar to the spectral measure in rigorous scattering theory \cite{Simon1978AnOO}
and allows us to express target clustering which is essential for improved spectral inequalities.
It is also closely related to an operator-valued function, which encodes projections of the guess vectors onto spectral subspaces of  $\matr{H}$ and is
introduced in Section \ref{sec:op_val_func}. Take together these two concepts enable an integrated spectral inequality. 

\subsubsection{Complete Operators}
In the following, we will be somewhat more specific on the hermitian operator $\matr{H}$ that is subject to study. 

\begin{definition}
    \label{def:com_op}
 Let $\matr{H}$ be an hermitian operator on a Hilbert space $\Hil$. The spectrum of $\matr{H}$ can be divided into a
 point spectrum $\sigma_p$ and a continuous spectrum $\sigma_c$.\footnote{
    The point spectrum consist of discrete eigenvalues, as they are familiar form finite-dimensional matrices, and accumulation points. 
 Physical operators typically only have a single accumulation point where the discrete spectrum merges into the continuous.}

    \begin{itemize}
        \item The \textbf{Point Spectrum (\(\sigma_p\))} consists of all eigenvalues \(\lambda\) where there exists a non-zero vector \(|\varphi\rangle \in \Hil\) such that \(\mathbf{H}|\varphi\rangle = \lambda |\varphi\rangle\). Each eigenvector \(|\varphi\rangle\) associated with an eigenvalue \(\lambda\) is normalizable and belongs to the Hilbert space \(\Hil\).
        \item The \textbf{Continuous Spectrum (\(\sigma_c\))} consists of those \(\lambda\) for which \(\mathbf{H} - \lambda \mathbf{I}\) is not invertible, yet does not yield any normalizable eigenvectors. Instead, the "eigenfunctions" associated with \(\sigma_c\) are generalized eigenfunctions, often lying in an extension of $\Hil$ (like distributions).
 These functions are parameterized by points in a phase space \(\Gamma \subseteq \mathbb{R}^d\). 
    \end{itemize}

 An hermitian operator \(\mathbf{H}\) is said to be \emph{complete in $\Hil$} if it satisfies the spectral theorem in the relaxed sense 
    \begin{align} 
\mathbf{I} = \sum_{\lambda_k \in \sigma_p} |\varphi_k \rangle \langle \varphi_k| + \int_{\sigma_c} \int_{\Gamma_\lambda} |\varphi_k \rangle \langle \varphi_k | dk d\lambda,   \label{eq:completeness}
\end{align}
where \(\mathbf{I}\) is the identity operator on $\Hil$, \(|\varphi_k\rangle\) are the eigenfunctions or generalized eigenfunctions and $\Gamma_\lambda \subset \Gamma$ is the subspace of the phase 
space that corresponds to $\lambda \in \sigma_c$.
\end{definition}
Roughly speaking, the main point from 
Definition \ref{def:com_op} is a guarantee, that we can expand the guess vectors in terms of the eigenfunctions of the operator $\matr{H}$ 
at least in a relaxed sense. The definition is exactly sufficiently general to encompass quantum mechanical operators that 
describe physical systems.\\
It is more than general enough to include finite-dimensional matrices and compact operators for $\matr{H}$.
This is easily seen as they have only a point spectrum. The
subsequently derived inequalities are also of interest in the context of finite-dimensional matrix inequalities. \\

An advantage of Definition \ref{def:com_op} is that we can jointly treat discrete and continuous spectra, which is necessary for a full-scale analysis.
Hamiltonians that describe physical systems are indeed a mixture of continuous and discrete spectra.
They do not have access to the spectral theorem, but the completeness relation in the sense of (\ref{eq:completeness}) is 
a satisfactory substitute.
As bound and scatter states are often in a superposition, the joint treatment is closer to physical reality \cite{Moyal_1949}.\\

We will frequently make use of a spectral factorization of the phase space in the integrations 
\begin{align*}
    \int_\Gamma f(k)  dk = \int_{\sigma_c} \int_{\Gamma_\lambda} f(k)  dk d\lambda.
\end{align*}
A slightly more careful account for the definition of a GEP in the context of dimension reduction is given.

\begin{definition}
    \label{def:GEP_com}
 Let $\matr{H}$ be an hermitian operator, complete in $\Hil$ as in Definition \ref{def:com_op}. 
 Consider guess vectors $|v_1\rangle, \cdots , |v_M\rangle \in \Hil$ with expansions 
    \begin{align*}
 |v_l\rangle & = \sum_{\lambda_k \in \sigma_p } \alpha_{lk} |\varphi_k\rangle + \int_{\sigma_c} \int_{\Gamma_\lambda }\alpha_l(k) |\varphi_k\rangle dk   d\lambda.
    \end{align*}
 The guess vector operator $\matr{\hat V}: \mathbb{C}^M \to \Hil$ is defined as 
    \begin{alignat*}{2}
 \matr{\hat V} & = \sum_{l=1}^{M} |v_l\rangle \langle e_l|\\
                   & =  \sum_{l=1}^{M} \sum_{\lambda_k \in \sigma_p }  \alpha_{lk} |\varphi_k\rangle \langle e_l|  + \sum_{l=1}^{M} \int_{\sigma_c} \int_{\Gamma_\lambda}  \alpha_l(k) |\varphi_k\rangle  \langle e_l| dk  d\lambda .
     \end{alignat*}
 The matrix representation of $\matr{\hat V}$ in the left basis $\{|\varphi_k\rangle\}$ and right basis $\{|e_l\rangle\}$ is denoted as $\matr{V}$. 
 The generalized eigenvalue problem that is induced by the operator $\matr{H}$ and the guess vector matrix $\matr{V}$
 is denoted as $(\matr{H}, \matr{V})$.
\end{definition}
The completeness relation $(\ref{eq:completeness})$ implies $\langle \varphi_k | \varphi_{k'} \rangle = \delta(k-k')$ and 
the entries of the guess vector matrix $\matr{V}_{kl}$ are given by  $a_{kl}$ (discrete) or $ a_l(k)$ (continious). \\

In order to leverage more information on the guess vectors for a sharper precision guarantee, 
a certain operator-valued function proves to be a valuable tool.

\subsubsection{An Operator-Valued Function for Guess Vector Projections}
\label{sec:op_val_func}

For an hermitian and complete GEP as in Definition \ref{def:GEP_com} we introduce an operator-valued function that projects 
the guess vectors on spectral subspaces of $\matr{H}$.\\ 

Let $\text{Hom}(\mathbb{C}^M, \Hil)$ denote the set of linear operators that map from $\mathbb{C}^M$ to $\Hil$ and 
let $\mathcal{P}(\mathbb{R})$ denote the power set of $\mathbb{R}$.\footnote{ That is $\mathcal{P}(\mathbb{R})$ is the set of all subsets of $\mathbb{R}$.}
For a given guess vector operator $\matr{\hat V}$ to a Hamiltonian $\matr{H}$ with spectrum $\{\sigma_p , \sigma_c\}$,
an operator-valued function $\matr{\hat V}\ast$ is introduced as,
\begin{align*}
 \matr{\hat V}\ast& : \mathcal{P}(\mathbb{R}) \to \text{Hom}(\mathbb{C}^M, \Hil),\\ 
 \matr{\hat V}I & =  \sum_{l=1}^{M} \sum_{\lambda_k \in I_p}    \alpha_{lk} |\varphi_k\rangle \langle e_l| +   \sum_{l=1}^{M} \int_{ I_c }   \int_{\Gamma_\lambda}  \alpha_l( k) |\varphi_k\rangle \langle e_l| dk d\lambda . 
\end{align*}
Here, the shorthand notation $I_p = I \cap \sigma_p$ and $I_c = I \cap \sigma_c$ was used.
In the sum of $\matr{\hat V}I$ only those terms are included, which correspond to eigenvalues $\lambda$ of $\matr{H}$ that are in $I_p \subset \mathbb{R}$.
The integral is taken over the phase space variables $k$ that have eigenvalues $\lambda$ in $I_c$. 
Since $\matr{H}$ is hermitian we have for example $\matr{\hat V}\mathbb{R} = \matr{\hat{V}}$. 
If, for example, $I = [E,E')$ is an interval we write $\matr{\hat{V}}[E,E')$ and
\begin{align*}
 \matr{\hat V} [E,E') & =  \sum_{l=1}^{M} \sum_{\lambda_k \in [E,E')}    \alpha_{lk} |\varphi_k\rangle \langle e_l| +   \sum_{l=1}^{M} \int_{[E,E') \cap \sigma_c}   \int_{\Gamma_\lambda}  \alpha_l( k) |\varphi_k\rangle \langle e_l| dk d\lambda
\end{align*}
We also exploit a short-hand notation $\matr{\hat{V}}^\uparrow(E) = \matr{\hat{V}}(-\infty,E )$ and $\matr{\hat{V}}^\downarrow(E) = \matr{\hat{V}}(E, \infty)$.
For example, 
\begin{align*}
 \matr{\hat{V}}^\uparrow(E) & =  \sum_{l=1}^{M} \sum_{\lambda_k <E }    \alpha_{lk} |\varphi_k\rangle \langle e_l| +   \sum_{l=1}^{M} \int_{ \inf {\sigma_c} }^E   \int_{\Gamma_\lambda}  \alpha_l( k) |\varphi_k\rangle \langle e_l| dk d\lambda.
\end{align*}
The definition of the operator-valued function $\matr{\hat{V}}\ast$ induces the definition of the matrix-valued function $\matr{V} \ast$.
For example $\matr{V} (E) = [\langle \varphi_k |\matr{\hat{V}}(E) | e_l \rangle ]_{kl}$
and accordingly for $\matr{V}[E,E'), \matr{V}^\uparrow(E),  \dots$.\\
Similarly, a matrix-valued function for spectral decompositions of the Gram matrix is induced,
\begin{align*}
 \matr{V}^2 \ast:  \mathcal{P}(\mathbb{R}) \to  \mathbb{C}^{M \times M}, \quad \matr{V}^2I = (\matr{V}I)^2.
\end{align*}

The dependence of generalized eigenvalues $\mu_i(\matr{H}, \matr{V})$ on the Hamiltonian $\matr{H}$ 
will no longer explicitly be denoted and instead we write $\mu_i(\matr{V})$.
The definition of $\matr{V}\ast$ as matrix-valued function also induces interpretations of
of $ \lambda_j(\matr{V}^2), \lambda_j^\uparrow(\matr{V}^2), \mu_j(\matr{V}), \mu_j^\downarrow(\matr{V}),\dots$
as functions defined on $\mathcal{P}(\mathbb{R})$. In Table \ref{tab:not_induced} a representative collection of the arising notation 
is given and will be subsequently used. 

\begin{table}[h]
    \centering
    \begin{tabular}{|c|c|}
    \hline
    \textbf{Notation} & \textbf{Definition} \\
    \hline   
    $\lambda_j(\matr{V}^2)[E,E')$ & $\lambda_j(\matr{V}^2[E,E'))$ \\
    \hline  
    $\lambda_j^\uparrow(\matr{V}^2)(E)$ & $\lambda_j((\matr{V}^\uparrow(E))^2)$ \\
    \hline
    $\mu_j(\matr{V})(E,E']$ & $\mu_j(\matr{V}(E,E'])$ \\
    \hline
    $\mu_j^\downarrow(\matr{V})(E)$ & $\mu_j(\matr{V}^\downarrow(E))$ \\
    \hline
    \end{tabular}
    \caption{Representative collection of notation that is induced by the matrix-valued function $\matr{V}\ast$.}
    \label{tab:not_induced}
\end{table}

We will also make use of special notation to denote the lowest non-vanishing eigenvalue of a Gramm matrix $\matr{V}^2$ 
by writing $\lambda_{m^*}(\matr{V}^2)$. The definition is similarly induced trough the matrix-valued function $\matr{V}\ast$,
\begin{align*}
    \lambda_{m^*}(\matr{V}^2)(E) := \min_{x \in \operatorname*{Ker}(\matr{V}^2(E))^\perp} \frac{\langle x, \matr{V}^2(E) x \rangle}{\langle x, x \rangle}.
\end{align*}

\subsubsection{The Spectral Measure}
\label{sec:spec_meas}

The final concept that is necessary to express the target clustering property and to
obtain approximation guarantees, that honor the full spectrum of an operator $\matr{H}$ in its unbounded and 
mixed nature, is given by the \emph{spectral measure}.\\

We briefly motivate: For an hermitian and complete GEP as in Definition \ref{def:GEP_com} consider the map 
\begin{align*}
    \int |\alpha(E)|^2 (\ast) & : \mathcal{P}(\mathbb{R}) \to \mathbb{R}, \\
    \int |\alpha(E)|^2 (I) & := \operatorname{Tr}[ (\matr{\hat V} I )^2]\\
    & = \sum_{\lambda_k \in I \cap \sigma_p} |a_{lk}|^2 + \int_{I \cap \sigma_c} \int_{\Gamma_\lambda} |a_l(k)|^2 dk d\lambda.
\end{align*}

 So far, $\int |\alpha(E)|^2$ is merely a sequence of symbols that is used to denote the map,
 but the map is indeed seen to be a measure on $\mathbb{R}$. 
 Its normalization can be directly obtained from the trace of the 
 Gram matrix 
    \begin{align*}
 \operatorname{Tr}[ (\matr{\hat V} \mathbb{R} )^2] = \operatorname{Tr}[\matr{\hat V}^2] = \operatorname{Tr}[\matr{V}^2].
    \end{align*}

 The following definition allows us to treat the measure form above as an integral.
    \begin{definition}[Spectral Measure]
        \label{def:spec_meas}
 For a GEP \((\matr{H}, \matr{V})\) as in Definition \ref{def:GEP_com}, the \emph{spectral measure} \( |\alpha(\cdot)|^2 : \mathbb{R} \to \mathbb{R}_{\geq 0} \) is defined as:
        \begin{align}
 |\alpha(E)|^2 = \sum_{\lambda_k \in \sigma_p} \delta(E - \lambda_k) \sum_{l=1}^M |a_{lk}|^2 + \int_{\sigma_c} \delta(E - \lambda) \int_{\Gamma_\lambda} \sum_{l=1}^M |a_l(k)|^2 \, dk \, d\lambda. \label{eq:spec_meas_def}
        \end{align}
        
 A spectral integral for a GEP with point spectra \(\sigma_p\) and continuous spectra \(\sigma_c\) is defined as follows. Let \(I \subset \mathbb{R}\) and \(f: \mathbb{R} \to \mathbb{C}\) be Riemann integrable on \(\sigma_c\) and a linear combination of delta functions $\delta(E -\lambda)$ on $\mathbb{R} \setminus \sigma_c $. Then, the integral is given by
        \begin{align*}
            \int_{I} f(E) \, dE = \int_{\sigma_d} \mathbb{1}_I(E) f(E) \, dE + \int_{\sigma_c} \mathbb{1}_I(E) f(E) \, dE,
        \end{align*}
 where the first integral over \(\sigma_p\) satisfies linearity and 
        \begin{align*}
            \int_{\sigma_d} \mathbb{1}_I(E) \delta(E - \lambda_k) \, dE = \begin{cases}
            1 & \text{if } \lambda_k \in I \cap \sigma_d \\
            0 & \text{otherwise}
            \end{cases}.
        \end{align*}
 The second integral is the Riemann integral over \(\sigma_c\).
        \end{definition}
 Here, $\mathbb{1}_I(E)$ denotes the indicator function of the set $I$. We solely introduced 
 the spectral integral to deal with measures, which are a mixture of discrete and continuous measures.\\

 Regularity of the spectral measure $|\alpha(E)|^2$ within the continuous spectrum is, in practice, not a problem. 
 Within statistical physics, $\Gamma_\lambda$ is often referred to as the \emph{energy shell}. For a Hamiltonian $\matr{H}$ that is 
a smooth function in the phase space variables $k$, the energy shells 
also vary smoothly in $\lambda$. This means that phase space variables $k$ and $k'$ that are close to one another
will lie in energy shells that have small energy differences. In particular, reasonable regularity assumptions on the 
continuous coefficients $a_l(k)$ of the guess vectors will imply that the spectral measure $|\alpha(E)|^2$ is also sufficiently regular within $\sigma_c$.\\

The following statement is an immediate consequence of the definition, but we want to highlight the value of the spectral measure. 
\begin{proposition}
    \label{prop:spec_meas}
 For a GEP as in Definition \ref{def:GEP_com} where $\matr{H}$ and a discrete spectrum $\sigma_d$ and continuous spectrum $\sigma_c$.
 For all $I \subset \mathbb{R}$ the spectral measure $|\alpha(E)|^2$ satisfies
    \begin{align}
        \int_{I} |\alpha(E)|^2 dE & = \operatorname{Tr}[\matr{V}^2 I].   \label{eq:spec_measur_meas}
    \end{align}
    It can be obtained as a derivate from cumulative distribution function of the Gramm matrix $\matr{V}^2$,
    \begin{align}
        |\alpha(E)|^2 = \frac{d}{dE'} \operatorname{Tr}[\matr{V}^2(-\infty, E']] \bigg|_{E' = E}. \label{eq:spec_measur_meas_derivate}
    \end{align}    
\end{proposition}
Equation (\ref{eq:spec_measur_meas}) and (\ref{eq:spec_measur_meas_derivate}) show that the spectral measure is analogous to a probability measure.
It describes the distribution of the guess vectors with respect to an energy decomposition. It is simultaneously a mixture of discrete and continuous measures and
the usage of delta distributions in the definition of $|a( E )|^2$ is merely a technicality to embed this joint nature. 
The function $\operatorname{Tr}[\matr{V}^2(-\infty, E']]$ is a monotone increasing function in $E$ as seen in Lemma \ref{lem:monotone_gram_eigenvalues}. 
Therefore it is indeed a valid cumulative distribution function, after we allow for delta pikes as derivatives at discontinuities.\\

We immediately obtain the following identities for the spectral measure. 
\begin{corollary}
 For a GEP and a spectral measure as in Proposition \ref{prop:spec_meas} we have
    \label{cor:spec_meas}
    \begin{align}
        \int_{-\infty}^{\infty} |\alpha(E)|^2 dE & = \operatorname{Tr}[\matr{V}^2]\\
        \lambda_1(\matr{V}^2 I) & \leq  \int_{I} |\alpha(E)|^2 dE \label{eq:spec_meas_ineq}
    \end{align}
 If the GEP has a signal-noise decomposition $\matr{V} = \matr{B} + \matr{N}$ to the spectral subspace $\mathcal{E}_{(E_a,E_b)}$, then the spectral measure has,
    \begin{align}
        \int_{-\infty}^{E_a} |\alpha(E)|^2 dE + \int^{\infty}_{E_b} |\alpha(E)|^2 dE = \operatorname{Tr}[\matr{N}^2]. \label{eq:spec_meas_noise}
    \end{align}
\end{corollary}
The overall normalization of the spectral measure can thus directly be computed from the Gram matrix $\matr{V}^2$ that 
is accessible to the computer. Inequality (\ref{eq:spec_meas_ineq}) follows from 
the simple estimate 
\begin{align*}
 \operatorname{Tr}[\matr{V}^2I] = \lambda_1( \matr{V}^2I) + \cdots + \lambda_M(\matr{V}^2I) \geq \lambda_1(\matr{V}^2I).
\end{align*}

To illustrate the accuracy of (\ref{eq:spec_meas_ineq}) we consider for simplicity the example of an interval $I$ that 
has no overlap with the continuous spectrum $I\cap \sigma_c = \varnothing$ and only contains discrete eigenvalues. 
The rank of $\matr{V}^2I$ is bounded by the number of eigenvalues in $I$, 
\begin{align*}
 \operatorname{rank}(\matr{V}^2I) \leq | \{ \lambda_k \in I \} |.
\end{align*}
If $I = (\lambda_k - dE ,\lambda_k +dE)$ and $dE$ is small enough the rank of $\matr{V}^2(I)$ will eventually be bounded by 
the multiplicity of the eigenvalue $\lambda_k$. Thus, the estimate $\lambda_1( \matr{V}^2I) \leq \operatorname{Tr}[\matr{V}^2I]$
becomes in the limit $\operatorname{len}(I) \to 0$ up to the multiplicity of $\lambda_k$ exact.\\

Identity (\ref{eq:spec_meas_noise}) follows immediately from Proposition \ref{prop:spec_meas} and the definition of the noise matrix $\matr{N}$,
\begin{align*}
 \matr{N}^2 = \matr{V}^2(-\infty,E_a] + \matr{V}^2[E_b,\infty).
\end{align*}

Theorem \ref{thm:spectral_inequalities} and inequality (\ref{eq:spec_meas_ineq}) are the main tools
from which we derive an integrated spectral inequality. 

\subsection{Integrated Spectral Inequalities}
To give the main result, some mild definitions are used to compress the notation. 
\subsubsection{Signal GEPs and Certifiable Signal GEPs}
\begin{definition}[signal GEP]
    \label{def:signal_GEP}
 A hermitian GEP $(\matr{H},\matr{V})$, 
 will be called a \emph{signal} to an energy interval $(E_a,E_b)$, if the guess vector Gramm matrix $\matr{V}^2$ 
 satisfies 
    \begin{align*}
        \int_{E_b}^{\infty} \frac{|\alpha(E)|^2 }{\lambda_{m^*}^\uparrow(\matr{V}^2)(E)} dE \leq 1 \qquad \text{and} \qquad \int_{-\infty}^{E_a} \frac{|\alpha(E)|^2 }{\lambda_{m^*}^\downarrow(\matr{V}^2)(E)} dE \leq 1.
    \end{align*}
 For such signal-GEPs, let $\matr{B} = \matr{V}(E_a,E_b)$ denote the pure signal contributions of the guess vectors.
\end{definition}
Recall that $\lambda_{m^*}^\uparrow(\matr{V}^2)(E)$ returns the smallest non-zero eigenvalue of the Gram matrix $\matr{V}^\uparrow(E)^2$.
Definition \ref{def:signal_GEP} is very mild. Any Gramm matrix 
that attempts to approximate the eigenvalues in $(E_a,E_b)$ with a precision guarantee should easily satisfy the prerequisites of 
a signal GEP. On the contrary, a signal GEP that does not satisfy the prerequisites of Definition \ref{def:signal_GEP} cannot be expected to deliver 
useful approximations to the eigenvalues in $(E_a,E_b)$. A slightly stronger but still mild condition is given in the following definition.
\begin{definition}[certifiable signal GEP]
    \label{def:certified_signal_GEP}
 Consider a signal GEP $(\matr{H},\matr{V})$ to an interval $(E_a,E_b)$ as
 in Definition \ref{def:certified_signal_GEP}. Let $m$ denote the rank of $\matr{V}^2$. The signal GEP is called \emph{certifiable}, if the guess vector Gram matrix $\matr{V}^2$
 satisfies
    \begin{align*}
        0\leq \frac{\int_{E_b}^{\infty} |\alpha(E)|^2 dE}{\lambda_{m}(\matr{V}^2)- \lambda_{1}^\downarrow(\matr{V}^2)[E_b]} \leq 1 \qquad \text{and} \qquad  0\leq \frac{\int_{-\infty}^{E_a} |\alpha(E)|^2 dE}{\lambda_{m}(\matr{V}^2) - \lambda^\uparrow_1(\matr{V})[E_a]} \leq 1
    \end{align*}
 Here, the notation $\lambda_{1}^\downarrow(\matr{V}^2)[E_b] := \lambda_{1}(\matr{V}[E_b, \infty)^2)$ and $\lambda_{1}^\uparrow(\matr{V}^2)[E_a] := \lambda_{1}(\matr{V}(-\infty, E_a]^2)$ was used.
\end{definition}
Given a GEP $(\matr{H}, \matr{V})$ of rank $m$ and an estimate on the leading eigenvalue $\lambda_1(\matr{N}^2)$ of the noise Gram matrix,
a sufficient condition to ensure that the GEP is a certifiable signal GEP is given by 
\begin{align}
    \lambda_m(\matr{V}^2) \geq (m+1) \lambda_1(\matr{N}^2). \label{eq:certifiable_condition}
\end{align}
This is easily seen from (\ref{eq:spec_meas_noise}) of Corollary \ref{cor:spec_meas},
\begin{align*}
    \int_{E_b}^{\infty} |\alpha(E)|^2 dE \leq \operatorname{Tr}[\matr{N}^2] \leq m \lambda_1(\matr{N}^2).
\end{align*}
In particular,
\begin{align*}
 \frac{\int_{E_b}^{\infty} |\alpha(E)|^2 dE}{\lambda_{m}(\matr{V}^2)- \lambda_{1}^\downarrow(\matr{V}^2)[E_b]} \leq \frac{m \lambda_1(\matr{N}^2)}{(m+1 -1 ) \lambda_1(\matr{N}^2)} \leq 1.
\end{align*}
Thus, the condition (\ref{eq:certifiable_condition}) shows that certifiable signal GEPs are indeed certifiable.
A GEP that attempts to give approximations of high precision, should easily satisfy (\ref{eq:certifiable_condition}).\footnote{
In these statements, we have ignored the so-called singularity problem of the Gram matrix $\matr{V}^2$, which is known in 
the numerical literature. It is often based on overestimating the number of eigenvalues in $(E_a,E_b)$
to "better span" the target subspace, but also leads to the disadvantage of singular Gram matrices.
However, as we will elaborate in Section \ref{sec:NumTheoSimbiosis}, it 
can be completely avoided, while still maintaining the advantage of slight overestimation of the number of eigenvalues in $(E_a,E_b)$.}

\subsubsection{Main Result}
\label{sec:IntSpecResults}

We are ready to give an integrated spectral inequality on the precision
of projection-based GEPs as a dimension reduction scheme. 
\begin{theorem}
    \label{thm:integrated_spectral_stability}
 Let the hermitian GEP $(\matr{H},\matr{V})$ be a signal to the energy interval $(E_a,E_b)$ as in Definition \ref{def:signal_GEP}.
 For all $i$ the following inequalities hold,
    \begin{align}
        \mu_i(\matr{V}) - \mu_i(\matr{B}) & \leq \int_{E_b}^{\infty} \frac{|\alpha(E)|^2 }{\lambda_{m^*}^\uparrow(\matr{V}^2)(E)} \left(E -  \mu_i(\matr{B}) \right) dE \label{eq:integrated_spec_ineq_up_thm}\\
        \mu_i(\matr{V}) - \mu_i(\matr{B}) & \geq \int_{-\infty}^{E_a} \frac{|\alpha(E)|^2 }{\lambda_{m^*}^\downarrow(\matr{V}^2)(E)} \left(E -  \mu_i(\matr{B}) \right) dE.  \label{eq:integrated_spec_ineq_down_thm}
    \end{align}
\end{theorem}
Compared to Theorem \ref{thm:specral_stability} of Section \ref{sec:SpecIneq}, the integrated spectral inequalities are of much more smooth nature. 
The precision guarantee of Theorem \ref{thm:integrated_spectral_stability} has no longer 
an explicit dependence on the target operator $\matr{H}$, that is subject to study. 
Instead, all the high-dimensional structure of the original problem is 
efficiently compressed into the spectral measure $|\alpha(E)|^2$ which is a one-dimensional function.  The inequalities (\ref{eq:integrated_spec_ineq_up_thm}) and (\ref{eq:integrated_spec_ineq_down_thm})
show that if the spectral measure is concentrated around the target interval $(E_a,E_b)$ and decays sufficiently fast at the tails, 
a sharp precision guarantee can be given.
As the spectral measure is specific to the guess vectors, a generating protocol 
also has the opportunity to derive a concentration guarantee on $|\alpha(E)|^2$.\\

We note that in the denominator $\lambda_{m^*}^\uparrow(\matr{V}^2)(E)$ of the bound in Theorem \ref{thm:integrated_spectral_stability} there, is no longer the smallest eigenvalue of the full Gram 
matrix $\matr{V}^2$, but the smallest non-vanishing eigenvalue of $\matr{V}^\uparrow(E)^2$. This can provide 
an explanation on numerical observations that have been made in the literature, as we will further elaborate in Section \ref{sec:NumericalPractices}.
However,  $\lambda_{m^*}^\uparrow(\matr{V}^2)(E)$ is not exactly accessible information in practice. 
Below we give a corollary which we believe is the most user-friendly implication of Theorem \ref{thm:integrated_spectral_stability}.

\begin{corollary}
    \label{cor:integrated_spectral_stability}
 Let the hermitian GEP $(\matr{H},\matr{V})$ be a certifiable signal to the energy interval $(E_a,E_b)$ as in Definition \ref{def:certified_signal_GEP}.
 For all $i$ the following inequalities hold,
    \begin{align}
        \mu_i(\matr{V}) - \mu_i(\matr{B}) & \leq \frac{\int_{E_b}^{\infty} \left(E -  \mu_i(\matr{B}) \right) |\alpha(E)|^2 dE   }{\lambda_{m}(\matr{V}^2)- \lambda_{1}^\downarrow(\matr{V}^2)[E_b]},  \label{eq:integrated_spec_ineq_up_cor} \\
        \mu_i(\matr{V}) - \mu_i(\matr{B}) & \geq \frac{\int_{-\infty}^{E_a}\left(E -  \mu_i(\matr{B}) \right)|\alpha(E)|^2 dE}{\lambda_{m}(\matr{V}^2) - \lambda^\uparrow_1(\matr{V})[E_a]}.   \label{eq:integrated_spec_ineq_down_cor}
    \end{align}
\end{corollary}
We can further estimate with $\lambda^\uparrow_1(\matr{V})[E_a], \lambda_{1}^\downarrow(\matr{V}^2)[E_b] \leq \lambda_1(\matr{N}^2)$
such that the inequalities only require control over $\lambda_1(\matr{N}^2)$ and $|\alpha(E)|^2$ to give an estimate. \\

The integrated spectral inequalities are applied to give a precision guarantee for Filter Diagonalization, which is subject to the next chapter. \\

\subsection{Proof Overview}

The following is the main Proposition from which Theorem \ref{thm:integrated_spectral_stability} will be derived.
\begin{proposition} 
 Let $(\matr{H},\matr{V})$ be a hermitian GEP. \\
 Let $E_b \in \mathbb{R}$ be such that $\matr{V}^\uparrow(E_b) \neq 0$. Then 
    \label{prop:integrated_spectral_inequality}
    \begin{align}
        \mu_i(\matr{V}) - \mu_i(\matr{V}^\uparrow(E_b)) & \leq \int_{E_b}^{\infty} \frac{|\alpha(E)|^2 dE}{\lambda_{m^*}^\uparrow(\matr{V}^2)(E)}\bigl( E - \mu_i( \matr{V}^\uparrow(E)) \bigr) dE.    \label{eq:integrated_spec_ineq_up}
    \end{align}
 Let $E_a \in \mathbb{R}$ be such that $\matr{V}^\downarrow(E_a) \neq 0$. Then 
    \begin{align}
        \mu_i(\matr{V}) - \mu_i(\matr{V}^\downarrow(E_a)) & \geq \int_{-\infty}^{E_a} \frac{|\alpha(E)|^2 dE}{\lambda_{m^*}^\downarrow(\matr{V}^2)(E)}\bigl( E - \mu_i( \matr{V}^\downarrow(E)) \bigr) dE. \label{eq:integrated_spec_ineq_down}
    \end{align}
\end{proposition}
\noindent We recall that in our notation $\lambda_{m^\star}(\matr{V}^2)(E)$ denotes the smallest non-zero eigenvalue of $\matr{V}^2(E)$.
The proof of Proposition \ref{prop:integrated_spectral_inequality} relies on a telescope sum 
and applications of the spectral inequalities of Theorem \ref{thm:spectral_inequalities} and is given in Section \ref{sec:proof_prop_integrated_spectral_inequality}.\\

The matrix valued function $\matr{V}\ast$ induces an interpretation of 
the eigenvalues $\mu_i(\matr{V})\ast$ as real-valued functions defined on $\mathcal{P}(\mathbb{R})$.
They are found to have a monotonicity property, that is given in Corollary \ref{cor:monotone_gen_eigenvalues}.

\begin{corollary}
    \label{cor:monotone_gen_eigenvalues}
 For an hermitian GEP $(\matr{H}, \matr{V})$ consider real-valued functions \\
    \begin{minipage}{0.5\textwidth}
        \begin{align*}
            &\mu_{i}^{\uparrow}(\matr{V})(\ast): \mathbb{R} \to \mathbb{R},\\
            &\mu_{i}^\uparrow(\matr{V})(E) = \mu_i(\matr{V}^\uparrow(E)),
        \end{align*}
        \end{minipage} \text{ and }
        \begin{minipage}{0.5\textwidth}
        \begin{align*}
            &\mu_{i}^{\downarrow}(\matr{V})(\ast): \mathbb{R} \to \mathbb{R},\\
            &\mu_{i}^\downarrow(\matr{V})(E) = \mu_i(\matr{V}^\downarrow(E)).
        \end{align*}
    \end{minipage}
 The function $\mu_{i}^\uparrow(\matr{V})(E)$ is monotonically increasing and $\mu_{i}^\downarrow(\matr{V})(E)$ is monotonically decreasing in $E$.
\end{corollary}
\noindent Corollary \ref{cor:monotone_gen_eigenvalues} is proven by an application of the spectral inequalities of Theorem \ref{thm:spectral_inequalities}.
An analogous monotonicity property holds for the eigenvalues of the Gram matrix $\matr{V}^2$ as seen in Corollary \ref{lem:monotone_gram_eigenvalues}.
\begin{corollary}
    \label{lem:monotone_gram_eigenvalues}
 For an hermitian GEP $(\matr{H}, \matr{V})$ consider real valued functions \\
    \begin{minipage}{0.5\textwidth}
        \begin{align*}
            &\lambda_{i}^{\uparrow}(\matr{V}^2)(\ast): \mathbb{R} \to \mathbb{R},\\
            &\lambda_{i}^\uparrow(\matr{V}^2)(E) = \lambda_i(\matr{V}^\uparrow(E)^2),
        \end{align*}
        \end{minipage} \text{ and }
        \begin{minipage}{0.5\textwidth}
        \begin{align*}
            &\lambda_{i}^{\downarrow}(\matr{V}^2)(\ast): \mathbb{R} \to \mathbb{R},\\
            &\lambda_{i}^\downarrow(\matr{V}^2)(E) = \lambda_i(\matr{V}^\downarrow(E)^2),
        \end{align*}
    \end{minipage}
 The function $\lambda_i^\uparrow(\matr{V}^2)(E)$ is monotonically increasing and $\lambda_i^\downarrow(\matr{V}^2)(E)$ is monotonically decreasing in $E$.
\end{corollary}
The monotonicity properties of Corollary \ref{lem:monotone_gram_eigenvalues} and Corollary \ref{cor:monotone_gen_eigenvalues} are used in the proof of Proposition \ref{prop:integrated_spectral_inequality}
and Theorem \ref{thm:integrated_spectral_stability}.
They also ensure that the integrands of Proposition \ref{prop:integrated_spectral_inequality} are indeed Riemann integrable, as is 
seen in Claim \ref{claim:riemann_integrability}.

\begin{claim}
    \label{claim:riemann_integrability}
 Let $(\matr{H},\matr{V})$ be an hermitian GEP. Define the minimal energy $E_{\min}$ as,
    \begin{align*}
 E_{\min}= \inf E \quad \text{ s.t. } \quad \lambda_{m^*}^\uparrow(\matr{V}^2)(E) > 0.
    \end{align*}
 Then the function $f^\uparrow(E)$ given by
    \begin{align*}
 f^\uparrow(E)= \frac{E - \mu_i^\uparrow(\matr{V})(E)}{\lambda_{m^*}^\uparrow(\matr{V}^2)(E)},
    \end{align*}
 is Riemann integrable over any compact interval $[E_1,E_2] \subset \mathbb{R}$
 with $E_2 > E_1 > E_{\min}$.
\end{claim}

\noindent The statement is a direct consequence of the fact that smooth functions of Riemann integrable functions are Riemann integrable.
 An analogous integrability statement holds for suitably defined $f^\downarrow(E)$ and $E_{\max}= \sup E  \text{ s.t. } \lambda_{m^*}^\downarrow(\matr{V}^2)(E) > 0$.
A more detailed derivation of Claim is given in Section \ref{sec:riemann_integrability}. \\

In Section \ref{sec:proof_prop_integrated_spectral_inequality}
we will also elaborate on how the delta distributions of the spectral measure are evaluated at discontinuities. 
For continuous spectra, $|\alpha(E)|^2$ will be sufficiently regular, if the coefficients $a_l(k)$ of the guess vectors are assumed to be sufficiently regular
and the energy shells $\Gamma_\lambda$ of $\matr{H}$ vary smoothly in $\lambda$.\footnote{E.g. $\matr{H} = \matr{P}^2$ has energy shells corresponding to spheres in the phase space variable $p$. }
Throughout the derivations, we will also make use of the following fact on the eigenvalues of hermitian GEPs.

\begin{fact}
    \label{fact:proper_eigenvalues_interval}
 Proper eigenvalues of a hermitian GEP have 
    \begin{align*}
 E_1 \leq \mu_i(\matr{V})(E_1,E_2) \leq E_2.
    \end{align*}
\end{fact}
Fact \ref{fact:proper_eigenvalues_interval} is essentially equivalent to 
spectral range statement of Fact \ref{fact:spec_range} in Section \ref{sec:RayleighQuotient} and 
does not require a proof. Finally, it is worth recalling Weyl's inequalities. For hermitian matrices $\matr{A}$ and $\matr{B}$
acting on a Hilbert space of dimensionality $m$ holds, 
\begin{align}
    \lambda_{i+j+1}(\matr{A}+\matr{B}) \leq \lambda_{i}(\matr{A}) + \lambda_{j}(\matr{B}) \leq \lambda_{i+j-m}(\matr{A} + \matr{B}). \label{eq:Weyl_classical}
\end{align}

In the next Section \ref{sec:derivations_integrated_spectral_stability}, the derivations are presented in an order consistent with logical dependencies.

\subsection{Derivations}
\label{sec:derivations_integrated_spectral_stability}

\subsubsection{Generalized Eigenvalues as Monotone Functions}

A slightly stronger version of Corollary \ref{cor:monotone_gen_eigenvalues} is proven.
\begin{corollaryprime}
    \label{cor:monotone_gen_eigenvalues2}
 Consider a hermitian GEP $(\matr{H},\matr{V})$ and $E_2\geq E_1$.\\
 Denote with  $m$ the rank of $\matr{V}^\uparrow(E_2)$. Then
    \begin{align*}
        \mu_{i}(\matr{V}^\uparrow(E_2)) - \mu_{i}(\matr{V}^\uparrow(E_1)) \geq \frac{\lambda_{m}(\matr{V}^2)[E_1,E_2)}{\lambda_1^\uparrow(\matr{V}^2)(E_2)} \left(\mu_{m*}({\matr{V}}[E_1,E_2))  - \mu_i(\matr{V}^\uparrow(E_1)) \right) \geq 0.
    \end{align*}
 Denote with  $m$ the rank of $\matr{V}^\uparrow(E_2)$. Then
    \begin{align*}
        \mu_{i}(\matr{V}^\downarrow(E_1)) - \mu_{i}(\matr{V}^\downarrow(E_2)) \leq \frac{\lambda_{m}(\matr{V}^2)(E_1,E_2]}{\lambda_1^\downarrow(\matr{V}^2)(E_1)} \left(\mu_{1}({\matr{V}}(E_1,E_2])  - \mu_i(\matr{V}^\downarrow(E_1)) \right) \leq 0.
    \end{align*}
 In particular, $\mu_{i}(\matr{V}^\uparrow(E))$ is a monotonically increasing function and $\mu_{i}(\matr{V}^\downarrow(E))$ is a monotonically decreasing function in $E$.
\end{corollaryprime}

\begin{proof}[\textit{Proof of Corollary \ref{cor:monotone_gen_eigenvalues2}}]
 Let $E_2\geq E_1$. The first inequality is proven first. The vector matrices $\matr{V}^\uparrow(E_1)$ and $\matr{V}[E_1,E_2)$ are spanned by eigenvectors of $\matr{H}$ to different eigenvalues
 and form in particular an orthogonal decomposition of $\matr{V}^\uparrow(E_2)$.
 Theorem \ref*{thm:spectral_inequalities} can be applied with $\matr{V} = \matr{V}^\uparrow(E_2)$, $\matr{B} = \matr{V}^\uparrow(E_1)$ and $\matr{N} = \matr{V} - \matr{B} =\matr{V}[E_1,E_2)$
 and $i=j$. 
 Let $m^*$ denote the rank of $\matr{V}[E_1,E_2)$. As $\matr{V}[E_1,E_2)$ is a projection of $\matr{V}^\uparrow(E_2)$, we have $m^* \leq m$.
 By Fact \ref{fact:proper_eigenvalues_interval} all proper eigenvalues $(\matr{H}, \matr{N})$ are in $[E_1,E_2)$. 
 Spurious eigenvalues $\mu_{j}( \matr{V})(E_1,E_2)$ can WLOG be set to $\mu_{m^*}( \matr{N}) \in[E_1,E_2)$.
 Thus we have for all eigenvalues $\mu_{j}( \matr{N}) \in[E_1,E_2)$. 
 Similarly, $\mu_{i}^\uparrow(\matr{V})(E_1) < E_1$ for all $i$ and $\mu_{m}( \matr{N}) - \mu_{i}(\matr{B}) > 0$.
 Therefore, Theorem \ref{thm:spectral_inequalities} can be applied with $s=1$ and inequality (\ref{eq:lower_spectral_Weyl}) gives 
    \begin{align*}
        \mu_{i}(\matr{V}^\uparrow(E_2)) - \mu_{i}(\matr{V}^\uparrow(E_1)) \geq \frac{\lambda_{m}(\matr{V}^2)(E_1,E_2)}{\lambda_1^\uparrow(\matr{V}^2)(E_2)} \left(\mu_{m*}({\matr{V}}(E_1,E_2))  - \mu_i(\matr{V}^\uparrow(E_1)) \right) \geq 0.
    \end{align*}
 For the second inequality Theorem \ref{thm:spectral_inequalities} is applied with $\matr{V} = \matr{V}^\downarrow(E_1)$, $\matr{B} = \matr{V}^\downarrow(E_2)$ and $\matr{N} = \matr{V} - \matr{B} =\matr{V}(E_1,E_2]$.
 Fact \ref{fact:proper_eigenvalues_interval} has  $\mu_{1}( \matr{V})(E_1,E_2] - \mu_{i}^\downarrow(\matr{V})(E_2) \leq 0$ for all $i$. Inequality (\ref{eq:upper_spectral_Weyl}) of Theorem \ref{thm:spectral_inequalities} can be 
 applied with $s=-1$ and gives at $i=j$,
    \begin{align*}
        \mu_{i}(\matr{V}^\downarrow(E_1)) - \mu_{i}(\matr{V}^\downarrow(E_2)) \leq \frac{\lambda_{m}(\matr{V}^2)(E_1,E_2]}{\lambda_1^\downarrow(\matr{V}^2)(E_1)} \left(\mu_{1}({\matr{V}}(E_1,E_2])  - \mu_i(\matr{V}^\downarrow(E_1)) \right) \leq 0.
    \end{align*}
\end{proof}

\subsubsection{Gram Matrix Eigenvalues as Monotone Functions}
A slightly stronger version of Corollary \ref{lem:monotone_gram_eigenvalues} is given. The proof is a simple application
of the classical Weyl inequalities on eigenvalues of hermitian matrices.
\begin{corollaryprime}
    \label{lem:monotone_gram_eigenvalues2}
 Consider a hermitian GEP $(\matr{H}, \matr{V})$. For $E_1 \leq E_2$ following inequalities hold,
    \begin{align*}
        \lambda_{i}^\uparrow(\matr{V}^2)(E_2) - \lambda_i^\uparrow(\matr{V}^2)(E_1) & \geq \lambda_m(\matr{V}^2[E_1,E_2) )\geq 0,\\
        \lambda_i^\downarrow(\matr{V}^2)(E_1) - \lambda_i^\downarrow(\matr{V}^2)(E_2)& \geq \lambda_m(\matr{V}^2(E_1,E_2] )\geq 0.
    \end{align*}
 In particular, $\lambda_i^\uparrow(\matr{V}^2)(E)$ is monotonically increasing and $\lambda_i^\downarrow(\matr{V}^2)(E)$ is monotonically decreasing in $E$.
\end{corollaryprime}

\begin{proof}[\textit{Proof of Corollary \ref{lem:monotone_gram_eigenvalues2}}]
 Only the first inequality is proven as the proof of the second inequality is analogous.\\
 We write $\matr{V}_1 = (\matr{V}^\uparrow(E_1))$, $\matr{V}^2_2 = (\matr{V}^\uparrow(E_2))$ and $\matr{N} = \matr{V}_2 -\matr{V}_1$. 
 By construction $[ \matr{V}_1,\matr{N}]$ is an orthogonal decomposition of $\matr{V}_2$, that is $\matr{V}_2 = \matr{V}_1 + \matr{N}$
 and $\matr{V}_1 \matr{N} = \matr{N} \matr{V}_1 = 0$.
 From the classical Weyl inequalities (\ref{eq:Weyl_classical}) we have
    \begin{align*}
        \lambda_i(\matr{V}_2^2)= \lambda_i(\matr{V}_1^2 + \matr{N}^2 )\geq \lambda_i(\matr{V}_1^2)  + \lambda_m(\matr{N}^2 ).
    \end{align*}
 Since $\matr{N}^2 = \matr{V}^2[E_1,E_2)$, the first inequality follows. The second inequality is analogously 
 derived with
 an orthogonal decomposition of $\matr{V}^\downarrow(E_1)$. 
\end{proof}

\subsubsection{Proof of Proposition \ref{prop:integrated_spectral_inequality}}
\label{sec:proof_prop_integrated_spectral_inequality}
The assumption imposed on $E_b$ and $E_a$ in Proposition \ref{prop:integrated_spectral_inequality} merely ensures that we are not dividing by zero in the integrals.

\begin{proof}[\textit{Proof of Proposition \ref{prop:integrated_spectral_inequality}}]
 Only the proof of the first inequality is given, as the proof of the second inequality is analogous. Let $(\matr{H},\matr{V})$ and $E_b$ be as in 
 Proposition \ref*{prop:integrated_spectral_inequality}.\\

 Let $dE>0$. With properties of the spectral measure from Proposition \ref{prop:spec_meas} and Corrolary \ref{cor:spec_meas} we have
    \begin{align}
        \lambda_1(\matr{V}^2)[E,E+  dE) \leq \operatorname*{Tr}[\matr{V}^2[E,E+  dE] ]= \int_{E}^{E+ d E} |\alpha(E)|^2 dE.  \label{eq:Riemann_int_enstimate}
    \end{align}
 The Fundamental Theorem of Calculus has 
    \begin{align}
 \lim_{ d E \to 0} \frac{1 }{d  E}  \int_{E}^{E+ d E} |\alpha(E)|^2 dE  = |\alpha(E)|^2. \label{eq:Fundamental_Theorem_limit}
    \end{align}

 Let $E^* > E_b$ and $N\in \mathbb{N}$. Set $d E = (E^* - E_b)/N$ and define $E_k = E_b + k d E$ for $k = 0, \cdots N-1$.
 The guess vector matrices $\matr{V}^\uparrow(E_{k})$ and $\matr{V}[E_k,E_{k+1})$ are an orthogonal decomposition of $\matr{V}^\uparrow(E_{k+1})$.
 Thus, the upper inequality (\ref{eq:upper_spectral_Weyl}) of Theorem \ref{thm:spectral_inequalities} can be applied with $\matr{V}=\matr{V}^\uparrow(E_{k+1})$,
    $\matr{B}=\matr{V}^\uparrow(E_{k})$ and $\matr{N}=\matr{V}[E_k,E_{k+1})$, where $i=j$ is taken. 
 Fact \ref{fact:proper_eigenvalues_interval} allows the estimate $\mu_1(\matr{V}(E_k,E_{k+1})) \leq E_{k+1}$.
 Fact \ref{fact:proper_eigenvalues_interval} also has  $\mu_i^\uparrow(\matr{V})(E_{k}) \leq E_k$ and, in particular,
    $ E_{k+1} - \mu_i^\uparrow( \matr{V})(E_{k}) \geq 0$.
 Thus (\ref{eq:upper_spectral_Weyl}) can be applied with $s=1$ and has
    \begin{align*}
        \mu_i^\uparrow(\matr{V})(E_{k+1}) - \mu_i^\uparrow(\matr{V})(E_k) & \leq \frac{\lambda_1(\matr{V}^2)[E_k, dE)}{\lambda_{m^*}^\uparrow(\matr{V}^2)(E_{k+1})}\bigl( E_{k+1} - \mu_i^\uparrow( \matr{V})(E_{k}) \bigr).
    \end{align*}
 Writing a telescope sum gives 
    \begin{align}
        \mu_i^\uparrow(\matr{V})(E^*) - \mu_i^\uparrow(\matr{V})(E_b) & = \sum_{k=0}^{N-1} \mu_i^\uparrow(\matr{V})(E_{k+1})- \mu_i^\uparrow(\matr{V})(E_{k})   \notag \\  
        & \leq \sum_{k=0}^{N-1} \frac{\lambda_1(\matr{V}^2)[E_k, dE)}{\lambda_{m^*}^\uparrow(\matr{V}^2)(E_{k+1})}\bigl( E_{k+1} - \mu_i^\uparrow( \matr{V})(E_{k}) \bigr)     \label{eq:Riemann_int_sum} \\
        & \leq \sum_{k=0}^{N-1}  \frac{\int_{E}^{E+ d E} |\alpha(E')|^2 dE' }{d  E}  \frac{ dE }{\lambda_{m^*}^\uparrow(\matr{V}^2)(E_{k+1})} \bigl( E_{k+1} - \mu_i^\uparrow( \matr{V})(E_{k}) \bigr).   \notag
    \end{align}
 In the last line (\ref{eq:Riemann_int_enstimate}) was used. Taking the limit $N \to \infty$ and using (\ref{eq:Fundamental_Theorem_limit}) gives
    \begin{align*}
        \mu_i^\uparrow(\matr{V})(E^*) - \mu_i^\uparrow(\matr{V})(E_b) &  \leq \int_{E_b}^{E^*} \frac{|\alpha(E)|^2}{\lambda_{m^*}^\uparrow(\matr{V}^2)(E)}\bigl( E - \mu_i^\uparrow( \matr{V})(E) \bigr) dE.
    \end{align*}
 The limit $E^* \to \infty$ yields the desired result,
    \begin{align*}
        \mu_i^\uparrow(\matr{V}) - \mu_i^\uparrow(\matr{V})(E_b) &  \leq \int_{E_b}^{\infty} \frac{|\alpha(E)|^2}{\lambda_{m^*}^\uparrow(\matr{V}^2)(E)}\bigl( E - \mu_i^\uparrow( \matr{V})(E) \bigr) dE.
    \end{align*}
\end{proof}
The Fundamental Theorem of Calculus as used in (\ref{eq:Fundamental_Theorem_limit}) indeed also applies 
to the delta distributions that appear in the definition of the spectral measure.
Let $\Theta$ denote the Heaviside step function, then we have in a distributional sense
\begin{align*}
 \lim_{ d E \to 0} \frac{1 }{d  E}  \int_{E-dE}^{E+ d E}\delta(E'-A) dE' = \lim_{ d E \to 0} \frac{1 }{d  E} \left( \Theta( E +d E - A) - \Theta(E- d E -A ) \right) = \delta(E-A).
\end{align*}
The continuous spectrum is of even less concern. Typically $\Gamma_\lambda$ is a manifold that varies smoothly in $\lambda$
and reasonable regularity assumptions on the guess vectors $a_l(k)$ 
result in a sufficiently regular energy measure $|\alpha(E)|^2$ within $E \in \sigma_c$. \\

Regarding the point spectrum, one issue is the evaluation of delta distributions at discontinuities.
In the point spectrum $|\alpha(E)|^2$ will consist of delta distributions centered at the eigenvalues $E = \lambda_i$ of the operator $\matr{H}$
and $\mu_i^\uparrow( \matr{V})(E)$ and $\lambda_m^\uparrow(\matr{V}^2)(E)$ are constant functions in between the delta peaks.
In particular, the delta distributions have to be evaluated at discontinuities.
There is no universal convention on the evaluation of delta distributions at discontinuities,
but a typical practice is to average the left- and right-hand sides. However, (\ref{eq:Riemann_int_sum}) has that 
the correct choice to maintain a rigorous inequality in (\ref{eq:integrated_spec_ineq_up}) is to evaluate the delta distributions at the righthand side of the discontinuity.
Analogously in the lower estimate (\ref{eq:integrated_spec_ineq_down}) the correct choice to evaluate a delta distribution 
at a discontinuity is the left-hand side. \\

Fortunately, we do not have to bother too much with these formalities.
In practice, the denominator $\lambda_m^\uparrow(\matr{V}^2)(E)$ is estimated by a term independent of the integration variable $E$
as seen in Corrolary \ref{cor:integrated_spectral_stability}.

\subsubsection{Proof of Theorem \ref{thm:integrated_spectral_stability}}

\begin{proof}[\textit{Proof of Theorem \ref{thm:integrated_spectral_stability}}]
 Let $(\matr{H},\matr{V})$ and $(E_a,E_b)$ be as in Theorem \ref{thm:integrated_spectral_stability}. Only the first inequality $(\ref{eq:integrated_spec_ineq_up_thm})$ is proven, as the proof for the second inequality is given analogously.\\
 Denote $\matr{B}^\uparrow = \matr{V}^\uparrow(E_b)$. 
 By Corollary \ref{cor:monotone_gen_eigenvalues} $\mu^\downarrow_i(\matr{B}^\uparrow)(E)$ is a monotone decreasing function in $E$. 
 In particular,
    \begin{align*}
        \mu^\downarrow_i(\matr{B}^\uparrow)(E_a) \leq \lim_{E^* \to \infty}\mu^\downarrow_i(\matr{B}^\uparrow)(- E^*) =\mu_i(\matr{B}^\uparrow).
    \end{align*}
 Noting $(\matr{B}^\uparrow)^\downarrow(E_a) =\matr{V}(E_a,E_b)=\matr{B} $ and multiplying with $-1$ gives
    \begin{align}
 - \mu_i(\matr{V}^\uparrow(E_b)) \geq - \mu_i(\matr{B}).  \label{eq:proof_4}
    \end{align}
 Corollary \ref{cor:monotone_gen_eigenvalues} also has $- \mu_i(\matr{V}^\uparrow(E)) \leq - \mu_i(\matr{V}^\uparrow(E_b)) $ for $E \leq E_b$. 
 Applying this in (\ref{eq:integrated_spec_ineq_up}) of Proposition \ref{prop:integrated_spectral_inequality} gives,
    \begin{align*}
        \mu_i(\matr{V}) - \mu_i(\matr{V}^\uparrow(E_b)) & \leq \int_{E_b}^{\infty} \frac{|\alpha(E)|^2 }{\lambda_{m^*}^\uparrow(\matr{V}^2)(E)}\bigl( E - \mu_i( \matr{V}^\uparrow(E)) \bigr) dE. \\
                                                        & \leq \int_{E_b}^{\infty} \frac{|\alpha(E)|^2}{\lambda_{m^*}^\uparrow(\matr{V}^2)(E)}\bigl( E - \mu_i( \matr{V}^\uparrow(E_b)) \bigr) dE.      
    \end{align*}
 Bringing the $\mu_i(\matr{V}^\uparrow(E_b)$-term to the left-hand side 
    \begin{align*}
        \mu_i(\matr{V}) - \mu_i(\matr{V}^\uparrow(E_b))  & \underbrace{\left(1 - \int_{E_b}^{\infty} \frac{|\alpha(E)|^2}{\lambda_{m^*}^\uparrow(\matr{V}^2)(E)} dE  \right)}_{:=1-\delta}  \leq \int_{E_b}^{\infty} \frac{|\alpha(E)|^2 }{\lambda_m^\uparrow(\matr{V}^2)(E)} E dE. 
    \end{align*}
 By assumption $\matr{V}$ is a signal to $(E_a,E_b)$, and thus $1-\delta \geq 0$.
 In particular, (\ref{eq:proof_4}) can be applied to obtain a lower estimate
    \begin{align*}
        \mu_i(\matr{V}) - \mu_i(\matr{B}) \left(1 - \int_{E_b}^{\infty} \frac{|\alpha(E)|^2 }{\lambda_{m^*}^\uparrow(\matr{V}^2)(E)} dE  \right) \leq \int_{E_b}^{\infty} \frac{|\alpha(E)|^2 }{\lambda_{m^*}^\uparrow(\matr{V}^2)(E)} E dE. 
    \end{align*}
 Finally arranging back,
    \begin{align*}
        \mu_i(\matr{V}) - \mu_i(\matr{B}) \leq \int_{E_b}^{\infty} \frac{|\alpha(E)|^2 }{\lambda_{m^*}^\uparrow(\matr{V}^2)(E)} \left(E -  \mu_i(\matr{B}) \right) dE.
    \end{align*}
 The second inequality is proven analogously.
\end{proof}

\subsubsection{Corollary \ref*{cor:integrated_spectral_stability} of Theorem \ref{thm:integrated_spectral_stability}}

\begin{proof}[\textit{Proof of Corollary \ref{cor:integrated_spectral_stability}}]
 Let $(\matr{H},\matr{V})$ and $(E_a,E_b)$ be as in Corollary \ref{cor:integrated_spectral_stability}. Only the first inequality (\ref{eq:integrated_spec_ineq_up_cor}) is proven, as the proof for the second inequality is given analogously.
 Every certifiable signal is found to be a signal, as seen below. Theorem \ref{thm:integrated_spectral_stability} has 
    \begin{align*}
        \mu_i(\matr{V}) - \mu_i(\matr{B}) \leq \int_{E_b}^{\infty} \frac{|\alpha(E)|^2}{\lambda_{m^*}^\uparrow(\matr{V}^2)(E)} \left(E -  \mu_i(\matr{B}) \right) dE.
    \end{align*}
 The guess vector matrix $\matr{V}$ can be written as 
 orthogonal decomposition $\matr{V} = \matr{V}^\uparrow(E) + \matr{V}^\downarrow[E]$. 
 Therefore $\matr{V}^\uparrow(E)^2 = \matr{V}^2 - \matr{V}^\downarrow[E]^2$ and 
 Weyl's inequality allows for the estimate 
    \begin{align*}
        \lambda_{m^*}(\matr{V}^\uparrow(E)^2) & \geq \lambda_m(\matr{V}^2) - \lambda_{1+m - m^*}(\matr{V}^\downarrow[E]^2) \\
        &  \geq \lambda_m(\matr{V}^2) - \lambda_{1}(\matr{V}^\downarrow[E]^2).
    \end{align*}
 In the last line $m-m^* \geq 0$ was used. For $E \geq E_b$ a decomposition has $\matr{V}^\downarrow[E]^2 = \matr{V}^2[E_b] - \matr{V}^2[E_b,E)$
 and permits through Weyl's inequality the estimate
    \begin{align*}
        \lambda_{1}(\matr{V}^\downarrow[E]^2) & \geq \lambda_1(\matr{V}^2[E_b]) - \lambda_{m}(\matr{V}^2[E_b,E)) \geq \lambda_1(\matr{V}^2[E_b]).
    \end{align*}
 Therefore,
    \begin{align}
        \lambda_{m^*}(\matr{V}^\uparrow(E)^2) & \geq \lambda_m(\matr{V}^2) - \lambda_{1}(\matr{V}^\downarrow[E]^2)  \label{eq:pluging2}\\
        & \geq \lambda_m(\matr{V}^2) - \lambda_{1}(\matr{V}^2[E_b]) \geq 0. \notag
    \end{align}
 In the last line the assumption that $\matr{V}$ is a certifiable signal to $(E_a,E_b)$ was used. The inequality
 also implies, that every certifiable signal is a signal.
 Plugging $(\ref{eq:pluging2})$ into the integral gives
    \begin{align}
        \mu_i(\matr{V}) - \mu_i(\matr{B}) & \leq \int_{E_b}^{\infty} \frac{|\alpha(E)|^2}{\lambda_m(\matr{V}^2) - \lambda_{1}(\matr{V}^\downarrow[E]^2)} \left(E -  \mu_i(\matr{B}) \right) dE,  \label{eq:cor_other_ineq}\\
        & \leq \int_{E_b}^{\infty} \frac{|\alpha(E)|^2}{\lambda_m(\matr{V}^2) - \lambda_{1}(\matr{V}^\downarrow[E_b]^2)} \left(E -  \mu_i(\matr{B}) \right) dE. \label{eq:cor_other_ineq_given}
    \end{align}
 As the second inequality is proven analogously, the proof is complete.
\end{proof}
In the proof, inequality (\ref{eq:cor_other_ineq}) is a sharper version of the inequality that is actually given in Corollary \ref{cor:integrated_spectral_stability}. It can
 also access the known Gram matrix $\matr{V}^2$ and $\lambda_{1}(\matr{V}^\downarrow[E]^2)$ can be controlled from a guess vector envelope. 
However, we believe that the loss in precision from (\ref{eq:cor_other_ineq}) to (\ref{eq:cor_other_ineq_given}) 
is for most application non essential and (\ref{eq:cor_other_ineq_given}) is more user friendly. \\

Wely's inequality also allows one to give the denominator of Theorem \ref{thm:integrated_spectral_stability} in terms of the signal Gramm matrix $\matr{B}^2$.
For this, we assume that the rank of the Gram matrix $\matr{V}^2$ is equal to the rank of the signal Gram matrix $\matr{B}^2$.
Then 
\begin{align*}
    \lambda_{m^*}(\matr{V}^\uparrow(E)^2) & = \lambda_{m}(\matr{V}^\uparrow(E)^2)= \lambda_{m}( \matr{B}^2 +  \matr{V}[E_b,E]^2 + \matr{V}^\uparrow[E_a]^2) \\
    & \geq \lambda_{m}( \matr{B}^2)  + \lambda_m(\matr{V}[E_b,E]^2 + \matr{V}^\uparrow[E_a]^2) \geq \lambda_{m}( \matr{B}^2).
\end{align*}
Protocols that can establish a lower bound on the signal strength of the guess matrix they generate could 
perhaps use this alternative to give guarantees that are independent from algorithmic convergence.

\newpage 

\section{Algorithmic Implications}
\label{sec:NumTheoSimbiosis}

Here, we will address the occurrence of singular Gramm matrices, that has repeatedly caused confusion and ambiguities
in the numerical literature. A simple solution to the problem is proposed, that 
effectively refines the guess vector space and improves the quality of the estimates of the spectral inequalities.
We will also propose a \emph{dimension detection} scheme for a
spectral subspace $\mathcal{E}$ of interest, that is motivated by the rigorous spectral inequalities of the 
previous two sections.

\subsection{The Singularity Problem }
In the literature of computational chemistry, many algorithms based on GEPs report the occurrence of complex eigenvalues in the solution of Hermitian GEPs.
Complex eigenvalues to a hermitian operator are indeed concerning. 
However, a clear explanation for their occurrence is often omitted, and the proposed practices on how to handle them lead to ambiguities.
For a concrete example, see the discussion of Black-Davidson-inspired algorithms in \cite{Zuev}.\footnote{In particular page 3.}
Another example in the context of Filter Diagonalization is found in Section 7.2 of \cite{Mandelshtam}.\\

In fact, the explanation for complex eigenvalues is quite straightforward. Often, a singular Gram matrix is the cause of the problem. 
Vectors in the kernel of the basis matrix $\matr{V}$ do not relate to the hermitian operator $\matr{H}$ and 
lead to arbitrary eigenvalues.\\

In the context of mathematical data science, GEPs occur in data dimension reduction and classification, 
commonly referred to as Generalized Linear Discriminant Analysis (Generalized LDA). 
In this context, the problem of singular Gram matrices in computational routines is more explicitly addressed 
and referred to as the \emph{singularity problem} \cite{GenLDA}.\\ 

A problem that is closely related to the singularity problem is that the true dimensionality of a target spectral subspace is often unknown.
A particular advanced example of a GEP-based protocol that is applied in the context of quantum chemistry 
is FEAST \cite{Baiardi_2021}.\footnote{The mathematical machinery of FEAST, that produces the guess vectors 
for a target spectral subspace $\mathcal{E}_{(a,b)}(\matr{H})$ is based on Cauchy's integral formula 
and numerical evaluation of contour integrals.} Numerically it is found that the attained precision 
of the method is best when the number of guess vectors is slightly larger than the true dimensionality of the target spectral subspace (overshooting).
An explanation of why slight overshooting attains better approximations is that it is easier to 
better "span" the target spectral subspace with more guess vectors.
However, first, the true dimensionality of the target spectral subspace has to be guessed through trial and 
error and after overshooting more iterations are needed to distinguish spurious from true eigenvalues.\\

Based on the theory of the preceding sections, simple numerical practices are proposed that ultimately come along with the 
following advantages:
\begin{enumerate}
    \item Avoidance of singular Gram matrices
    \item Maintaining the advantages of slight overshooting
    \item Efficient dimension detection with a guarantee
    \item Access to rigorous precision guarantees
\end{enumerate}
We assume that a GEP-based protocol $\mathfrak{P}$ provides a bound
on the noise $\lambda_1(\matr{N}^2)$ of the guess vector matrix that it generartes.

\subsection{An $\varepsilon$-Nullspace to the Rescue}

\label{sec:epsilon_nullspace}
In the theoretical treatment of the previous sections, it was pointed out, that singular Gram matrices are not a concern for the spectral analysis of GEPs
as the nullspace of the Gram matrix can be ignored. It was the rank of the Gram matrix and the smallest non-zero eigenvalues, that was of interest.\\

As singular Gram matrices have been of some concern in the numerical literature, it is instructive to transfer this insight to numerical applications.
For matrices that are subject to noise, 
the strict definition of the kernel of a matrix is no longer suitable and requires a relaxed notion.
A \emph{$\varepsilon$-nullspace} is defined as 
\begin{align*}
 \operatorname{Ker}_\varepsilon (\matr{V}^2)= \left\{ x \in \mathbb{C}^M \quad s.t. \quad \left|\frac{\langle x, \matr{V}^2 x \rangle}{\langle x, x\rangle } \right|\leq \varepsilon \right\}.
\end{align*}
It is important to note that the $\varepsilon$-nullspace is not a vector space.\footnote{This means
 $x_1,x_2  \in \operatorname{Ker}_\varepsilon (\matr{V}^2)$ does not imply $ \alpha_1 x_1 + \alpha_2 x_2 \in \operatorname{Ker}_\varepsilon (\matr{V}^2)$
for all $\alpha_1, \alpha_2 \in \mathbb{C}$.} 
The $\varepsilon$-nullspace of a hermitian matrix is only related to the spectral subspace 
associated with the interval $[-\varepsilon, \varepsilon]$,
\begin{align*}
 \mathcal{E}_{[-\varepsilon, \varepsilon]}(\matr{A}) \subset \operatorname{Ker}_\varepsilon(\matr{A}).
\end{align*}
But the converse of the inclusion is a priori not true. However, for Gram matrices a related substitute property for the converse holds,
\begin{align}
 \mathcal{E}_{[0, \varepsilon]}(\matr{V}^2)^\perp \cap \operatorname{Ker}_\varepsilon(\matr{V}^2) =\varnothing.  \label{eq:epsilon_nullspace_orthogonal}
\end{align}
While the vector space property of a nullspace is not stable under $\varepsilon$ relaxations, the orthogonality property expressed 
in (\ref{eq:epsilon_nullspace_orthogonal}) is. It is also this orthogonality relation that is of greater numerical relevance. 
It ensures that after we have projected onto the spectral subspace $\mathcal{E}_{>\varepsilon}(\matr{V}^2)$ of well-conditioned eigenvectors, the Gram matrix is no longer singular.
We use the concept of an $\varepsilon$-nullspace to motivate practices that allow us to distinguish noise from signal in the guess vector matrix $\matr{V}$ and that result 
in better approximations. \\

Consider a protocol $\mathfrak{P}$ that generates guess vectors for a target spectral subspace $\mathcal{E}$.
From now on the number of guess vectors $M$ is acknowledged as an important parameter and explicitly denoted.
For example, $\matr{V}_M = [v_1, v_2, \cdots, v_M]$ and 
accordingly for a signal-noise decomposition as in Definition \ref{def:ortho_decomp}, $\matr{V}_M = [\matr{B}_M, \matr{N}_M]$.
It is assumed that the protocol $\mathfrak{P}$ provides bounds on
the noise of the Gram matrices that it generates,
\begin{align*}
    \lambda_1(\matr{N}_M^2) \leq \varepsilon_M, \qquad \text{for all} \quad M.
\end{align*}
Naturally, the sequence of noise estimates $\{\varepsilon_M \}$ is expected to be increasing in $M$. 

\subsubsection{Dimension Detection}
\label{sec:dimension_detection}
In the following, we illustrate how the spectral inequalities of the previous two Sections can be accessed, after the matrices $\matr{V}_M^\dagger \matr{H} \matr{V}_M, \matr{V}_M \in \mathbb{C}^{m\times m}$
have been obtained from the protocol $\mathfrak{P}$, in order to give precision guarantees on 
eigenvalue approximations. We will also suggest a \emph{guess vector space refinement} 
that maintains the advantages of slight overshooting and simultaneously avoids unfavorable precision corruption due to ill-conditioned eigenvectors.\\

The eigenvalues of the Gram matrix $\matr{V}_M^2$ are computed and listed in 
descending order $\lambda_1(\matr{V}_M^2) \geq \lambda_2(\matr{V}_M^2) \geq \cdots \geq \lambda_M(\matr{V}_M^2)$.
The detected dimensionality of the target spectral subspace is defined as 
\begin{align}
 m_{\text{detect}}= \min_m \lambda_m(\matr{V}_M^2) \quad \text{s.t.} \quad \lambda_m(\matr{V}_M^2) \geq \varepsilon_M. \label{eq:detected_dimension}
\end{align}
The following claim justifies the suggestion on how to detect the dimensionality of the target spectral subspace $\mathcal{E}$.
\begin{claim}
    \label{claim:detected_dimension}
 The detected dimensionality (\ref{eq:detected_dimension}) is lower than or equal to the true dimensionality of the target spectral subspace.
\end{claim}
The proof is a simple application of Weyl's inequality (\ref{eq:Weyl_classical}) and the sufficient condition is that $\varepsilon_M$ is larger than $\lambda_1(\matr{N}_M^2)$.
\begin{proof}
Weyl's inequality on eigenvalues of hermitian matrices has
\begin{align*}
    \lambda_{ m_{\text{detect}}}(\matr{V}^2) & = \lambda_{ m_{\text{detect}}}(\matr{B}^2 + \matr{N}^2) \leq \lambda_{ m_{\text{detect}}}(\matr{B}^2) + \lambda_{1}(\matr{N}^2).
\end{align*}
Solving for $\lambda_{ m_{\text{detect}}}(\matr{B}^2)$ and using the assumption (\ref{eq:detected_dimension}) gives 
\begin{align*}
    \lambda_{ m_{\text{detect}}}(\matr{B}^2) \geq \lambda_{ m_{\text{detect}}}(\matr{V}^2) - \lambda_{1}(\matr{N}^2) >0.
\end{align*}
The last inequality implies that the rank of $\matr{B}^2$ is at least $M_{\text{detect}}$. As the column vectors of $\matr{B}$ are in $\mathcal{E}$,
$$\operatorname*{span}(\matr{B}) \subset \mathcal{E}$$
this implies in particular that the space $\mathcal{E}$ is at least $M_{\text{detect}}$ dimensional.
\end{proof}

If for example, $m_{\text{detect}} = M$, it is proven, that the target space $\mathcal{E}$ is at least $M$-dimensional and
the guessed dimensionality needs to be increased. In the ideal scenario, $m_{\text{detect}}$ is smaller than $M$
and one observes a significant drop in the spectrum of the Gram matrix, 
\begin{align}
    \lambda_{M_{\text{detect}} }(\matr{V}_M^2) \gg  \lambda_{M_{\text{detect}} +1}(\matr{V}_M^2).  \label{eq:significant_drop}
\end{align}
A sophisticated protocol for guess vector generation motivates at least some kind of independence among the guess vectors
that can be maintained as long as the number of guess vectors is smaller than the true dimensionality of the target spectral subspace.
Then, (\ref{eq:significant_drop}) allows great confidence that the true dimensionality has correctly been detected.

\subsubsection{Advantage of Overshooting}
\label{sec:overshooting}
Using more guess vectors than the true dimensionality of the target spectral subspace
allows for more freedom to represent the true eigenvectors within the target spectral subspace
and can yield better approximations. \\

This can also be formalized through the conditioning of the eigenvalue $\lambda_{m_{\text{true}}}(\matr{V}_M^2)$
of the Gram matrix and Weyl inequalities. Here $m_{\text{true}}$ denotes the true dimensionality of the target spectral subspace $\mathcal{E}$.
We will denote with $\matr{V}_{M_\text{true}}^2$ the Gram matrix, that is obtained from the first $m_{\text{true}}$ guess vectors
of $v_1, v_2, \cdots,v_{m_\text{true}}, \cdots , v_M$ and is then through zero padding embedded in $\mathbb{C}^{M \times M}$.
Weyl inequalities have
\begin{align*}
    \lambda_{m_{\text{true}}}(\matr{V}_{M}^2) = \lambda_{m_{\text{true}}}(\matr{V}_{m_\text{true}}^2 + \matr{V}_{M}^2 - \matr{V}_{m_\text{true}}^2 ) \geq \lambda_{m_{\text{true}}}(\matr{V}_{m_\text{true}}^2) + \lambda_M(\matr{V}_{M}^2 - \matr{V}_{m_\text{true}}^2 ). 
\end{align*}
By construction, we have $\matr{V}_{M} \matr{V}_{m_\text{true}}^\dagger = \matr{V}_{m_\text{true}}^\dagger \matr{V}_{m_\text{true}}$ due to the zero padding. 
In particular,
\begin{align*}
 (\matr{V}_{M} - \matr{V}_{m_\text{true}})^2 & =   \matr{V}_{M}^2 - \matr{V}_{M} \matr{V}_{m_\text{true}}^\dagger - \matr{V}_{m_\text{true}}\matr{V}_{M}^\dagger + \matr{V}_{m_\text{true}}^2 \\
    & = \matr{V}_{M}^2 - \matr{V}_{m_\text{true}}^2 \geq 0. 
\end{align*}
Therefore, $\lambda_{m_{\text{true}}}(\matr{V}_{M}^2) \geq \lambda_{m_{\text{true}}}(\matr{V}_{M_\text{true}}^2) $, and the conditioning of 
the $m_{\text{true}}$-th eigenvalue of the Gram matrix can only be improved by a larger number of guess vectors.\\

These considerations suggest a numerical practice that we will refer to as \emph{guess vector space refinement}.
It includes the advantages of 
dimension detection and overshooting while it, nevertheless, avoids the singularity problem. 

\subsubsection{A Guess Vector Space Refinement}
\label{sec:guess_vector_space_refinement}

After the Gram matrix $\matr{V}_M^2$ has been obtained from the protocol $\mathfrak{P}$, 
we suggest, to diagonalize $\matr{V}_M^2$ and to obtain the eigenvalues 
$\lambda_1(\matr{V}_M^2) \geq \lambda_2(\matr{V}_M^2) \geq \cdots \geq \lambda_M(\matr{V}_M^2)$ in a descending order.
The corresponding eigenvectors are stored as columns in a unitary matrix $\matr{U}_M=[u_1, u_2, \cdots, u_M]$.
\begin{align*}
 \matr{U}^\dagger \matr{V}_M^\dagger \matr{V}_M \matr{U} = \operatorname{diag}(\lambda_1(\matr{V}_M^2), \cdots, \lambda_M(\matr{V}_M^2)).
\end{align*}
The geometrical meaning of the guess vector matrix $\matr{V}_M \matr{U} = [\tilde v_1, \cdots , \tilde v_M]$ is that the original guess vectors $v_1, v_2, \cdots, v_M$ have 
been rearranged as linear combinations into new vectors $\tilde v_i$. These vectors are maximally well conditioned, in descending order as can 
be seen from the min-max characterization of eigenvalues. \\

The suggestion is to only consider the first $m_{\text{detect}}$ reshuffled vectors $\tilde v_1, \tilde v_2, \cdots, \tilde v_{m_{\text{detect}}}$, 
for the GEP, which aims to approximate the target spectral subspace $\mathcal{E}$.
The vecotors $\tilde v_j = \matr{V} u_j$ with $\lambda_j(\matr{V}_M^2) \leq \varepsilon_M$ span the $\varepsilon_M$-kernel of the guess vector Gram matrix.
It is instructive to remove them from further analysis, as they carry very little information about the signal,
\begin{align*}
    \varepsilon_M \geq u_i^\dagger \matr{V}^2 u_i = x_i^\dagger (\matr{B}^2 + \matr{N}^2) x_i \geq x_i^\dagger \matr{B}^2 x_i.
\end{align*}
As previously pointed out, $\mathcal{E}_{[0, \varepsilon]}(\matr{V}^2) \subset \operatorname{Ker}_{\varepsilon}(\matr{V}^2)$. Therefore, all linear combinations of ill-conditioned vectors 
will remain ill-conditioned and carry negligible information about the signal.
These ill-conditioned vectors can also have unfavorable effects on the well-conditioned eigenvectors, as illustrated in Section 3 of Ref.\cite{stewart_perturbation_1978}.\\

The new GEP that corresponds to the refined guess vector space $\matr{V}_* = [\tilde v_1, \tilde v_2, \cdots, \tilde v_{m_{\text{detect}}}]$
is easily obtained by transforming the original left-hand side matrix $\matr{H}_{V_M}$ that a protocol $\mathfrak{P}$ has provided
with $\matr{U}$ and then to only consider the upper left $m_{\text{detect}} \times m_{\text{detect}}$ block,
\begin{align*}
 \matr{H}_{\matr{V}_*} = \left[\matr{U}^\dagger \matr{H}_{V_M} \matr{U}\right]_{i,j =1}^{m_{\text{detect}}}.
\end{align*}
The matrix $\matr{V}_*$ is already diagonalized after the unitary transformation, and we denote $\matr{\Lambda}_{\matr{V}_*} = \operatorname{diag}(\lambda_1(\matr{V}_M^2), \cdots, \lambda_{m_{\text{detect}}}(\matr{V}_M^2))$.\footnote{The proof 
of the Generalized Spectral Theorem in Section \ref{sec:proof_specral_theorem} shows how unitary transformations are used to diagonalize GEPs.}
The generalized eigenvalues of the GEP $(\matr{H}_{\matr{V}_*})$ are then obtained as the eigenvalues of 
\begin{align*}
 \matr{\Lambda}_{\matr{V}_*}^{-1} \matr{H}_{\matr{V}_*}.
\end{align*}
After such a guess vector space refinement, the approximations of the target spectral subspace $\mathcal{E}$ have 
been obtained from a well-conditioned GEP $(\matr{H}, \matr{V}_*)$, and the spectral inequalities of Sections \ref{sec:SpecIneq} and \ref{sec:IntegratedSpectralInequalities} can be accessed
to obtain precision guarantees.\\

We denote the signal noise decomposition of $\matr{V}_*$ as $[\matr{B}_*, \matr{N}_*]$. While 
$\matr{N}_*$ is only $m_{\text{detect}}$ dimensional, it has been generated from the $M$-dimensional noise matrix $\matr{N}_M$,
which the protocol initially provided.
In particular, the correct way to estimate $\lambda_1(\matr{N}_*^2)$ is to use the noise bound $\varepsilon_M$ and not $\varepsilon_{m_{\text{detect}}}$.
Naturally, the noise bounds $\varepsilon_M$ are increasing with $M$.\\

However, the conditioning of the Gram matrix $\matr{V}_*^2$ can only improve with increasing number of guess vectors as we have seen in Section \ref{sec:overshooting}.
However, eventually, there is a turning point, where the increase in the noise bound $\varepsilon_M$ outweighs the improvement in the conditioning of the Gram matrix $\matr{V}_*^2$.
Therefore, we expect that for most applications the used number of guess vectors $M$ that reaches optimal precisions is only slightly larger than the true dimensionality of the target spectral subspace
or has a threshold.\\

In summary, the simple guess vector space refinement enables us to incorporate information from a number of guess vectors that exceeds the true dimensionality of the target spectral subspace, while still avoiding the singularity problem.
A drop in the spectrum of Gram matrix $\matr{V}_M^2$ is used to detect the dimensionality of the target spectral subspace
and Claim \ref{claim:detected_dimension} guarantees that the true dimensionality is not underestimated. With the well-conditioned GEP $(\matr{H}_{\matr{V}_*}, \matr{V}_*)$
and $\varepsilon_M$ precision guarantees can be obtained from the spectral inequalities of the previous sections.\\

By now all results of Section \ref{sec:ResHigh} have been established. The proposed guess vector space refinement differs from the current practices in the literature.
In the last subsection, we will briefly elaborate on the current practice of post-removing spurious eigenvalues.\\
 We will also give a small example that shows how a GEP-based approach with multiple guess vectors can compare favorably to a single Rayleigh quotient optimisation for estimating extreme eigenvalues.

\subsection{Additional Remarks on Numerical Practices}

\subsubsection{Argument against Filtering Spurious Eigenvalues}
\label{sec:NumericalPractices}

The practice of the guess vector space refinement is in contrast to the current practice of filtering spurious eigenvalues from the GEP $(\matr{H},\matr{V}_M)$.
Here, one accepts that the guess vector matrix $\matr{V}_M$ becomes singular after the true dimensionality of the target spectral subspace has been overestimated.
Then, one also obtains spurious eigenvalues, infiltrates the solutions, and does not relate to the target spectral subspace of interest. 
Either through suspicious complex contributions in the calculated eigenvalues or through extra iterations that detect unstable oscillations, the spurious eigenvalues 
are heuristically identified and removed.\\

We argue against this practice. Even after the spurious eigenvalues have been removed, they have already corrupted the precision of the approximation.
In Section 3 of \cite{stewart_perturbation_1978} it is illustrated in an example, that an ill-conditioned 
eigenvector of a GEP can cause drastic noise amplification on a well-conditioned eigenvector.\\

This is also reflected in the spectral inequalities of Section \ref{sec:SpecIneq}, which give 
precision statemnts with a factor $\lambda_m(\matr{V}^2)^{-1}$. However,
we also note that the integrated spectral inequality of Theorem \ref{thm:integrated_spectral_stability}
softens the impact of spurious eigenvalues in comparison,
\begin{align*}
    \mu_i(\matr{V}) - \mu_i(\matr{B}) & \leq \int_{E_b}^{\infty} \frac{|\alpha(E)|^2 }{\lambda_{m^*}^\uparrow(\matr{V}^2)(E)} \left(E -  \mu_i(\matr{B}) \right) dE.
\end{align*}
In particular, here the condition factor is not given by the lowest eigenvalue of the full Gram matrix $\lambda_m(\matr{V}^2)$, but by the smallest non-vanishing eigenvalue of $\matr{V}^\uparrow(E)^2$.
This softens the impact of spurious Gram matrices in the context of projection-based GEPs.
It can partially explain why very slight overshooting of the true dimensionality has, despite the singularity problem, reached 
satisfying approximations in the literature. \\

However, we nevertheless suggest abandoning the practice of over-estimate and post-select.
A guess vector space refinement as proposed in Section \ref{sec:guess_vector_space_refinement} is simple to implement and 
enables rigorous precision guarantees. To our knowledge, the freedom to select a larger number of guess vectors and to
refine the guess vector space has not been explored in the literature. As discussed it enables more efficient dimension detection and 
attained precision in the eigenvalues.

\subsubsection{Extremal Eigenvalues through Projection}

We recall that the generalized eigenvalues of a GEP $(\matr{H}, \matr{V})$ are contained within the spectral range of $\matr{H}$ (Fact \ref{fact:spec_range}).
Equivalently formulated, a generalized eigenvalue $\mu_i$ is contained within the range of the Rayleigh quotient of $\matr{H}$ on a Hilbert space $\Hil$ that contains the guess vectors of $\matr{V}$.
But from the other direction, the generalized eigenvalues of a GEP do
not relate in a similar critical manner to the individual Rayleigh quotients of the guess vectors.\\

Consider the simple example 
\begin{align*}
 \matr{H} = \begin{bmatrix}
 -1 & 0 & 0 \\
        0 & 0 & 0 \\
        0 & 0 & 1
        \end{bmatrix}, \qquad \matr{V} = [v_1,v_2] = \begin{bmatrix}
            1 & 1 \\
            1 & 0 \\
            1 & -1
        \end{bmatrix}.
\end{align*}
Both $v_1$ and $v_2$ have individual Rayleigh quotients equal to zero. Therefore not particularly close 
to the lowest eigenvalue of $\matr{H}$. The GEP $(\matr{H}, \matr{V})$ has 
\begin{align*}
    \begin{bmatrix}
        0 & -2  \\
 -2 & 0 
        \end{bmatrix} x = \mu   \begin{bmatrix}
            3 & 0  \\
            0 & 2 
            \end{bmatrix}x.
\end{align*}
The solutions are readily obtained at $\mu_{\pm} = \pm \sqrt{\frac{2}{3}}$. The approximation of the lowest eigenvalue
is relative to the poor Rayleigh quotient of the guess vectors encouraging. \\

This small example illustrates, that even if one only aims to approximate the lowest eigenvalue of a Hamiltonian, more precision can 
be obtained by using multiple guess vectors. Right away aiming for a slightly larger number of lowest-lying eigenvalues
allows spending fewer resources on individual Rayleigh quotient optimization, which has a tendency to quickly flatten in precision gain.

\section{Concluding Remarks}

Throughout this chapter, we have developed a novel approximation theory for spectral analysis, inspired by the extensive numerical literature in Quantum Chemistry.
Numerical experiments and heuristic folk theorems 
are vital to scientific progress. Eventually, however, a firm mathematical ground becomes invaluable to formulate a clearer vision of the next 
challenges ahead.\\

The rigorous error estimates that we have developed here, facilitate a more comprehensive analysis of the optimal allocation of computational resources.
Many theoretical models that give rise to high-dimensional computational challenges are themselves only approximations of nature.
Without a clear uncertainty quantification, it remains unclear whether the accuracy of an approximate solution exceeds the inherent precision of the model itself. 
Therefore, the spectral inequalities presented here enable a more informed and efficient distribution of resources. \\

In addition, we found that error estimates are also of algorithmic relevance. They allow us to distinguish between noise and signal in the computational data
and to determine the true dimensionality of the subspace of interest. \\

Finally, we would like to emphasise the generality of the framework. Many scientific problems can be formulated as spectral problems.
Therefore, the theoretical framework promises to be relevant in a broad range of scientific disciplines.
In Chapter \ref{chap:FilterDiag}, we will apply the presented approximation theory to a signal processing routine, and thus demonstrate its practical utility.
\newpage


\section{Appendix}

\subsection{Generalized Spectral Theorem \ref{thm:gen_specral_theorem}}
\label{sec:proof_specral_theorem}

We give two proofs for the Generlaized Spectral Theorem \ref{thm:gen_specral_theorem}. 
The first proof is based on pure linear algebra and gives insights to computational aspects on how the eigendecomposition is obtained.
The second proof is shorter and relies on usual spectral theorem with respect to an inner product space that is induced by the GEP. 

\subsubsection{Proof I: Pure Linear Algebra}
We assume first that the Gramm matrix $\matr{V}^2$ is positive definite and argue for the singular case afterwards.
We start by establishing that the eigenvalues of a hermitian GEP are real and that their eigenvectors satisfy an orthogonality relation.\\

Let $\{ (\mu_1,x_1), \cdots (\mu_k,x_k)\}$ be a collection of eigenpairs that solve 
\begin{align*}
    \matr{H}_V x_i = \mu_i \matr{V}^2 x_i. 
\end{align*}
Left multiplying the equation with $x_j^\dagger$ and using that $\matr{H}_V$ is self adjoint gives
\begin{align}
    \mu_i \langle x_j, \matr{V}^2 x_i \rangle = \langle x_j, \matr{H}_V x_i \rangle = \langle \matr{H}_V x_j, x_i \rangle = \conj{\mu}_j \langle x_j, \matr{V}^2 x_i \rangle. \label{eq:spec_ortho}
\end{align}
For $i =j$ this implies  $\mu_i = \conj{\mu_i}$. Therefore all eigenvalues $\mu_i$ of a hermitian and non singular GEP must be real.
Here the assumption $\langle x,\matr{V}^2 x \rangle >0$ was used.  
If $\mu_i \neq \mu_j$ eq.\ref{eq:spec_ortho} implies that $\langle x_i, \matr{V}^2 x_j \rangle =0$. 
Thus eigenvectors to different eigenvalues are orthogonal with respect to the inner product $\langle \cdot, \cdot \rangle_V$,
that is induced by the basis matrix. 
For degenerate eigenvalues, the eigenvectors can be orthonormalized with respect to $\langle \cdot, \cdot \rangle_V$ by a Gramm schmidt procedure. 
In summary all eigenvectors of a hermitian non degenerate GEP can be orthonormalized such that we have
\begin{align}
    \langle x_i, \matr{V}^2 x_j \rangle = \delta_{ij} \label{eq:specral_ortho2}.
\end{align}

Next we ask for an analog of a complete eigendecomposition a hermitian GEP.
That is, the question of wether every hermitian GEP of rank $m$ has $m$ solutions with eigenvectors 
orthonormal in the sense of eq.\ref{eq:specral_ortho2}. The question formalizes in matrix from as 
\begin{align}
    \matr{V}^\dagger \matr{H} \matr{V} \matr{\Phi} = \matr{V}^\dagger \matr{V} \matr{\Phi} \matr{\Lambda} \quad \label{eq:GEP_matForm}
\end{align}
where $\matr{\Phi} \in \mathbb{C}^{m \times m}$ contains the generalized eigenvectors in its columns
and $\matr{\Lambda}$ is a diagonal matrix containing the eigenvalues of the GEP. The $\matr{V}^2$-orthonormality relation expressed in
eq.\ref{eq:specral_ortho2} imposes on $\matr{\Phi}$ the constraint
\begin{align*}
    \matr{\Phi}^\dagger \matr{V}^\dagger \matr{V} \matr{\Phi} = \matr{I}.
\end{align*}

Left multiplying eq.\ref{eq:GEP_matForm} with $\matr{\Phi}^\dagger$ gives 
following formulation of the spectral problem,

\begin{equation*}
    \begin{aligned}
        \matr{\Phi}^\dagger \matr{V}^\dagger \matr{H} \matr{V} \matr{\Phi} &= \matr{\Phi}^\dagger \matr{V}^\dagger \matr{V} \matr{\Phi} \matr{\Lambda}  \quad
        &\begin{aligned}
            & \text{ with } && \matr{\Phi}\in \mathbb{C}^{m \times m}, \matr{\Lambda}\in \mathbb{R}^{m\times m} \text{ diagonal.}\\
            & \text{ subject to} && \matr{\Phi}^\dagger \matr{V}^\dagger \matr{V} \matr{\Phi} = \matr{I}.
        \end{aligned}
    \end{aligned}
\end{equation*}

To ensure the constraint on $\matr{\Phi}$ the eigen decomposition of the Gram matrix is considered
\begin{align*}
    \matr{\Phi}^\dagger \matr{V}^\dagger \matr{V} \matr{\Phi}  = \matr{\Phi}^\dagger \matr{U}_V \matr{\Lambda}_V  \matr{U}_V^\dagger  \matr{\Phi}.
\end{align*}
Here $\matr{U}_V$ is a unitary matrix and $\matr{\Lambda}_V$ is a diagonal matrix containing the eigenvalues of $\matr{V}^\dagger \matr{V}$.
We pick $\matr{\Phi} = \matr{U}_V \matr{\Lambda}_V^{-1/2} \matr{U}$ where $\matr{U}$ is some unitary matrix to be chosen later. 
Then we have 
\begin{align*}
    \matr{\Phi}^\dagger \matr{V}^\dagger \matr{V} \matr{\Phi}  = \matr{U}^\dagger \matr{\Lambda}_V^{-1/2} \matr{\Lambda}_V \matr{\Lambda}_V^{-1/2} \matr{U} = \matr{I},
\end{align*}
where we used in the last equality that $\matr{U}$ is a unitary matrix.\\ 
With this our problem reduces to 
\begin{align*}
    \matr{\Phi}^\dagger \matr{V}^\dagger \matr{H} \matr{V} \matr{\Phi} = \matr{\Lambda}.
\end{align*}
With $\matr{\Phi} = \matr{U}_V \matr{\Lambda}_V^{-1/2} \matr{U}$ we have 
\begin{align*}
    \matr{\Phi}^\dagger \matr{V}^\dagger \matr{H} \matr{V} \matr{\Phi} = \matr{U}^\dagger \left( \matr{\Lambda}_V^{-1/2} \matr{U}_V^\dagger \matr{V}^\dagger \matr{H} \matr{V} \matr{U}_V \matr{\Lambda}_V^{-1/2} \right) \matr{U}
\end{align*}
We denote the matrix in brace as $\matr{A}$ and observe that it is self adjoint and has in particular an eigen decomposition.
Therefore we can write
\begin{align*}
    \matr{\Phi}^\dagger \matr{V}^\dagger \matr{H} \matr{V} \matr{\Phi} = \matr{U}^\dagger \matr{U}_A \matr{\Lambda}_A \matr{U}^\dagger_A  \matr{U} = \matr{\Lambda}
\end{align*}
From this equation we see that picking  $\matr{U} = \matr{U}_A$ gives a valid solution. The hermitian and non singular generalized eigenvalue problem is thus solved by 
$\matr{\Phi} = \matr{U}_V \matr{\Lambda}_V^{-1/2} \matr{U}_A^\dagger$ and $\matr{\Lambda} = \matr{\Lambda}_A$.\\

If now $\matr{V}^2$ is not positive definite, the GEP can be be brought into block diagonal form by an unitary transformation $\matr{\tilde U}$, 
according to the decomposition $\operatorname*{Ker}(\matr{V}^2)^\perp \oplus \operatorname*{Ker}(\matr{V}^2) $. 
Since we have $\operatorname*{Ker}(\matr{V}^2) \subset \operatorname*{Ker}(\matr{H}_V)$, there exists indeed a unitary matrix $\matr{\tilde U}$ such that
\begin{align*}
    \matr{\tilde U}^\dagger \matr{H}_V \matr{\tilde U} = \begin{bmatrix}
        \matr{\tilde H}_V & 0 \\
        0 & 0
    \end{bmatrix} \quad \text{and} \quad \matr{\tilde U}^\dagger \matr{V}^2 \matr{\tilde U} =\begin{bmatrix}
        \matr{\tilde H}_V & 0 \\
        0 & 0
    \end{bmatrix}
\end{align*}
Left multiplying the GEP in matrix form then with $\matr{\tilde U}^\dagger$ and inserting the identity matrix in the form of $\matr{\tilde U} \matr{\tilde U}^\dagger$ gives
\begin{align}
    \matr{H}_V \matr{\Phi} & =  \matr{V}^2  \matr{\Phi}  \matr{ \Lambda}  \label{eq:GEP_degen_Matform} \\
    \Leftrightarrow \quad \matr{\tilde U}^\dagger  \matr{H}_V \matr{\tilde U}  \matr{\tilde U}^\dagger \matr{\Phi} & = \matr{\tilde U}^\dagger \matr{V}^2 \matr{\tilde U} \matr{\tilde U}^\dagger  \matr{\Phi}  \matr{ \Lambda} \notag \\
    \begin{bmatrix}
        \matr{\tilde H}_V & 0 \\
        0 & 0 
    \end{bmatrix} \matr{\tilde U}^\dagger \matr{\Phi}  & =
        \begin{bmatrix}
            \matr{ \tilde V}^2 & 0 \\
            0 & 0
        \end{bmatrix} \matr{\tilde U}^\dagger  \matr{\Phi}  \matr{ \Lambda}.  \label{eq:GEP_degen_block}
\end{align}
We carry out the previous procedure on the non degenerate part of the generalized eigenvalue problem $(  \matr{\tilde  H}_V , \matr{\tilde  V}^2)$ that is of dimensionality $m$ and
find $m$ eigenvectors $\tilde{x}_i$ that are orthonormal with respect to $\langle \cdot, \cdot \rangle_{\tilde V}$. 
The smaller dimensional eigenvectors $\tilde{x}_i$ are trivially embedded into the full space $\mathbb{C}^M$ with $x_i = \tilde{x}_i \oplus 0$.
We choose an arbitrary basis $\{n_1, \cdots , n_{M-m}\}$ of the null space of $\matr{V}^2$ and
denote $$\matr{\tilde \Phi} = [x_1, \cdots, x_m, n_1, \cdots , n_{M-m}] \in \mathbb{C}^{m \times m}.$$
The choice $\matr{\tilde U}^\dagger \matr{\Phi} =\matr{\tilde \Phi} $ and $\matr{\Lambda} = \operatorname*{diag}( \mu_1, \cdots, \mu_m, 0, \cdots, 0)$
solves then eq. \ref{eq:GEP_degen_block}. Therefore the solution to the original GEP given in eq.\ref{eq:GEP_degen_Matform} is then given by  $ \matr{\Phi} = \matr{\tilde U} \matr{\tilde \Phi}$.
For convenience, 
\begin{align*}
    \matr{H}_V \matr{\Phi} &  = \matr{H}_V \matr{\tilde U} \matr{\tilde \Phi} = \matr{V}^2 \matr{\tilde U} \matr{\tilde \Phi} \matr{\Lambda} = \matr{V}^2 \matr{\Phi} \matr{\Lambda}.
\end{align*}
In the third equality we used that $(\matr{\tilde \Phi}, \matr{\Lambda})$ is a solution 
of to the GEP of $(\matr{H}_V \matr{\tilde U}, \matr{\tilde V}^2 \matr{\tilde U})$ in matrix form as can be seen from left multiplying eq. \ref{eq:GEP_degen_block} with $\matr{\tilde U}$.\\
We also have the generalized unitary relation
\begin{align*}
    \matr{\Phi}^\dagger \matr{V}^2 \matr{\Phi} & = \matr{\tilde \Phi}^\dagger \matr{\tilde U}^\dagger \matr{V}^2 \matr{\tilde U} \matr{\tilde \Phi} = \matr{\tilde \Phi}^\dagger \begin{bmatrix}
        \matr{ \tilde V}^2 & 0 \\
        0 & 0
    \end{bmatrix}\matr{\tilde \Phi} = \matr{I}_m.
\end{align*}
With this the spectral theorem for hermitian Generalized Eigenvalue Problems is proven.
\qed

The proof also shows, that even in the degenerate case the columns of $\matr{\Phi}$ can be chosen, such that they are linearly independent. 
Therefore there exists a basis of $\mathbb{C}^M$ consisting of eigenvectors of the hermitian GEP.\\

\subsubsection{Proof II: GEP specific Hilbert Space}
\label{sec:proof_specral_theorem2}
    Let $(\matr{H}, \matr{V})$ be a hermitian GEP such as in Theorem \ref{thm:gen_specral_theorem}.
    Assume that $\matr{V}^2$ is positive definite, the singular case follows trivially. 
    Let $\langle \cdot, \cdot \rangle_\Hil$ denote inner product on $\Hil$. We consider the Hilberspace $\mathbb{C}^m_V$ with inner product $\langle \cdot, \cdot \rangle_V$ given by
    \begin{align*}
        \langle x, y \rangle_V = \langle \matr{V} x, \matr{V} y \rangle_\Hil = x^\dagger \matr{V}^\dagger \matr{V} y.
    \end{align*}
    Since $\matr{H}$ is by assumption hermitian in $\Hil$, it is also hermitian in $\mathbb{C}^m_V$ in the sense of the quantum mechanical convention,\footnote{Here the operator acts after the bra on the right. Formally the quantum mechanical convention 
    can be used to define the action of $\matr{H}$ in $\mathbb{C}^m_V$ through its action of on a orthonormal basis. Then it is found to be indeed linear and hermitian such that the spectral theorem applies.}
    \begin{align*}
        \langle x| \matr{H}| y \rangle_V := \langle \matr{V} x, \matr{H} \matr{V} y \rangle_\Hil = \langle \matr{H} \matr{V}  x,  \matr{V} y \rangle_\Hil.
    \end{align*}
    By the spectral theorem for finite dimensional Hilbert spaces, there exists a orthonormal basis of $\mathbb{C}^m_V$ consisting of eigenvectors of $\matr{H}$ acting 
    in $\mathbb{C}^m_V$. Let $\{e_i\}$ be the canonical coordinate system in Euclidean space $\mathbb{C}^m$ and let $\phi_i$ denote the representation of the eigenbasis of with respect to $\{e_i\}$.
    Let $\matr{\Phi}$ be the matrix containing the eigenvectors $\phi_i$ in its columns and let $\matr{\Lambda}$ be the diagonal matrix containing the eigenvalues $\mu_i$,
    such that we have  $\matr{\Phi} e_i =\phi_i$ and $\matr{\Lambda} e_i = \mu_i e_i$. 
    The orthonormality of the eigenvectors in $\mathbb{C}^m_V$ implies
    \begin{align*}
        \matr{\Phi}^\dagger \matr{V}^\dagger \matr{V} \matr{\Phi} = \matr{I}.
    \end{align*}
    For the representation of $ \matr{V}^\dagger  \matr{H} \matr{V} \matr{\Phi}$  and $ \matr{V}^\dagger \matr{V} \matr{\Phi} \matr{\Lambda}$ in Euclidean space with respect to the basis $\{e_i\}$ we have
    \begin{align*}
        \left(\matr{V}^\dagger  \matr{H} \matr{V} \matr{\Phi} \right)_{ij} & = e^\dagger_i \matr{V}^\dagger  \matr{H} \matr{V} \matr{\Phi} e_j  = \langle \matr{V} e_i ,    \matr{H} \matr{V} \phi_j \rangle_\Hil = \langle e_i| \matr{H}|  \phi_j \rangle_V = \mu_j \langle e_i, \phi_j \rangle_V \\
        & = \mu_j  e_i^\dagger \matr{V}^\dagger \matr{V} \matr{\Phi} e_j =   e_i^\dagger \matr{V}^\dagger \matr{V} \matr{\Phi}  \matr{\Lambda} e_j = \left(\matr{V}^\dagger \matr{V} \matr{\Phi} \matr{\Lambda}\right)_{ij}. 
    \end{align*}
    In the singular case where $\matr{V}^2$ is only positive semi definite, the same argument can be applied to the non degenerate part of the GEP, while the degnerate part can be ignored after a block diagonalization.
\qed

\subsection{Fact \ref{fact:spec_range}, spectral range}
\label{sec:proof_spec_range}
    Let $(\matr{H}, \matr{V})$ be a hermitian and compact GEP as in Fact \ref{fact:spec_range}. 
    By the Generalized Spectral Theorem \ref*{thm:gen_specral_theorem} we can express all proper eigenvalues $\mu_i$ as a Rayleigh quotient of the eigenvectors $x_i$,
    \begin{align*}
        \mu_i(\matr{H}, \matr{V}) = \frac{x_i^\dagger \matr{V}^\dagger \matr{H} \matr{V} x_i }{x_i^\dagger \matr{V}^\dagger \matr{V} x_i}. 
    \end{align*}
    But this is the classical Rayleigh quotient of the operator $\matr{H}$ to the vector $\matr{V} x_i$. 
    By the variational principle the Rayleigh quotient of a compact operator is bounded by its maximal and minimal eigenvalue. Therefore,
    \begin{align*}
        \lambda_{\text{min}}(\matr{H}) \leq \mu_m(\matr{H}, \matr{V}) \leq  \mu_1(\matr{H}, \matr{V}) \leq \lambda_{\text{max}}(\matr{H}) 
    \end{align*}
    and the fact follows.\\

    A second more explicit proof that some readers might prefer: By the classical spectral theorem for hermitian and compact operators, the exist an orthonormal basis $\varphi_k$ of $\Hil$ consisting of 
    eigenvectors of $\matr{H}$. Therefore we can expand any $\matr{V} x$ in terms of the eigenvectors $\varphi_k$,
    \begin{align*}
        \matr{V} x = \sum_{k} \alpha_k \varphi_k.
    \end{align*}
    Plugging into the Rayleigh quotient of the leading eigenvalue $\mu_1$ we have
    \begin{align*}
        \mu_1(\matr{H}, \matr{V}) = \frac{\sum_k \lambda_k |\alpha_k|^2  }{\sum_k |\alpha_k|^2 } \leq \lambda_{\text{max}} \frac{\sum_k |\alpha_k|^2  }{\sum_k |\alpha_k|^2 } =\lambda_{\text{max}}
    \end{align*}
    Here we also used that proper eigenvectors have $\sum_k |\alpha_k|^2 >0$. Analogously we can lower bound lowest eigenvalue $\mu_m$ by $\lambda_{\text{min}}$. With
    this 
    \begin{align*}
        \mu(\matr{H}, \matr{V}) \subset [\lambda_{\text{min}}(\matr{H}), \lambda_{\text{max}}(\matr{H})].
    \end{align*}
    follows.
\subsection{Fact \ref{fact:gen_eig_to_normal}, GEP eigenvalue correspondence}
\label{sec:ProofFact}
Since we have $\operatorname*{span}\{b_1, \cdots, b_M \} = \mathcal{E}_{\mathcal{I}}$,
we can expand eigenvectors $\varphi_k$ with $k\in \mathcal{I}$ in terms of the guess vectors $b_i$,
\begin{align*}
    \varphi_k &= \sum_{i=1}^m z_{ki} b_i = \matr{B} z_k.
\end{align*}
The eigenpairs $(\mu_i, x_i)$ of the GEP $(\matr{H}, \matr{B})$ have
\begin{align}
    \matr{B}^\dagger \matr{H}  \matr{B} x_i = \mu_i \matr{B}^\dagger \matr{B} x_i \label{eq:FactProof1}
\end{align}
Left multiplying eq.\ref{eq:FactProof1} with $z_k^\dagger$ we find for the left and right hand side,
\begin{align*}
    z_k^\dagger \matr{B}^\dagger \matr{H}  \matr{B} x_i & = \varphi^\dagger  \matr{H}  \matr{B} x_i = \lambda_k \varphi^\dagger_k  \matr{B} x_i\\
    \mu_i z_k^\dagger \matr{B}  \matr{B}^\dagger x_i & = \mu_i \varphi^\dagger_k  \matr{B} x_i.
\end{align*}
Equality of the left and right hand side implies that either $\mu_i = \lambda_k$ or $\varphi_k^\dagger  \matr{B} x_i =0$.
For $x_i \not \in \operatorname*{Ker}(\matr{B})$ at least one eigenvector $\varphi_k$ must have non zero overlap with $\matr{B} x_i$, since $\{{\varphi}_{k\in\mathcal{I}} \}$ is 
a ONB of $\operatorname*{span}\{b_1, \cdots, b_M \}$. 
Therefore we conclude that all proper generalized eigenvalues $\mu_i$ of $(\matr{H}, \matr{B})$ must coincide with an eigenvalue $\lambda_k$ of $\matr{H}$.

\subsection{Theorem \ref{thm:min_max_GEP}, min-max characterization}
\label{sec:proof_min_max_GEP}
The following is an adaption of Poincare's lemma for generalized eigenvalue problems.

\begin{lemma}[Poincare]
    Let $\mathcal{S}_k$ be a subspace of $\mathbb{C}^m$ of dimension dimension $k$. Then there exist vectors $x\in \mathcal{S}_k$ and $y \in \mathcal{S}_k^\perp $ such that
    \begin{alignat*}{2}
       \frac{\langle x, \matr{H}_V x \rangle}{ \langle x, \matr{V}^2 x\rangle }  & \leq \mu_k, \quad \text{ and } \qquad \frac{\langle y, \matr{H}_V y \rangle}{ \langle y, \matr{V}^2 y\rangle }& \geq \mu_{k+1}.
    \end{alignat*}
    \label{lem:Poincare}
\end{lemma}
\begin{proof}[\textit{Proof of Lemma \ref{lem:Poincare}}]
    Let $v_1, \cdots v_m$ be the eigenvectors of the GEP, with eigenvalues $\mu_i$ in descending order. The subspace
    $\mathcal{N} = \operatorname{span}\{ v_{k},v_{k+1}  \ldots,  v_{m}\}$ has codimension $k-1$ and 
    $\mathcal{S}_k$  has by assumption dimension $k$. Thus the two subspaces have non trivial overlap. 
    Choose $x\in \mathcal{K} \cap \mathcal{N}$, normalized with respect to $\matr{V}^2$. We have 
    \begin{alignat*}{2}
        x & = \sum_{i=k}^m \alpha_i v_i, \quad && \text{ since } x \in \mathcal{N} \\
        1 & = \sum_{i=k}^m |\alpha_i|^2 \quad && \matr{V}^2\text{-normalization}
    \end{alignat*}
    From the eigenvector properties of $v_i$ and $\matr{V}^2$ orthogonality we have
    \begin{align*}
        x^\dagger \matr{H}_V x = x^\dagger \left(\sum_{i=k}^M \alpha_i \mu_i  \matr{B} v_i \right)= \sum_{j,i=k}^M \conj{\alpha_j}\alpha_i \mu_i  v_j^\dagger \matr{B} v_i = \sum_{i=k}^M |\alpha_i|^2 \mu_i \leq \mu_k.
    \end{align*}
    For the second inequality we consider $\mathcal{N} = \operatorname{span}\{ v_{1},  \ldots, v_{k+1} \}$ which has dimension $k+1$ while 
    $\mathcal{S}^\perp_k$ has codimension $k$. 
    For $\matr{V}^2$-normalized $y \in \mathcal{K}^\perp \cap \mathcal{N}$ we find
    \begin{align*}
        y^\dagger \matr{H}_V y = \sum_{i=1}^{k+1} \mu_i |\beta_i|^2 \geq \mu_{k+1}.
    \end{align*}
\end{proof}

\begin{proof}[\textit{Proof of Theorem \ref{thm:min_max_GEP}}]
    Focusing on the first equality. we consider an arbitrary $k$ dimensional subspace $\mathcal{S}_k$ of $\mathbb{C}^m$. 
    Poincare's lemma for GEPs gives a bound on the right hand side by $\mu_k$,
    \begin{align*}
        \operatorname*{max}_{\substack{ S_k }} \operatorname*{min}_{\substack{ x \in S_k }} \frac{\langle x, \matr{V}^\dagger \matr{H} \matr{V} x \rangle}{   \langle  x , \matr{V}^\dagger\matr{V} x \rangle} \leq \mu_k.
    \end{align*}
    The bound is achieved by choosing by choosing $\mathcal{S}_k = \operatorname*{span}\{ v_1, \cdots v_k\}$.
    The second part follows analogously from the second part of Poincare's lemma.
\end{proof}

\subsection{Riemann Integrability for Integrated Spectral Inequality}
\label{sec:riemann_integrability}

Riemann integrability of $f^\uparrow(E)= \frac{E - \mu_i^\uparrow(\matr{V})(E)}{\lambda_{m^*}^\uparrow(\matr{V}^2)(E)}$ is established from the fact that monotone functions are Riemann integrable
and that $f^\uparrow(E)$ is a regular function of Riemann integrable functions. For completeness, a formal proof is included.

\begin{proof}[\textit{Proof of Claim \ref{claim:riemann_integrability}}]
Let $f^\uparrow(E)$ and $E_{\min}$ as in Claim \ref{claim:riemann_integrability}.
Let $E_1,E_2 \in \mathbb{R}$ be such that $E_2 > E_1 > E_{\min}$ and consider the interval $[E_1,E_2]$.
We show that the function $f^\uparrow(E)$ is bounded on $[E_1,E_2]$ and continuous almost everywhere, and thus, Riemann integrable.\\

By Fact \ref{fact:proper_eigenvalues_interval}, $\mu_i^\uparrow(\matr{V})(E) \leq E$ and therefore
\begin{align*}
 E_1 -E_2 \leq E - \mu_i^\uparrow(\matr{V})(E),
\end{align*}
for all $E \in [E_1,E_2]$. By Corollary \ref{cor:monotone_gen_eigenvalues} $\mu_i^\uparrow(\matr{V})(E)$ is a monotonically increasing function in $E$ and thus
\begin{align*}
 E - \mu_i^\uparrow(\matr{V})(E) \leq E_2 -\mu_i^\uparrow(\matr{V})(E_1).
\end{align*}
Eigenvalues of finite-dimensional matrix problems are finite, $|\mu_i^\uparrow(\matr{V})(E_1)|<\infty$.\footnote{There are multiple analogs of
classical bounds on eigenvalues in the setting of Generalized Eigenvalue Problems, such as the Gerschgorin Circle Theorem \cite{stewart_perturbation_1978}}
Let $m_1$ denote the rank of $\matr{V}^\uparrow(E_1)$ and $m_2$ the rank of $\matr{V}^\uparrow(E_2)$.
There are at most $m_2 - m_1$ discontinuities $\lambda_{m^*}^\uparrow(\matr{V}^2)(E)$ on $[E_1,E_2]$, due 
to a changing rank of $\matr{V}^\uparrow(E)$.\footnote{ Here we also used that 
the rank of $\matr{V}^\uparrow(E_1)$ is a monotone increasing function, as seen e.e. from Lemma \ref{lem:monotone_gram_eigenvalues}. Therefore $m_2-m_1 \geq 0$.} Denote with $\mathcal{P}=\{[ a_0=E_1,a_1),[a_1,a_2),\dots,[a_{m_2-m_1},E_2]\}$ the partition of $[E_1,E_2]$ 
into intervals on which the rank of $\matr{V}^\uparrow(E_1)$ is constant. 
By Corollary \ref{cor:monotone_gen_eigenvalues} $\lambda_{m^*}^\uparrow(\matr{V}^2)(E)$ is a monotone increasing function in $E$
on all intervals $[a_k,a_{k+1})$ of the partition $\mathcal{P}$. In particular $\lambda_{m^*}^\uparrow(\matr{V}^2)(E)$ attains its minimum on $[E_1,E_2]$ at
one of the points $a_k$. We write $\kappa =\min\limits_{E\in[E_1,E_2]}\lambda_{m^*}^\uparrow(\matr{V}^2)(E)$ and have by definition $\kappa > 0$.
Here, we have also used the assumption that $E_1 > E_{\min}$, which ensures that $E$ is not so low that $\matr{V}^\uparrow(E)$ becomes equal to zero. In summary
\begin{align*}
 \frac{E_1 -E_2 }{\kappa} \leq f^\uparrow(E) \leq \frac{ E_2 -\mu_i^\uparrow(\matr{V})(E_1)}{\kappa}
\end{align*}
and $f^\uparrow(E)$ is bounded on $[E_1,E_2]$. \\

The generalized eigenvalue $\mu_i^\uparrow(\matr{V})(E)$ is a monotone increasing function in $E$ and bounded on the interval $[E_1,E_2]$.
In particular the set of discontinuities of $\mu_i^\uparrow(\matr{V})(E)$ has measure zero on $[E_1,E_2]$. 
The function $\lambda_{m^*}^\uparrow(\matr{V}^2)(E)$ is, due to discontinuities in the rank of $\matr{V}^\uparrow(E)$, 
only a piecewise monotone functions in $E$, but its set of discontinuities has nevertheless measure zero.
The set of discontinuities of $f^\uparrow(E)$ is contained in the union of discontinuities of $\mu_i^\uparrow(\matr{V})(E)$ and $\lambda_{m^*}^\uparrow(\matr{V}^2)(E)$
and has, in particular measure zero. Therefore $f^\uparrow(E)$ is continuous almost everywhere on $[E_1,E_2]$.
We have thus established that $f^\uparrow(E)$ is Riemann integrable on $[E_1,E_2]$. Since $E_2, E_1$ was up to $E_2 > E_1 > E_{\min}$ arbitrary
the claim follows.
\end{proof}

Analogously it shown, that $f^\downarrow(E)$ defined as
\begin{align*}
 f^\downarrow(E)= \frac{E - \mu_i^\downarrow(\matr{V})(E)}{\lambda_{m^*}^\downarrow(\matr{V}^2)(E)}
\end{align*}
is Riemann integrable on any compact interval $[E_1,E_2] \subset \mathbb{R}$ with $E_1 < E_2 \leq E_{\max}$.


\chapter{Prolate Filter Diagonalization}
\label{chap:FilterDiag}

\section{Signal Processing through a Symmetry between Harmonic Analysis and Quantum Mechanics}
\subsection{Introduction}
In Chapter \ref{chap:DimRedSepAna}, we emphasized the significance of spectral analysis of operators in the natural sciences. 
Equally important, however, is spectral analysis in the context of frequency decompositions of signals.
Our ability to accurately perceive and interpret the world—whether through high-precision experiments or long-distance communication—
depends on our capacity to decompose signals into their frequency components.
A major limitation for the degree of precision with which the spectrum of a signal can be determined is the fact that signals can only be measured for a finite time.\\

Fortunately, the problems of determining the spectral decompositions of operators, in terms of their eigenvalues and eigenvectors, 
and the spectral decompositions of functions, in terms of their Fourier transforms, are fundamentally interconnected. 
A deeper dualism between harmonic analysis and functional analysis has long been recognized across different fields and has lead to various breakthroughs \cite{Neumann, ProI, Cramér1992}.\\

An early anticipation of this deeper symmetry can be found in the work of Bochner \cite{Bochner1932}. 
Bochner saw, that functions whose Fourier transform is a probability measure, induce a positive definite convolution operator.
A familiar example where Bochners Theorem can be applied 
to show  positivity of a convolution operator is given by the prolate spheroidal integral operator $\BL_\bw \TL_T$. \\

A more recent insight into the connection between harmonic analysis and quantum mechanics can be recognized within 
numerical quantum chemistry, and is of significant computational interest. 
Through a quantum mechanical interpretation, the Harmonic Inversion Problem (HIP) of functions whose Fourier transform is a probability measure 
can be mapped to a hermitian generalized eigenvalue problem. In particular, the 
problem of determining the frequency decomposition of a signal $C(t)$ can access the approximation theory 
that has been developed in Chapter \ref{chap:DimRedSepAna}.\\

We will briefly elaborate, on how Filter Diagonalization has been developed in the context of quantum chemistry.

\subsection{Related Work}
While Filter Diagonalization (FD) in the end turned out as a pure signal processing method, it is of no surprise 
that it has been developed in the context of quantum chemistry, as here signals have 
indeed the interpretation of Hamiltonian time evolution.\\

Much in the spirit of the intuition we outlined in  Section \ref{sec:intuition_GEP},
Neuhauser proposed to project the spectral problem of a high-dimensional Hamiltonian onto a low-dimensional subspace through 
time-evolved states that are processed through filter functions \cite{Neuhauser1990BoundSE}. It was later noticed, that all the information necessary
to compute the matrix elements of the GEP is contained in the signal $C(t)$, such that there is 
actually no need to access the generating Hamiltonian \cite{wall_extraction_1995}. Noteworthy refinements and 
further adaptations of the numerical method have been developed by Manndelshtam \cite{Mandelshtam}. Levitina and Brändas have already 
argued for prolates as the optimal choice of filter functions \cite{LEVITINA20091448}. Here we also first observed the usage of Chebyshev systems 
to detect the number of frequencies in a band. \\

Unfortunately, the truncation error estimates 
for integrals derived in \cite{LEVITINA20091448} are found to be incorrect and contradict the optimality properties of the prolates. It also appears
that the authors were unaware of the catastrophic cancellation effects in the generation of prolates and their eigenvalues, and their numerical demonstration was not convincing.
However, the problem of the cancellation effects was already addressed and solved by then in \cite{Buren2002AccurateCO}.

\subsection{Result Highlight and Overview}

In this chapter, we establish Filter Diagonalization as an approximation method, that can access the 
theoretical framework of Chapter \ref{chap:DimRedSepAna}. We derive a precision guarantee for 
Filter Diagonalization through prolate spheroidal wave functions.

\begin{theorem}
    \label{thm:PFD_precision}
 Consider the signal $C(t) = \sum_k |a_k|^2 e^{i\omega_k t}$ which consists of discrete frequencies $\omega_k \in \mathbb{R}$
 and the interval $B = [ \bw - \omega^*,  \bw +\omega^* ]$.
 Consider the protocol, that generates a GEP with matrix elements,
    \begin{align*}
 \matr{V}^2_{sl} & = \int_{-T}^{T}  \int_{-T}^{T} e^{-i \omega^*(\tau -t)} \prlt_s(\tau) \prlt_l(t) C(\tau - t) d\tau dt\\
 (\matr{H}_V)_{sl} & = i \int_{-T}^{T}  \int_{-T}^{T} e^{-i \omega^*(\tau -t)} \prlt_s(\tau) \prlt_l(t) \partial_t C(\tau - t) d\tau dt,
    \end{align*}
 where $\{\prlt_n\}$ is the sequence of $\bw T$-prolates. Let $m$ be the number of frequencies $\omega_k$ in the band $B$. 
 As in Algorithm \ref{alg:1} obtain a refined $m$ dimensional 
 GEP $(\matr{H}_{V_*}, \matr{V}^2_*)$ from $(\matr{H}_{V_M}, \matr{V}^2_M)$. 
 The eigenvalues of $(\matr{H}_{V_*}, \matr{V}^2_*)$ are denoted as 
    $\tilde \omega_k$ and satisfy inequalities,\footnote{ As usually we assume in (\ref{eq:prlt_FD_precision}) that $\omega_k \in B$ and that $\omega_l$ and $\omega_k$ are indexed in the same order, eg. 
    $\omega_k < \omega_{k+1} < \cdots$ implies $ \tilde \omega_k < \tilde \omega_{k+1} < \cdots$.}
    \begin{align}
 \tilde \varepsilon_M  \frac{\sum\limits_{\omega_l < \omega^* - \bw } (\omega_l - \omega_k)|a_l|^2 }{\lambda_m(\matr{V}^2_M) - \tilde \varepsilon_M C(0)}   \leq \tilde \omega_k -\omega_k \leq \tilde \varepsilon_M \frac{\sum\limits_{\omega_l > \omega^* + \bw } (\omega_l - \omega_k)|a_l|^2 }{\lambda_m(\matr{V}^2_M) - \tilde \varepsilon_M C(0)}.  \label{eq:prlt_FD_precision}
    \end{align}
 Here $\tilde \varepsilon_M = 2\pi \sum_{l=0}^{M-1}  \gamma_l(1-\gamma_l) C_{\text{extra},l}$ and $C_{\text{extra},l}$ is as in Theorem \ref{thm:prlt_bound}.
\end{theorem}

The error parameter $\tilde \varepsilon_M$ of Prolate Filter Diagonalization (PFD) is significantly close to zero if
 $M \ll 2\bw T / \pi$. In equation (\ref{eq:asymptotic_prlt_env}) of Section \ref{sec:asymptotic_prlt_env} we also 
given an asymptotic formula for $\tilde \varepsilon_M(c)$ for large $c= \bw T$.\\

In Section \ref{sec:FilterDiag_eps_approx_solver} we provide an expression to approximate the amplitudes $|a_k|^2$. 
Proposition \ref{prop:filter_diagonalization_amplitude_precision} gives precision guarantees for the approximated amplitudes. \\

In practice, the number of frequencies $m$ in the signal $C(t)$ is a priori unknown and has to be detected as described in Algorithm \ref{alg:1}
and Section \ref{sec:dimension_detection}.
Due to the small $\varepsilon$-estimates provided by PFD and strong independence of relations of the prolate basis, the technique enables highly sensitive dimension detection. 
Section \ref{sec:detectability_frequencies}  introduces practical equations, that quantify the signal strength of a filter system 
specific to a set of frequencies.\\

A notable aspect of Filter Diagonalization is that the computational complexity to approximate
$N$ frequencies scales linearly $N$. This is in sharp contrast to the computational cost of direct matrix diagonalisation, which quickly becomes impractical 
for large matrices.
Once all frequencies and amplitudes are approximated, they can be used in the lower and upper estimates in equation (\ref{eq:prlt_FD_precision})
to compute sharp precision estimates.

\subsection{Overview of the detailed Exposition}

In Section \ref{sec:method_development_infinite_time} we formally derive Filter Diagonalization through a quantum mechanical interpretation of signals in the theoretical context of infinite time access.
While this section has an introductory character, we introduce some new perspectives and mathematical formalism that is not common in the literature for Filter Diagonalization.
In particular, we formalize FD through an Alternant matrix, that is induced by the filter functions, which is of significant practical advantage as seen in Section \ref{sec:amplitude_uncertainty_quantification} and \ref{sec:detectability_frequencies}.\\

Section \ref{sec:finite_time_access} is the core of this chapter and establishes FD as an $\varepsilon$-approximation method as 
introduced in Definition \ref{asump:protocol}. In particular, FD can access the approximation theory of Chapter \ref{chap:DimRedSepAna} along 
with its precision guarantees. A more general version of Theorem \ref{thm:PFD_precision} is given in Theorem \ref{thm:filter_diagonalization} and permits to 
use any sequence of filter functions. In Section \ref{sec:amplitude_uncertainty_quantification} precision guarantees for approximated amplitudes are derived.
In Section \ref{sec:detectability_frequencies} we develop a concept, that allows to quantify the ability of a filter system to distinguish frequencies.\\

In Section \ref{sec:ProlateFilterDiag} we establish Theorem \ref{thm:PFD_precision} for the special case of Prolate Filter Diagonalization by 
applying the supremum bound we derived in Chapter \ref{chap:prlt_bound}.

\newpage 

\section{Method Development in the Playground of Infinite Time Access}

\label{sec:method_development_infinite_time}
\subsection{Quantum Mechanical Interpretation of Signals}

Within Physics and Chemistry, the Harmonic Inversion Problem (HIP)
often arises in the context of spectroscopy. The signal $C(t)$ 
is then a autocorrelation function of a quantum system described by a Hamiltonian $\matr{H}$
and the frequencies $\omega_k$ correspond 
to the energy levels of the system,
\begin{align}
    C(t) = \langle \Psi(t), \Psi(0) \rangle = \langle \Psi(0), e^{i \matr{H} t } \Psi(0) \rangle. \qquad \label{eq:autocorrelation_quantum}
\end{align}
Here $\Psi(0)  \equiv  \Psi$ is the initial state of the system,
\begin{align*}
    \Psi(t) = e^{-i\matr{H}t } \Psi(0) = \sum_k a_k e^{-i\omega_k t } \varphi_k,
\end{align*}
and $\varphi_k$ are the eigenstates of $\matr{H}$ to eigenvalue $\omega_k$.
However, any signal that is a Fourier transform of a measure, can be endowed with such a quantum mechanical 
interpretation and seen as an autocorrelation function. For simplicity, we will focus here on
signals consisting of discrete frequencies without an accumulation point. Natural generalizations will be subject to a future work.\\

We will remain in the quantum mechanical interpretation and treat the signal $C(t)$ as an autocorrelation function 
to a wavepacket $\Psi(t)$ evolving under a Hamiltonian $\matr{H}$. From this viewpoint, the HIP is mapped 
to a generalized eigenvalue problem. From a signal processing perspective,
the quantum mechanical interpretation is mere auxiliary construct, as we find that all the information necessary to compute the GEP is contained in the signal $C(t)$.\\

With respect to applications in quantum chemistry, the interpretation of $C(t)$ as autocorrelation function is not just a mathematical construct. But the 
observation that the GEP can be solely computed from a signal $C(t)$ without accessing the Hamiltonian allows to avoid high dimensional vector multiplications 
that would remain necessary if we were to remain in the picture of a wave vector evolving under $\matr{H}$.\\

The spirit of Filter Diagonalization as introduced here, leans on a long-standing 
intuition of a deeper dualism between harmonic analysis and quantum mechanics. Here we mean quantum mechanics 
as a mathematical theory close to functional analysis, that deals with vectors of a Hilbert space that are time evolved under an 
operator exponential. To be instructive, we first familiarize ourselves with the machinery in the un-physical scenario where 
we assume infinite access to the signal $C(t)$.

\subsection{From Harmonic Inversion to Generalized Eigenvalue Problem}

We assume access to a signal $C(t)= \sum_{k=1}^N |a_k|^2 e^{i\omega_k t}$ for all $t\in \mathbb{R}$. 
Naturally the frequencies are considered to be distinct from another, $\omega_k \neq \omega_l$. 
Let $\matr{H}: \Hil \to \Hil$ be an auxillary hermitian operator, acting on an auxiliary Hilbert space $\Hil$.
We require that the spectrum of $\matr{H}$ has a subset of eigenvalues $\omega_k$ that coincide with the frequencies of the signal $C(t)$ and denote the corresponding eigenvectors as $\varphi_k$.
Let $\Psi(t) = e^{i\matr{H}t} \Psi_0 = \sum_{k=0}^N a_k e^{-i \omega_k t }\varphi_k$ be the time evolving wave packet that generates the signal, in 
the sense of Hamiltonian time evolution,
\begin{align*}
    C(t) = \langle \Psi(t), \Psi_0 \rangle = \langle \Psi_0, e^{i \matr{H} t } \Psi_0 \rangle.
\end{align*}
The aim is to determine all frequencies of $C(t)$ within an interval $[-\bw +\omega^*, \bw +\omega^*]$.
We denote the number of frequencies in this interval by $m^*$, and assume that $m^* < \infty$.\\

To determine the frequencies of $C(t)$, the spectral information contained in the time evolved wavepacket $\Psi(t)$ is used 
to generate a guess vector space, that coincides with an spectral subspace of $\matr{H}$.
The spectral subspace of interest is spanned by all eigenvectors of $\matr{H}$ that have 
eigenvalues in the interval $[-\bw + \omega^*, \bw +\omega^*]$ and non zero overlap with $\Psi(0)$.
We denote this subspace by $\mathcal{E}_{\omega^*, \bw} (\matr{H}_\Psi)$ and want to generate guess vectors $v_l$ such that
\begin{align}
    \operatorname{span}(v_1, \cdots, v_M) = \mathcal{E}_{\omega^*, \bw} (\matr{H}_\Psi) . \quad \label{eq:span_condition}
\end{align}
The true dimensionality of $\mathcal{E}_{\omega^*, \bw} (\matr{H}_\Psi)$ and thus number of frequencies in $[-\bw +\omega^*, \bw +\omega^*]$ is denoted by $m^*$, and is a priori 
unknown. The guess vectors are stored in a matrix $\matr{V} = [v_1, \cdots, v_M] \in \mathbb{C}^{N\times M}$ and used to project the eigenvalue problem of $\matr{H}$ onto to 
an generalized eigenvalue problem (GEP) of dimension $M$,
\begin{align} 
    \matr{V}^\dagger \matr{H} \matr{V} \matr{\Phi} =  \matr{V}^2 \matr{\Phi} \matr{\Lambda}(\vec{\mu}). \label{eq:GEP_exact}
\end{align}
Here $\matr{\Lambda}(\vec{\mu}) = \operatorname{diag}(\mu_1, \cdots, \mu_M)$ is a diagonal matrix of the eigenvalues $\mu$.
If condition $(\ref{eq:span_condition})$ is satisfied, the GEP will yield eigenvalues $\mu$ that coincide with the frequencies $\omega_k \in [-\bw + \omega^*, \bw +\omega^*]$ of the signal.\\

To generate the guess vectors $v_i$ a sequence of \emph{filter functions} $\{f_l\} \subset \Ls^2_{\infty}$ is used.
We denote the Fourier transforms of the filter functions by $\{F_l\}$.\footnote{ Sometimes we will also refer to the Fourier transfroms $\{F_l\}$ as the filter functions.}
The guess vectors of choice are time overlaps of the evolved wave vector $\Psi(t) e^{i\omega^* t}$ with filter functions $f_l(t)$,
\begin{align}
    v_l & = \int_{-\infty}^\infty f_l(t) \Psi(t) e^{i\omega^* t} dt = \sum_k \int_{-\infty}^\infty  f_l(t) a_k e^{-i (\omega_k-\omega^*) t}  \varphi_k dt \notag \\
        & = \sum_k   F_l(\omega_k - \omega^*) a_k  \varphi_k = \sum_k   F_l(\omega_k^*) a_k  \varphi_k. \label{eq:guess_vector}
\end{align}
In the last line a shorthand notation $\omega_k^* = \omega_k - \omega^*$ for the shifted frequencies was introduced. The resulting 
guess vector operator $\matr{V} = [v_1, \cdots, v_M]$ has matrix elements with respect to the eigenbasis of $\matr{H}$,
\footnote{In our notation we will often not distinguish between the operator and its matrix representation.}
\begin{align*}
    \matr{V}_{kl} = \langle \varphi_k | \matr{V} e_l \rangle = F_l(\omega_k^*) a_k.
\end{align*}
We sometimes denote the matrix by  $\matr{V}_M$ to indicate the guess dimensionality $M$ of the matrix. As the method is iterative,
$\matr{V}_M$ contains all the information of $\matr{V}_{M-1}$.
We recognize in $(F_l(\omega_k^*))_{kl} \in \mathbb{C}^{N\times M}$ an \emph{Alterant matrix} of the filter functions evaluated at the frequencies of the signal. \\

Consider the operator-valued function defined as,
\begin{align*}
    \matr{\mathfrak{F}}_M : \mathbb{R}^\star \to \mathbb{C}^{ \star \times M}, \qquad \vec{\omega} \mapsto \matr{\mathfrak{F}}_M(\vec{\omega}) = (F_l(\omega_k))_{kl}.
\end{align*}
The star $\star$ indicates that the input dimensionality of $\matr{\mathfrak{F}}_M$ is not fixed, but adapted to the input vector.\footnote{
This could be more formally elaborated, by defining $\mathbb{R}^\star$ slightly analogous to the quantum mechanical 
Fock space but we will not pursue this here.} The subscript $M$ indicates the guess dimensionality of the matrix,  which may be omitted if it is clear from the context or to emphasize
 the generality of the operator, to which the guess dimensionality is a mere input parameter.
We refer to $\matr{\mathfrak{F}}$ as the \emph{Filter Operational}, corresponding to the filter functions $\{F_l\}$.\\

With the filter operational the guess matrix $\matr{V}$ can be factorized in amplitude and frequency dependent matrices. 
Let $\vec{a} = (a_1, \cdots, a_N)^T$ and $\vec{\omega} = (\omega_1, \cdots, \omega_N)^T$ and be the vectors with the energy coefficients $a_i$ corresponding
to frequencies $\omega_i$ of the wave packet $\Psi(0)$. Let $\matr{\Lambda}(\vec{a})= \operatorname{diag}(a_1, \cdots, a_N)$ denote the diagonal matrix of $\vec{a}$.
The guess vector matrix $\matr{V}$ can be written as
\begin{align*}
    \matr{V} = (F_l(\omega_k^*) a_k)_{kl} = \matr{\Lambda}(\vec{a}) \matr{\mathfrak{F}}(\vec{\omega^*}).
\end{align*}

In order to ensure that the eigenvalues $\mu$ of the GEP (\ref{eq:GEP_exact}) coincide with the frequencies $\omega_k \in [-\bw +\omega^*, \bw +\omega^*]$ of the signal, condition $(\ref{eq:span_condition})$
must be met. 
The guess vector space is guaranteed to be a subspace of $\mathcal{E}_{[-\bw, \bw]} (\matr{H}_\Psi)$, if 
the filter functions $f_l$ are band limited to $[-\bw, \bw]$, that is $f_l \in \BL_\bw$. Indeed, $(\ref{eq:guess_vector})$ 
has $v_l \in \mathcal{E}_{[-\bw, \bw]} (\matr{H}_\Psi)$ if $F_l(\omega_k^*) = 0$ for $|\omega_k^*| > \bw$. 
Thus a sequence of band limited filter functions $\{f_l\}$ guarantees
\begin{align*}
    \operatorname{span}(v_1, \cdots, v_M) \subseteq \mathcal{E}_{\omega^*, \bw} (\matr{H}_\Psi) .
\end{align*}
for all $M$. It remains to ensure that equality in $(\ref{eq:span_condition})$ can be achieved for $M$ high enough and to detect 
the true number of frequencies $m^*$ in target interval $[-\bw, \bw]$.

\subsection{Dimension Detection through Chebyshev Systems}
\label{sec:dimension_detection_Chebyshev}

The alterant matrix $\matr{\mathfrak{F}}(\vec{\omega})$ is a practical object in this context of Filter Diagonalization. 
A heuristic method to detect the dimensionality of $\mathcal{E}_{\omega^*, \bw} (\matr{H}_\Psi)$ would be to increase $M$ for a few iterations without the rank $m$ of $\matr{V}^2$ increasing, and then to assume $m=m^*$.
However, for maximal information gain of each iteration the guess vectors $v_l$ should be as independent as possible. And by requiring 
the filter functions  $F_l$ to satisfy a certain independence relation, a precise dimension detection mechanism is feasible.
\begin{definition}
    A set of functions $F_1, \cdots , F_M$  is a \emph{Chebyshev systems} on $[-\bw, \bw]$ if for all distinct $\omega_1, \cdots, \omega_M \in [-\bw, \bw]$ the Alterant matrix 
    $\matr{\mathfrak{F}}(\vec{\omega})$ is non singular
    \begin{align*}
        \operatorname{det}(\matr{\mathfrak{F}}(\vec{\omega})) = \left|\begin{array}{cccc} F_1\left(\omega_1\right) & F_2\left(\omega_1\right) & \cdots & F_M\left(\omega_1\right) \\ F_1\left(\omega_2\right) & F_2\left(\omega_2\right)& \cdots & F_M\left(\omega_2 \right) \\ \vdots & \vdots & \ddots & \vdots \\ F_1\left(\omega_M\right) & F_2\left(\omega_M\right) & \cdots & F_M\left(\omega_M\right)\end{array} \right|\neq 0.
    \end{align*}    We call a sequence $\{F_l\}$ a \emph{complete Chebyshev systems} if $F_1, \cdots , F_M$ is a Chebyshev system for all $M$. 
\end{definition}
A sequence of filter functions $\{F_l\}$  that is a complete Chebyshev system on $[-\bw, \bw]$ guarantees a full rank of the Gramm matrix $\matr{V}^2$ for all 
guess dimensionalities smaller or equal to the true dimensionality $M\leq m^*$.\\

If $\{f_l\}$ is also bandlimitted, the dimensionality of the guess space $\mathcal{E}_{[-\bw, \bw]} (\matr{H}_\Psi)$ can therefore 
be detected form the Gram matrix, by determining the largest guess dimensionality $M$ such that $\matr{V}^2_M$ is non-singular,
\begin{alignat*}{2}
    \operatorname{det}(\matr{V}^2_M) & > 0 \qquad \text{ for } M \leq m^*  \qquad &&\text{(Chebyshev system)} \\
    \operatorname{det}(\matr{V}^2_M) & = 0 \qquad \text{ for } M > m^*   \qquad &&\text{(band limited)}.
\end{alignat*}

In summary, a sequence of $\bw$-band limited filter functions $\{f_l\}$, whose Fourier transforms form a complete Chebyshev system in $[-\bw, \bw]$
can determine all frequencies $\omega_k$ in $[-\bw + \omega^*, \bw + \omega^*]$ of the signal $C(t)$, 
by solving a GEP of dimensionality equal to the number of frequencies in the interval of interest
\begin{align}
    \matr{V}^\dagger \matr{H} \matr{V} \matr{\Phi} =  \matr{V}^2 \matr{\Phi} \matr{\Lambda}(\tilde \omega).  \label{eq:GEP_PG}
\end{align}
Here $\tilde \omega$ is the vector containing all $m^*$ frequencies of $C(t)$ in $[-\bw + \omega^*, \bw + \omega^*]$,
and $\matr{\Phi}=[x_1,\cdots, x_{m^*}]$ is the matrix of eigenvectors of the GEP. 
With $\matr{\Phi}$ and the determined frequencies $\tilde \omega$ we can also obtain the 
amplitudes $\tilde{a}^2$ that correspond to $\tilde \omega$. 

\subsection{Amplitude Extraction}
\label{sec:amplitude_extraction}
Even though the Hilbert space $\Hil$ is mere auxiliary construct, we become briefly 
a bit more precise in the notation to see how the amplitudes $|a_k|^2$ can be computed from the GEP.
The guess vector operator is denoted as $\matr{\hat V}: \mathbb{C}^{m^*} \to \Hil$ 
and its matrix representation with respect to left basis $\{ \varphi_k \}$ and right basis $\{e_i\}$ as $\matr{V}$. 
We denote with $\vec{a}$ and $ \vec{\omega}$ the $N$-dimensional vecotrs, consisting of all frequencies and amplitudes in 
$C(t)$ and with  $\tilde a $ and $ \tilde \omega $ the $m^*$ dimensional sub-vectors, that correspond to the spectral subspace 
$\mathcal{E}_{\omega^*, \bw} (\matr{H}_\Psi)$.
Since the filter functions are band limited, the guess operator $\matr{\hat V}$ will only have 
a non trivial image in the target spectral subspace $\mathcal{E}_{\omega^*, \bw} (\matr{H}_\Psi) \subset \Hil$. 
Therefore, $\matr{\hat V}$ is isomorphic to a square matrix $\mathbb{C}^{m^* \times m^*}$,
\begin{alignat*}{2}
    &\matr{\hat V}: \mathbb{C}^{m^*} \to \Hil && \equiv \mathbb{C}^{m^*} \to \mathcal{E}_{\omega^*, \bw} (\matr{H}_\Psi) \\
    &\matr{V}  =  \matr{\Lambda}(\vec{a}) \matr{\mathfrak{F}}_{m^*}(\vec{\omega}) && \equiv \matr{\Lambda}(\tilde a) \matr{\mathfrak{F}}_{m^*}(\tilde \omega^*).
\end{alignat*}
We denote the eigenvectors to the GEP (\ref{eq:GEP_PG}) by $x_i$. 
The usual normalization of the eigenvectors of a GEP has $\matr{\Phi} \matr{V}^2 \matr{\Phi} = \matr{I}$ (Section \ref{sec:GeneralizedSpectralTheorem}).
In particular $\{\matr{\hat V} x_i\}$ will form an orthonormal set of vectors. Since we have that $\mu_i = \omega_i$ it follows 
from the uniqueness of the eigenvectors of the GEP, that $\matr{\hat V} x_i = \varphi_k$.
This implies on the matrix representation with respect to $\{\varphi_k\}$ that $\matr{V} x_i = e_{\sigma(i)}$ for all $i$ and some permutation $\sigma$. In particular, 
\begin{alignat}{3}
    & \qquad   \matr{V}   \matr{\Phi}  = \matr{\Lambda}(\tilde a) && \matr{\mathfrak{F}}_{m^*}(\tilde \omega^*) \matr{\Phi}  && = \matr{I}  \notag \\
    & \Longleftrightarrow \qquad && \matr{\mathfrak{F}}_{m^*}(\tilde{\omega}^*) \matr{\Phi}&& = \matr{\Lambda}(\tilde a)^{-1}.  \label{eq:amplitude_computation}
\end{alignat}
Thus after the frequencies $\tilde \omega$ are determined from the GEP, the alternant matrix $\matr{\mathfrak{F}}_{m^*}(\tilde \omega^*)$ can be used to compute the amplitudes $\tilde a^2$.
If $\omega_k$ is the eigenvalue of $\varphi_k$ then $\matr{V} x_i = e_{\sigma(i)}$ also implies $\omega_k = \tilde \omega_{\sigma(i)}$. Therefore, the ordering 
of the amplitudes $\tilde a$ obtained through eq.(\ref{eq:amplitude_computation}) agrees with the ordering of the frequencies $\tilde \omega$. 
Thus the signal $C(t)$ can be reconstructed from the GEP and an exact frequency decomposition obtained. The reconstruction
of the signal obtained through the amplitudes and frequencies from a Filter Diagonalization in the band $[-\bw + \omega^*, \bw + \omega^*]$,
effectively corresponds to the application of an ideal band pass filter $\mathcal{B}_{\omega^*,\bw}$ to the signal $C(t)$.
\begin{align*}
   \mathcal{B}_{\omega^*,\bw} C(t) = \sum_{k=1}^{m^*} \tilde{a}_k^2 e^{i \tilde \omega_k t}
\end{align*}
Here $\mathcal{B}_{\omega^*,\bw}$ is the projection operator onto the space of functions with Fourier transform supported in $[-\bw + \omega^*, \bw + \omega^*]$.

\subsection{Back to Signal Processing}

\label{sec:back_to_signal_processing}
It remains to be shown that the hermitian operator $\matr{H}$ along with the Hilbert space $\Hil$ and the time evolved wavepacket $\Psi(t)$ are
indeed pure auxiliary constructs. All information necessary to compute the matrix elements of the GEP is 
contained in the signal $C(t)$. For the Gram matrix $\matr{V}^2$ we have,
\begin{align*}
    \matr{V}_{sl}^2 = \langle v_s, v_l \rangle & = \int \int^\infty_{-\infty} e^{-i \omega^*(\tau-t)} \conj{f_s}(\tau) f_l(t) \langle \Psi(\tau), \Psi(t) \rangle dt d\tau\\
    & = \int \int^\infty _{-\infty}e^{-i \omega^*(\tau-t)} \conj{f_s}(\tau) f_l(t) C(\tau - t) dt d\tau.
\end{align*}
The action of the Hamiltonian $\matr{H}$ on the wavepacket $\Psi(t)$ is 
related to the derivative of the signal 
\begin{align*}
    \langle \Psi(\tau), \matr{H}\Psi(t) \rangle & = \langle \Psi(\tau), \matr{H} e^{- i \matr{H} t }\Psi(0) \rangle =  i \partial_t \langle \Psi(\tau), e^{- i \matr{H} t }\Psi(0) \rangle\\
                                        & = i \partial_t C(\tau - t).
\end{align*}
Thus, the left-hand side $\matr{V}^\dagger \matr{H} \matr{V}= \matr{H_V}$ of the GEP has matrix elements,
\begin{align*}
    ( \matr{H_V})_{sl} & = \langle v_s,\matr{H} v_l \rangle  = \int \int^\infty_{-\infty} e^{-i \omega^*(\tau-t)} \conj{f_s}(\tau) f_l(t) \langle \Psi(\tau), \matr{H} \Psi(t) \rangle dt d\tau\\
    & =  i \int \int^\infty_{-\infty} e^{-i \omega^*(\tau-t)} \conj{f_s}(\tau) f_l(t)  \partial_t C(\tau - t) dt d\tau \\
    &= \int \int^\infty_{-\infty}   e^{-i \omega^*(\tau-t)} \conj{f_s}(\tau)  C(\tau - t)  \left( \omega^* f_l(t)  -i f_l'(t) \right) dt d\tau.
\end{align*}
In the last line a partial integration was applied. \\

We have thus derived, that the problem of finding frequency decompositions of the signal $C(t)$ can be mapped to generalized eigenvalue problem.
Every signal $C(t)$ induces a convolution operator $\matr{T}_C$ defined through,
\begin{align*}
     \matr{T}_C f(\tau) = \int_{-\infty}^\infty C(\tau -t) f(t) dt.
\end{align*}
The matrix elements of the GEP that dermines the frequencies of $C(t)$ within the band $[-\bw , \bw ]$ are then given by
\begin{align}
    ( \matr{H_V})_{sl}  = \langle f_s, -i \frac{\partial}{\partial \tau }  \matr{T}_C f_l \rangle, \qquad \qquad \matr{V}_{sl}^2  = \langle f_s, \matr{T}_C f_l \rangle.  \label{eq:FD_Schrödinger}
\end{align}
where we assume the the inner product of $\Ls^2_{\infty}$. The operator $i \partial_\tau$, that appears on the left-hand side of the GEP
 is recognized from the left-hand side of the Schrödinger equation.

\paragraph{Remarks}
Chebyshev systems already play a crucial role in approximation theory. 
In particular in the numerical evaluation
in the numerical evaluation of integrals and interpolation they form a work horse. We find that the independence relation, that they ensure 
is also valuable asset in the context of signal processing and dimension reductions for spectral analysis. The alternant matrix that is induced by 
a filter system has previously not been used in the context of Filter Diagonalization. However, we find it to be a valuable object in the 
uncertainty quantification of the method. In real world applications that are subject to noise, a mere independence relation though a nonzero determinant
is not quite enough. The independence relation should be strong enough such that we can distinguish between the signal and noise, for accurate signal detection.
As seen in Section \ref{sec:detectability_frequencies}, the alternant matrix can be used to quantify the signal strength of a filter system.\\

Lastly, the perspective of Filter Diagonalization through equation (\ref{eq:FD_Schrödinger}) is new and useful for unexplored generalizations of the method.
Perhaps, it can in itself provide a new perspective on Quantum Mechanics, that is not based on an attempt to develop a theory that aims to describe 
the physical world as it is. \\

A viewpoint that equation (\ref{eq:FD_Schrödinger}) may enable, is to develop a theory, that describes how we \emph{observe the reality}. 
Namely, through signals. And in this Section, we have seen that the time evolution of every signal whose Fourier transform is a probability measure, can 
be described by a unitary operator as implied by the Schrödinger equation.
While it would be certainly interesting to see this speculation through, it is not the focus of this work.\\



\newpage 

\section{Filter Diagonalization with Finite Time Access}
\label{sec:finite_time_access}

Physical reality unfortunately does not allow infinite access to a signal.
In practice, it is not possible to take a full Fourier transform of actual data.  
In the literature of Wavelet Theory and related fields, 
it is well-established, that the challenge of spectral analysis despite finite time access to the signal,
requires approximate methods.
Filter Diagonalization shows promise to be an optimal approach for such an information processing task. \\

In fact, exact frequency recovery is possible despite finite time access to the signal,
and corresponds to the diagonalization of a GEP of dimensionality equal to the number of frequencies in the signal.\\

We assume access to a signal $C(t)= \sum_{k=1}^N |a_k|^2 e^{i\omega_k t}$ for a finite time $t\in [-2T,2T]$.
Analogous to the infinite time case, we consider an auxiliary Hilbert space $\Hil$, an operator $\matr{H}: \Hil \to \Hil$ and 
a wavepacket $\Psi(t) = e^{i\matr{H}t} \Psi_0 = \sum_{k=1}^N a_k e^{i \omega_k t }\varphi_k$, that generates the signal $C(t)$.
The following provides the definition of filter systems that will be used. 

\begin{definition}
    \label{def:filter_system}
    A \emph{filter system} is a sequence of square integrable functions $\{F_l\}_{l=1}^\infty \subset \Ls_\infty^2$.
    A filter system is accompanied by a matrix valued function $\matr{\mathfrak{F}}_M$ defined as
    \begin{align*}
        \matr{\mathfrak{F}}_M: \mathbb{R}^\star \to \mathbb{C}^{\star \times M}, \qquad \omega \mapsto \matr{\mathfrak{F}}_M(\omega) = (F_l(\omega_k))_{kl}.
    \end{align*}
    The \emph{filter envelope of degree M} is the function 
    \begin{align*}
        \mathfrak{F}^{\text{env}}_M (\omega)= \sum_{l=1}^M |F_l(\omega)|^2.
    \end{align*}
    The time dual $\{f_l\}$ of a filter system  is the sequence of inverse Fourier transforms of $\{F_l\}$,
    \begin{align*}
        \{f_l\}_{l=1}^\infty = \{\FT^{-1}[F_l] \}_{l=1}^\infty.
    \end{align*}
    If the functions $f_l$ are supported only on $[-T,T]$, we call 
    $\left( \{f_l\}, \{F_l\} \right)$ a \emph{$T$-filter system}.
    If the sequence of functions $\{F_l\}$ forms a complete Chebyshev system on $[-\bw, \bw]$ we call it a \emph{$\bw$-filter system}.
\end{definition}

We will synonymously refer to the filter functions in time space $\{f_l\}$ or in frequency space $\{F_l\}$ as the filter system, 
as the Fourier transform is a bijection. A time-limited filter system has $ f_l = \TL_T f_l$ and allows signal processing with finite time access.
The guess vectors $v_l$ are then given by

\begin{align*}
    v_l & = \int_{-T}^{T} f_l(t) \Psi(t) e^{i\omega^* t} dt = \sum_k \int_{-T}^{T}  f_l(t) a_k e^{-i (\omega_k-\omega^*) t}  \varphi_k dt\\
        & = \sum_{k=1}^N a_k \int_{-\infty}^{\infty}  e^{-i (\omega_k-\omega^*) t}  \TL_T f_l(t)   \varphi_k dt \\
        & = \sum_{k=1}^N a_k \FT[\TL_T f_l](\omega_k-\omega^*) \varphi_k = \sum_{k=1}^N a_k F_l(\omega_k-\omega^*) \varphi_k.
\end{align*}
Analogous to Section \ref{sec:back_to_signal_processing} integral representations for the GEP matrix elements can be derived.
\begin{align}
    \matr{V}_{sl}^2 = \langle v_s, v_l \rangle & = \int^T_{-T} \int^T_{-T} e^{-i \omega^*(\tau-t)} \conj{f_s}(\tau) f_l(t) C(\tau - t) dt d\tau.    \label{eq:GEP_V_el_time}\\
    & = \sum_{k=1}^N |a_k|^2 \conj{F_s}(\omega_k - \omega^*) F_l( \omega_k - \omega^*)       \label{eq:GEP_V_el_freq}\\
    ( \matr{H_V})_{sl} = \langle v_s,\matr{H} v_l \rangle  & =  i \int^T_{-T} \int^T_{-T} e^{-i \omega^*(\tau-t)} \conj{f_s}(\tau) f_l(t)  \partial_t C(\tau - t) dt d\tau \label{eq:GEP_H_el_time} \\
    & = \sum_{k=1}^N \omega_k |a_k|^2 \conj{F_s}(\omega_k - \omega^*) F_l( \omega_k - \omega^*)  \label{eq:GEP_H_el_freq}
\end{align}
Equations (\ref{eq:GEP_V_el_time}) and (\ref{eq:GEP_H_el_time}) are the time domain representations of the GEP matrix elements, that can be 
computed from the signal $C(t)$. Equations (\ref{eq:GEP_V_el_freq}) and (\ref{eq:GEP_H_el_freq}) are the frequency domain representations of the GEP matrix elements,
that are helpful for analytical derivations.\\

Compared to the infinite time access case, we no longer have the freedom, to select a frequency range with perfect 
precision. Since the filter functions are already time limited, they cannot be simultaneously band limited.
In particular other frequencies $\omega_k \notin (E_a,E_b)$ will contribute to the matrix elements of the GEP as seen in eq.(\ref{eq:GEP_V_el_freq}) and (\ref{eq:GEP_H_el_freq}).

\subsubsection{Exact Frequency Recovery despite Finite Time Access}

We can nevertheless still obtain exact frequency recovery, if the signal  $C(t)$ consist 
of a finite number of frequencies $N$. In many use cases, this is at least approximately true.
Let $\bw_C = \frac{1}{2}(\omega_{\max}-\omega_{\min})$ be the half-bandwidth of the signal $C(t)$. Then a $\bw_C T$-Filter system
can determine all frequencies of $C(t)$ by solving a GEP of dimensionality equal to the number of frequencies in the signal. \\

We will simply select the bandwidth $\bw$ of the filter system to be larger than the bandwidth $\bw_C$ of the signal $C(t)$ and 
choose a filter center $\omega^*$ that is approximately in the middle of the signal bandwidth. Then the GEP matrices are iteratively grown 
until $\matr{V}^2_{M+1} = \matr{\mathfrak{F}}_{M+1}( \vec{\omega}^*)^\dagger \matr{\Lambda}(\vec{a})^2 \matr{\mathfrak{F}}_{M+1}( \vec{\omega}^*)$ is singular.
This will happen at $M=N$, since the filter sequence is assumed to be a Chebyshev system. 
The generalised eigenvalues of the GEP $(\matr{H}_{V_N} , \matr{V}_N)$ will then perfectly coincide with the frequencies of the signal $C(t)$. \\

Thus, any arbitrary small but finite amount of access to the signal $C(t)$ is sufficient to
perfectly recover all the frequencies of the signal, by solving a GEP of dimensionality equal to the number of frequencies. In a way, this contradicts our intuition of what is physically possible. 
On the other hand, the identity theorem has that any analytic function, is uniquely defined by its values on an arbitrary small interval.
From this perspective, Filter Diagonalization can be seen as constructive instance of the identity theorem within Harmonic Analysis.
In contrast to an infinite Taylor expansion, only a finite dimensional GEP has to be solved.
It returns the frequencies and amplitudes that allow to extrapolate the signal $C(t)$ from a short time access into all of time.
Does Filter Diagonalization allow us to see infinitely far into the past and future? \\

Not exactly, at least not in such literal sense. We have made some assumptions in the above exposition, that will not perfectly hold a physical scenario. We assumed 
continuous access to the signal $C(t)$. In practice, it is not possible to continuously measure a signal, not even for a short time
and let alone with infinite precision. And additionally, for signal that consists of an infinite number of frequencies, the GEP will have infinite dimensionality.

This dampens the enthusiasm of Filter Diagonalization being able to perfectly extrapolate from a short observation time into all the past and future.
But nevertheless, input data $C(t)$ of finite precision will transfer favorably to a finite precision in the output 
of the method.\footnote{This will be elaborated in a future extension, that incorporates uncertainty quantification in the numerical evaluation of the GEP signals 
through sampling formulas. }\\

The true power of filter diagonalization lies in the favorable relationship between precision and computational cost of the method.
In practice, signals consist of a lot of frequencies, which would result in a high dimensional GEP that 
is respectively expansive to solve. However, in Filter Diagonalization we can allow for a small but controlled error 
and effectively adjust the computational cost of the method through a frequency window.

\subsection{Filter Diagonalization as $\varepsilon$-Approximate Solver}
\label{sec:FilterDiag_eps_approx_solver}
Through the quantum mechanical interpretation, 
Filter Diagonalization fits into the scheme of an $\varepsilon$-dimension reduction protocol of Chapter \ref{chap:DimRedSepAna}. 
The high-dimensional operator $\matr{H}$ is any fictitious Hamiltonian
that has a spectrum, that coincides with the one of the signal $C(t)$. An important precision parameter 
is the time access $4T$ to the signal. \\

We assume that a user has specified an interval $B = [\omega^* - \bw , \omega^* + \bw ]$ of interest, and wants to determine 
all frequencies of the signal $C(t)$ in this interval. Equations (\ref{eq:GEP_V_el_time}) and (\ref{eq:GEP_H_el_time}) are used 
to iteratively grow a GEP $(\matr{H}_{V_M}, \matr{V}_M)$. 
The resulting Gram matrix has a signal noise decomposition
\begin{align}
    \matr{V}^2_{sl} &= \underbrace{\sum_{\omega_k \in B} |a_k|^2 \conj{F_s}(\omega_k^*) F_l( \omega_k^*)}_{=\matr{B}^2_{sl}} + \underbrace{\sum_{\omega_k \notin B} |a_k|^2 \conj{F_s}(\omega_k^*) F_l( \omega_k^*)}_{=\matr{N}^2_{sl}}. \label{eq:FD_Gram_matrix_decomposition}
\end{align}
We can derive a bound on the operator norm of the noise matrix, that is independent on the spectral details of the signal. 

\begin{lemma}
    For the noise Gram matrix as in (\ref{eq:FD_Gram_matrix_decomposition}) we have the bound
    \begin{align*}
        \lambda_1(\matr{N}^2_M) & \leq \operatorname*{Tr}[\matr{N}^2_M] \leq  \sum_{\omega_k \notin B } |a_k|^2 \operatorname*{sup}_{\omega \notin [-\bw, \bw] } \sum_{l=1}^M |F_l(\omega)|^2 ,
    \end{align*}
    where
    \begin{align*}
        \sum_{\omega_k \notin I } |a_k|^2 \leq C(0).
    \end{align*}
\end{lemma}
In the simple proof we aim to convey some geometrical opportunities that arise from the interpretation of a hidden high-dimensional operator $\matr{H}$, that generates the signal $C(t)$.
\begin{proof}
    Let $\matr{P}_{B}$ be the projection operator onto eigenvectors of $\matr{H}$ with eigenvalues 
    $\omega \in B$ and let $\matr{P}_{N} = \matr{1} - \matr{P}_{B}$ be the projection operator on the orthogonal complement.
    From cyclical property of the trace we have,
    \begin{align*}
        \operatorname*{Tr}[\matr{N}^2_M] & = \operatorname*{Tr}[\matr{\mathfrak{F}}_M(\vec{\omega}^*)^\dagger \matr{\Lambda}(\vec{a}) \matr{P}_{N}^2 \matr{\Lambda}(\vec{a}) \matr{\mathfrak{F}}_M(\vec{\omega}^*)]\\
        & =  \operatorname*{Tr}[\matr{P}_{N} \matr{\Lambda}(\vec{a}) \matr{\mathfrak{F}}_M(\vec{\omega}^*)  \matr{\mathfrak{F}}_M(\vec{\omega}^*)^\dagger \matr{\Lambda}(\vec{a})^\dagger \matr{P}_{N}] \\
        & = \sum_{\omega_k \notin B } |a_k|^2 \sum_{l=1}^M |F_l(\omega_k^*)|^2 \\
        & \leq \sum_{\omega_k \notin B } |a_k|^2 \operatorname*{sup}_{\omega \notin [-\bw, \bw]  } \sum_{l=1}^M |F_l(\omega)|^2.
    \end{align*}
    The second identity follows from $C(0) = \sum_{k=1}^N |a_k|^2$.
\end{proof}
Thus, noise estimates of the generated guess vector matrix, can be directly obtained from the Filter Envelope of the Filter System.\\

In many signal processing scenarios of physical interest the signal $C(t)$ contains frequencies at various different scales.
Especially for such complicated cases, Filter Diagonalization provides a powerful tool. Rather that iterating various error estimates for different scales,
the integrated spectral inequalities of Theorem \ref{thm:integrated_spectral_stability} allow to swiftly capture all the noise analysis in a single bound.\\

To apply the integrated spectral inequalities, we require the spectral measure of Section {\ref{sec:spec_meas}}.
In the context of Filter Diagonalization, the spectral measure is related to the filter envelope of the filter system.
We denote with $|a(\omega)|^2$ the inverse Fourier transform of the signal $C(t)$,
\begin{align*}
    |a(\omega)|^2 & = \FT^{-1}[C](\omega) = \frac{1}{2\pi}\int_{-\infty}^{\infty} C(t) e^{i\omega t} dt \\
    &= \sum_{k=1}^N |a_k|^2 \delta(\omega - \omega_k).
\end{align*}
The spectral measure specific to a signal and a filter system is then given by

\begin{align*}
    |\alpha(\omega)|^2 = |a(\omega)|^2  \mathfrak{F}^{\text{env}}_M (\omega).
\end{align*}
By the Definition \ref{def:spec_meas} of the spectral integral and delta distributions, we have indeed 
\begin{align}
    \int_{I} |\alpha(\omega)|^2 d\omega = \sum_{\omega_k \in I_d} |a_k|^2 \mathfrak{F}^{\text{env}}_M (\omega_k ) + \int_{I_c} |a(\omega)|^2 \mathfrak{F}^{\text{env}}_M (\omega) d\omega 
\end{align}
for all $I \subset \mathbb{R}$.\footnote{Recall that the left-hand side is the spectral integral. This means for $I=\{\omega_k\}$ a single point at a discrete frequency is well-defined and not equal to zero.} 
Where $I_d$ is the set of discrete frequencies and $I_c$ is the continuous part of $C(t)$.
For a purely discrete spectrum, the integrated spectral inequalities of Theorem \ref{thm:integrated_spectral_stability}
write in the context of Filter Diagonalization as,
\begin{align*}
    \int_{\omega^* + \bw}^{\infty} \left(\omega - \omega_k \right) |\alpha(\omega)|^2 d\omega = \sum_{\omega_l > \omega^* + \bw} (\omega_l - \omega_k) |a_l|^2 \mathfrak{F}^{\text{env}}_M (\omega_l).
\end{align*}
The theoretical treatment can be adjusted to allow for continuous spectra in the signal $C(t)$
and will be presented in a future work.\\

\noindent We are ready to formulate Filter Diagonalization as an $\varepsilon$-dimension reduction protocol.

\begin{theorem}
    \label{thm:filter_diagonalization}
    Consider a signal $C(t)= \sum_{k=1}^N |a_k|^2 e^{i\omega_k t}$ with $|a_k|^2, \omega \in \mathbb{R}$ that is given for  $t\in [-2T, 2T]$ and
    a $\bw T$-filter system $( \{f_l\}, \{F_l\})$ as in Definition \ref{def:filter_system}.\\
    Let $\mathfrak{P}$ be the protocol, that generates for a given frequency band $B = [\omega^* - \bw, \omega^* + \bw]$ a GEP $(\matr{H}_{V_M}, \matr{V}_M)$ 
    with matrix elements, 
    \begin{align}
        \matr{V}_{sl}^2  & = \int^T_{-T} \int^T_{-T} e^{-i \omega^*(\tau-t)} \conj{f_s}(\tau) f_l(t) C(\tau - t) dt d\tau.    \label{eq:GEP_V_el_time2}\\
        ( \matr{H_V})_{sl}   & =  i \int^T_{-T} \int^T_{-T} e^{-i \omega^*(\tau-t)} \conj{f_s}(\tau) f_l(t)  \partial_t C(\tau - t) dt d\tau \label{eq:GEP_H_el_time2}.
    \end{align}
    Then $\mathfrak{P}$ defines an $\varepsilon$-dimension reduction protocol as in Definition \ref{asump:protocol}.
    The $\varepsilon$-estimates for the noise Gram matrix have 
    \begin{align*}
        \operatorname*{Tr}(\matr{N}^2_M) <  C(0) \sup_{\omega \not\in [-\bw, \bw]} \sum_{l=1}^M F_l (\omega)^2 =: \varepsilon_M.
    \end{align*}
\end{theorem}

The following algorithm can be used extract the frequencies and amplitudes of the signal $C(t)$ from the GEP.
For the sake of exposition, we assume that the correct number of frequencies $m$ in the frequency band $B$ of 
interest is known in Algorithm \ref{alg:2}. In practice, the GEP can be iteratively grown and a strong drop in the spectrum of the Gram matrix $\matr{V}^2_M$
can be used to detect the number of frequencies in the signal.\footnote{See Algorithm \ref{alg:1}, Section \ref{sec:dimension_detection} and Section \ref{sec:dimension_detection_Chebyshev}}
In Section \ref{sec:detectability_frequencies} we will discuss 
how the detectability of frequencies can be quantified through the Alterant matrix $\matr{\mathfrak{F}}_M(\vec{\omega}^*)$.

\begin{algorithm}[H]
    \caption{Filter Diagonalization}
    \label{alg:2}
    Consider a signal $C(t)$, a frequency band $ B =[\omega^* - \bw, \omega^* + \bw]$ and a $\bw T$-filter system $( \{f_l\}, \{F_l\})$ as in Theorem \ref{thm:filter_diagonalization}.
    Let $m$ denote the number of frequencies $\omega_k$ in the band $B$. 
    \begin{algorithmic}[1] 
        \State Pick $M \geq m$ \label{line:pickM}
        \Indent
            \State Compute $\matr{V}_M^2$ and $\matr{H}_{V_M}$ through (\ref{eq:GEP_V_el_time2}) and (\ref{eq:GEP_H_el_time2}) 
            \State Compute $\matr{U}$ s.t. $ \matr{U}^\dagger \matr{V}^2_M  \matr{U} = \operatorname{diag}(\lambda_1(\matr{V}_M^2), \cdots , \lambda_M(\matr{V}_M^2))$
            \State Let $\matr{U}_m =[u_1,\cdots, u_m]$  \Comment{Pick leading $m$ eigenvectors of $\matr{V}_M^2$} 
            \State Solve the GEP   $ \quad \matr{U}_m^\dagger\matr{H}_{V_M} \matr{U}_m \matr{\Phi} = \matr{U}_m^\dagger \matr{V}_M^2 \matr{U}_m \matr{\Phi} \matr{\Lambda}(\tilde \omega)$
            \State Set $\matr{\Lambda}( |\tilde a|^2) = \operatorname{diag}( \matr{\mathfrak{F}}_{M}(\tilde{\omega}^*)^{-\dagger}\matr{V}^2_M \matr{\mathfrak{F}}_{M}(\tilde{\omega}^*)^{-1} )$  \Comment{Approximated amplitudes} \label{line:amplitude_computation}
        \EndIndent
        \State Return $\tilde \omega$, $|\tilde a|^2$, $\lambda_{m}(\matr{V}_M^2)$, $M$
        \end{algorithmic}
    \end{algorithm}
In Algorithm \ref{alg:2} the vectors $\tilde \omega$ and $|\tilde a|^2$ contain the approximated frequencies and amplitudes of the signal $C(t)$
within the bandwidth of interest.
In line \ref{line:amplitude_computation} the right inverse of $\matr{\mathfrak{F}}_{M}(\tilde{\omega}^*)$ is denoted as $\matr{\mathfrak{F}}_{M}(\tilde{\omega}^*)^{-1}$.
Line \ref{line:amplitude_computation} alternatively reads as $|\tilde a_k |^{2} = ( \matr{\mathfrak{F}}_{M}(\tilde{\omega}^*)^{-\dagger}\matr{V}^2 \matr{\mathfrak{F}}_{M}(\tilde{\omega}^*)^{-1})_{kk}$.
It is important to honor the ordering of the determined frequencies as they are assigned to amplitudes. \\

We can access the approximation theory that was developed in Chapter \ref{chap:DimRedSepAna} to obtain precision guarantees for the approximated frequencies of the signal $C(t)$.

\begin{corollary}
    \label{cor:filter_diagonalization_precision}
    Let $ \omega_1 \geq \cdots \geq  \omega_m $  be the exact frequencies of the signal $C(t)$ in the interval $B = [\omega^* - \bw, \omega^* + \bw]$.
    The approximated frequencies $\tilde \omega_1 \geq \cdots \geq \tilde \omega_m$ obtained through Algorithm \ref{alg:2} have precision guarantees:
    \begin{itemize}
        \item If $C(t)$ only has frequencies in $[\omega_{\min}, \omega_{\max}]$ then         
        \begin{align}
            \frac{ \lambda_1(\matr{N}^2_M) }{\lambda_{m}(\matr{V}_M^2)}(\omega_{\min} - \omega_k) \leq \tilde{\omega}_k - \omega_k & \leq \frac{ \lambda_1(\matr{N}^2_M) }{\lambda_{m}(\matr{V}_M^2)}(\omega_{\max} - \omega_k)   \label{eq:FiltDiag_freq_precision_simp}
        \end{align}
        \item If $\lambda_{m}(\matr{V}_M^2) > (m+1)\lambda_1(\matr{N}^2_M)$ then         \begin{align}
            \frac{\sum\limits_{\omega_l < \omega^* - \bw} (\omega_l - \omega_k) |a_l|^2 \mathfrak{F}^{\text{env}}_M (\omega_l)}{\lambda_{m}(\matr{V}_M^2) - \lambda_1(\matr{N}^2_M)} \leq \tilde{\omega}_k - \omega_k & \leq \frac{\sum\limits_{\omega_l > \omega^* + \bw} (\omega_l - \omega_k) |a_l|^2 \mathfrak{F}^{\text{env}}_M (\omega_l)}{\lambda_{m}(\matr{V}_M^2)- \lambda_{1}(\matr{N}^2_M)}. \label{eq:FiltDiag_freq_precision_inte}
        \end{align}
    \end{itemize}
\end{corollary}

The method can now be iterated for different frequency bands of interest. It is important to highlight the following:\\

\emph{The computational compexity to approximate $N$ frequencies in a signal $C(t)$ through Filter Diagonalization scales linearly in $N$.}\\

To achieve this, we partition the frequency range $[\omega_{\min}, \omega_{\max}]$ into sub-intervals, each containing approximately $m$ frequencies.
Consequently, to approximate all $N$ frequencies of the signal $C(t)$, approximately $N/m$ generalized eigenvalue problems (GEPs) must be solved. 
The overall computational cost is then estimated as, 
\begin{align*}
    \sim  \frac{N}{m} m^3 = \mathcal{O}(N m^2).
\end{align*}
All additional steps, including the numerical evaluation of integrals, generation of filter functions, computation of the Alternant matrix, 
and more, can be incorporated into a scaling constant.\footnote{
    In most physical applications, the frequencies correspond to eigenvalues of high-dimensional Hamiltonians.    
Direct diagonalization
scales here with $N^3$ and quickly becomes intractable. Furthermore, there is limited flexibility in selectively targeting a range of eigenvalues of interest.}\\

After all frequencies and amplitudes have been determined, they can be used in (\ref{eq:FiltDiag_freq_precision_inte})
to compute precision guarantees. Rather than the exact values, the approximated frequencies $\tilde \omega_k$  and amplitudes $|\tilde a_k|^2$
are used in the error bounds. We write for the error in approximated values $\delta \omega =\omega - \tilde \omega$.
The error that results from using approximate values rather than exact values can be treated in perturbation sense.\footnote{
    For most applications, the error in the precision guarantee that is induced by using approximate values in the bounds can be readily neglected. 
    For high precision applications, one could derive rigorous bounds from the presented in equalities, that are independent from the exact values. Together Proposition \ref{prop:filter_diagonalization_amplitude_precision} 
    and (\ref{eq:FiltDiag_freq_precision_inte}) define a well-defined system of inequalities in $2N$ variables, $\delta \omega$ and $\delta a$.}\\


In Section \ref{sec:amplitude_uncertainty_quantification} we derive precision guarantees for the approximated amplitudes.
The error in the approximated amplitudes is guaranteed to be small if,
\begin{enumerate}
    \item The noise Gramm matrix $\lambda_1(\matr{N}_M^2)$ is small.
    \item The error in the frequency approximations $\delta \omega$ is small.
    \item The diagonal entries of $\matr{\mathfrak{F}}_{M}(\tilde{\omega}^*)^{-2}$ are well conditioned. 
\end{enumerate}
These criterias apply for the amplitudes individually. This means, that if a single 
diagonal element of  $\matr{\mathfrak{F}}_{M}(\tilde{\omega}^*)^{-2}$ is singular, 
only the precision of the corresponding amplitude is affected and not all. Criteria 
1 implies 2 through Corollary \ref{cor:filter_diagonalization_precision}. Criteria 3 can be explicitly 
verified from the filter matrix evaluated at the approximated frequencies. The meaning 
of $(\matr{\mathfrak{F}}_{M}(\tilde{\omega}^*)^{-2})_{kk}$ is closely tied to the detectability 
of the frequency $\omega_k$ in the signal $C(t)$ as we elaborate in Section \ref{sec:detectability_frequencies}.
The detailed precision inequalities are given in Proposition \ref{prop:filter_diagonalization_amplitude_precision}.

\subsection{Amplitude Uncertainty Quantification}
\label{sec:amplitude_uncertainty_quantification}
Here we derive error bounds for approximated amplitudes $|\tilde a_k|^2$ obtained through Algorithm \ref{alg:2}.
In the numerical literature of Filter Diagonalization, the amplitudes are extracted from the eigenvectors of the GEP,
similar to what we have seen in the toy example of infinite time access in Section \ref{sec:amplitude_extraction} \cite{Mandelshtam,wall_extraction_1995}.
However, bounding the error in the eigenvector of a perturbed eigenvalue problem typically seems to be more challenging than bounding the error in the eigenvalue itself \cite{stewart_perturbation_1978}.\\

In order to overcome this more intricate challenge we lean on the structure of the Alternant matrix $\matr{\mathfrak{F}}_{M}(\tilde{\omega}^*)$. 
As seen in Section \ref{sec:method_development_infinite_time} the guess vectors of FD are generated through the filter operator $\matr{\mathfrak{F}}_{M}$. While the exact input vector to 
$\matr{\mathfrak{F}}_{M}$ is not known, the operator itself is well known. And regularity properties of the filter system convey the intuition that
the guess vectors generated through $\matr{\mathfrak{F}}_{M}$ at the exact frequencies should not deviate much from the guess vectors generated at the approximated frequencies.
In the following proof, we find that this intuition can be equipped with mathematical precision through the diagonal elements of $\matr{\mathfrak{F}}_{M}(\tilde{\omega}^*)^{-2}$.\\

Let $\vec{\omega}_{\text{all}}^*$ be the vector of all frequencies $\omega_k$ shifted by the filter center $\omega^*$
and $\vec{a}_{\text{all}}$ the vector of all amplitudes $|a_k|^2$ in the signal $C(t)$. Let $ \vec{\omega}^*, \vec{a} \in \mathbb{R}^m$ be the vectors of the frequencies and amplitudes in the frequency band $B$.
Let $\matr{P}_N$ and $ \matr{P}_B$ be the projection operators onto the noise and signal subspace of the GEP. 
By inserting the identity $\matr{I} = \matr{P}_B + \matr{P}_N$ the Gram matrix $\matr{V}^2$ can be decomposed into signal and noise part,
\begin{align}
    \matr{V}^2_M & = \matr{\mathfrak{F}}_{M}(\vec{\omega}_{\text{all}}^*)^\dagger \left( \matr{P}_{B} + \matr{P}_{N }\right) \matr{\Lambda}(\vec{a}_{\text{all}} )^2 \matr{\mathfrak{F}}_{M}(\vec{\omega}_{\text{all}}^*) \notag \\
    & = \underbrace{\matr{\mathfrak{F}}_{M}(\vec{\omega}^*)^\dagger \matr{\Lambda}(\vec{a})^2 \matr{\mathfrak{F}}_{M}(\vec{\omega}^*)}_{=\matr{B}^2} + \matr{N}^2.  \label{eq:V2_decomp}
\end{align}
Note that $ \matr{\mathfrak{F}}_{M}(\vec{\omega}_{\text{all}}^*) $ maps into a $N$ dimensional vector space while  $\matr{\mathfrak{F}}_{M}(\vec{\omega}^*)$ maps into an $m$ dimensional vector space.
Let $\delta \omega \in \mathbb{R}^m $ be the vector that describes the deviations of the approximated frequencies
\begin{align*}
    \vec{\omega}^*  =  \tilde{\omega}^* + \delta \omega.
\end{align*}
Inserting into $\matr{B}^2$ gives 
\begin{align}
    \matr{B}^2 & = \matr{\mathfrak{F}}_{M}(\tilde{\omega}^* + \delta \omega )^\dagger \matr{\Lambda}(\vec{a})^2 \matr{\mathfrak{F}}_{M}(\tilde{\omega}^* + \delta \omega)  \label{eq:alternant_1 }\\
    & = \underbrace{\matr{\mathfrak{F}}_{M}(\tilde{\omega}^* )^\dagger \matr{\Lambda}(\vec{a})^2 \matr{\mathfrak{F}}_{M}(\tilde{\omega}^*)}_{:=  \tilde{\matr{B}}^2 } + \delta \matr{B}^2.  \label{eq:alternant_expansion }
\end{align}
The last equation follows from an expansion of the filter matrix $\matr{\mathfrak{F}}_{M}(\tilde{\omega}^* )$ 
through the mean value theorem: The elements of the Alternant matrix have
\begin{align}
    F_l(\tilde{\omega}_k^*+ \delta \omega_k)- F_l(\tilde{\omega}_k^*) = \int_{\tilde{\omega}_k^*}^{\tilde{\omega}_k^* +\delta \omega_k} F_l'(\omega) d\omega = \delta \omega_k F_l'(\tilde{\omega}_k^* +  \delta s_{lk}).
\end{align}
Here $\tilde{\omega}_k^* +  \delta s_{lk} \in [\tilde{\omega}_k^*, \tilde{\omega}_k^* +\delta \omega_k]$ and thus $|\delta s_{lk}| \leq |\delta \omega_k|$.
The Filter matrix $\matr{\mathfrak{F}}_{M}(\tilde{\omega}^* + \delta \omega )$
rewrites as
\begin{align*}
    \matr{\mathfrak{F}}_{M}(\tilde{\omega}^* + \delta \omega ) = \matr{\mathfrak{F}}_{M}(\tilde{\omega}^* ) + \matr{\Lambda}(\delta\omega) \matr{\mathfrak{F}}_{M}'(\tilde{\omega}^*, \delta s).
\end{align*}
Here $\matr{\mathfrak{F}}_{M}'(\tilde{\omega}^*, \delta s)$ is the matrix with entries $[F_l'(\tilde{\omega}_k^* +  \delta s_{lk})]_{kl}$.
Inserting this expansion of the Alternant matrix into (\ref{eq:alternant_1 }) proves (\ref{eq:alternant_expansion }) and has 
\begin{align*}
    \delta \matr{B}^2 =  \matr{\mathfrak{F}}_{M}'(\tilde{\omega}^*, \delta s)^\dagger \matr{\Lambda}(\delta\omega)  \matr{\Lambda}(\vec{a})^2 \matr{\mathfrak{F}}_{M}(\tilde{\omega}^*) + \matr{\mathfrak{F}}_{M}(\tilde{\omega}^*)^\dagger \matr{\Lambda}(\vec{a})^2 \matr{\Lambda}(\delta\omega) \matr{\mathfrak{F}}_{M}'(\tilde{\omega}^*, \delta s) + \\
    \matr{\mathfrak{F}}_{M}'(\tilde{\omega}^*, \delta s)^\dagger \matr{\Lambda}(\delta\omega)  \matr{\Lambda}(\vec{a})^2 \matr{\Lambda}(\delta\omega) \matr{\mathfrak{F}}_{M}'(\tilde{\omega}^*, \delta s).
\end{align*}
Since we assume complete Chebyshev filter system, $\matr{\mathfrak{F}}_{M}(\tilde{\omega}^*)$ has full row rank and in particular a
right inverse, denoted as $\matr{\mathfrak{F}}_{M}(\tilde{\omega}^*)^{-1}$.
Then $\matr{\mathfrak{F}}_{M}(\tilde{\omega}^*)^{-\dagger}$ is the left inverse of $\matr{\mathfrak{F}}_{M}(\tilde{\omega}^*)^{\dagger}$.
Equation (\ref{eq:V2_decomp}) rewrites as 
\begin{align*}
    \matr{\Lambda}(\vec{a})^2 - &\matr{\mathfrak{F}}_{M}(\tilde{\omega}^*)^{-\dagger} \matr{V}^2_M \matr{\mathfrak{F}}_{M}(\tilde{\omega}^*)^{-1} = \matr{\mathfrak{F}}_{M}(\tilde{\omega}^*)^{-\dagger} \matr{N}^2 \matr{\mathfrak{F}}_{M}(\tilde{\omega}^*)^{-1} + \\
   &  \matr{\mathfrak{F}}_{M}(\tilde{\omega}^*)^{-\dagger}  \matr{\mathfrak{F}}_{M}'(\tilde{\omega}^*, \delta s)^\dagger \matr{\Lambda}(\vec{a})^2 \matr{\Lambda}(\delta\omega)  + \matr{\Lambda}(\vec{a})^2 \matr{\Lambda}(\delta\omega) \matr{\mathfrak{F}}_{M}'(\tilde{\omega}^*, \delta s) \matr{\mathfrak{F}}_{M}(\tilde{\omega}^*)^{-1}  + \\
    &\matr{\mathfrak{F}}_{M}(\tilde{\omega}^*)^{-\dagger}  \matr{\mathfrak{F}}_{M}'(\tilde{\omega}^*, \delta s)^\dagger  \matr{\Lambda}(\vec{a})^2 \matr{\Lambda}(\delta\omega)^2   \matr{\mathfrak{F}}_{M}'(\tilde{\omega}^*, \delta s) \matr{\mathfrak{F}}_{M}(\tilde{\omega}^*)^{-1}.
\end{align*}
We used that diagonal matrices commute.
We left and right multiply the equation with $e_k^\dagger$ and $e_k$. Taking the absolute value of the resulting equation and applying the 
triangle inequality on the right-hand side gives
\begin{align*}
    \big| |a_k|^2 - & \left(\matr{\mathfrak{F}}_{M}(\tilde{\omega}^*)^{-\dagger} \matr{V}^2_M \matr{\mathfrak{F}}_{M}(\tilde{\omega}^*)^{-1} \right)_{kk} \big| \leq \lambda_1(\matr{N}^2_M) \left(\matr{\mathfrak{F}}_{M}(\tilde{\omega}^*)^{-\dagger} \matr{\mathfrak{F}}_{M}(\tilde{\omega}^*)^{-1} \right)_{kk} \\
    & + \delta \omega_k 2 |a_k|^2  \left| \left(\matr{\mathfrak{F}}_{M}'(\tilde{\omega}^*, \delta s) \matr{\mathfrak{F}}_{M}(\tilde{\omega}^*)^{-1}\right)_{kk}\right| + \lambda_1(\matr{\Lambda}(\vec{a})^2 \matr{\Lambda}(\delta\omega)^2) \left(  \left(\matr{\mathfrak{F}}_{M}'(\tilde{\omega}^*, \delta s) \matr{\mathfrak{F}}_{M}(\tilde{\omega}^*)^{-1} \right)^2\right)_{kk}.
\end{align*}
In the first and third term 
\begin{align*}
    \left| x^\dagger \matr{A}^\dagger \matr{B}^2 \matr{A} x \right| & \leq \lambda_1(\matr{B}^2) \left| x^\dagger \matr{A}^\dagger \matr{A} x \right|,
\end{align*}
was used. 
The third perturbation parameter $\lambda_1(\matr{\Lambda}(\vec{a})^2 \matr{\Lambda}(\delta\omega)^2) $ returns the largest error 
in frequency approximation weighted by the amplitude it corresponds to. \\

\noindent We derived that 
\begin{align}
    |\tilde{a}_k|^{2} = \left(\matr{\mathfrak{F}}_{M}(\tilde{\omega}^*)^{-\dagger} \matr{V}^2_M \matr{\mathfrak{F}}_{M}(\tilde{\omega}^*)^{-1} \right)_{kk}  \label{eq:amplitude_approximation2}
\end{align}
forms indeed good approximation for the true amplitude $|a_k|^2$. For 
algorithmic implementations the ordering of the obtained frequencies $\tilde \omega =(\tilde \omega_1, \cdots, \tilde \omega_m )^\top $ must be honored when matched to 
the amplitudes $|\tilde a|^2 = (|\tilde a_1|^2, \cdots, |\tilde a_m|^2)^\top$.\\
Equation (\ref{eq:amplitude_approximation2}) can be rewritten in matrix form, 
\begin{align*}
    \matr{\Lambda}(|\tilde{a}_k|^2) = \operatorname*{diag}\left(\matr{\mathfrak{F}_{M}}(\tilde{\omega}^*)^{-\dagger} \matr{V}^2 \matr{\mathfrak{F}_{M}}(\tilde{\omega}^*)^{-1} \right) .
\end{align*}
The following statement has been established.

\begin{proposition}
    \label{prop:filter_diagonalization_amplitude_precision}
    Consider approximated amplitudes obtained through Algorithm \ref{alg:2}
    \begin{align*} 
        |\tilde a_k |^{2} = ( \matr{\mathfrak{F}}_{M}(\tilde{\omega}^*)^{-\dagger}\matr{V}^2_M \matr{\mathfrak{F}}_{M}(\tilde{\omega}^*)^{-1})_{kk},
    \end{align*}
    and let $ \delta \omega =  \vec{\omega} - \tilde{\omega}$ be the vector of errors in the approximated frequencies.
    Then the following
    precision guarantees hold
    \begin{align}
        \left||a_k|^2 - |\tilde a_k|^2 \right| \leq  & \lambda_1(\matr{N}_M^2) \left( \matr{\mathfrak{F}}_{M}(\tilde{\omega}^*)^{-2}\right)_{kk}
        + 2\delta\omega_k |a_k|^2  \left| \left(\matr{\mathfrak{F}}_{M}'(\tilde{\omega}^*, \delta s) \matr{\mathfrak{F}}_{M}(\tilde{\omega}^*)^{-1}\right)_{kk}\right|   \notag \\
         & + \sup_l \left(|a_l|^2 \delta\omega_l^2 \right)  \left(\left(\matr{\mathfrak{F}}_{M}'(\tilde{\omega}^*, \delta s) \matr{\mathfrak{F}}_{M}(\tilde{\omega}^*)^{-1}\right)^2 \right)_{kk} .   \label{eq:amplitude_precision}
    \end{align}
    Here $\delta s \in \mathbb{R}^{m \times M}$ is a matrix with entries $|(\delta s)_{kl} - \delta \omega_k| \leq |\delta \omega_k|$,
    $\matr{\mathfrak{F}}_{M}'(\tilde{\omega}^*, \delta s)$ is a matrix with entries $[F_l'(\tilde{\omega}_k^* +  \delta s_{kl})]_{kl}$ and $ \matr{\mathfrak{F}}_{M}(\tilde{\omega}^*)^{-1}$ is 
    the right inverse.
\end{proposition}

Depending on the precision requirements, the matrix $\matr{\mathfrak{F}}_{M}'(\tilde{\omega}^*, \delta s)$ can be explicitly approximated 
by setting $\delta s =0$. This approximation is justified, if the error in the frequency approximations $\delta \omega$ can be guaranteed to be small
through the precision bounds of Corollary \ref{cor:filter_diagonalization_precision}.\\


One could also further estimate
\begin{align*}
    \left(\left(\matr{\mathfrak{F}}_{M}'(\tilde{\omega}^*, \delta s) \matr{\mathfrak{F}}_{M}(\tilde{\omega}^*)^{-1}\right)^2 \right)_{kk} \leq \lambda_1(\matr{\mathfrak{F}}_{M}'(\tilde{\omega}^*, \delta s)^2)  \left( \matr{\mathfrak{F}}_{M}(\tilde{\omega}^*)^{-2}\right)_{kk}.
\end{align*}
The importance of $\left( \matr{\mathfrak{F}}_{M}(\tilde{\omega}^*)^{-2}\right)_{kk}$ in the precision guarantees of the amplitudes
can also be motivated from the \emph{detectability} of a frequency in a signal.

\subsection{Detectability of Frequencies}
\label{sec:detectability_frequencies}

Consider the example of exact frequencies $\vec{\omega}^*$ feed into the filter system $\matr{\mathfrak{F}}_{M}(\vec{\omega}^*)$.
How well can an eigenvector $|\varphi_k\rangle$ of the hidden Hamiltonian be represented in the filter system?
We denote with $E_{\text{info}}(k) = \|\matr{\mathfrak{F}}_{M}(\vec{\omega}^*) x_k  \|^2$ the 
signal energy that a respesenation $x_k$ of the hidden eigenvector $|\varphi_k \rangle$ causes through the filter system.
Analogously we write $E_{\text{input}}(k) = \|x_k\|^2$ and ask 
for the relation between $E_{\text{info}}(k)$ and $E_{\text{input}}(k)$.\\

\noindent Let $ |\varphi_k\rangle$ be normalized and $x_k$ be scaled such that,
\begin{align}
 |\varphi_k\rangle = \matr{\mathfrak{F}}_{M}(\vec{\omega}^*) x_k.  \label{eq:signal_enegy}
\end{align}
Left multiplying with $\matr{\mathfrak{F}}_{M}(\vec{\omega}^*)^{-1}$ gives
\begin{align*}
 \matr{\mathfrak{F}}_{M}(\vec{\omega}^*)^{-1} |\varphi_k\rangle =  \matr{\mathfrak{F}}_{M}(\vec{\omega}^*)^{-1}  \matr{\mathfrak{F}}_{M}(\vec{\omega}^*)  x_k = x_k.
\end{align*}
Here we assumed wlog $x_k \in \operatorname*{Ker}(\matr{\mathfrak{F}}_{M}(\vec{\omega}^*))^\perp$.
Therefore,
\begin{align*}
    \|x_k\|^2 =  \langle \varphi_k | \matr{\mathfrak{F}}_{M}(\vec{\omega}^*)^{-\dagger}  \matr{\mathfrak{F}}_{M}(\vec{\omega}^*)^{-1} |\varphi_k\rangle\\
 = \left( \matr{\mathfrak{F}}_{M}(\vec{\omega}^*)^{-\dagger} \matr{\mathfrak{F}}_{M}(\vec{\omega}^*)^{-1} \right)_{kk}.
\end{align*}
We have thus derived the scaling factor between the input and output energy of the filter system,
\begin{align*}
 E_{\text{info}}(k) = \frac{E_{\text{input}}(k)}{\left(\matr{\mathfrak{F}}_{M}(\vec{\omega}^*)^{-2} \right)_{kk} }.
\end{align*}
Thus, $\left(\matr{\mathfrak{F}}_{M}(\vec{\omega}^*)^{-2} \right)_{kk}$ captures the signal amplification of the frequency $\omega_k$ 
through the filter system. 
By taking the square of (\ref{eq:signal_enegy}), we see that a low
value for $\left(\matr{\mathfrak{F}}_{M}(\vec{\omega}^*)^{-2} \right)_{kk}$ is equivalent to a high Rayleigh quotient of the eigenvector representation $x_k$
with respect to $\matr{\mathfrak{F}}_{M}(\vec{\omega}^*)^2$. \\

We call $\left(\matr{\mathfrak{F}}_{M}(\vec{\omega}^*)^{-2} \right)_{kk}^{-1} $ the \emph{detectability} of the frequency $\omega_k$ through the filter system.
Note that the detectability of a frequency is a collective property, as the value of $\left(\matr{\mathfrak{F}}_{M}(\vec{\omega}^*)^{-2} \right)_{kk}^{-1} $ also depends on the other frequencies in the signal.
We formally still need to justify that a high value for $\left(\matr{\mathfrak{F}}_{M}(\vec{\omega}^*)^{-2} \right)_{kk}^{-1} $ indeed results in a good detectability.\\

For a frequency $\omega_k$ to be detectable though Filter Diagonalization and Algorithm \ref{alg:1}, we need its signal energy to be larger than the noise energy $\varepsilon_M$. 
Weyl's inequality gives
\begin{align*}
    \lambda_{m}(\matr{V}_M^2) \geq \lambda_{m}(\matr{B}_M^2) + \lambda_M(\matr{N}_M^2) \geq  \lambda_{m}(\matr{B}_M^2).
\end{align*}
Thus, a sufficient condition to ensure that Algorithm \ref{alg:1} detects the true dimensionality is 
that signal is stronger than the noise, $\lambda_{m}(\matr{B}_M^2) \geq \varepsilon_M$. \\

\noindent The representations $x_k$
of the hidden eigenvectors are orthogonal in the inner product induced by the Alternant matrix:
Wlog we normalize $x_k$ such that, 
\begin{align*}
 \matr{B}_M x_k =  a_k |\varphi_k\rangle,  \implies \matr{\Lambda}(\vec{a})_M\matr{\mathfrak{F}}_{M}(\vec{\omega}^*)   x_k = a_k e_k.
\end{align*}
In particular, 
\begin{align*}
 \matr{\mathfrak{F}}_{M}(\vec{\omega}^*)   x_k = e_k.
\end{align*}
The representations of the hidden eigenvectors span the vector space since we assume, that the filter system is a complete Chebyshev system.
As just derived, the vectors $x_k$ satisfy the orthogonality relation 
\begin{align}
    \langle x_l,  \matr{B}^2 x_k \rangle = \langle x_l ,\matr{\mathfrak{F}}_{M}(\vec{\omega}^*)^2 x_k \rangle =0 
\end{align}
for $k \neq l$.
In particular, through the variational principle and an expansion in the $\{ x_k\}$ basis we find,
\begin{align*}
    \lambda_{m}(\matr{B}_M^2) & = \min_{k } \frac{|a_k|^2 \langle x_k  ,\matr{\mathfrak{F}}_{M}(\vec{\omega}^*)^2 x_k \rangle}{ \langle x_k, x_k \rangle} \\
        & = \min_{k } \frac{|a_k|^2}{\left(\matr{\mathfrak{F}}_{M}(\vec{\omega}^*)^{-2} \right)_{kk}}.
\end{align*}
Thus, if $\left(\matr{\mathfrak{F}}_{M}(\vec{\omega}^*)^{-2} \right)_{kk}^{-1}$ is large, a small amplitude $|a_k|^2$ is sufficient to ensure that the 
frequency can be detected despite the presence of noise.\\

\newpage

\section{Prolate Filter Diagonalization}
\label{sec:ProlateFilterDiag}

In Section \ref{sec:FilterDiag_eps_approx_solver} we have developed an approximation theory for Filter Diagonalization,
that comes along with precision guarantees for the approximated frequencies and amplitudes, and that can be applied for various different 
choices of filter systems. The natural remaining question is to ask for the optimal filter system.

\subsection{Prolates are the Optimal Filter System}

One way to formalize the search for optimal filter functions is to seek the $T$-time limited function that is most concentrated in the $\bw$-band.
This is the familiar optimization problem we solved in Section \ref{sec:Fourier_concentration_time_limited},
and it represents just one of the many paths that lead to the discovery of Prolate Fourier Theory.\\

As eigenfunctions of an integral operator with a totally positive kernel, prolates inherit strong independence properties.\footnote{According to Bochner's theorem, Fourier transforms of probability measures are totally positive (Theorem 23 in \cite{Bochner1932}),
Specifically, $\rho_\bw = \FT^{-1}[\chi_\bw]$ is totally positive.} Eigenfunctions of integral operators with totally positive kernels form complete Chebyshev systems \cite{Karlin1966}.
Consequently prolates defined by $\BL_T \TL_\bw \tilde \prlt_n = \gamma_n \tilde \prlt_n$ from a complete Chebyshev system in $[-\bw, \bw]$.\footnote{
    Alternatively the Chebyshev property can also be 
    seen from their characterization as eigenfunctions to a Sturm-Liouville operator.}\\

Thus, the optimal filter systems are provided by the prolate spheroidal wave functions.
In order to extract frequencies from a signal given in $[-2T, 2T]$ with a band of length $2\bw$ 
we choose a $\bw T$-prolate sequence $\TL_T  \BL_\bw \TL_T \prlt_n = \gamma_n \TL_T \prlt_n$ as filter functions,
\begin{align*}
    \{f_l\}_{l=0}^\infty = \{ \TL_T \prlt_l\}_{l=0}^\infty , \quad \{F_l\}_{l=0}^\infty  = \{ \conj{\mu}_l \tprlt_l\}_{l=0}^\infty.
\end{align*}
Here  $\FT[ \TL_T \prlt_n]= \conj{\mu}_n \tprlt_n$ from Fact \ref{fact:prolate_FT} was used. 
This choice of filer functions has,
\begin{align*}
    \| f_n \|^2 = \|\TL_T \prlt_n\|^2 = \gamma_n.
\end{align*}
Therefore, the contribution of the prolate filter function to the Gram matrix $\matr{V}^2$, is weighted by its concentration in the $\bw$-band.
We will from now on refer to Filter Diagonalization through a filter system given by prolates as \emph{Prolate Filter Diagonalization}.
All that is left, is to access the precision guarantees of Section \ref{sec:FilterDiag_eps_approx_solver}
and to analyze the Filter Envelope $\mathfrak{F}^{\text{env}}_M(\omega)$ for the specific case of a prolate filter system.

\subsection{The Prolate Filter Envelope}
\label{sec:prolate_filter_envelope}

For the prolate filter envelope we have 

\begin{align*}
    \mathfrak{F}^{\text{env}}_M(\omega) = \sum_{l=0}^{M-1} |\conj{\mu}_l \tprlt_l(\omega)|^2 = 2 \pi \sum_{l=0}^{M-1} \gamma_n|\tprlt_l(\omega)|^2
\end{align*}

In Chapter \ref{chap:prlt_bound} we have derived a supremum bound for prolates outside of their concentration interval. 
Theorem \ref{thm:prlt_bound} gives for the prolate filter envelope
\begin{align*}
    \sup_{\omega \not \in  [-\bw, \bw]} \mathfrak{F}^{\text{env}}_M(\omega) & =  2 \pi  \sup_{\omega \not \in  [-\bw, \bw]} \sum_{l=0}^{M-1}  \gamma_l |\tprlt_l(\omega)|^2\\
    & \leq    2 \pi  \sum_{l=0}^{M-1}  \gamma_l(1-\gamma_l) C_{\text{extra},l}. 
\end{align*}
By (\ref{eq:lambda_neg_C_estimate}) we can estimate $C_{\text{extra},l}$ for $l < 2\bw T/ \pi$ as,
\begin{align*}
    C_{\text{extra},l} & \leq \sqrt{\|\prlt'\|_\infty^2 + \frac{\lambda_n^2}{4T^2}} - \frac{\lambda_n}{2T}\\
                    & \leq \sqrt{\bw^2  + \frac{\lambda_n^2}{4T^2}} - \frac{\lambda_n}{2T}.
\end{align*}
The numerical evaluation of prolates and the eigenvalues $\gamma_n$ was long subject to instabilities due 
to catastrophic cancellations. However, the problem has been identified and resolved in \cite{Buren2002AccurateCO} and recently also 
came along with published code. Similarly, the initial method to compute the eigenvalues $\lambda_n$ through Bouwkamps method 
was cumbersome and has been long replaced by more efficient methods \cite{hodge_eigenvalues_1970}.\footnote{However, it seems 
that not everyone was aware of \cite{hodge_eigenvalues_1970} and similar methods have been rediscovered in the literature, e.g. \cite{MiyazakiNUMERICALCO}.}\\

\noindent Overall we have for the $\varepsilon_M$-estimate of prolate filter diagonalization,
\begin{align}
    \operatorname*{Tr}[\matr{N}^2_M] & < C(0) \sup_{\omega \not \in  [-\bw, \bw]} \mathfrak{F}^{\text{env}}_M(\omega)  \leq 2\pi C(0) \sum_{l=0}^{M-1}  \gamma_l(1-\gamma_l) C_{\text{extra},l}  \notag \\
    & <  2\pi C(0) \sum_{l=0}^{M-1}  \gamma_l(1-\gamma_l) \left( \sqrt{\bw^2  + \frac{\lambda_n^2}{4T^2}} - \frac{\lambda_n}{2T}\right).  \label{eq:eps_computer}
\end{align}
The $\varepsilon_M$-estimate in line (\ref{eq:eps_computer}) can be efficiently computed by subroutines. However, we also include here 
an asymptotic estimate for the prolate envelope,
\begin{align}
    \sup_{\omega \not \in  [-\bw, \bw]} \mathfrak{F}^{\text{env}}_M(\omega) <  M \bw \frac{ \pi^{\frac{3}{2}} 2^{3 M} c^{M+\frac{1}{2}} e^{-2 c}}{(M-1)!}\left[1-\frac{6 M^2+ 50M  -21 }{32 c}+O\left(c^{-2}\right)\right].   \label{eq:asymptotic_prlt_env}
\end{align}
Equation (\ref{eq:asymptotic_prlt_env}) follows from the asymptotic expressions given in Section \ref{sec:asymptotic_bounds}. 
The derivation is given in Section \ref{sec:asymptotic_prlt_env}.
Inserting the bound on the prolate envelope (\ref{eq:asymptotic_prlt_env}) into the estimates of Corollary \ref{cor:filter_diagonalization_precision}
gives the precision guarantees for Prolate Filter Diagonalization,
\begin{align}
    \tilde \varepsilon_M  \frac{\sum\limits_{\omega_l < \omega^* - \bw } (\omega_l - \omega_k)|a_l|^2 }{\lambda_m(\matr{V}^2_M) - \varepsilon_M}   \leq \tilde \omega_k -\omega_k \leq \tilde \varepsilon_M \frac{\sum\limits_{\omega_l > \omega^* + \bw } (\omega_l - \omega_k)|a_l|^2 }{\lambda_m(\matr{V}^2_M) - \varepsilon_M},
\end{align}
with $\tilde \varepsilon_M = 2\pi \sum_{l=0}^{M-1}  \gamma_l(1-\gamma_l) C_{\text{extra},l}$. In particular,  $\tilde \varepsilon_M$ is very close to $0$ for $M \ll \bw T$.\\

Theorem \ref{thm:PFD_precision} is now finally established and is based on our main results from Chapter \ref{chap:prlt_bound} and \ref{chap:DimRedSepAna}. 
We are not aware of any other method for identifying the frequencies of a signal, given only a finite observation time, that offers a similarly strong precision guarantee.

\subsection{Conclusion and Outlook}
In this chapter, we have formulated Filter Diagonalization as an approximation theory, that can access 
precision guarantees for the approximated frequencies and amplitudes. 
We believe that Prolate Filter Diagonalization is currently among the most accurate methods available for identifying frequencies of a signal given a finite observation time.
The precision guarantee given in Theorem \ref{thm:PFD_precision} for PFD is characterized by fundamental approximation parameters $1-\gamma_n$ 
that describe optimal concentration properties of the filter functions.\\

For a nearly perfect precision, 
in the approximated frequencies, the essential requirement is that the number of frequencies $m$ in each scanned frequency band $B_{\omega^*, \bw}$ has 
$m \ll 2 \bw T /\pi $. The computational cost to then determine $N$ frequencies through PFD only scales as $O(N m^2)$.
The ability of the method to face intricate scenarios, that involve frequencies that only have small amplitudes or are particularly close to each other,
can be analyzed through the frequency detectability of a filter system, which was introduced in Section \ref{sec:detectability_frequencies}.\\

In future work, we will include an uncertainty quantification of the method that applies the prolate sampling formulas 
to evaluate the GEP integrals. Further generalizations of the method, to analyze signals with continuous spectra, are of 
particular interest and will be considered. \\

Finally, we would like to emphasize the theoretical value of Filter Diagonalization. The method is based on a symmetry between harmonic analysis 
and quantum mechanics. While these fields already have a long history of mutual influence, Filter Diagonalization gave a new perspective that is 
of particular computational power. The quantum mechanical roots of the method turn out to be a mere mathematical auxiliary, which in itself is at least of 
pedagogical interest for the foundations of quantum mechanics.\\

\newpage 

\section{Appendix}
\subsection*{Asymptotic Estimate for the Prolate Filter Envelope}
\label{sec:asymptotic_prlt_env}

\begin{proof}[Proof of Equation \ref{eq:asymptotic_prlt_env}]
    In Section \ref{sec:asymptotic_bounds} we have derived
    \begin{align}
        (1-\gamma_n) C_{\text{extra},n} = \bw \frac{4 \sqrt{\pi} 2^{3 n} c^{n+\frac{3}{2}} e^{-2 c}}{n!}\left[1-\frac{6 n^2+ 62n +35 }{32 c}+O\left(c^{-2}\right)\right], \label{eq:sup_extra_asympt2}
    \end{align}
    Using Theorem \ref{thm:prlt_bound}, we can establish an asymptotic bound for the prolate filter envelope,
    \begin{align*}
        \sup_{\omega \not \in  [-\bw, \bw]} \mathfrak{F}^{\text{env}}_M(\omega) & \leq 2\pi  \sum_{l=0}^{M-1}  \gamma_l (1-\gamma_n) C_{\text{extra},n}  <  2\pi \sum_{l=0}^{M-1}  (1-\gamma_n) C_{\text{extra},n} \\
        & < 2\pi \bw  M  (1-\gamma_{M-1}) C_{\text{extra},M-1} \\
        & = 2\pi M \bw \frac{4 \sqrt{\pi} 2^{3 (M-1)} c^{M+\frac{1}{2}} e^{-2 c}}{(M-1)!}\left[1-\frac{6 (M-1)^2+ 62(M-1) +35 }{32 c}+O\left(c^{-2}\right)\right] 
    \end{align*}
    In the second line, we used the fact that the sequence $(1-\gamma_{l}) C_{\text{extra},l}$ is monotonically increasing in $l$ for $l \leq 2c/ \pi$ in leading asymptotic order, as can be
    seen from (\ref{eq:sup_extra_asympt2}).
\end{proof}

\newpage 

\section*{Postface}
\addcontentsline{toc}{section}{Postface}

In the final chapter of this thesis we have presented a novel protocol to determine the frequencies of a signal $C(t)= \sum_k |a_k|^2 e^{i \omega_k t}$, that is only given for a finite amount of time
and that comes along with a fundamental precision guarantee. 
The protocol does not require expensive hardware, is computationally efficient, and is straightforward to implement.\\

Thus Prolate Filter Diagonalization can be readily applied to address many real-world problems, such as evaluation of experimental data, signal processing, 
numerical and theoretical analysis and many more. \\

Addressing this applied scenario has been the central focus of this thesis. After reflecting on the journey 
that would be required to understand optimal approximation in this context, we do feel greatly excited about the 
 theoretical work of addressing applied problems. It has become evident that a cross-disciplinary 
 approach, rather than isolation within a single field, is essential to succeed in real-world challenges.\\

Approximation theory represents a thrilling research field, and we are looking forward to the further development of the techniques presented 
here and their applications.

\printbibliography


\end{document}